\documentclass[11pt]{article}
\usepackage{latexsym}
\usepackage{amssymb}
\usepackage{amsthm}
\usepackage{amsmath}
\usepackage{mathrsfs}
\usepackage{graphicx}
\usepackage[latin1]{inputenc}
\newtheorem{definition}{Definition}[section]
\newtheorem{theorem}{Theorem}[section]
\newtheorem{prop}[theorem]{Proposition}
\newtheorem{coro}[theorem]{Corollary}
\newtheorem{lemma}[theorem]{Lemma}
\newtheorem{remark}[theorem]{Remark}

\newcommand{\R}{\mathbb{R}}             
\newcommand{\N}{\mathbb{N}}             
\newcommand{\Z}{\mathbb{Z}}             %
\newcommand{\C}{\mathbb{C}}             
\renewcommand{\H}{\mathcal{H}}          
\newcommand{\D}{\mathcal{D}}            
\newcommand{\B}{\mathcal{B}}            
\newcommand{\M}{\mathcal{M}}            
\newcommand{\G}{\mathcal{G}}            
\renewcommand{\S}{\mathbb{S}}           
\renewcommand{\L}{\mathcal{L}}

\newcommand{\h}{\mathfrak{h}}

\newcommand{\ls}{\mathfrak{l}}

\newcommand{\e}{\epsilon}
\newcommand{\half}{\frac{1}{2}}

\newcommand{\Ga}{\Gamma^1}         
\newcommand{\Gb}{\Gamma^2}         
\newcommand{\Gc}{\Gamma^3}         

\renewcommand{\l}{\langle}
\renewcommand{\r}{\rangle}

\newcommand{\F}{\mathcal{F}}           
\newcommand{\DS}{\mathbb{D}_{\S^2}}    

\newcommand{\go}{\Gamma_0(\lambda)}    

\newcommand{\muk}{\mu_{kl}(\lambda)}

\newcommand{\fj}{f_j (X, \lambda,k, z)}
\newcommand{\fun}{f_1 (X, \lambda,k, z)}

\newcommand{\funp}{f_1^+ (X, \lambda,k, z)}
\newcommand{\gj}{g_j (X, \lambda,k, z)}

\newcommand{\alun}{a_{L1}(\lambda,k, z)}
\newcommand{\ald}{a_{L2}(\lambda,k, z)}
\newcommand{\alt}{a_{L3}(\lambda,k, z)}
\newcommand{\alq}{a_{L4}(\lambda,k, z)}

\newcommand{\ds}{\displaystyle}

\newcommand{\Section}[1]{\section{#1} \setcounter{equation}{0}}

\author{Thierry Daudé \footnote{Laboratoire AGM, Département de Mathématiques, UMR CNRS 8088, Université de Cergy-Pontoise, 2 Avenue Adolphe Chauvin, 95302 Cergy-Pontoise, FRANCE. Email Adress: thierry.daude@u-cergy.fr} \, and François Nicoleau
    \footnote{Laboratoire Jean Leray, UMR 6629, Université de Nantes, 2, rue de la Houssinière, BP
    92208, 44322 Nantes Cedex 03, FRANCE. Email adress: nicoleau@math.univ-nantes.fr \hfill \break
    Research  supported in part by the French National Research Project NOSEVOL, No. ANR- 2011 BS0101901 and by the French National Research Project AARG, No. ANR-12-BS01-012-01}}

\title{Direct and inverse scattering at fixed energy for massless charged Dirac fields by Kerr-Newman-de Sitter black holes}
\date{\today}

\begin{document}

\maketitle


\begin{abstract}
In this paper, we study the direct and inverse scattering theory at fixed energy for massless charged Dirac fields evolving in the exterior region of a Kerr-Newman-de Sitter black hole. In the first part, we establish the existence and asymptotic completeness of time-dependent wave operators associated to our Dirac fields. This leads to the definition of the time-dependent scattering operator that encodes the far-field behavior (with respect to a stationary observer) in the asymptotic regions of the black hole: the event and cosmological horizons. We also use the miraculous property (quoting Chandrasekhar) - that the Dirac equation can be separated into radial and angular ordinary differential equations - to make the link between the time-dependent scattering operator and its stationary counterpart. This leads to a nice expression of the scattering matrix at fixed energy in terms of stationary solutions of the system of separated equations. In a second part, we use this expression of the scattering matrix to study the uniqueness property in the associated inverse scattering problem at fixed energy. Using essentially the particular form of the angular equation (that can be solved explicitely by Frobenius method) and the Complex Angular Momentum technique on the radial equation, we are finally able to determine uniquely the metric of the black hole from the knowledge of the scattering matrix at a fixed energy.

\vspace{0.5cm}
\end{abstract}


\Section{Introduction}

In this paper, we pursue our investigation of inverse scattering problems in black hole spacetimes initiated in \cite{DN1} and continued in \cite{DN2, DN3}. Roughly speaking, the question we adressed in this serie of works was: can we determine the geometry of black hole spacetimes by observing waves at infinity? Here infinity means infinity from the point of view of an observer located in the exterior region of a black hole and stationary with respect to it. For instance, from this point of view, it is well known that the event horizon (and also the cosmological horizon if any) of the black hole is perceived as an asymptotic region of spacetime by our observer, a phenomenon due to the intense effect of gravity near an horizon. Observing waves at the "infinities" of a black hole, in particular measuring how a black hole scatters incoming scalar or electromagnetic waves (or affects the trajectories of classical particles), is one of the few ways to get informations on the main characteristics of the black hole\footnote{Among the other interesting ways, the analysis of gravitational waves and/or of Hawking radiations could be also used as the starting point for an inverse problem since these phenomena are believed to be measurable in the near future.}. Indeed, since black holes are by essence invisible, we can only study them by indirect means. Positive answers were given to this question in the case of Reissner-Nordstr\"om black holes from the knowledge of the high energies of scattered Dirac waves, or in the case of Reissner-Nordstr\"om-de-Sitter black holes from the knowledge of scattered Dirac waves with a fixed nonzero energy. Whereas spherically symmetric black holes were studied in the papers \cite{DN1,DN2,DN3}, we adress here the same problem but in the case of more complicated geometrical objects: rotating black holes.

Precisely, we shall consider the family of Kerr-Newman-de-Sitter black holes (KN-dS) which are exact solutions of the Einstein-Maxwell equations and describe electrically charged rotating black holes with positive cosmological constant. In Boyer-Lindquist coordinates $(t,r,\theta,\varphi)$, the exterior region of a KN-dS black hole is described by the four-dimensional manifold
\begin{equation} \label{Intro-Manifold}
  \M = \mathbb{R}_{t} \times ]r_-,r_+[_r \times \S_{\theta,\varphi}^{2},
\end{equation}
equipped with the Lorentzian metric (having signature $(+,-,-,-)$)
\begin{equation} \label{Intro-Metric}
  g = \frac{\Delta_r}{\rho^2} \left[ dt - \frac{a\sin^2 \theta}{E} \, d\varphi \right]^{2} - \frac{\rho^2}{\Delta_r} dr^{2} - \frac{\rho^2}{\Delta_\theta} d\theta^{2} - \frac{\Delta_\theta \sin^2 \theta}{\rho^2} \left[ a \,dt - \frac{r^2 + a^2}{E} d\varphi \right]^{2},
\end{equation}
where
$$
  \rho^2 = r^2 + a^2 \cos^2 \theta, \quad \quad E = 1 + \frac{a^2 \Lambda}{3},
$$
$$
 \Delta_r = (r^2 + a^2) (1 - \frac{\Lambda r^2}{3}) - 2Mr + Q^2, \quad \ \ \Delta_\theta = 1 + \frac{a^2 \Lambda \cos^2 \theta}{3}.
$$

The three parameters $M>0$, $Q \in \R$ and $a \in \R$ appearing above are interpreted as the mass, the electric charge and the angular momentum per unit mass of the black hole whereas the parameter $\Lambda > 0$ is the positive cosmological constant of the universe. The values $r_\pm$ appearing in (\ref{Intro-Manifold}) are the two larger distinct roots of the fourth order polynomial $\Delta_r$ such that $\forall r \in (r_-, r_+), \ \Delta_r(r) > 0$. Note that the hypersurfaces $\{ r = r_\pm \}$ are singularities of the metric since $\Delta_r$ appears at the denominator in (\ref{Intro-Metric}). These hypersurfaces are called event horizon for $\{r=r_-\}$ and cosmological horizon for $\{r=r_+\}$. We emphasize that these singularities are mere coordinate singularities. The event and cosmological horizons correspond in fact to regular null hypersurfaces which can be crossed one way, but would require speed greater than that of light to be crossed the other way. Hence their names: horizons.

The Boyer-Lindquist coordinates fit well to the point of view of \emph{stationary observers}, that is the class of observers who move on worldlines of constant $r$ and $\theta$ and with a uniform angular velocity $\omega$, i.e. $\varphi = \omega t + const$. Indeed, the variable $t$ measures the perception of time of such observers when located far from the event and cosmological horizons. Think typically of a telecospe on earth aimed at the black hole. Since this idea corresponds well to the kind of inverse scattering experiment we have in mind, we shall work with this coordinates system eventhough it makes appear singularities at the event and cosmological horizons. But let us see what exactly happens there?

The answer is that, from the point of view of stationary observers, the event and cosmological horizons are unreachable regions. More precisely, recall that KN-dS black holes possess two families of incoming and outgoing shearfree null geodesics - called principal null geodesics - that foliate entirely the spacetime. They are generated by
\begin{equation} \label{Intro-NullVectors}
  V^\pm = \frac{r^2 + a^2}{\Delta_r} \left( \partial_t + \frac{a E}{r^2 + a^2} \partial_\varphi \right) \pm \partial_r,
\end{equation}
They should be thought of as the trajectories of light-rays aimed at - or coming from - the event and cosmological horizons. A crucial point is that the event and cosmological horizons are never reached in a finite time $t$ by the principal null geodesics. This means that these horizons are \emph{asymptotic regions} of space from the point of view of the stationary observers defined above.

To better understand this point, we work with a new radial variable $x$ defined by
$$
  \frac{dx}{dr} = \frac{r^2 + a^2}{\Delta_r},
$$
in such a way that, in the $t - x$ plane, the incoming and outgoing principal null geodesics become simply straightlines $x = \pm t + c$. In particular, in the new coordinates system $(t,x,\theta, \varphi)$, the event and cosmological horizons are then pushed away to $\{x = -\infty\}$ and $\{ x = +\infty\}$ respectively. This encodes the "asymptoticness" property mentioned above. Summarizing, a KN-dS black hole will be from now on described by the $4$-dimensional manifold
\begin{equation} \label{Intro-Man}
  \B = \R_t \times \R_x \times \S^2_{\theta, \varphi},
\end{equation}
equipped with the metric
\begin{equation} \label{Intro-Met}
  g = \frac{\Delta_r}{\rho^2} \Big[ dt - \frac{a\sin^2 \theta}{E} \, d\varphi \Big]^{2} - \frac{\rho^2 \Delta_r}{(r^2 + a^2)^2} dx^2 - \frac{\rho^2}{\Delta_\theta} d\theta^{2} - \frac{\Delta_\theta \sin^2 \theta}{\rho^2} \Big[ a \,dt - \frac{r^2 + a^2}{E} d\varphi \Big]^{2}.
\end{equation}

As waves for our inverse scattering problem, we shall consider massless charged Dirac fields evolving in the exterior region of a KN-dS black hole described by (\ref{Intro-Man}) and (\ref{Intro-Met}). Using the stationarity and global hyperbolicity of the spacetime, we shall see (after several simplifications) that such Dirac fields obey a PDE of evolution, generically denoted here by
\begin{equation} \label{Intro-DiracEq}
  i \partial_t \psi = H \psi,
\end{equation}
on the cylinder $\Sigma = \R_x \times \S^2_{\theta, \varphi}$ corresponding to the level spacelike hypersurface $\{t = 0\}$. Here $H$ stands for the resulting time-independent Dirac Hamiltonian, a first order matrix-valued PDE on $\Sigma$. Due to our choice of studying Dirac fields, a nice Hilbert space framework is at our disposal and the Hamiltonian $H$ that generates the evolution can be viewed as a selfadjoint operator on a fixed Hilbert space. We stress the fact that this wouldn't be true for scalar or electromagnetic fields, \textit{i.e.} for fields having integral spin, because of the well known superradiance phenomenon (see \cite{Ba5, FKSY, GGH}). In that case, there would be no positive conserved quantity along the evolution and the Hamiltonian couldn't be viewed as a selfadjoint operator on a fixed Hilbert space. The right framework seems to be the one given by Krein spaces (see \cite{FKSY, GGH} for more details on that point). This is the main reason that explains our choice to study Dirac fields for which no superradiance phenomenon occurs. But let us now describe more thoroughly the corresponding Hamiltonian $H$. It will be shown to take the form
\begin{equation} \label{Intro-Op-H}
  H = J^{-1} H_0, \qquad H_0 = \Ga D_x + a(x) H_{\S^2} + c(x, D_\varphi),
\end{equation}
where $\Ga = diag(1,-1)$, $D_x = -i \partial_x$, $H_{\S^2}$ denotes a certain angular Dirac operator on the $2$-sphere $\S^2$ and the matrix-valued potential $J^{-1}$ is given by
$$
  J^{-1} = \frac{1}{1 + \alpha(x,\theta)^2} \left( I_2 - \alpha(x,\theta) \Gc \right),
$$
where $\Gc$ is a Dirac matrix that anti-commutes with $\Ga$. At last, the potentials $a(x), c(x,D_\varphi), \alpha(x,\theta)$ are given in term of the metric (\ref{Intro-Metric}) by
$$
a(x) = \frac{\sqrt{\Delta_r}}{r^2 + a^2}, \quad c(x, D_\varphi) = \frac{aE}{r^2 + a^2} D_\varphi + \frac{qQr}{r^2 + a^2}, \quad \alpha(x,\theta) = \frac{\sqrt{\Delta_r}}{\sqrt{\Delta_\theta}} \frac{a \sin\theta}{r^2 + a^2}.
$$
Here $q$ denotes the electric charge of the Dirac field.

On one hand, the Hamiltonian $H_0$ is selfadjoint on the Hilbert space\footnote{Note that we work in an unweighted $L^2$ space in order to make the Dirac operator as independent as possible of the geometry.}
$$
  \H = L^2(\R \times \S^2, dx d\theta d\varphi; \,\C^2),
$$
used to represent our Dirac spinors and equipped with the usual scalar product. On the other hand, a residue of superradiance phenomenon entails that the Hamiltonian $H$ is selfadjoint on the slightly modified Hilbert space
$$
  \G = L^2(\R \times \S^2, dx d\theta d\varphi; \,\C^2),
$$
equipped with the scalar product $\l.,.\r_\G = (.,J.)_\H$ where
$$
  J = I_2 + \alpha(x,\theta) \Gc.
$$

We shall show that $\sup_{\theta} |\alpha(x,\theta)|$ is exponentially decreasing at both horizons $\{ x = \pm\infty\}$. Hence the full Dirac Hamiltonian $H$ can be viewed as a "small" non spherically symmetric perturbation of order $1$ of the Hamiltonian $H_0$ which in turn is composed of
\begin{itemize}
 \item a differential operator $\Ga D_x + \frac{a E}{r^2 + a^2} D_\varphi$ that - roughly speaking - corresponds asymptotically to transport along the outgoing and incoming principal null geodesics,
 \item an angular matrix-valued differential operator $H_{\S^2}$ weighted by a scalar potential $a(x)$ that is exponentially decreasing at both horizons $\{ x = \pm\infty\}$,
 \item a scalar perturbation $\frac{qQr}{r^2 + a^2}$ caused by the interaction between the electric charge of the black hole and that of the Dirac fields.
\end{itemize}
We emphasize here once again that the full Hamiltonian $H$ is a perturbation of order $1$ of the Hamiltonian $H_0$ which, at first sight, breaks the obvious symmetries of $H_0$ (see below for more details on this point). We list now a few other generic properties of the Dirac Hamiltonian $H$ that are important to understand the nature of our inverse problem.

The Dirac Hamiltonian $H$ shares many properties with a canonical Dirac Hamiltonian on the cylinder $\Sigma = \R_x \times \S^2_{\theta, \varphi}$ equipped with a Riemanniann metric having two asymptotically hyperbolic ends $\{x = \pm \infty\}$. More precisely, the Hamiltonians $H$ and $H_0$ above can be regarded as perturbations of order $1$ of a canonical Dirac Hamiltonian on $\Sigma$ equipped with the spherically symmetric simple metric
$$
  h = dx^2 + \frac{1}{a(x)^2} \left( d\theta^2 + \sin^2\theta d\varphi^2 \right),
$$
and $a(x)$ exponentially decreasing at both horizons $\{ x = \pm\infty\}$. We refer to \cite{DN3}, Theorem 1.2 and the remark before it, for a few words on this simpler model and the corresponding inverse results.

Nevertheless, our manifold is - in some sense - simpler than a general Riemanniann manifold with two asymptotically hyperbolic ends since it possesses symmetries. It is obvious from (\ref{Intro-Metric}) that the spacetime $\M$ is invariant by time translation, meaning that the vector field $\partial_t$ is Killing and thus generates a continuous group of isometry, and also has cylindrical symmetry, meaning that the vector field $\partial_\varphi$ is Killing. The first symmetry was used to express the Dirac equation as an evolution equation generated by a \emph{time-independent} Hamiltonian $H$. The second symmetry allows us to decompose the Dirac equation onto a Hilbert basis of angular modes $\{e^{ik\varphi}, \ k \in \frac{1}{2} + \Z\}$ simplifying thus the problem.

Although there is no other symmetry (in the sense of the existence of globally defined Killing vector fields on $\M$), KN-dS black holes have hidden symmetries that can be used to considerably simplify the problem. In particular, they are responsible for the separability of the Dirac equation (\ref{Intro-DiracEq}) into coupled systems of ODEs, a property crucially used in this paper. We mention that these hidden symmetries come from the existence of a Killing tensor \cite{AB}, itself coming from the existence of a conformal Killing-Yano tensor \cite{Fr}. For the separability of Dirac equation in spacetimes admitting such structures, we refer to \cite{Kam}.


A scattering theory for scalar waves on general asymptotically hyperbolic manifolds is now well studied and well known. We refer for instance to \cite{IK} for an extensive introduction to the direct and inverse scattering results in this class of manifolds. In this paper, we first establish a direct scattering theory for massless Dirac fields evolving in the exterior region of a KN-dS black hole. We follow the papers \cite{HaN, Da2} where similar results were obtained in Kerr or Kerr-Newman black holes (that is $\Lambda = 0$ in our model).

To understand the scattering properties of the Dirac fields, it is important to have in mind that there are two similar but distinct asymptotic regions: the event and cosmologial horizons. At late times (and from the point of view of stationary observers), the Dirac fields scatter towards these asymptotic regions. Indeed, the pure point spectrum of $H$ will be shown to be empty, meaning that the energy of Dirac fields cannot remain trapped in compact sets between the two horizons. Moreover and as usual in scattering theory, the Dirac fields seem to obey simpler PDEs in the asymptotic regions (always from the point of view of stationary observers). Precisely, when $x \to \pm \infty$, we formally see that
$$
  H \to H_\pm = \Ga D_x + \frac{aE}{r_\pm^2 + a^2} D_\varphi + \frac{qQr_\pm}{r_\pm^2 + a^2}.
$$
The simpler Hamiltonians $H_\pm$ generate the asymptotic dynamics at the event and cosmological horizons respectively. In order to separate the two asymptotic regions, we use the projectors onto the positive and negative spectrum of $\Ga$
$$
  P_\pm = \mathbf{1}_{\R^\pm}(\Ga).
$$
The operators $P_\pm$ are in fact the asymptotic velocity operators associated to the unperturbed Dirac Hamiltonian $\Ga D_x$. That these operators are useful to separate the event and cosmological horizons was already used in \cite{HaN,Da1,Da2}. Then we define the future (+) and past (-) asymptotic Hamiltonians
$$
  H^+ = H_+ P_+ + H_- P_-, \quad \quad H^- = H_+ P_- + H_- P_+,
$$
which are clearly selfadjoint on $\H$. We are then able to prove the following Theorem

\begin{theorem} \label{Intro-WO}
  The global wave operators
  $$
    W^\pm(H,H^\pm,I_2) = s-\lim_{t \to \pm \infty} e^{itH} e^{-itH^\pm},
  $$
  exist as operators from $\H$ to $\G$ and are asymptotically complete, \textit{i.e.} they are isometries from $\H$ to $\G$ and their inverse wave operators given by
  $$
    (W^\pm(H,H^\pm,I_2)^* = W^\pm(H^\pm,H,J) = s-\lim_{t \to \pm \infty} e^{itH^\pm} J e^{-itH},
  $$
  also exist as operators from $\G$ to $\H$.
\end{theorem}

The Theorem \ref{Intro-WO} will be proved in several steps in Section \ref{Time-Dependent}. The main technical point is to obtain a Limiting Absorption Principle (LAP) for the Hamiltonian $H$. This will be done in Appendix \ref{LAP} by means of a Mourre theory. Once a LAP for $H$ is obtained, the existence and asymptotic completeness of the above wave operators follow from standard arguments of Kato's $H$-smooth operator theory (see for instance \cite{RS}).

The global time-dependent scattering operator is then defined by the standard rule
\begin{equation} \label{Intro-TD-ScatOp}
  S = (W^+(H,H^+,I_2))^* \ W^-(H,H^-,I_2),
\end{equation}
which is clearly a unitary operator from $\H$ to $\H$. Note that, due to our choice of unweighted Hilbert space $\H$, we can naturally compare two different scattering operators $S$ and $\tilde{S}$ associated to two different KN-dS black holes since they act on the same Hilbert space. In this paper however, we are interested in inverse scattering problem at a fixed energy. That's why we won't take the full scattering operator $S$ given by (\ref{Intro-TD-ScatOp}) as the starting point of our inverse problem. Instead, from $S$, we construct the corresponding time-dependent scattering matrix $S(\lambda)$ at energy $\lambda$. Since the asymptotic Hamiltonians $H^\pm$ are not identical, there is no canonical way to define such a scattering matrix. Nevertheless, we follow the usual route and introduce the following unitary operators on $\H$
\begin{equation} \label{Intro-F+}
  F^+ \psi(\lambda) = \frac{1}{\sqrt{2\pi}} \int_\R  \left( \begin{array}{cc} e^{-ix (\lambda - \Omega_+(D_\varphi))}&0\\0&e^{ix (\lambda - \Omega_-(D_\varphi))} \end{array} \right) \psi(x) dx,
\end{equation}
and
\begin{equation} \label{Intro-F-}
  F^- \psi(\lambda) = \frac{1}{\sqrt{2\pi}} \int_\R  \left( \begin{array}{cc} e^{-ix (\lambda - \Omega_-(D_\varphi))}&0\\0&e^{ix (\lambda - \Omega_+(D_\varphi))} \end{array} \right) \psi(x) dx,
\end{equation}
where $\Omega_\pm(D_\varphi) = \frac{aE D_\varphi + qQr_\pm}{r^2_\pm + a^2}$. These operators diagonalize the asymptotic Hamiltonians $H^+$ and $H^-$ respectively. Hence, we define the global scattering matrix at energy $\lambda$ in a natural way by the rule
\begin{equation} \label{Intro-TD-Scat}
  S = (F_+)^* S(\lambda) F_-.
\end{equation}
Note that $S(\lambda)$ is a unitary operator on $\H_{\S^2} = L^2(\S^2; \C^2)$.

The question we adress in this paper can be now more precisely stated. Does the knowledge of $S(\lambda)$ for a fixed energy $\lambda \in \R$ determine uniquely a KN-dS black hole? In fact, we can refine considerably our uniqueness inverse result as follows. Note first that the scattering matrix $S(\lambda)$ has the following $2$ by $2$ matrix structure
$$
  S(\lambda) = \left[ \begin{array}{cc} T^L(\lambda) & R(\lambda) \\ L(\lambda) & T^R(\lambda) \end{array} \right].
$$
The operators $T^L(\lambda)$ and $T^R(\lambda)$ are the transmission operators at a fixed energy $\lambda \in \R$ whereas $R(\lambda)$ and $L(\lambda)$ are the reflection operators from the right and from the left respectively at a fixed energy $\lambda$. The formers measure the part of signal that is transmitted from one horizon to the other in a scattering process. The latters measure the part of a signal that is reflected from one end to itself in a scattering process. Note at last that all of them act on the Hilbert space $L^2(\S^2, d\theta d\varphi; \C)$.

To determine uniquely a KN-dS black hole, it will be enough to know either the transmission operators $T^{L/R}(\lambda)$, or one of the reflection coefficients $R(\lambda)$ or $L(\lambda)$ at a fixed energy. In fact, we can even obtain a better result using the cylindrical symmetry of the spacetime. Clearly, this entails that all the operators (Dirac Hamiltonian, wave and scattering operators) under study can be decomposed onto the Hilbert sum of angular modes $\{ e^{ik\varphi} \}_{k \in 1/2 + \Z}$. Note our choice of half-integers $k \in \frac{1}{2} + \Z$ since we want some anti-periodic conditions in the variable $\varphi$. Recall that we are working with spinors that change sign after a complete rotation. Then, we have
$$
  \H_{\S^2} = \oplus_{k \in \frac{1}{2} + \Z} \H_{\S^2}^k, \quad \H_{\S^2}^k = L^2((0,\pi),  d\theta; \C^2),
$$
and the scattering matrix can be decomposed as the orthogonal sum
$$
  S(\lambda) = \oplus_{k \in \frac{1}{2} + \Z} S_k(\lambda), \quad S_k(\lambda) = \left[ \begin{array}{cc} T_k^L(\lambda) & R_k(\lambda) \\ L_k(\lambda) & T_k^R(\lambda) \end{array} \right],
$$
where $T^{L/R}_k(\lambda), R_k(\lambda)$ and $L_k(\lambda)$ are the corresponding transmission and reflection operators acting on the Hilbert space $\ls = L^2((0,\pi),  d\theta; \C)$.

The main inverse result of this paper states that the knowledge of $T_k^{L/R}(\lambda)$ or $R_k(\lambda)$ or $L_k(\lambda)$ at a fixed energy $\lambda \in \R$ and for two different angular modes $k \in \frac{1}{2} + \Z$ is enough to determine uniquely a KN-dS black hole. Precisely, we shall prove

\begin{theorem} \label{Intro-Main}
Let $(M,Q^2,a,\Lambda)$ and $(\tilde{M},\tilde{Q}^2,\tilde{a},\tilde{\Lambda})$ be the parameters of two a priori different KN-dS black holes. Let $\lambda \in \R$ and denote by $S(\lambda)$ and $\tilde{S}(\lambda)$ the corresponding scattering matrices at fixed energy $\lambda$. More generally, we shall add a symbol $\ \tilde{}$ to all the relevant scattering quantities corresponding to the second black hole. Assume that both reduced transmission operators $T_k^{R}(\lambda)$, $T_k^{L}(\lambda)$ or one reduced reflection operators $R_k(\lambda)$ or $L_k(\lambda)$ are known, \textit{i.e.}
\begin{eqnarray}
  T_{k}^{R/L}(\lambda) & = & \tilde{T}_{k}^{R/L}(\lambda), \nonumber \\
  R_{k}(\lambda) & = & \tilde{R}_{k}(\lambda), \label{Intro-MainAssumption} \\
  L_{k}(\lambda) & = & \tilde{L}_{k}(\lambda), \nonumber
\end{eqnarray}
as operators on $\ls = L^2((0,\pi), d\theta; \C)$ and for \emph{two} different values of $k \in \frac{1}{2} + \Z$. Then the parameters of the two black holes coincide, \textit{i.e.}
$$
  M = \tilde{M}, \ a = \tilde{a}, \ Q^2 = \tilde{Q}^2, \ \Lambda = \tilde{\Lambda}.
$$
\end{theorem}

In fact we obtain much better results in the course of the proof. Indeed, not only do we recover four parameters, but we are in fact able  to determine scalar functions depending on the radial variable $x$. Precisely, the function
\begin{equation} \label{Function}
 \frac{\lambda - c(x,k)}{a(x)},
\end{equation}
is determined up to certain diffeomorphisms. Since (\ref{Function}) is known for two different $k$, we determine in fact two different functions. Finally, if we specialize to the case of Kerr-de-Sitter black hole, that is if we set $Q = 0$ in the metric (\ref{Intro-Metric}), then we determine the functions
\begin{equation} \label{Function1}
  a(x), \quad c(x,k),
\end{equation}
separately up to a discrete set of translations in the variable $x$. We emphasize that the recovery of (\ref{Function}) and (\ref{Function1}) from one of the scattering operators at a fixed energy $\lambda$ is the main result of this paper. We recover thus the potentials that appear in the expression of the Hamiltonian $H$ or rather $H_0$. We then use the particular form of the functions (\ref{Function}) or (\ref{Function1}) to determine uniquely the parameters $M, a, Q^2, \Lambda$.

Let us finish this introduction saying a few words on the strategy of the proof of Theorem \ref{Intro-Main}. Since we need to introduce many other objects and notations for the full proof, we shall restrict ourselves here to give only the main steps. However, we devote the entire Section \ref{MainResult} to sketch the proof of our main Theorem in much more details, once we have introduced all the necessary background.

The first step is to use the hidden symmetries of the equation to simplify further the problem. As already said, the Dirac equation in KN-dS black holes can be separated into systems of ODEs. Precisely, it will be shown in Section \ref{Sep-Var} that the stationary equation at a fixed energy $\lambda \in \R$
\begin{equation} \label{Intro-StatEq}
  H \psi = \lambda \psi,
\end{equation}
can be simplified as follows.

\begin{theorem} \label{Intro-Separation}
  Denote by $I$ the set $(1/2 + \Z) \times \N^*$. Then, for all $\lambda \in \R$, there exists a Hilbert decomposition of the energy space $\H$ as
  $$
    \H = \bigoplus_{(k,l) \in I} \H_{kl}(\lambda),
  $$
  where
  $$
    \H_{kl}(\lambda) = L^2(\R ; \C^2) \otimes Y_{kl}(\lambda) \ \simeq \ L^2(\R,\C^2),
  $$
  with the following properties.
  \begin{itemize}
  \item The $Y_{kl}(\lambda) = Y_{kl}(\lambda)(\theta,\varphi)$ are $L^2$ solutions of the eigenvalues angular equation
  \begin{equation} \label{Intro-AngularEq}
    \left[ H_{\S^2} - \lambda \frac{a \sin\theta}{\sqrt{\Delta_\theta}} \Gc \right] Y_{kl}(\lambda) = \muk Y_{kl}(\lambda).
  \end{equation}
  More precisely, the $Y_{kl}(\lambda)$ are the normalized eigenfunctions of the selfadjoint angular operator
  \begin{equation} \label{Intro-AngularOp}
    A_{\S^2}(\lambda) = H_{\S^2} - \lambda \frac{a \sin\theta}{\sqrt{\Delta_\theta}} \Gc.
  \end{equation}
  on $\H_{\S^2} = L^2(\S^2: \C^2)$ associated to its positive eigenvalues $\mu_{kl}(\lambda)$ ordered in such a way that $\forall k \in \frac{1}{2} +\Z, \ \forall l \in \N^*$, $\muk < \mu_{k(l+1)}(\lambda)$.
  \item For all $\psi = \sum_{(k,l) \in I} \psi_{kl}(x) \otimes Y_{kl}(\lambda) \in \H$, the stationary equation (\ref{Intro-StatEq}) is equivalent to the countable family of one-dimensional radial stationary equations
  \begin{equation} \label{Intro-RadialEq}
   \left[ \Ga D_x + \mu_{kl}(\lambda) a(x) \Gb + c(x,k) \right] \psi_{kl}(x) = \lambda \, \psi_{kl}(\lambda),
  \end{equation}
  parametrized by the generalized angular momentum $\muk$.
  \end{itemize}
\end{theorem}

Before proceeding further, note that the generalized angular momenta $\muk$ - that play here the role of constants of separation - do depend on the energy $\lambda$, making the separation of variables process highly non trivial. Now let us fix an energy $\lambda \in \R$. Using the corresponding decomposition of the Hilbert space $\H$ given in Theorem \ref{Intro-Separation}, the scattering matrix $S(\lambda)$ can also be written as the orthogonal sum of $2$ by $2$ unitary matrices in the following way
$$
  S(\lambda) = \oplus_{(k,l) \in I} S_{kl}(\lambda),
$$
where

$$
  S_{kl}(\lambda) = \left[ \begin{array}{cc} T(\lambda,k,\muk) & R(\lambda,k,\muk) \\ L(\lambda,k,\muk) & T(\lambda,k,\muk) \end{array} \right].
$$

The coefficients $T(\lambda,k,\muk)$, $R(\lambda,k,\muk)$ and $L(\lambda,k,\muk)$ are naturally interpreted as the transmission and reflection coefficients associated to the one-dimensional Dirac equation (\ref{Intro-RadialEq}). This is why we emphasize their dependence with respect to $\lambda$, $k$ and $\muk$ since the equations (\ref{Intro-RadialEq}) do depend on these parameters disjointly. Also, these transmission and reflection coefficients can be expressed in terms of determinants of particular fundamental solutions of (\ref{Intro-RadialEq}) - called Jost functions - having prescribed asymptotic behaviours at the horizons. One important point to mention here is that we have at our disposal very precise expressions for the coefficients $T(\lambda,k,\muk)$, $R(\lambda,k,\muk)$ and $L(\lambda,k,\muk)$ as power series in the generalized angular momentum $\muk$. These exact expressions will allow us to make accurate calculations of the asymptotics of the scattering coefficients needed in the next steps.

Summarizing at this stage, we see that the transmission operators $T_k^{L/R}(\lambda)$ and the reflection operators $R_k(\lambda), \ L_k(\lambda)$ which are known by our main assumption (\ref{Intro-MainAssumption}), can be considerably simplified if we decompose them onto the generalized spherical harmonics $Y_{kl}(\lambda)$, itselves eigenfunctions of the angular equation (\ref{Intro-AngularEq}). Notice that, and this is one of the main peculiarity of this model (with respect to \cite{DN3}), these generalized spherical harmonics also depend on the parameters of the KN-dS black hole and thus are a priori unknown.

The second step is then to show that, although the previous remark entails that we don't know the scattering coefficients $T(\lambda,k,\muk)$, $R(\lambda,k,\muk)$ and $L(\lambda,k,\muk)$ a priori, it can be shown that - rougly speaking - our main assumption (\ref{Intro-MainAssumption}) implies that for large enough $l >> 1$, one of the following condition is satisfied.
\begin{eqnarray}
  T(\lambda,k,\muk) & = & \tilde{T}(\lambda,k,\tilde{\mu}_{kl}(\lambda)), \nonumber \\
  R(\lambda,k,\muk) & = & \tilde{R}(\lambda,k,\tilde{\mu}_{kl}(\lambda)), \label{Intro-Cons1} \\
  L(\lambda,k,\muk) & = & \tilde{L}(\lambda,k,\tilde{\mu}_{kl}(\lambda)). \nonumber \
\end{eqnarray}
Moreover, for large enough $l >> 1$, we can also show
\begin{equation} \label{Intro-Cons2}
  Y_{kl}(\lambda) = \tilde{Y}_{kl}(\lambda).
\end{equation}
This will be proved in Section \ref{ApplicationFrobenius}. We emphasize that the above results are not at all immediate and are a consequence of the following intermediate results and arguments:
\begin{itemize}
\item A re-interpretation of the coefficients $T(\lambda,k,\muk)$, $|R(\lambda,k,\muk)|^2$ and $|L(\lambda,k,\muk)|^2$ as eigenvalues associated to certain operators (for instance to $T_k^L(\lambda)$, $R_k(\lambda)^* R_k(\lambda)$ and $L_k(\lambda) L_k(\lambda)^*$) which are supposed to be known by (\ref{Intro-MainAssumption}).
\item A detailed analysis of the asymptotics of $T(\lambda,k,\muk)$, $R(\lambda,k,\muk)$ and $L(\lambda,k,\muk)$ when $\muk \to \infty$. This leads to the proof that these coefficients are strictly increasing with respect to the generalized angular momentum $\muk$.
\item A proof that for fixed $k \in \frac{1}{2} + \Z$, the generalized angular momenta increases as the integers, \textit{i.e.} there exists $0 < c < C$ such that $c l \leq \muk \leq C l$ for $l$ large enough. This is done in Appendix \ref{Estimate-mukl}.
\item As a consequence of all these results put together, the scattering coefficients $T(\lambda,k,\muk)$, $|R(\lambda,k,\muk)|^2$ and $|L(\lambda,k,\muk)|^2$ are shown to be simple eigenvalues of the operators $T_k^L(\lambda)$, $R_k(\lambda)^* R_k(\lambda)$ and $L_k(\lambda) L_k(\lambda)^*$ for $l$ large enough, associated with the eigenfunctions $Y_{kl}^1(\lambda)$ or $Y_{kl}^2(\lambda)$, components of $Y_{kl}(\lambda)$. This leads to the above results.
\end{itemize}

The third step detailed in Section \ref{InverseAngular} consists in solving an inverse problem for the angular equation (\ref{Intro-AngularEq}) from the knowledge of the generalized spherical harmonics $Y_{kl}(\lambda)$ given in (\ref{Intro-Cons2}). This can be done for two reasons. First, we only have to recover two parameters since it turns out that the angular operator (\ref{Intro-AngularOp}) only depends on $a$ and $\Lambda$. Second, the angular equation (\ref{Intro-AngularEq}) is a system of ODEs of Fuschian type having weakly singularities at the north and south poles of $\S^2$ in our coordinate system. In consequence, this equation can be explicitely solved by Frobenius method. In particular, we are able in Section \ref{Frobenius} to construct (under the form of singular power series in the variable $\theta$) the generalized spherical harmonics $Y_{kl}(\lambda)$ as the unique $L^2$ solutions of (\ref{Intro-AngularEq}). From the asymptotics of $Y_{kl}(\lambda)$ when $\theta \to 0$ (that we are able to calculate) and (\ref{Intro-Cons2}), we easily prove that $a$ and $\Lambda$ are uniquely determined. From this in turn, we deduce that all the quantites depending only on $a$ and $\Lambda$ are uniquely determined. This is the case for instance for the generalized angular momenta $\muk$ and $\tilde{\mu}_{kl}(\lambda)$ which are thus shown to coincide under the assumption (\ref{Intro-MainAssumption}). This will be used in the last step of the proof of our main Theorem.

The fourth and main idea of this paper is to use the Complex Angular Momentum (CAM) method (see \cite{Re, Ne}  for a presentation of the method in the historical setting of radial Schr\"odinger operators) to get more informations from the scattering coefficients $T(\lambda,k,\muk)$, $R(\lambda,,k,\muk)$ and $L(\lambda,,k,\muk)$ known for $(k,l) \in (\frac{1}{2} + \Z) \times \N^*$. This method - already used in \cite{DN3} in a simpler setting - consists in allowing the generalized angular momentum $\muk$ to be complex and in using the particular analytic properties of the coefficients $T(\lambda,k,z), \ R(\lambda,k,z), L(\lambda,k,z)$ with respect to $z \in \C$. One of the main steps is to prove that certain quantities such as $\frac{1}{T(\lambda,k,z)}$ are entire functions of exponential type with respect to $z$ and belong to the Nevanlinna class of analytic functions. We refer to the beginning of Section \ref{Nevanlinna} for a precise definition. The important point here is that such functions are uniquely determined by their values on a sequence $\alpha_l$ of complex numbers satisfying a M\"untz condition such as
$$
  \sum_{l} \frac{1}{1 + |\alpha_l|} = \infty.
$$
Since the generalized angular momenta $\muk$ verify the above M\"untz condition (see Appendix \ref{Estimate-mukl}) and thanks to our previous results, we can show that the Nevanlinna function $\frac{1}{T(\lambda,k,z)}$ is known not only for $z = \muk, \ k \in \frac{1}{2} + \Z$ and for all $l \in \N^*$, but for all $z \in \C$. This enlarges considerably the amount of informations we can extract from the initial transmission operator $T_k^{L/R}(\lambda)$. Similarly, we can show that $R(\lambda,k,z)$ and $L(\lambda,k,z)$ are known for all $z \in \C$. From this novel amount of informations, we can straightforwardly obtain a first uniqueness result {\it{localized in energy}}. Precisely, we prove

\begin{theorem} \label{Intro-Uniqueness-Localized}
  Let $(M,Q^2,a,\Lambda)$ and $(\tilde{M},\tilde{Q}^2,\tilde{a},\tilde{\Lambda})$ be the parameters of two a priori different KN-dS black holes. We denote by $I$ a (possibly small) open interval of $\R$. Assume that for all $\lambda \in I$ and for two different $k \in \frac{1}{2} + \Z$, one of the following conditions holds
\begin{eqnarray*}
  L(\lambda,k,\muk) &=& \tilde{L}(\lambda,k,\muk),  \\
  R(\lambda,k,\muk) &=& \tilde{R}(\lambda,k,\muk),
\end{eqnarray*}
for all $l \in \mathcal{L}_k$ where the sets $\mathcal{L}_k \subset \N^* $ satisfy a M\"untz condition
$$
  \sum_{l \in \mathcal{L}_k} \frac{1}{l} = \infty.
$$
Then, we have
\begin{equation} \label{Intro-Uniqueness-ac}
  a(x) = \tilde{a}(x), \quad c(x,k) = \tilde{c}(x,k).
\end{equation}
In particular, using the particular form of the potential $a(x)$ and $c(x,k)$, we can show that the parameters of the two black holes coincide.
\end{theorem}
We refer to Section \ref{Complexification} for a proof of the above results on the CAM method.

Eventually, to obtain a uniqueness result from the scattering coefficients at a \emph{fixed} energy, we need more informations on the properties of the scattering coefficients with respect to the complexified angular momentum $z$. In particular, using a convenient change of variable (a Liouville transformation) and the corresponding form of the radial Dirac equation (\ref{Intro-RadialEq}), we shall obtain precise asymptotics of the scattering data $T(\lambda, k, z), \ R (\lambda, k ,z), \ L (\lambda, k ,z)$ when $z \to +\infty$. These asymptotics, the results of the previous CAM method together with a standard technique (as exposed first in \cite{FY} and used in this setting in \cite{DN3}) will lead to the unique determination of certain scalar functions depending on the radial variable (up to certain diffeomorphisms). Once again, from the explicit form of these functions, we prove the uniqueness of the parameters of the two black holes, that is our main Theorem \ref{Intro-Main}. All these last results will be proved in Sections \ref{AsymptoticsSD} and \ref{Inverse}.


\Section{Kerr-Newman-de-Sitter black holes}

Kerr-Newman-de-Sitter black holes (KN-dS) are exact solutions of the Einstein-Maxwell equations that describe electrically charged rotating black holes with positive cosmological constant. In Boyer-Lindquist coordinates, the exterior region of a KN-dS black hole is described by
the four-dimensional manifold
$$
  \M = \mathbb{R}_{t} \times ]r_-,r_+[_r \times \S_{\theta,\varphi}^{2},
$$
equipped with the Lorentzian metric (having signature $(+,-,-,-)$)
\begin{equation} \label{Metric}
  g = \frac{\Delta_r}{\rho^2} \left[ dt - \frac{a\sin^2 \theta}{E} \, d\varphi \right]^{2} - \frac{\rho^2}{\Delta_r} dr^{2} - \frac{\rho^2}{\Delta_\theta} d\theta^{2} - \frac{\Delta_\theta \sin^2 \theta}{\rho^2} \left[ a \,dt - \frac{r^2 + a^2}{E} d\varphi \right]^{2},
\end{equation}
where
$$
  \rho^2 = r^2 + a^2 \cos^2 \theta, \quad E = 1 + \frac{a^2 \Lambda}{3},
$$
$$
 \Delta_r = (r^2 + a^2) (1 - \frac{\Lambda r^2}{3}) - 2Mr + Q^2, \quad \Delta_\theta = 1 + \frac{a^2 \Lambda \cos^2 \theta}{3}.
$$

The three parameters $M>0$, $Q \in \R$ and $a \in \R$ appearing above are interpreted as the mass, the electric charge and the angular momentum per unit mass of the black hole whereas the parameter $\Lambda > 0$ is the cosmological constant of the universe. The electromagnetic potential $1$-form $A$ solution of the Maxwell equation is also given by
\begin{equation} \label{ElecPot}
  A = A_a dx^a = - \frac{Q r}{\rho^2} \Big(dt - \frac{a \sin^2 \theta}{E} d\varphi \Big).
\end{equation}

For later use, we give another expression for the metric $g$ where the coordinates $1$-forms have been isolated.
\begin{eqnarray} \label{MetricG}
  g & = & \frac{\Delta_r - \Delta_\theta a^2 \sin^2\theta}{\rho^2} dt^2 - \frac{2a\sin^2\theta}{E \rho^2} (\Delta_r - \Delta_\theta a^2 \sin^2\theta) dt d\varphi \\
    &   & \qquad - \frac{\rho^2}{\Delta_r} dr^{2} - \frac{\rho^2}{\Delta_\theta} d\theta^{2} - \frac{\sin^2\theta}{E^2 \rho^2} \left( \Delta_\theta (r^2 + a^2)^2 - \Delta_r a^2 \sin^2\theta \right) d\varphi^2.
\end{eqnarray}

The geometry of $(\M, g)$ crucially depends on the possible sign taken by the function $\Delta_r$. We are interested here in describing KN-dS black holes, \textit{i.e.} in the case where the function $\Delta_r$ has three simple positive roots $0<r_c<r_-<r_+$ and thus only a negative one $r_n = -(r_c + r_- + r_+) <0$. As studied in \cite{BC1}, this is the case if the following conditions are fulfilled
\begin{eqnarray}
  (1) & \frac{a^2 \Lambda}{3} \leq 7 - 4 \sqrt{3}, \label{Cond-Lambda} \\
  (2) & M^-_{crit} < M < M^+_{crit}, \label{Cond-M}
\end{eqnarray}
with
\begin{eqnarray*}
  M^\pm_{crit} & = & \frac{1}{\sqrt{18 \Lambda}} \left( \Big( 1 - \frac{a^2 \Lambda}{3} \Big) \pm \sqrt{\Big( 1 - \frac{a^2 \Lambda}{3} \Big)^2 - 4 \Lambda ( a^2 + Q^2)} \,\right)^2 \\ & & \quad \left( 2 \Big( 1- \frac{a^2 \Lambda}{3} \Big)^2 \mp \sqrt{\Big( 1 - \frac{a^2 \Lambda}{3} \Big)^2 - 4 \Lambda (a^2 + Q^2)} \, \right).
\end{eqnarray*}

In what follows, we shall always assume that the conditions (\ref{Cond-Lambda}) and (\ref{Cond-M}) are fulfilled. In this case, the hypersurfaces $\{r = r_j\}, \ j = +, -, c$ appear as \emph{singularities} of the metric (\ref{Metric}). We recall that they are merely \emph{coordinate singularities} that could be easily removed using another coordinate systems. The hypersurface $\{r=r_c\}$ is called the Cauchy horizon whereas the hypersurfaces $\{r=r_-\}$ and $\{r=r_+\}$ are called the event and cosmological horizons respectively. It can be shown (\cite{AM, W}) that they are regular null hypersurfaces that can crossed one way but would require speed greater than that of light to be crossed the other way. Hence their name: horizons.

\begin{remark}
  In contrast with the "coordinate" singularities $\{r = r_j\}, \ j=c,\pm$, there is a \emph{curvature singularity} located at each point of the ring $\left\{r = 0, \ \theta = \frac{\pi}{2} \right\}$. Here some scalars obtained by contracting the Riemann tensor explode.
\end{remark}

\begin{remark}
  When $M = M^-_{crit}$, the Cauchy and event horizons coalesce and the metric (\ref{Metric}) represents an extremal KN-dS black hole. Similarly, when $M = M^+_{crit}$, the event and cosmological horizons coalesce and we also obtain a black hole solution. We won't study these special cases in this paper and thus assume (\ref{Cond-M}) with stric inequalities. Also, when the parameter $a = 0$, we recover the family of \emph{spherically symmetric Reissner-Nordstr\"om-de-Sitter} black holes and when the parameters $a = Q = 0$, we obtain the more familiar family of \emph{Schwarzschild-de-Sitter} black holes. The former class of solutions has been the object of the paper \cite{DN3}.
\end{remark}

In this paper, we shall only consider the exterior region of a KN-dS black hole, that is the region lying between the event and cosmological horizons $\{ r_- < r < r_+ \}$. Note that the function $\Delta_r$ is positive there. Let us list a few geometrical properties of this region. \\
\begin{itemize}
\item It is a \emph{"stationary" axisymmetric} spacetime. Indeed the vector fields $\partial_t$ and $\partial_\varphi$ are clearly Killing vectors that generate the time-translation and cylindrical symmetries of the spacetime. We write "stationary" between commas since, strictly speaking, the stationarity only happens in the region far from the event and cosmological horizons (see the third item below). Observe that the lack of spherical symmetry appears in the crossed term $(\dots) \, dt d\varphi$ in (\ref{MetricG}) and the dependence of the metric $g$ with respect to the azimuthal variable $\theta$. The family of KN-dS black holes is thus a non trivial generalization of the family of spherically symmetric Reissner-Nordstr\"om black holes studied in \cite{DN3}.

\item The exterior region $\M$ is \emph{globally hyperbolic} meaning that $\M$ is foliated by Cauchy hypersurfaces, given here by the level hypersurfaces $\Sigma_t = \{t = \textrm{const}.\}$ of the time function $t$. This, together with the stationarity of the spacetime, entail that the Dirac equation we shall study in this paper will naturally take the form of an evolution equation with initial data on the spacelike hypersurface $\Sigma_0$ and with associated Dirac Hamiltonian independent of the time coordinate $t$. General existence Theorems on hyperbolic systems of PDEs by Leray assert that the corresponding Cauchy problem will be well posed.

\item The vector field $\partial_t$ is \emph{timelike} except in two toroidal regions surrounding the event and cosmological horizons: the \emph{ergospheres} defined as the regions where $r_- < r < r_+$ and $\Delta_r - \Delta_\theta a^2 \sin^2 \theta < 0$. Conversely the vector fields $\partial_r, \partial_\theta$ are \emph{spacelike} everywhere in the exterior region whereas $\partial_\varphi$ is \emph{spacelike} except within the two ergospheres where it becomes \emph{timelike}. Note that the presence of ergospheres indicates the lack of stationarity of the exterior region of a KN-dS black hole, that is there exists no globally defined Killing vector field that is timelike in $\M$. This entails analytical difficulties when studying fields with integral spin since we cannot find any quantities conserved along the evolution such that the corresponding dynamics is then generated by a selfadjoint Hamiltonian. This is also the root of the superradiance phenomenon that quantifies the possibility for fields with integral spin to extract (rotational) energy from the black hole. We shall avoid this difficulty here by considering Dirac fields (having spin $1/2$) for which no superradiance phenomenon occurs.

\item The spacetime $(\M, g)$ is described from the point of view of \emph{stationary observers}, that is the class of observers who move on worldlines of constant $r$ and $\theta$ and with a uniform angular velocity $\omega$, i.e. $\varphi = \omega t + const$. When $\omega = 0$, these observers are called \emph{static}. Choosing $r_- \, << \, r \, << \, r_+$, that is a region far from the event and cosmological horizons and their ergospheres, we note that the function $t$ is proportional to the proper time of such observers. Thus \emph{the variable $t$ measures the perception of time of stationary observers} located far from the event and cosmological horizons. This is the point of view we shall adopt in this paper.

\item The exterior region $\M$ is of Petrov type D, that is it possesses exactly two principal null vectors given by (see \cite{AM, SU})
\begin{equation} \label{NullVectors}
  V^\pm = \frac{r^2 + a^2}{\Delta_r} \left( \partial_t + \frac{a E}{r^2 + a^2} \partial_\varphi \right) \pm \partial_r,
\end{equation}
that generate the whole family of principal null geodesics. The spacetime is entirely foliated by this family of shearfree null geodesics which should be thought as the trajectories of light-rays aimed at - or coming from - the event and cosmological horizons. As a consequence of the term $1/\Delta_r$ in the expression (\ref{NullVectors}), we see that that the event and cosmological horizons are never reached in a finite time $t$ by the principal null geodesics. We thus point out that these horizons are \emph{asymptotic regions} of spacetime from the point of view of the stationary observers defined above. As a consequence, we stress the fact that we won't need to add boundary conditions at the event and cosmological horizons to study the Cauchy problem of Dirac fields evolving the exterior region of a KN-dS black hole.
\end{itemize}

In order to simplify the later analysis and encode the fact that the event and cosmological horizons are asymptotic regions in the Boyer-Linquist coordinates, we introduce a new radial variable, the tortoise or Regge-Wheeler type coordinate $x$, by the requirement
$$
  \frac{dx}{dr} = \frac{r^2 + a^2}{\Delta_r}.
$$
By integration, we find that
\begin{equation} \label{RW}
  x = \frac{1}{2\kappa_-} \ln(r-r_-) + \frac{1}{2\kappa_+} \ln(r_+ - r) + \frac{1}{2\kappa_c} \ln(r-r_c) + \frac{1}{2\kappa_n} \ln(r-r_n) + \,c \ \in \R,
\end{equation}
where the $\kappa_j$'s are the quantities given by
$$
  \kappa_j = \frac{\Delta'_r(r_j)}{2(r_j^2 + a^2)}, \ j = -,+,c,n,
$$
and $c$ is any constant of integration. The constants $\kappa_\pm$ are called the surface gravities of the event and cosmological horizons and will be of importance later in this paper. Also, the constant of integration $c$ will play an important role in the definition of the scattering matrix at fixed energy (see Subsection \ref{Direct-Stat-Scat}).

In the new coordinates sytem $(t,x,\theta, \varphi)$, we notice that the event and cosmological horizons are then pushed away to $\{x = -\infty\}$ and $\{ x = +\infty\}$ respectively. From (\ref{NullVectors}), we also note that the principal null directions are then expressed as
\begin{equation} \label{NullVectorsX}
  V^\pm = \frac{r^2 + a^2}{\Delta_r} \left( \partial_t \pm \partial_x + \frac{a E}{r^2 + a^2} \partial_\varphi \right).
\end{equation}
Hence, in the $t - x$ plane, the principal null geodesics are simply the straightlines $x = \pm t + c_0$ with $c_0$ any constant, mimicking in this plane the Minkowski spacetime. The addition of the $\partial_\varphi$ term in (\ref{NullVectorsX}) shows that the principal geodesics spin around the event and cosmological horizons when they get closed to them. Note that the speed of rotation of the principal geodesics is given by $\Omega_- = \frac{aE}{r_-^2 + a^2}$ at the event horizon whereas it is $\Omega_+ = \frac{aE}{r_+^2 + a^2}$ at the cosmological horizon.

Summarising, we shall then work on the $4$-dimensional manifold $\B = \R_t \times \R_x \times \S^2_{\theta, \varphi}$ equipped with the Lorentzian metric
$$
  g = \frac{\Delta_r}{\rho^2} \Big[ dt - \frac{a\sin^2 \theta}{E} \, d\varphi \Big]^{2} - \frac{\rho^2 \Delta_r}{(r^2 + a^2)^2} dx^2 - \frac{\rho^2}{\Delta_\theta} d\theta^{2} - \frac{\Delta_\theta \sin^2 \theta}{\rho^2} \Big[ a \,dt - \frac{r^2 + a^2}{E} d\varphi \Big]^{2}.
$$
The manifold $\B$ is a "stationary" axisymmetric globally hyperbolic spacetime having two asymptotic regions: the event $\{x = -\infty\}$ and cosmological $\{x = +\infty\}$ horizons. These horizons have particular geometries of asymptotically hyperbolic type when restricted to the spacelike hypersurfaces $\Sigma_t$ (see Remark \ref{AsympHyp} below). We shall now study how massless charged Dirac fields evolve in this spacetime and analyse precisely their scattering properties.


\Section{The massless charged Dirac equation}


\subsection{Hamiltonian formulation} \label{Hamilton-Form}

We use the form of the massless charged Dirac equation in its Hamiltonian formulation obtained in \cite{BC1}, formulae (4.1)-(4.7). This form of the equation is well suited to understand and use the separation of variables mentionned in the introduction. Note also that these authors calculated an expression for the "massive" charged Dirac equation on a KN-dS background. In this paper however, we consider the case of "massless" charged Dirac fields only to avoid additional technical difficulties in the later analysis.

Let us introduce some notations. The matrices $\Ga, \Gb, \Gc$ denote the usual $2 \times 2$ Dirac matrices that satisfy the anticommutation relations
\begin{equation} \label{AntiCom}
  \Gamma^i \Gamma^j + \Gamma^j \Gamma^i = 2 \delta_{ij} I_2, \quad \forall i,j=1,2,3,
\end{equation}
where $\delta_{ij}$ stands for the Kronecker symbol. More specifically, we shall work with the following representations of the Dirac matrices
\begin{equation} \label{DiracMatrices}
  \Ga = \left( \begin{array}{cc} 1&0 \\0&-1 \end{array} \right), \quad \Gb = \left( \begin{array}{cc} 0&1 \\1&0 \end{array} \right), \quad \Gc = \left( \begin{array}{cc} 0&i \\-i&0 \end{array} \right).
\end{equation}
Note that this is not the representation used in \cite{BC1}. This leads simply to an expression of the Dirac equation that is unitarily equivalent to the one in \cite{BC1} and particularly convenient for our purpose. We shall also denote by $D_x, D_\theta, D_\varphi$ the differential operators $-i\partial_x, -i \partial_\theta, -i\partial_\varphi$.

The massless charged Dirac fields are represented by $2$-components spinors with finite "energy" belonging to
$$
  L^2\left(\R\times \S^2, \frac{\sin\theta}{\sqrt{\Delta_\theta}} dx d\theta d\varphi; \C^2 \right)
$$
and the Hamiltonian form of the Dirac equation reads (see \cite{BC1}).
$$
  i \partial_t \phi =  \mathbb{D} \phi,
$$
with
\begin{equation} \label{Op-D}
  \mathbb{D} = \left[ \Big( 1 - \frac{\Delta_r}{\Delta_\theta} \frac{a^2 \sin^2\theta}{(r^2 + a^2)^2} \Big)^{-1} \Big( I_2 - \frac{\sqrt{\Delta_r}}{\sqrt{\Delta_\theta}} \frac{a \sin\theta}{r^2 + a^2} \Gc \big) \right] \mathbb{D}_0,
\end{equation}
and
\begin{eqnarray}
  \mathbb{D}_0 & = & \Ga D_x  + \frac{\sqrt{\Delta_r}}{r^2 + a^2} \left[ \sqrt{\Delta_\theta} \left( \Gb \Big( D_\theta -i\frac{\cot \theta}{2}\Big)  + \Gc \frac{E}{\Delta_\theta \sin\theta} D_\varphi \right) \right] \nonumber \\
               &   & \hspace{2cm} + \left[ \frac{aE}{r^2 + a^2} D_\varphi + \frac{qQr}{r^2 + a^2} \right]. \label{Op-D0}
\end{eqnarray}
Here the parameter $q$ is the electric charge of the Dirac fields.

Let us introduce further notations in order to obtain a synthetic form of the above Hamiltonian. We introduce the scalar (differential) potentials
\begin{equation} \label{Pot-a-c}
  a(x) = \frac{\sqrt{\Delta_r}}{r^2 + a^2}, \quad c(x, D_\varphi) = \frac{aE}{r^2 + a^2} D_\varphi + \frac{qQr}{r^2 + a^2},
\end{equation}
and the matrix-valued multiplication operator
\begin{equation} \label{Pot-J}
  J = I_2 + \alpha(x,\theta) \Gc, \quad \alpha(x,\theta) = \frac{\sqrt{\Delta_r}}{\sqrt{\Delta_\theta}} \frac{a \sin\theta}{r^2 + a^2} \Gc.
\end{equation}
It is shown in \cite{BC1}, (4.16) that $\eta = \ds\sup_{x, \theta} \alpha(x,\theta) < 1$. Since $\| \Gc \|_\infty = 1$, we conclude that the operator $J$ is invertible and a short calculation shows that
\begin{equation}\label{Jmoinsun}
  J^{-1} = (1 - \alpha^2) (I_2 - \alpha \Gc) = \Big( 1 - \frac{\Delta_r}{\Delta_\theta} \frac{a^2 \sin^2\theta}{(r^2 + a^2)^2} \Big)^{-1} \Big( I_2 - \frac{\sqrt{\Delta_r}}{\sqrt{\Delta_\theta}} \frac{a \sin\theta}{r^2 + a^2} \Gc \big).
\end{equation}
Hence $J^{-1}$ is precisely the term in front of $\mathbb{D}_0$ in (\ref{Op-D}) and the Hamiltonians $\mathbb{D}, \mathbb{D}_0$ can be written as
$$
  \mathbb{D} = J^{-1} \mathbb{D}_0, \quad \mathbb{D}_0 = \Ga D_x + a(x) \left[ \sqrt{\Delta_\theta} \left( \Gb \Big( D_\theta -i\frac{\cot \theta}{2}\Big)  + \Gc \frac{E}{\Delta_\theta \sin\theta} D_\varphi \right) \right] + c(x,D_\varphi).
$$

Let us continue our simplifications. In order to work in a Hilbert space that does not depend explicitly on the parameters of the black hole (recall that we want to identify the scattering matrices associated to two a priori different black holes in the later inverse problem), we consider the weighted spinor
$$
  \psi = \frac{\sin\theta}{\Delta_\theta^{1/4}} \phi.
$$
Hence, the new spinor $\psi$ belongs to the Hilbert space
$$
  \H = L^2(\R \times \S^2, dx d\theta d\varphi; \, \C^2),
$$
and is easily shown to satisfy the evolution equation
\begin{equation} \label{DiracEq}
  i \partial_{t} \psi = H \psi,
\end{equation}
where the Hamiltonian $H$ is expressed as
\begin{equation} \label{Op-H}
  H = J^{-1} H_0, \qquad H_0 = \Ga D_x + a(x) H_{\S^2} + c(x, D_\varphi),
\end{equation}
Here, $H_{\S^2}$ denotes an angular Dirac operator on the $2$-sphere $\S^2$ which, in the spherical coordinates $(\theta,\varphi)$, takes the form
\begin{equation} \label{HS2}
  H_{\S^2} = \sqrt{\Delta_\theta} \left[ \Gb D_\theta + \Gb \frac{i\Lambda a^2 \sin(2\theta)}{12 \Delta_\theta}  + \frac{\Gc}{\sin{\theta}} D_\varphi + \Gc \frac{ \Lambda a^2 \sin(\theta)}{3 \Delta_\theta} D_\varphi \right].
\end{equation}

Let us make a few comments on the Dirac Hamiltonian $H$. From (\ref{Pot-J}) and (\ref{Op-H}), $H$ can be written as $H = H_0 + \alpha(x,\theta) H_0$. We shall show in (\ref{Asymp-J}) that $\sup_{\theta} |\alpha(x,\theta)|$ is exponentially decreasing as $x \to \pm \infty$. Hence $H$ can be viewed as a "small", non spherically symmetric perturbation of order $1$ of the Hamiltonian $H_0$ which in turn is composed of
\begin{itemize}
\item A differential operator $\Ga D_x + \frac{a E}{r^2 + a^2} D_\varphi$ which, in view of (\ref{NullVectorsX}) and (\ref{DiracMatrices}), simply corresponds to transport along the outgoing and incoming principal null geodesics.
\item An angular differential operator $H_{\S^2}$ weighted by a scalar potential $a(x)$ (depending on the radial variable only). Note that $H_{\S^2}$ is a slight perturbation of the usual Dirac operator on $\S^2$ whose expression in our spherical coordinates and with our choice of weight on the spinor is
\begin{equation} \label{DS2}
  \DS = \Gb D_\theta + \frac{\Gc}{\sin{\theta}} D_\varphi.
\end{equation}
\item A scalar perturbation $\frac{qQr}{r^2 + a^2}$ caused by the interaction between the electric charge of the black hole and that of the Dirac fields.
\end{itemize}

Let us emphasize that the lack of spherical symmetry of the black hole is encoded both by the operator $J$ which is a slight non spherically symmetric perturbation of the identity $I_2$ and the presence of an extra differential operator $D_\varphi$ in the expression (\ref{Op-H}) of $H_0$. However, we can use the "cylindrical" symmetry of the problem to simplify the equation. Decomposing the Hilbert space $\H$ onto the angular modes $\{e^{ik\varphi}\}_{k \in 1/2 + \Z}$, i.e.
\begin{equation} \label{H-k}
  \H = \oplus_{k \in 1/2 + \Z} \H_k, \qquad \H_k = L^2(\R \times (0,\pi), dx d\theta; \C^2)
\end{equation}
it is clear that the Hamiltonian $H$ lets invariant each $\H_k$ and we have
\begin{equation} \label{Hk}
  H_{|\H_k} := H^k = J^{-1} H_0^k, \qquad H_{0 |\H_k} := H_0^k = \Ga D_x + a(x) H_{\S^2}^k + c(x,k).
\end{equation}
Thus the differential operator $c(x,D_\varphi)$ in (\ref{Op-H}) becomes $c(x,k), \ k \in 1/2 + \Z$ and can be treated as a mere potential.

To understand the scattering properties of the Dirac fields, it is necessary to know the asymptotics of the different potentials entering in the expression of the Hamiltonian $H^k$. We first introduce the notations
\begin{equation} \label{Omega}
  \Omega_\pm(D_\varphi) = \frac{aE D_\varphi + qQr_\pm}{r^2_\pm + a^2}, \quad \Omega_\pm(k) = \frac{aEk + qQr_\pm}{r^2_\pm + a^2}.
\end{equation}
Using (\ref{RW}), (\ref{Pot-a-c}), (\ref{Pot-J}) and (\ref{Omega}), it is straightforward to show that the asymptotics of the scalar potentials $a(x), c(x,k)$ and the matrix-valued potential $J(x,\theta)$ are given by
\begin{eqnarray} \label{Asymp-a}
a(x) &=& a_{\pm} e^{\kappa_{\pm}x} + O( e^{3\kappa_{\pm}x}), \ x \rightarrow \pm \infty, \\
a'(x) &=& a_{\pm} \kappa_{\pm} e^{\kappa_{\pm}x} + O( e^{3\kappa_{\pm}x}), \ x \rightarrow \pm \infty,
\end{eqnarray}
and for $\forall k \in 1/2 + \Z$,
\begin{eqnarray} \label{Asymp-c}
c(x,k) &=& \Omega_{\pm}(k) + c_{\pm} e^{2\kappa_{\pm}x} + O( e^{4\kappa_{\pm}x}), \ x \rightarrow \pm \infty, \\
c'(x,k) &=& 2\kappa_{\pm} c_{\pm} e^{2\kappa_{\pm}x} + O( e^{4\kappa_{\pm}x}), \ x \rightarrow \pm \infty,
\end{eqnarray}
and
\begin{equation} \label{Asymp-J}
\sup_{\theta} \|J(x,.) - I_2\|_\infty  =  \sup_\theta |\alpha(x,\theta)| = j_{\pm}  e^{\kappa_{\pm}x} + O( e^{3\kappa_{\pm}x}), \ x \rightarrow \pm \infty,
\end{equation}
where the quantities $a_\pm, \ c_\pm, \ j_\pm$ are constants that only depend on the parameters of the black hole. We won't need an explicit expression for these constants in the next analysis.

\begin{remark} \label{AsympHyp}
  Since the surface gravities $\kappa_- > 0$ and $\kappa_+ < 0$, we see that the potentials $a(x)$, $c(x,k)$ and $\sup_{\theta} \|J-I_2\|_\infty (x)$ decay exponentially at both horizons $\{ x \to \pm \infty \}$. This exponential decay of the potentials reflect the geometry of asymptotically hyperbolic type near the horizons as explained in the introduction.
\end{remark}

We now summarize the spectral results on the Hamiltonians $(H, H_0)$ obtained in \cite{BC1, BC2}. We first recall that the spinors $\psi$ satisfying (\ref{DiracEq}) belong to the "energy" space $L^2(\R \times \S^2, dx d\theta d\varphi; \, \C^2)$. We shall write
$$
  \H = L^2(\R \times \S^2, dx d\theta d\varphi; \, \C^2),
$$
when this space is equipped with the usual $L^2$ inner product $(.,.)$, whereas we write
$$
  \G = L^2(\R \times \S^2, dx d\theta d\varphi; \, \C^2),
$$
when it is equipped with the \emph{modified} inner product $<.,.> = (.,J.)$. Then $\H$ and $\G$ are Hilbert spaces with equivalent norms since $J$ is a bounded invertible matrix-valued potential. We have \cite{BC1}

\begin{theorem} \label{Spectra-H-H0}
The Hamiltonians $H_0$ and $H$ are selfadjoint on $\H$ and $\G$ respectively on their natural domains
$$
  D(H) = D(H_0) = \{ \psi \in L^2(\R \times \S^2, dx d\theta d\varphi; \, \C^2), \ H_0 \psi \in L^2(\R \times \S^2, dx d\theta d\varphi; \, \C^2) \}.
$$
Moreover
\begin{equation} \label{AC-Spectrum}
  \sigma(H_0) = \sigma_{ac}(H_0) = \R, \qquad \sigma(H) = \sigma_{ac}(H) = \R.
\end{equation}
In particular
\begin{equation} \label{PP-Spectrum}
  \sigma_{pp}(H_0) = \sigma_{pp}(H) = \emptyset.
\end{equation}
\end{theorem}
The main point in Thm \ref{Spectra-H-H0} is the absence of pure point spectrum and singular continuous spectrum for the Hamiltonians $H_0$ and $H$. It is worth mentioning that the absence of pure point spectrum is a consequence of the separability of the Dirac equation into systems of radial and angular ODEs. Note that the separation of variables is also used in \cite{BC1,BC2} to prove the absence of singular continuous spectrum by relying on a standard decomposition method for the absolutely continuous spectrum as exposed in \cite{We}. In this paper, we shall give an alternative proof of the absence of singular continuous spectrum for $H, H_0$ by establishing a Limiting Absorption Principle (LAP) for these operators by means of a Mourre theory (see Appendix \ref{LAP-Mourre}). These LAP will entail in turn a complete scattering theory for theses Hamiltonians (see Section \ref{Direct-scat}).

%
%

\subsection{Separation of variables} \label{Sep-Var}

Let us pause a moment (before studying the direct and inverse scattering theory for the Hamiltonian $H$) and study in much more details the separation of variables procedure for the Dirac equation (\ref{DiracEq}) since it will be in the heart of our analysis. We first introduce an additional notation. Looking at (\ref{Pot-J}), we can write
\begin{equation} \label{J}
  J = I_2 + \alpha(x,\theta) \Gc, \quad \alpha(x,\theta) = a(x) b(\theta),
\end{equation}
where
\begin{equation} \label{Pot-a-b}
   a(x) = \frac{\sqrt{\Delta_r}}{r^2 + a^2}, \quad b(\theta) = \frac{a \sin \theta}{\sqrt{\Delta_\theta}}.
\end{equation}
Then, for $\lambda \in \R$, we consider the "stationary" Dirac equation
\begin{equation} \label{StationaryEq}
  H \psi = \lambda \psi.
\end{equation}
Using (\ref{Op-H}) and (\ref{J}), we can re-write (\ref{StationaryEq}) as
\begin{eqnarray}
  H \psi = \lambda \psi & \Longleftrightarrow & J^{-1} (H_0 - \lambda J) \psi = 0, \nonumber \\
  & \Longleftrightarrow & \left[ \Ga D_x + a(x) H_{\S^2} + c(x,D_\varphi) - \lambda ( I_2 + a(x) b(\theta) \Gc) \right] \psi = 0, \nonumber \\
  & \Longleftrightarrow & \left[ \Ga D_x + c(x,D_\varphi) - \lambda + \ a(x) (H_{\S^2} - \lambda b(\theta) \Gc) \right] \psi =0. \label{St-1}
\end{eqnarray}
We denote by $H(\lambda)$ the above stationary operator, \textit{i.e.}
\begin{equation} \label{Hlambda}
  H(\lambda) = \Ga D_x + c(x,D_\varphi) - \lambda + \ a(x) (H_{\S^2} - \lambda b(\theta) \Gc).
\end{equation}
The stationary equation $H \psi = \lambda \psi$ is thus equivalent to $H(\lambda) \psi = 0$.

To proceed further, we need to study the angular part of $H(\lambda)$. We introduce the notations
\begin{equation} \label{AS2}
A_{\S^2}(\lambda) = H_{\S^2} - \lambda b(\theta) \Gc
\end{equation}
and
$$
  \H_{\S^2} = L^2(\S^2, d\theta d\varphi; \C^2).
$$
We are interested in studying the angular "eigenvalues" equation on $\H_{\S^2}$
\begin{equation} \label{Angular-Eq}
  A_{\S^2}(\lambda) \psi_{kl} = \mu_{kl}(\lambda) \psi_{kl}.
\end{equation}
Following \cite{BC1,BC2}, we can prove

\begin{theorem} \label{Spectrum-A}
For all $\lambda \in \R$, the operator $A_{\S^2}(\lambda)$ is selfadjoint on $\H_{\S^2}$ and has pure point spectrum. More precisely, for all $k \in \frac{1}{2} + \Z$ and $l \in \Z^*$, there exists a sequence of eigenvalues $\mu_{kl}(\lambda) \in \R$  of $A_{\S^2}(\lambda)$ and associated normalized eigenfunctions $Y_{kl}(\lambda) \in \H_{\S^2}$ such that \\

$(1) \quad \H_{\S^2} = \oplus_{(k,l) \in I} \ Span(Y_{kl}(\lambda)), \qquad I = (\frac{1}{2} + \Z) \times \Z^*$, \\

$(2) \quad A_{\S^2}(\lambda) Y_{kl}(\lambda) = \mu_{kl}(\lambda) Y_{kl}(\lambda)$, \\

$(3) \quad D_\varphi Y_{kl}(\lambda) = k Y_{kl}(\lambda)$.
\end{theorem}

Hence, the operators $D_\varphi$ and $A_{\S^2}(\lambda)$ possess a common basis of eigenfunctions $Y_{kl}(\lambda)$ that we can use for the separation of variables. Before this and for later use, let us make a few comments on the above result. To prove Thm \ref{Spectrum-A}, Belgiorno and Cacciatori in \cite{BC1,BC2} first decompose $\H_{\S^2}$ onto the angular modes $\{e^{i k \varphi}\}$ for all $k \in 1/2 + \Z$, that are the eigenfunctions (with associated eigenvalues $k$) for the selfadjoint operator $D_\varphi$ with anti-periodic boundary conditions. Then
$$
  \H_{\S^2} = \oplus_{k \in 1/2 + \Z} \H_{\S^2}^k, \qquad \H_{\S^2}^k = L^2((0,\pi), d\theta; \C^2) := \L,
$$
and clearly, the reduced subspaces $\H_{\S^2}^k$ remain invariant under the action of $A_{\S^2}(\lambda)$. For each $k \in 1/2 + \Z$, we denote
\begin{equation} \label{Aklambda}
  A_k(\lambda) = A_{\S^2}(\lambda)_{|\H_{\S^2}^k},
\end{equation}
or more explicitly
\begin{equation} \label{Ak}
  A_k(\lambda) = \sqrt{\Delta_\theta} \left[ \Gb D_\theta + \Gb \frac{i\Lambda a^2 \sin(2\theta)}{12 \Delta_\theta}  + \frac{k \Gc}{\sin{\theta}} + \Gc \frac{ \Lambda a^2 k \sin(\theta)}{3 \Delta_\theta} - \lambda \frac{a \sin \theta}{\Delta_\theta} \Gc \right].
\end{equation}
The study of the selfadjoint operator $A_k(\lambda)$ on $\L$ is the object of the Sections 5.1.1 and 6.1 in \cite{BC2}. Using an appropriate form for the eigenvalue equation
\begin{equation} \label{Angular-Eq-k}
  A_k(\lambda) u(\theta) = \mu_k(\lambda) \, u(\theta),
\end{equation}
and an analogue of Pr\"ufer transformation for Dirac system (see \cite{We}), they are able to show that the operator $A_k(\lambda)$ has "discrete simple" spectrum. For each $l \in \Z^*$, we denote by $\mu_{kl}(\lambda)$ the simple eigenvalues of $A_k(\lambda)$ and $u_{kl}^\lambda(\theta)$ the corresponding eigenfunctions. Note that, since the spectrum of $A_k(\lambda)$ is discrete, it has no accumulation point and thus
$$
  \forall k \in 1/2 + \Z, \quad |\mu_{kl}(\lambda)| \to +\infty, \quad \textrm{as} \ l \to \infty.
$$
Finally, the eigenfunctions $Y_{kl}(\lambda)$ of $A_{\S^2}(\lambda)$ appearing in Thm \ref{Spectrum-A} are thus given by
\begin{equation} \label{Ykl}
  \forall k \in 1/2 + \Z, \ l \in \Z^*, \quad Y_{kl}(\lambda)(\theta, \varphi) = u_{kl}^\lambda(\theta) e^{ik\varphi}.
\end{equation}

The choice of labeling the eigenvalues and eigenfunctions of $A_k(\lambda)$ with an index $l \in \Z^*$ is convenient for the following reason. Observe first that it follows from the anticommutation relation (\ref{AntiCom}) that the Dirac matrix $\Ga$ \emph{anticommutes} with the operators $A_k(\lambda)$ (and thus with the operator $A_{\S^2}(\lambda)$). Hence for each admissible $(k,l) \in I = (\frac{1}{2} + \Z) \times \Z^*$, we have
$$
  A_k(\lambda) (\Ga Y_{kl}(\lambda)) = - \Ga A_k(\lambda) Y_{kl}(\lambda) = - \mu_{kl}(\lambda) \Ga Y_{kl}(\lambda),
$$
from which it follows that to each positive eigenvalue $\mu_{kl}(\lambda)$ and associated eigenfunction $Y_{kl}(\lambda)$, there is a corresponding negative eigenvalue $-\mu_{kl}(\lambda)$ and associated eigenfunction $\Ga Y_{kl}(\lambda)$. From this, for each $k \in 1/2 + \Z$, we decide to label the positive eigenvalues of $A_k(\lambda)$ (in increasing order) by $\mu_{kl}(\lambda)$ with $l \in \N^*$, \textit{i.e.}
$$
  0 < \mu_{k1}(\lambda) < \mu_{k2}(\lambda) < \dots
$$
Conversely, the negative eigenvalues are thus labelled by the negative values of $l$ by the requirement
\begin{equation} \label{Index-l}
  \forall l \in \N^*, \quad \mu_{k,-l}(\lambda) = -\mu_{kl}(\lambda), \qquad Y_{k,-l}(\lambda) = \Ga Y_{kl}(\lambda).
\end{equation}

Let us come back to the separation of variables. Our aim is to decompose the Hilbert space $\H$ in an Hilbert sum of reduced Hilbert spaces that remain invariant under the action of the operator $H(\lambda)$ given by (\ref{Hlambda}) and allows separation of variables. We proceed as follows.

According to Thm \ref{Spectrum-A}, first note that
\begin{eqnarray*}
  \H & = & L^2(\R \times \S^2, dx d\theta d\varphi; \C^2) =  L^2(\R, dx) \otimes L^2(\S^2, d\theta d\varphi ;\C^2), \\
     & = & \bigoplus_{(k,l) \in (1/2 + \Z) \times \Z^*} \Big( L^2(\R, dx) \otimes Y_{kl}(\lambda) \Big).
\end{eqnarray*}
The reduced Hilbert subspaces $L^2(\R, dx) \otimes Y_{kl}(\lambda)$ - that clearly remain invariant under the action of the angular operator $A_{\S^2}(\lambda)$ - does not remain invariant however under the radial part $\Ga D_x$ of the operator $H(\lambda)$ since $\Ga$ sends $Y_{kl}(\lambda)$ onto $\Ga Y_{kl}(\lambda) = Y_{k,-l}(\lambda)$ by our convention (\ref{Index-l}). Instead we consider the decomposition
$$
  \H = \bigoplus_{(k,l) \in (1/2 + \Z) \times \N^*} \Big( L^2(\R,dx) \otimes Span(Y_{kl}(\lambda), Y_{k,-l}(\lambda)) \Big).
$$

We claim that the reduced Hilbert spaces $L^2(\R,dx) \otimes Span(Y_{kl}(\lambda), Y_{k,-l}(\lambda))$ remain invariant under the action of the operator (\ref{Hlambda}) and allows separation of variables. To see this, we first remark that the reduced subspaces $L^2(\R,dx) \otimes Span(Y_{kl}(\lambda), Y_{k,-l}(\lambda))$ for $(k,l) \in (1/2 + \Z) \times \N^*$ can be identified with the tensorial products
$$
  L^2(\R,dx;\C^2) \otimes Y_{kl}(\lambda).
$$
Indeed, using our convention (\ref{Index-l}) and the particular form of $\Ga$ (see (\ref{DiracMatrices})), we have for all $\psi \in L^2(\R,dx) \otimes Span(Y_{kl}(\lambda), Y_{k,-l}(\lambda))$,
$$
  \psi = \psi_1(x) Y_{kl}(\lambda) + \psi_2(x) Y_{k,-l}(\lambda) = \big( \psi_1(x) I_2 + \psi_2(x) \Ga \big) Y_{kl}(\lambda),
$$
which can be written in components
\begin{equation} \label{1to2}
  \psi = \left( \begin{array}{c} \psi_1(x) + \psi_2(x) \\ \psi_1(x) - \psi_2(x) \end{array} \right) \otimes Y_{kl}(\lambda).
\end{equation}
Hence $\psi \in L^2(\R,dx;\C^2) \otimes Y_{kl}(\lambda)$.

Conversely, for all $\psi \in L^2(\R,dx;\C^2) \otimes Y_{kl}(\lambda)$, we have
\begin{equation} \label{2to1}
  \psi = \left( \begin{array}{c} \psi_1(x) \\ \psi_2(x) \end{array} \right) \otimes Y_{kl}(\lambda) = \frac{1}{2} \big( \psi_1(x) + \psi_2(x) \big) Y_{kl}(\lambda) + \frac{1}{2} \big( \psi_1(x) - \psi_2(x) \big) Y_{k,-l}(\lambda),
\end{equation}
and thus $\psi \in L^2(\R,dx) \otimes Span(Y_{kl}(\lambda), Y_{k,-l}(\lambda))$.

In fact, we easily see that the map
$$
  \begin{array}{rcl} \Theta: \ L^2(\R,dx) \otimes Span(Y_{kl}(\lambda), Y_{k,-l}(\lambda)) & \longrightarrow & L^2(\R ; \C^2) \otimes Y_{kl}(\lambda), \\
                     \psi_1(x) Y_{kl}(\lambda) + \psi_2(x) Y_{k,-l}(\lambda)     & \longrightarrow & \frac{1}{\sqrt{2}} \left( \begin{array}{c} \psi_1(x) + \psi_2(x) \\ \psi_1(x) - \psi_2(x) \end{array} \right) \otimes Y_{kl}(\lambda), \end{array}
$$
is unitary. For ease of notations, we shall work with the reduced Hilbert spaces
\begin{equation} \label{Hkl}
  \H_{kl}(\lambda) := L^2(\R ; \C^2) \otimes Y_{kl}(\lambda) \ \simeq \ L^2(\R,\C^2),
\end{equation}
for each $(k,l) \in I := (1/2 + \Z) \times \N^*$. We then have the following decomposition for the full Hilbert space $\H$
\begin{equation} \label{HilbertDecomposition}
  \H = \bigoplus_{(k,l) \in (1/2 + \Z) \times \N^*} \H_{kl}(\lambda).
\end{equation}

Now it is a direct calculation to show that the subspaces $\H_{kl}(\lambda)$ remain invariant under the action of the stationary operator (\ref{Hlambda}). For instance, if
$$
  \H_{kl}(\lambda) \ni \psi = \psi_{kl} \otimes Y_{kl}(\lambda) = \left( \begin{array}{c} \psi_1(x) \\ \psi_2(x) \end{array} \right) \otimes Y_{kl}(\lambda),
$$
then by (\ref{2to1})
$$
  \psi = \frac{1}{2} \big( \psi_1(x) + \psi_2(x) \big) Y_{kl}(\lambda) + \frac{1}{2} \big( \psi_1(x) - \psi_2(x) \big) Y_{k,-l}(\lambda),
$$
and thus
\begin{eqnarray*}
  \Ga \psi & = & \frac{1}{2} \big( \psi_1(x) + \psi_2(x) \big) \Ga Y_{kl}(\lambda) + \frac{1}{2} \big( \psi_1(x) - \psi_2(x) \big) Y_{k,l}(\lambda), \ \textrm{by (\ref{Index-l}) and} \ (\Ga)^2 = I_2, \\
                & = & \left( \begin{array}{c} \psi_1(x) \\ -\psi_2(x) \end{array} \right) \otimes Y_{kl}(\lambda), \ \textrm{by (\ref{1to2})}, \\
                & = & (\Ga \psi_{kl}) \otimes Y_{kl}(\lambda) \ \in \H_{kl}(\lambda), \ \textrm{by (\ref{DiracMatrices})}.
\end{eqnarray*}
Hence, the operator $\Ga D_x$ appearing in (\ref{Hlambda}) lets invariant $\H_{kl}(\lambda)$ and its action on $\psi = \psi_{kl} \otimes Y_{kl}(\lambda) \in \H_{kl}(\lambda)$ is given by
\begin{equation} \label{Gamma1-Ykl}
  \Ga D_x \psi = (\Ga D_x \psi_{kl}) \otimes Y_{kl}(\lambda).
\end{equation}

Similarly, we have
\begin{eqnarray*}
  A_{\S^2}(\lambda) \psi & = & \frac{1}{2} \big( \psi_1(x) + \psi_2(x) \big) \muk Y_{kl}(\lambda) - \frac{1}{2} \big( \psi_1(x) - \psi_2(x) \big) \muk Y_{k,-l}(\lambda), \ \textrm{by Thm \ref{Spectrum-A}}, \\
                & = & \muk \left( \begin{array}{c} \psi_2(x) \\ \psi_1(x) \end{array} \right) \otimes Y_{kl}(\lambda), \ \textrm{by (\ref{1to2})}, \\
                & = & (\muk \Gb \psi_{kl}) \otimes Y_{kl}(\lambda) \ \in \H_{kl}(\lambda), \ \textrm{by (\ref{DiracMatrices})}.
\end{eqnarray*}
Hence, the angular operator $A_{\S^2}(\lambda)$ appearing in (\ref{Hlambda}) lets invariant $\H_{kl}(\lambda)$ and its action on $\psi = \psi_{kl} \otimes Y_{kl}(\lambda) \in \H_{kl}(\lambda)$ is given by
\begin{equation} \label{AS2-Ykl}
  A_{\S^2}(\lambda) \psi = (\muk \Gb \psi_{kl}) \otimes Y_{kl}(\lambda).
\end{equation}

Summarising  the whole stationary operator $H(\lambda)$ lets invariant the $\H_{kl}(\lambda)$'s and denoting $I = (1/2 + \Z) \times \N^*$, we have for all
$$
  \psi(x,\theta,\varphi) = \ds\sum_{(k,l) \in I} \psi_{kl}(x) \otimes Y_{kl}(\theta,\varphi)  \in \H,
$$
\begin{eqnarray}
  H(\lambda) \psi = 0 & \Longleftrightarrow & \left[ \Ga D_x + c(x,D_\varphi) - \lambda + \ a(x) A_{\S^2}(\lambda) \right] \psi =0, \nonumber \\
  & \Longleftrightarrow & \sum_{(k,l) \in I} \left[ \big( \Ga D_x + c(x,k) - \lambda + \mu_{kl}(\lambda) a(x) \Gb \big) \psi_{kl}(x) \right] \otimes Y_{kl} = 0, \label{St-2} \\
  & \Longleftrightarrow & \forall (k,l) \in I, \quad \left[ \Ga D_x + \mu_{kl}(\lambda) a(x) \Gb + c(x,k) \right] \psi_{kl}(x) = \lambda \, \psi_{kl}(x). \nonumber
\end{eqnarray}
In consequence, after decomposition onto the reduced subspaces $\H_{kl}(\lambda)$, the stationary equation $H(\lambda) \psi = 0$ is equivalent to a countable family of radial stationary equations given by
\begin{equation} \label{Radial-Eq}
  H_{kl}(\lambda) \psi_{kl}(x) = \lambda \, \psi_{kl}(x),
\end{equation}
where the one-dimensional operators $H_{kl}(\lambda)$ are given by
\begin{equation} \label{Radial}
  H_{kl}(\lambda) = \Ga D_x + \mu_{kl}(\lambda) a(x) \Gb + c(x,k),
\end{equation}
and are parametrized by the energy $\lambda \in \R$ of the system and the angular momenta $\muk$ for each $(k,l) \in I$.

This is the separation of variables for the massless Dirac equation in KN-dS black holes. In order to have a synthetic view of the procedure, we summarize these results into a theorem.
\begin{theorem}[Separation of variables] \label{Separation}
  Denote by $I$ the set $(1/2 + \Z) \times \N^*$. Then, for all $\lambda \in \R$, there exists a Hilbert decomposition of the energy space $\H$ as
  $$
    \H = \bigoplus_{(k,l) \in I} \H_{kl}(\lambda),
  $$
  where
  $$
    \H_{kl}(\lambda) = L^2(\R ; \C^2) \otimes Y_{kl}(\lambda) \ \simeq \ L^2(\R,\C^2),
  $$
  and the $Y_{kl}(\lambda)$ are the eigenfunctions of the angular operator (defined in (\ref{AS2})) by
  $$
    A_{\S^2}(\lambda) = H_{\S^2} - \lambda b(\theta) \Gc
  $$
  and associated to its positive eigenvalues $\mu_{kl}(\lambda)$. Moreover, the reduced subspaces $\H_{kl}(\lambda)$ remain invariant under the action  of the stationary Dirac Hamiltonian given in (\ref{St-1}) by
  $$
    H(\lambda) = \Ga D_x + c(x,D_\varphi) - \lambda + \ a(x) (H_{\S^2} - \lambda b(\theta) \Gc)
  $$
   and for all $\psi = \sum_{(k,l) \in I} \psi_{kl}(x) \otimes Y_{kl}(\lambda)$, the stationary equation $H(\lambda) \psi = 0$ is equivalent to the countable family of one-dimensional (radial) stationary equations
  $$
  H_{kl}(\lambda) \psi_{kl}(x) = \lambda \, \psi_{kl}(\lambda),
  $$
  where
  $$
    H_{kl}(\lambda) = \Ga D_x + \mu_{kl}(\lambda) a(x) \Gb + c(x,k).
  $$
\end{theorem}

\subsection{Estimates on the eigenvalues of the angular operator} \label{Est-Eigen}

We end this section with a brief study of the angular operator $A_{\S^2}(\lambda)$ and more specifically of its eigenvalues $\mu_{kl}(\lambda)$. In order to apply the Complex Angular Momentum method (see Section \ref{Complexification}), we are interested in obtaining a rough estimate on the distribution of the eigenvalues $\mu_{kl}(\lambda)$ when $\lambda \in \R$ and $k \in 1/2 + \Z$ are fixed et for large $l \in \N^*$. Precisely, we want to prove that
$$
  \sum_{l \in \N^*} \frac{1}{\mu_{kl}(\lambda)} = +\infty.
$$
This result is a consequence of the following estimate on the growth of the eigenvalues $\muk$.

\begin{prop}[Estimate on $\muk$] \label{Growth-mukl}
For all $\lambda \in \R$, for all $k \in \frac{1}{2} + \Z$ and for all $l \in \N^*$, there exist constants $C_1$ and $C_2$ independent of $k,l$ such that
$$
  \left( 2 - e^{\frac{1}{26}} \right) \left(|k| - \frac{1}{2} + l\right) - C_1 |k| - C_2 - |a \lambda| \leq \mu_{kl}(\lambda) \leq e^{\frac{1}{26}} \left(|k| - \frac{1}{2} + l\right) + C_1 |k| + C_2 + |a \lambda|.
$$
We thus conclude that for fixed $\lambda \in \R$ and $k \in \frac{1}{2} + \Z$,
$$
  \sum_{l \in \N^*} \frac{1}{\mu_{kl}(\lambda)} = + \infty.
$$
\end{prop}
\begin{proof}
  The proof of this Proposition follows from the theory of analytic perturbation due to Kato \cite{Ka} and is given in Appendix \ref{Estimate-mukl}.
\end{proof}

%
%
%
%

\Section{The direct scattering problem} \label{Direct-scat}

In this section, we formulate the direct scattering theory associated to massless Dirac fields evolving in the exterior region of a KN-dS black hole. In a first part, a stationary expression of the scattering matrix at fixed energy is given in terms of particular solutions of the separated equations obtained in Section \ref{Sep-Var} and the physical relevance of its components are explained. We emphasize that the different behaviors of the Dirac Hamiltonian $H$ (given by (\ref{Op-H})) at the horizons $\{x = \pm \infty \}$ as well as the presence of the long-range potential $c(x,D_\varphi)$ (given by (\ref{Pot-a-c})) make the definition of the stationary scattering matrix quite involved and non canonical. Another important feature of this stationary scattering matrix is its dependence on the choice of Regge-Wheeler coordinate $x$ which, as mentioned in the introduction, is defined up to a constant of integration only. We thus give the expression of the whole one-parameter family of stationary scattering matrices that describe the same black hole and decide to identify them in the statement of the inverse problem in Section \ref{MainResult}. In a second part, we make the link between the stationary expression of the scattering matrix and its time-dependent counterpart that turns out to be much more intrinsically and naturally defined. This helps us to better understand for instance the dependence of the stationary scattering matrix on the choice of the constant of integration in the definition of the Regge-Wheeler variable $x$. At last, the existence of the time-dependent scattering matrix is proved essentially by means of a Mourre theory (similar to \cite{Da2, HaN}) and its stationary expression is then obtained by a somewhat standard procedure as presented in \cite{Y}.

%
%

\subsection{Stationary expression of the scattering matrix} \label{Direct-Stat-Scat}

Recall that we work on the Hilbert space $\H = L^2(\R \times \S^2, dx d\theta d\varphi, \C^2)$ and that the evolution of our Dirac fields $\psi \in \H$ is generated by the Hamiltonian (see Section \ref{Hamilton-Form} and more precisely (\ref{Pot-a-c}), (\ref{Pot-J}) and (\ref{Op-H}) for the notations)
$$
  H = J^{-1} H_0, \quad H_0 = \Ga D_x + a(x) H_{\S^2} + c(x,D_\varphi),
$$
where
$$
  J^{-1} = (1 - \alpha^2) (I_2 - \alpha \Gc), \quad \alpha(x,\theta) = a(x) b(\theta) = \frac{\sqrt{\Delta_r}}{r^2 + a^2} {\frac{a \sin^2\theta}{\sqrt{\Delta_\theta}}}.
$$
The stationary scattering at energy $\lambda \in \R$ is thus governed by the stationary equation
\begin{equation} \label{Stat-Eq-Global}
  H \psi = \lambda \psi,
\end{equation}
which, as explained in Section \ref{Separation}, can be re-written as
$$
  H(\lambda) \psi = 0,
$$
where the Hamiltonian $H(\lambda)$ defined in (\ref{Hlambda}) is given by
$$
 H(\lambda) = \Ga D_x + c(x,D_\varphi) - \lambda + \ a(x) (H_{\S^2} - \lambda b(\theta) \Gc).
$$
In view of Theorem \ref{Separation}, we decompose the Hilbert space $\H$ onto a Hilbert sum of reduced Hilbert spaces $\H_{kl}(\lambda)$ each ones corresponding to a generalized spherical harmonics $Y_{kl}(\lambda)$, \textit{i.e.} the eigenfunctions of the angular operator $A_{\S^2}(\lambda)$. We recall that the indices $(k,l)$ run over the set $I = (\frac{1}{2} + \Z) \times \N^*$ and that the reduced Hilbert spaces $\H_{kl}(\lambda)$ are all isomorphic to $\h := L^2(\R; \C^2)$ and remain invariant under the action of the Hamiltonian $H(\lambda)$. Thus the global stationary equation (\ref{Stat-Eq-Global}) is equivalent to a countable family of one-dimensional stationary equations given by
\begin{equation} \label{Rad-0}
  H_{kl}(\lambda) \psi_{kl}(x) = \lambda \, \psi_{kl}(\lambda),
\end{equation}
where
\begin{equation} \label{Rad-1}
  H_{kl}(\lambda) = \Ga D_x + \mu_{kl}(\lambda) a(x) \Gb + c(x,k).
\end{equation}
Here the spinors $\psi_{kl}$ belong to $\h$ and $\mu_{kl}(\lambda)$ are the eigenvalues of the angular operator $A_{\S^2}(\lambda)$ obtained after separation of variables and studied in Section \ref{Sep-Var}, Thm \ref{Separation}. We recall at last that the potentials $a(x)$ and $c(x,k)$ for $k \in \frac{1}{2} + \Z$ satisfy the asymptotics (\ref{Asymp-a}) and (\ref{Asymp-c}) respectively, that is $a(x)$ is exponentially decreasing at both horizons $\{x = \pm \infty \}$ whereas $c(x,k)$ tends to different and nonzero constants $\Omega_\pm(k)$ given in (\ref{Omega}) at each horizon.

Hence, to contruct the scattering matrix associated to (\ref{Stat-Eq-Global}), it suffices to construct the family of reduced scattering matrices associated to the one-dimensional equations (\ref{Rad-0}) parametrized by the angular modes $k \in 1/2 + \Z$ and angular momenta $\mu_{kl}(\lambda)$. Then the global scattering matrix will be defined by summation of these reduced scattering matrices over the generalized spherical harmonics $Y_{kl}(\lambda)$.

\subsubsection{The simplified reduced stationary scattering matrix} \label{Simplified-SM}

In what follows, we shall work on a separated equation (\ref{Rad-0}) and thus we assume that the indices $(k,l) \in I = (\frac{1}{2} + \Z) \times \N^*$ are fixed and that we work on a generalized spherical harmonic $Y_{kl}(\lambda)$. A direct scattering theory - and more precisely the construction of the scattering matrix at fixed energy $\lambda$ - associated to the one-dimensional Dirac equation (\ref{Rad-0})-(\ref{Rad-1}) is almost standard and we would like to apply the nice framework given in \cite{AKM}. However, since the potential $c(x,k)$ is long-range and thus does not belong to $L^1(\R; \C^2)$, the equation (\ref{Rad-1}) does not enter exactly this framework. To remedy this situation, let us \emph{remove} the potential $c(x,k)$ by introducing the family of unitary transforms on $\h$
\begin{equation} \label{Unitary-Uk}
  \forall k \in 1/2 + \Z, \quad U_k = e^{-iC(x,k) \Ga},
\end{equation}
where
\begin{equation} \label{Pot-C}
  C(x,k) = \int_{-\infty}^x [c(s,k) - \Omega_-(k)] ds + \Omega_-(k) x + K,
\end{equation}
is a primitive of the potential $c(x,k)$ parametrized by the constant of integration $K$. Consider now the spinor $\phi_{kl} = U_k^{-1} \psi_{kl}$. The radial equation (\ref{Rad-0}) for $\phi_{kl}$ then becomes
\begin{equation} \label{Rad-2}
  \left[ \Ga D_x + \mu_{kl}(\lambda) a(x) \left( e^{iC(x,k)\Ga} \Gb e^{-iC(x,k) \Ga} \right) \right] \phi_{kl} = \lambda \phi_{kl}.
\end{equation}
We introduce the notation
$$
  V_k(x) = a(x) e^{iC(x,k)\Ga} \Gb e^{-iC(x,k) \Ga},
$$
which simplifies into
$$
  V_k(x) = a(x) e^{2iC(x,k)\Ga} \Gb = \left( \begin{array}{cc} 0&a(x) e^{2iC(x,k)}\\ a(x) e^{-2iC(x,k)}&0 \end{array} \right),
$$
thanks to (\ref{AntiCom}) and (\ref{DiracMatrices}). Hence the radial equation (\ref{Rad-2}) becomes
\begin{equation} \label{Rad-3}
  \left[ \Ga D_x + \mu_{kl}(\lambda) V_k(x) \right] \phi_{kl} = \lambda \phi_{kl}.
\end{equation}
Note that the potential $V_k(x)$ now decays exponentially at both horizons $\{x=\pm\infty\}$ by (\ref{Asymp-a}) and is thus a \emph{short-range} potential. We shall define the scattering matrices associated to (\ref{Rad-3}) for any real value of the eigenvalues $\mu_{kl}(\lambda)$. Precisely, let us set $z = -\mu_{kl}(\lambda) \in \R$ and $q_k(x) = a(x) e^{2iC(x,k)}$. We are led to consider the "general" one-dimensional stationary equation
\begin{equation} \label{Stat-Eq-Ak}
  \left[ \Ga D_x - z V_k(x) \right] \phi_{kl} = \lambda \phi_{kl},
\end{equation}
with
\begin{equation} \label{Pot-V}
V_k(x) = \left( \begin{array}{cc} 0&q_k(x) \\ \bar{q}_k(x) &0 \end{array} \right), \qquad q_k(x) = a(x) e^{2iC(x,k)}.
\end{equation}

The equation (\ref{Stat-Eq-Ak}) enters precisely the framework given in \cite{AKM} except for the additional coupling constant $z \in \R$. The equation (\ref{Stat-Eq-Ak}) is also quite similar to the one in our previous paper \cite{DN3} and we shall follow here the results obtained therein. The \emph{natural} scattering matrix associated to (\ref{Stat-Eq-Ak}) can be defined in terms of particular solutions called Jost functions. These are $2 \times 2$ matrix-valued solutions of (\ref{Stat-Eq-Ak}) with prescribed asymptotic behaviors at $x = \pm \infty$ given by
\begin{eqnarray}
  F_L(x,\lambda,k,z) & = & e^{i\Ga \lambda x} ( I_2 + o(1)), \ \ x \to +\infty, \label{FL} \\
  F_R(x,\lambda,k,z) & = & e^{i\Ga \lambda x} (I_2 + o(1)), \ \ x \to -\infty. \label{FR}
\end{eqnarray}
From (\ref{Stat-Eq-Ak}), (\ref{FL}) and (\ref{FR}), it is easy to see that such solutions (if there exist) must satisfy the integral equations
\begin{equation} \label{IE-FL}
  F_L(x,\lambda,k,z) = e^{i\Ga \lambda x} - i z \Ga \int_x^{+\infty} e^{-i\Ga \lambda (y-x)} V_k(y) F_L(y,\lambda,k,z) dy,
\end{equation}
\begin{equation} \label{IE-FR}
  F_R(x,\lambda,k,z) = e^{i\Ga \lambda x} + i z \Ga \int_{-\infty}^x e^{-i\Ga \lambda (y-x)} V_k(y) F_R(y,\lambda,k,z) dy.
\end{equation}
But, since the potential $V_k$ belongs to $L^1(\R)$, it follows that the integral equations (\ref{IE-FL}) and (\ref{IE-FR}) are uniquely solvable by iteration and using that $\|V_k(x)\| (:= \sup_{1 \leq i,j \leq 2} |(V_k)_{ij}(x)|) = a(x)$, the Jost functions are shown to satisfy the estimates
$$
  \|F_L(x,\lambda,k,z)\| \leq e^{|z| \int_x^{+\infty} a(s) ds}, \quad \|F_R(x,\lambda,k,z)\| \leq e^{|z| \int_{-\infty}^x a(s) ds}.
$$
Moreover we can prove
\begin{lemma} \label{DetFL}
  For $\lambda \in \R$ and $z \in \R$, either of the Jost solutions $F_L(x,\lambda,k,z)$ and $F_R(x,\lambda,k,z)$ forms a fundamental matrix of (\ref{Stat-Eq-Ak}) and has determinant equal to 1. Moreover, the following equalities hold
\begin{eqnarray}
  F_L(x,\lambda,k,z)^* \, \Ga \, F_L(x,\lambda,k,z) & = & \Ga, \label{RelationFL}\\
  F_R(x,\lambda,k,z)^* \, \Ga \, F_L(x,\lambda,k,z) & = & \Ga, \label{RelationFR}
\end{eqnarray}
where $^*$ denotes the matrix conjugate transpose.
\end{lemma}
\begin{proof}
See \cite{AKM}, Proposition 2.2.
\end{proof}

Since the Jost solutions are fundamental matrices of (\ref{Stat-Eq-Ak}), there exists a $2\times2$ matrix $A_L(\lambda,k,z)$ such that $F_L(x,\lambda,k,z) = F_R(x,\lambda,k,z) \, A_L(\lambda,k,z)$. From (\ref{FR}) and (\ref{IE-FL}), we get the following integral expression for $A_L(\lambda,k,z)$
\begin{equation} \label{ALRepresentation}
  A_L(\lambda,k,z) = I_2 - i z \Ga \int_\R e^{-i\Ga \lambda y} V_k(y) F_L(y,\lambda,k,z) dy.
\end{equation}
Moreover, the matrix $A_L(\lambda,k,z)$ satisfies the following equality (see \cite{AKM}, Proposition 2.2)
\begin{equation} \label{AL-Relation}
  A_L^*(\lambda,k,z) \Ga A_L(\lambda,k,z) = \Ga, \quad \forall \lambda \in \R, \ z \in \R.
\end{equation}
Using the matrix notation
$$
  A_L(\lambda,k,z) = \left[\begin{array}{cc} a_{L1}(\lambda,k,z)&a_{L2}(\lambda,k,z)\\a_{L3}(\lambda,k,z)&a_{L4}(\lambda,k,z) \end{array} \right],
$$
the equality (\ref{AL-Relation}) can be written in components for all $\lambda, z \in \R$ as
\begin{equation} \label{ALUnitarity}
  \left. \begin{array}{ccc} |a_{L1}(\lambda,k,z)|^2 - |a_{L3}(\lambda,k,z)|^2 & = & 1, \\
  |a_{L4}(\lambda,k,z)|^2 - |a_{L2}(\lambda,k,z)|^2 & = & 1, \\
  a_{L1}(\lambda,k,z) \overline{a_{L2}(\lambda,k,z)} - a_{L3}(\lambda,k,z) \overline{a_{L4}(\lambda,k,z)} & = & 0. \end{array} \right.
\end{equation}

The matrices $A_L(\lambda,k,z)$ encode all the scattering information of the equation (\ref{Stat-Eq-Ak}). Precisely, following \cite{AKM, DN3}, we define the scattering matrix $\hat{S}(\lambda,k,z)$ associated to (\ref{Stat-Eq-Ak}) by
\begin{equation} \label{SR-SM1}
  \hat{S}(\lambda,k,z) = \left[ \begin{array}{cc} \hat{T}(\lambda,k,z)&\hat{R}(\lambda,k,z)\\ \hat{L}(\lambda,k,z)&\hat{T}(\lambda,k,z) \end{array} \right],
\end{equation}
where
\begin{equation} \label{SR-SM2}
  \hat{T}(\lambda,k,z) = \frac{1}{a_{L1}(\lambda,k,z)}, \quad \hat{R}(\lambda,k,z) = - \frac{a_{L2}(\lambda,k,z)}{a_{L1}(\lambda,k,z)}, \quad \hat{L}(\lambda,z) = \frac{a_{L3}(\lambda,k,z)}{a_{L1}(\lambda,k,z)}.
\end{equation}
The relations (\ref{ALUnitarity}) lead to the unitarity of the scattering matrix $\hat{S}(\lambda,k,z)$. Precisely, we have proved the following Lemma
\begin{lemma} \label{Unitarity-SM-Simplified}
  For each $\lambda \in \R$ and $z \in \R$, let the scattering matrix $\hat{S}(\lambda,k,z)$ be defined by (\ref{SR-SM1})-(\ref{SR-SM2}). We have
    \begin{equation} \label{SCUnitarity}
    \left. \begin{array}{ccc} |\hat{T}(\lambda,k,z)|^2 + |\hat{R}(\lambda,k,z)|^2 & = & 1, \\
    |\hat{T}(\lambda,k,z)|^2 + |\hat{L}(\lambda,k,z)|^2 & = & 1, \\
    \hat{T}(\lambda,k,z) \overline{\hat{R}(\lambda,k,z)} + \hat{L}(\lambda,k,z) \overline{\hat{T}(\lambda,k,z)} & = & 0. \end{array} \right.
  \end{equation}
\end{lemma}


\begin{remark}
The quantities $\hat{T}$ and $\hat{R}, \hat{L}$ are called the transmission and reflection coefficients respectively associated to equation (\ref{Stat-Eq-Ak}).
The former measures the part of a signal transmitted from an horizon to the other in a scattering process whereas the
latters measure the part of a signal reflected from an horizon to itself (event horizon for $\hat{L}$ and cosmological horizon
for $\hat{R}$)\footnote{Whence the notations $\hat{L}$ for \emph{left} reflection coefficient since the event horizon is
located "on the left"' at $x=-\infty$ and $\hat{R}$ for \emph{right} reflection coefficient since the cosmological horizon is located "on the right" at $x=+\infty$.}.
\end{remark}


\subsubsection{The simplified reduced time-dependent scattering matrix}

Before proceeding further, let us make the link between the scattering matrix $\hat{S}(\lambda,k,z)$ given by (\ref{SR-SM1}) and its time-dependent counterpart expressed in terms of time-dependent wave operators. Consider the Hamiltonian
\begin{equation} \label{HkzHat}
  \hat{H}_{kz} = \Ga D_x - z V_k(x),
\end{equation}
appearing in (\ref{Stat-Eq-Ak}) and the free Hamiltonian
\begin{equation} \label{Hinfty}
  H_\infty = \Ga D_x.
\end{equation}

Using a standard Fourier transform, it is immediate that the Hamiltonian $H_\infty$ is selfadjoint on $\h = L^2(\R; \C^2)$ and has purely absolutely continuous spectrum covering the whole line $\R$. Moreover, since the potential $V_k(x)$ is globally bounded and exponentially decreasing at $x = \pm \infty$, it is easy to see (by Kato-Rellich Theorem \cite{RS}) that $\hat{H}_{kz}$ is also selfadjoint on $\h$, and by a standard trace class method \cite{RS}, that $\hat{H}_{kz}$ has purely absolutely spectrum given by the whole line $\R$ and that the standard wave operators
\begin{equation} \label{WO-Ak}
  W^\pm(\hat{H}_{kz}, H_\infty) = s-\lim_{t \to \pm \infty} e^{it\hat{H}_{kz}} e^{-itH_\infty},
\end{equation}
exist and are asymptotically complete on $\h$, \textit{i.e.} they are isometries on $\h$ and their inverse wave operators given by
$$
 (W^\pm(\hat{H}_{kz}, H^\infty))^* = W^\pm(H_\infty, \hat{H}_{kz}) := s-\lim_{t \to \pm \infty} e^{itH_\infty} e^{-it\hat{H}_{kz}},
$$
also exist on $\h$. .

Thus we can define the scattering operator associated to the pair of Hamiltonians $(\hat{H}_{kz}, H_\infty)$ by the usual rule
\begin{equation} \label{S-Ak}
  \hat{S}(k,z) = \left( W^+(\hat{H}_{kz}, H_\infty) \right)^* W^-(\hat{H}_{kz}, H_\infty).
\end{equation}
The operator $\hat{S}(k,z)$ is unitary on $\h$ and commutes with $H_\infty$. We thus can decompose the scattering operator $\hat{S}(k,z)$ on the energy representation of the Hamiltonian $H_\infty$. For this, let us introduce the unitary transform $\F$ on $\h$ by
\begin{equation} \label{F}
  (\F \psi)(\lambda) = \frac{1}{\sqrt{2\pi}} \int_\R e^{-i\Ga x \lambda} \psi(x) dx.
\end{equation}
Then the transform $\F$ diagonalizes the Hamiltonian $H_\infty$, i.e. $\F H_\infty \F^* = M_\lambda$ where $M_\lambda$ is the operator of multiplication by $\lambda$. Hence we define the scattering matrix $\hat{S}(\lambda,k,z)$ at energy $\lambda$ by the rule
\begin{equation} \label{SM-Ak}
  \hat{S}(k,z) = \F^* \hat{S}(\lambda,k,z) \F,
\end{equation}
where now the scattering matrix $\hat{S}(\lambda,k,z)$ defined in (\ref{SM-Ak}) is a unitary $2 \times 2$ - matrix.

\begin{theorem} \label{Link-Stat-TimeDep-kz}
The time-dependent reduced scattering matrix $\hat{S}(\lambda,k,z)$ constructed through the chain of identities (\ref{WO-Ak})-(\ref{SM-Ak}) coincides exactly with the stationary reduced scattering matrix $\hat{S}(\lambda,k,z)$ defined by (\ref{SR-SM1})-(\ref{SR-SM2}) and constructed with the help of the Jost functions (\ref{FL})-(\ref{FR}) and scattering data $A_L(\lambda,k,z)$ given by (\ref{ALRepresentation}). Whence the correspondence between the stationary (\ref{SR-SM1}) and time-dependent (\ref{SM-Ak}) scattering matrices associated to the stationary equation (\ref{Stat-Eq-Ak}) and our use of the same notation for both objects.
\end{theorem}
\begin{proof}
This follows from the results in \cite{AKM} and more precisely \cite{DM}.
\end{proof}


\subsubsection{The physical reduced scattering matrix}

Let us now come back to the original and physical problem. The true scattering data are indeed not associated to the equation (\ref{Stat-Eq-Ak}) but instead to the general stationary equation (see (\ref{Rad-1}))
\begin{equation} \label{Stat-Eq}
  \left[ \Ga D_x  - z a(x) \Gb + c(x,k) \right] \psi = \lambda \psi,
\end{equation}
where $z \in \R$ is any real number playing the role here of the angular momenta $\mu_{kl}(\lambda)$. The subtlety in dealing with the one-dimensional equation (\ref{Stat-Eq}) is that there is no canonical choice for the associated scattering matrix due to the long-range potential $c(x,k)$. Recall indeed that $c(x,k)$ tends to two distinct constants $\Omega_\pm(k)$ when $x$ tends to $\pm \infty$ respectively whereas $a(x)$ is exponentially decreasing at both horizons. We shall nevertheless define a "natural" scattering matrix associated with (\ref{Stat-Eq}) as follows. Introduce first the Hamiltonian
\begin{equation} \label{Hkz}
  H_{kz} = \Ga D_x  - z a(x) \Gb + c(x,k).
\end{equation}

Recall now the asymptotics of the potentials $a(x)$ and $c(x,k)$ in the above description of $H_{kz}$.  We have
$$
a(x) = a_{\pm} e^{\kappa_{\pm}x} + O( e^{3\kappa_{\pm}x}), \ x \rightarrow \pm \infty,
$$
and
$$
  \forall k \in 1/2 + \Z, \quad c(x,k) = \Omega_{\pm}(k) + c_{\pm} e^{2\kappa_{\pm}x} + O( e^{4\kappa_{\pm}x}), \ x \rightarrow \pm \infty,
$$
where $\Omega_\pm(k)$ are given by (\ref{Omega}). This leads to define the following \emph{asymptotic} selfadjoint operators on $\h$
$$
  H^{\pm \infty}_k = \Ga D_x + \Omega_\pm(k),
$$
and we thus have
\begin{eqnarray*}
  H_{kz} & = & H^{+\infty}_k + O(e^{2 \kappa_+ x}), \quad x \to + \infty, \\
  H_{kz} & = & H^{-\infty}_k + O(e^{2 \kappa_- x}), \quad x \to - \infty.
\end{eqnarray*}

Hence, the operator $H_{kz}$ is a short-range perturbation of $H^{- \infty}_k$ at the event horizon and of $H^{+\infty}_k$ at the cosmological horizons. We thus expect that at the event horizon, the Hamiltonian $H_{kz}$ can be compared with the asymptotic Hamiltonian $H^{-\infty}_k$ whereas at the cosmological horizon, $H_{kz}$ can be compared with $H^{+\infty}_k$. In order to separate the two asymptotic regions, we introduce (as in \cite{Da1}) the projections
$$
  P_\pm = \mathbf{1}_{\R^\pm}(\Ga).
$$

Here the projectors $P_\pm$ allow us to separate easily the part of the Dirac fields falling into the event horizon and the part escaping toward the cosmological horizon. They correspond to the \emph{asymptotic velocity operators} associated to the Hamiltonian $H_\infty =  \Ga D_x$ in the langage of \cite{DG}. For instance, when $t$ tends $+\infty$, we shall compare the part of the Dirac fields escaping to the cosmological horizon with the simpler dynamic generated by $H^{+\infty}_k P_+$ whereas the part of the Dirac fields falling into the event horizon will be compared with the simpler dynamic generated by $H^{-\infty}_k P_-$.

Having this in mind, let us finally introduce the full \emph{asymptotic Hamiltonians}
\begin{equation} \label{Hkpm}
  H^\pm_k = H^{+\infty}_k P_\pm + H^{-\infty}_k P_\mp.
\end{equation}
A complete direct scattering theory was obtained in \cite{DN2} and also in \cite{Da1} for the pair of Hamiltonians $(H_{kz}, H_k^\pm)$. We summarize these results here. The Hamiltonians $H_{kz}$ and $H_k^\pm$ are selfadjoint on the Hilbert space $\h = L^2(\R; \C^2)$. Moreover, according to the analysis given in \cite{Da1}, we can prove the following Theorem
\begin{theorem} \label{WO-Hkz-Hkpm}
  For each $k \in \frac{1}{2} + \Z$ and $z \in \R$, the Hamiltonians $H_{kz}$ and $H^\pm_k$ have purely absolutely continuous spectra, precisely
  $$
    \sigma(H_{kz})= \sigma_{ac}(H_{kz}) = \R, \quad \sigma(H^\pm_k)= \sigma_{ac}(H^\pm_k) = \R,
  $$
  and the wave operators
  \begin{equation} \label{WO-Rad}
    W^\pm(H_{kz}, H^\pm_k) := s-\lim_{t \to \pm \infty} e^{itH_{kz}} e^{-itH^\pm_k},
  \end{equation}
  exist on $\h$ and are asymptotically complete, \textit{i.e.} they are isometries on $\h$ and their inverse wave operators given by
  $$
    (W^\pm(H_{kz}, H^\pm_k))^* = W^\pm(H^\pm_{k}, H_{kz}) := s-\lim_{t \to \pm \infty} e^{itH^\pm_{k}} e^{-itH_{kz}},
  $$
  also exist on $\h$.
\end{theorem}

Thus we can define the scattering operator associated to the pair of Hamiltonians $(H_{kz}, H_k^\pm)$ by
\begin{equation} \label{S-Rad}
  S(k,z) = \left( W^+(H_{kz}, H_k^+) \right)^* W^-(H_{kz}, H_k^-),
\end{equation}
which is clearly unitary on $\h$.

We now define the scattering matrix $S(\lambda,k,z)$ at fixed energy $\lambda$ associated to the scattering operator (\ref{S-Rad}). The difference between the previous situation and the current one is that now the free Hamiltonians $H_k^-$ and $H_k^+$ differ and thus the definition of $S(\lambda,k,z)$ is non-canonical. Following the usual approach, we first introduce the unitary transforms $F_k^\pm$ on $\h$ by
\begin{eqnarray}
  F_k^+ \psi(\lambda) & = & \F \left( \begin{array}{cc} e^{i\Omega_+(k) x}&0\\0&e^{-i\Omega_-(k) x} \end{array} \right) \psi(\lambda), \nonumber \\
                      & = & \frac{1}{\sqrt{2\pi}} \int_\R  \left( \begin{array}{cc} e^{-ix (\lambda - \Omega_+(k))}&0\\0&e^{ix (\lambda - \Omega_-(k))} \end{array} \right) \psi(x) dx, \label{F+k}
\end{eqnarray}
and
\begin{eqnarray}
  F_k^- \psi(\lambda) & = & \F \left( \begin{array}{cc} e^{i\Omega_-(k) x}&0\\0&e^{-i\Omega_+(k) x} \end{array} \right) \psi(\lambda), \nonumber \\
                      & = & \frac{1}{\sqrt{2\pi}} \int_\R  \left( \begin{array}{cc} e^{-ix (\lambda - \Omega_-(k))}&0\\0&e^{ix (\lambda - \Omega_+(k))} \end{array} \right) \psi(x) dx. \label{F-k}
\end{eqnarray}
Then it is a straightforward calculation to show that the transforms $F_k^\pm$ diagonalize the Hamiltonians $H_k^\pm$, i.e. $F_k^\pm H_k^\pm (F_k^\pm)^* = M_\lambda$ where $M_\lambda$ is the operator of multiplication by $\lambda$. Similarly to (\ref{SM-Ak}), we thus define in a natural way the scattering matrix $S(\lambda,k,z)$ associated to (\ref{Stat-Eq}) by
\begin{equation} \label{SM-Rad}
  S(k,z) = (F_k^+)^* S(\lambda,k,z) F_k^-.
\end{equation}
Hence $S(\lambda,k,z)$ is obviously a unitary $2 \times 2$ - matrix.

It remains to make the link between the \emph{physical} scattering matrix $S(\lambda,k,z)$ defined by (\ref{SM-Rad}) and the \emph{simplified} or \emph{unphysical} scattering matrix $\hat{S}(\lambda,k,z)$ obtained in (\ref{SR-SM1}) in terms of the Jost functions solutions of (\ref{Stat-Eq-Ak}). For this, we shall express the wave operators (\ref{WO-Rad}) in terms of the wave operators (\ref{WO-Ak}). Remarking the identity $H_{kz} = U_k^{-1} \hat{H}_{kz} U_k$ where $U_k$ is the unitary transform defined in (\ref{Unitary-Uk}), we get
\begin{eqnarray}
  W^\pm(H_{kz}, H_k^\pm) & = & U_k U_k^{-1} s-\lim_{t \to \pm \infty} e^{itH_{kz}} U_k U_k^{-1} e^{-itH_k^\pm}, \nonumber \\
  & = & U_k s-\lim_{t \to \pm \infty} e^{it U_k^{-1} H_{kz} U_k} U_k^{-1} e^{-itH_k^\pm}, \nonumber \\
  & = & U_k s-\lim_{t \to \pm \infty} e^{it \hat{H}_{kz}} e^{-itH_\infty} e^{itH_\infty} U_k^{-1} e^{-itH_k^\pm}, \nonumber \\
  & = & U_k W^\pm( \hat{H}_{kz}, H_\infty) G_k^\pm, \label{WO-1}
\end{eqnarray}
where we have used the chain-rule in the last step and where we have defined
\begin{equation} \label{Gkpm}
  G_k^\pm = s-\lim_{t \to \pm \infty} e^{itH_\infty} U_k^{-1} e^{-itH_k^\pm}.
\end{equation}

Let us study the operators $G_k^\pm$ separatly. We prove
\begin{lemma} \label{G-k}
  Denote by $\beta(k)$ the constant
  \begin{equation} \label{Betak}
    \beta(k) = \int_{-\infty}^0 \left[ c(s,k) - \Omega_-(k) \right] ds + \int_0^{+\infty} \left[ c(s,k) - \Omega_+(k) \right] ds,
  \end{equation}
  which is well defined by (\ref{Asymp-c}). Then we have
  \begin{equation} \label{G-k-1}
    G_k^\pm = e^{i\Ga \left[ \left( \Omega_+(k) + \beta(k) + K \right) P_\pm \ + \ \left( \Omega_-(k) + K \right) P_\mp \right]},
  \end{equation}
  where $K$ is the constant of integration in (\ref{Pot-C}) and $P_\pm = \mathbf{1}_{\R^\pm}(\Ga)$. Moreover, we have
  \begin{equation} \label{G-k-2}
    \F G_k^+ = e^{i\Ga K} \left( \begin{array}{cc} e^{i\beta(k)}&0\\0&1 \end{array} \right) F_k^+,
  \end{equation}
  and
  \begin{equation} \label{G-k-3}
    \F G_k^- = e^{i\Ga K} \left( \begin{array}{cc} 1&0\\0&e^{-i\beta(k)} \end{array} \right) F_k^-.
  \end{equation}
\end{lemma}
\begin{proof}
  We only prove (\ref{G-k-1}) and (\ref{G-k-2}) in the case $G_k^+$ since the proof for $G_k^-$ is similar. Note first that the projectors $P_\pm$ have the explicit form
  $$
    P_+ = \left( \begin{array}{cc} 1&0\\0&0 \end{array} \right), \quad P_+ = \left( \begin{array}{cc} 0&0\\0&1 \end{array} \right),
  $$
  thanks to (\ref{DiracMatrices}) and that $P_\pm \Ga = \pm I_2$ by definition of $P_\pm$. Hence we can make explicit computations as follows.
  \begin{eqnarray}
   e^{itH_\infty} U_k^{-1} e^{-itH_k^+} & = &  e^{it\Ga D_x} e^{iC(x,k) \Ga} e^{-it \Ga D_x} e^{-it \left( \Omega_+(k) P_+ + \Omega_-(k) P_- \right)}, \nonumber \\
   & = & e^{itD_x} e^{iC(x,k)} e^{-itD_x} e^{-it \Omega_+(k)} P_+ + e^{-itD_x} e^{-iC(x,k)} e^{itD_x} e^{-it \Omega_-(k)} P_-, \nonumber \\
   & = & e^{iC(x+t,k)} e^{-it \Omega_+(k)} P_+ + e^{-iC(x-t,k)} e^{-it \Omega_-(k)} P_-. \label{Gk0}
  \end{eqnarray}
  On one hand, we have
  \begin{eqnarray}
    \lim_{t \to +\infty} e^{-iC(x-t,k)} e^{-it \Omega_-(k)} P_- & = & \lim_{t \to +\infty} e^{-i \int_{-\infty}^{x-t} [c(s,k) - \Omega_-(k)]ds - i \Omega_-(k) x -iK} P_-, \nonumber \\
    & = & e^{-i(\Omega_-(k) x + K) P_-}. \label{Gk1}
  \end{eqnarray}
  On the other hand, using that
  \begin{equation} \label{CxkBetak}
    C(x,k) = -\int_x^{+\infty} [c(s,k) - \Omega_+(k) ] ds + \Omega_+(k) x \ + \ \beta(k) + K,
  \end{equation}
  a similar explicit calculation as (\ref{Gk1}) shows that
  \begin{eqnarray}
    \lim_{t \to +\infty} e^{iC(x+t,k)} e^{-it \Omega_-(k)} P_+ & = & \lim_{t \to +\infty} e^{-i \int_{x+t}^{+\infty} [c(s,k) - \Omega_+(k)]ds + i \Omega_+(k) x + i \beta(k) + iK} P_+, \nonumber \\
    & = & e^{i(\Omega_+(k) x + \beta(k) + K) P_+}. \label{Gk2}
  \end{eqnarray}
  Putting (\ref{Gk0}) - (\ref{Gk2}) together, we get (\ref{G-k-1}).

  At last, (\ref{G-k-2}) is a standard calculation using the definitions of $\F$, $F_k^\pm$ and (\ref{G-k-1}) that we leave to the reader.
\end{proof}

We can now make the link between the physical and simplified scattering matrices (\ref{SM-Rad}) and (\ref{SR-SM1}). Using (\ref{WO-1}), (\ref{S-Ak}), (\ref{SM-Ak}) and Lemma \ref{G-k}, we compute
\begin{eqnarray}
  S(k,z) & = & \left( W^+(H_{kz}, H_k^+) \right)^* W^-(H_{kz}, H_k^-), \nonumber \\
         & = & \left( U_k W^+(\hat{H}_{kz}, H_\infty) G_k^+ \right)^* U_k W^-(\hat{H}_{kz}, H_\infty) G_k^-, \nonumber \\
         & = & \left( G_k^+ \right)^* \hat{S}(k,z) G_k^-, \nonumber \\
         & = & \left( \F G_k^+ \right)^* \hat{S}(\lambda, k,z) \F G_k^-, \nonumber \\
         & = & \left( F_k^+ \right)^* \left( \begin{array}{cc} e^{-i\beta(k)}&0\\0&1 \end{array} \right) e^{-i\Ga K} \hat{S}(\lambda, k,z) e^{i\Ga K} \left( \begin{array}{cc} 1&0\\0&e^{-i\beta(k)} \end{array} \right) F_k^-. \label{SM-1}
\end{eqnarray}
Hence, identifying (\ref{SM-Rad}) with (\ref{SM-1}), we finally get
$$
  S(\lambda,k,z) = \left( \begin{array}{cc} e^{-i\beta(k)}&0\\0&1 \end{array} \right) e^{-i\Ga K} \hat{S}(\lambda, k,z) e^{i\Ga K} \left( \begin{array}{cc} 1&0\\0&e^{-i\beta(k)} \end{array} \right),
$$
which in terms of the components (\ref{SR-SM1}) of $\hat{S}(\lambda,k,z)$ can be written as
\begin{equation} \label{Link-SM}
  S(\lambda,k,z) = \left( \begin{array}{cc} e^{-i\beta(k)} \hat{T}(\lambda,k,z) & e^{-2i(\beta(k) + K)} \hat{R}(\lambda,k,z)\\ e^{2iK} \hat{L}(\lambda,k,z) & e^{-i\beta(k)} \hat{T}(\lambda,k,z) \end{array} \right).
\end{equation}

For the scattering matrix $S(\lambda,k,z)$, we shall use the notation
$$
  S(\lambda,k,z) = \left[ \begin{array}{cc} T(\lambda,k,z)&R(\lambda,k,z)\\ L(\lambda,k,z)&T(\lambda,k,z) \end{array} \right].
$$
Hence we obtain the identifications between the components of the physical and unphysical scattering matrices
\begin{eqnarray*}
  T(\lambda,k,z) & = & e^{-i\beta(k)} \hat{T}(\lambda,k,z), \\
  R(\lambda,k,z) & = & e^{-2i(\beta(k) + K)} \hat{R}(\lambda,k,z), \\
  L(\lambda,k,z) & = & e^{2iK} \hat{L}(\lambda,k,z).
\end{eqnarray*}
From Lemma \ref{Unitarity-SM-Simplified} and the above equalities, it is immediate to check the unitarity of the scattering matrix $S(\lambda,k,z)$ as was expected. Precisely we have

\begin{lemma} \label{Unitarity-SM}
  For each $\lambda \in \R$ and $z \in \R$, let the scattering matrix $\hat{S}(\lambda,k,z)$ be defined by (\ref{SR-SM1})-(\ref{SR-SM2}). We have
    \begin{equation}
    \left. \begin{array}{ccc} |T(\lambda,k,z)|^2 + |R(\lambda,k,z)|^2 & = & 1, \\
    |T(\lambda,k,z)|^2 + |L(\lambda,k,z)|^2 & = & 1, \\
    T(\lambda,k,z) \overline{R(\lambda,k,z)} + L(\lambda,k,z) \overline{T(\lambda,k,z)} & = & 0. \end{array} \right.
  \end{equation}
\end{lemma}

\subsubsection{The global scattering matrix}

Let us continue this section by coming back to the original scattering matrix associated to the separated equation (\ref{Rad-0})-(\ref{Rad-1}) and by defining the global scattering matrix at fixed energy $\lambda$.

We first state in the form of a Theorem the definition of the unphysical reduced scattering matrix on a generalized spherical harmonic $Y_{kl}(\lambda)$ as well as the corresponding unphysical global scattering matrix simply obtained by summation over the generalized spherical harmonics $Y_{kl}(\lambda)$.

\begin{theorem} \label{SM-Unphysical}
  For all indices $(k,l) \in I = (\frac{1}{2} + \Z) \times \N^*$ and for all $\lambda \in \R$, denote by $\mu_{kl}(\lambda)$ the eigenvalues of the angular operator $A_{\S^2}(\lambda)$ introduced in Thm \ref{Separation} and define using (\ref{SR-SM1})-(\ref{SR-SM2})
  \begin{equation} \label{SM-Red-Unphysical}
    \hat{S}_{kl}(\lambda) = \left[ \begin{array}{cc} \hat{T}_{kl}(\lambda)& \hat{R}_{kl}(\lambda)\\ \hat{L}_{kl}(\lambda)& \hat{T}_{kl}(\lambda) \end{array} \right] := \hat{S}(\lambda,k,\mu_{kl}(\lambda)),
  \end{equation}
  the unphysical reduced scattering matrix at energy $\lambda$ constructed with the help of the Jost functions (\ref{FL}) associated to (\ref{Rad-3})).

  The global unphysical scattering matrix $\hat{S}(\lambda)$ for all energies $\lambda \in \R$ is then defined as the unitary operator from $\H_{\S^2} = L^2(\S^2; \C^2)$ onto itself in the following way. For all $\psi \in \H_{\S^2}$, we can decompose $\psi$ onto the generalized spherical harmonics $Y_{kl}(\lambda)$ by
  $$
    \psi = \sum_{(k,l) \in I} \psi_{kl} \otimes Y_{kl}(\lambda), \qquad \psi_{kl} \in \C^2,
  $$
  (see Thm \ref{Separation}) and the global scattering matrix $\hat{S}(\lambda)$ is defined by
  \begin{equation} \label{SM-Global-Unphysical}
    \hat{S}(\lambda) \psi = \sum_{(k,l) \in I} \left( \hat{S}_{kl}(\lambda) \psi_{kl} \right) \otimes Y_{kl}(\lambda).
  \end{equation}
\end{theorem}

We continue defining the physical global scattering matrix at energy $\lambda$ associated to the stationary equation (\ref{Stat-Eq-Global}).

\begin{theorem} \label{SM-Physical}
Using the same notations as in Thm \ref{SM-Unphysical} and using (\ref{SM-Rad}), define
  \begin{equation} \label{SM-Red-Physical}
    S_{kl}(\lambda) = \left[ \begin{array}{cc} T_{kl}(\lambda)& R_{kl}(\lambda)\\ L_{kl}(\lambda)& T_{kl}(\lambda) \end{array} \right] := S(\lambda,k,\mu_{kl}(\lambda)).
  \end{equation}
  The unitary $2 \times 2$ - matrix $S_{kl}(\lambda)$ is the physical reduced scattering matrix at fixed energy $\lambda$ and it is expressed in terms of the unphysical and simplified scattering matrix $\hat{S}_{kl}(\lambda)$ by the formula
  $$
    S_{kl}(\lambda) = \left( \begin{array}{cc} e^{-i\beta(k)}&0\\0&1 \end{array} \right) e^{-i\Ga K} \hat{S}_{kl}(\lambda) e^{i\Ga K} \left( \begin{array}{cc} 1&0\\0&e^{-i\beta(k)} \end{array} \right),
  $$
  or in components
  \begin{eqnarray*}
  T_{kl}(\lambda) & = & e^{-i\beta(k)} \hat{T}_{kl}(\lambda), \\
  R_{kl}(\lambda) & = & e^{-2i(\beta(k) + K)} \hat{R}_{kl}(\lambda), \\
  L_{kl}(\lambda) & = & e^{2iK} \hat{L}_{kl}(\lambda),
\end{eqnarray*}
where $K$ is any constant of integration in (\ref{Unitary-Uk}) and
$$
    \beta(k) = \int_{-\infty}^0 \left[ c(s,k) - \Omega_-(k) \right] ds + \int_0^{+\infty} \left[ c(s,k) - \Omega_+(k) \right] ds.
$$

Associated to the stationary equation (\ref{Stat-Eq-Global}), we then define the physical global scattering matrix $S(\lambda)$ for all energies $\lambda \in \R$ as the unitary operator from $\H_{\S^2} = L^2(\S^2; \C^2)$ onto itself in the following way. For all $\psi \in \H_{\S^2}$, we can decompose $\psi$ onto the generalized spherical harmonics $Y_{kl}(\lambda)$ by
$$
    \psi = \sum_{(k,l) \in I} \psi_{kl} \otimes Y_{kl}(\lambda), \qquad \psi_{kl} \in \C^2,
$$
(see Thm \ref{Separation}) and the physical global scattering matrix $S(\lambda)$ is defined by
\begin{equation} \label{SM-Global-Physical}
    S(\lambda) \psi = \sum_{(k,l) \in I} \left( S_{kl}(\lambda) \psi_{kl} \right) \otimes Y_{kl}(\lambda).
\end{equation}
\end{theorem}

The global scattering matrix $S(\lambda)$ defined in (\ref{SM-Global-Physical}) will be the main object of investigation in this paper. We aim to identify the metric of the KN-dS black hole from the knowledge of $S(\lambda)$ for a fixed energy $\lambda \in \R$. In fact, we shall be able to prove better results.

First the knowledge of the restrictions of $S(\lambda)$ on two angular modes $\{e^{ik\varphi}\}$ and $\{e^{ik'\varphi}\}$ with $k \ne k' \in 1/2 + \Z$ is enough to identify the metric of the KN-dS black hole. Let us here define these objects.

\begin{definition} \label{SMk}
We denote by $S_k(\lambda)$ the unitary operator on $\mathfrak{L} = L^2((0,\pi), \C^2)$ defined as the restriction of the scattering matrix (\ref{SM-Global-Physical}) at energy $\lambda$ on the angular mode $\{e^{ik\varphi}\}$ with $ k \in \frac{1}{2} + \Z$. Precisely, for all $\psi \in \H_{\S^2}$, we can decompose $\psi$ as
  $$
   \psi = \sum_{k \in \frac{1}{2} + \Z} \psi_{k}(\theta) \otimes e^{ik\varphi}, \qquad \psi_k(\theta)  \in \mathfrak{L} = L^2((0,\pi), \C^2),
  $$
  and we have
  \begin{equation} \label{SM-k}
    S(\lambda) \psi = \sum_{k \in \frac{1}{2} + \Z} \left( S_k(\lambda) \psi_{k} \right) \otimes e^{ik\varphi},
  \end{equation}
  with
  \begin{equation} \label{SM-kl}
    S_k(\lambda) \psi_{k} = \sum_{l \in \N^*} \left( S_{kl}(\lambda) \psi_{kl} \right) \otimes u_{kl}^\lambda(\theta),
  \end{equation}
  where
  $$
    \psi_k(\theta) = \sum_{l \in \N^*} \psi_{kl} \otimes u_{kl}^\lambda(\theta), \quad \textrm{and} \ \quad Y_{kl}(\lambda)(\theta,\varphi) = u_{kl}^\lambda(\theta) e^{ik\varphi}.
  $$
  Clearly, the operators $S_k(\lambda)$ are unitary operators on $\mathfrak{L}$ by construction.
\end{definition}

Second, we don't need to know the whole reduced scattering matrices $S_k(\lambda)$ with $k \in 1/2 + \Z$, but only one of its diagonal - the transmission operators - or anti-diagonal - the reflection operators - parts. Let us define here these objects and state their main properties.

\begin{definition} \label{TkRk}
For all $k \in 1/2 + \Z$ and $\lambda \in \R$, we denote by $T_k^L(\lambda), T_k^R(\lambda), R_k(\lambda)$ and $L_k(\lambda)$ the transmission and reflection operators as the operators acting on $\mathfrak{l} = L^2((0,\pi); \C)$ such that
\begin{equation} \label{Link-Sk-TkRk}
  S_k(\lambda) = \left[ \begin{array}{cc} T_k^L(\lambda) & R_k(\lambda) \\ L_k(\lambda) & T_k^R(\lambda) \end{array} \right],
\end{equation}
where $S_k(\lambda)$ is the reduced scattering matrix defined in Definition \ref{SMk}. More precisely, let us introduce the operators
$$
  P_1 \,: \ \begin{array}{ccc} \mathfrak{L} & \longrightarrow & \mathfrak{l}, \\
                           (\psi_1, \psi_2) & \longrightarrow & \psi_1, \end{array}, \quad \quad
  P_2 \,: \ \begin{array}{ccc} \mathfrak{L} & \longrightarrow & \mathfrak{l}, \\
                           (\psi_1, \psi_2) & \longrightarrow & \psi_2, \end{array},
$$
as well as their adjoints
$$
  P_1^* \,: \ \begin{array}{ccc} \mathfrak{l} & \longrightarrow & \mathfrak{L}, \\
                           \psi_1 & \longrightarrow & (\psi_1,0), \end{array}, \quad \quad
  P_2^* \,: \ \begin{array}{ccc} \mathfrak{l} & \longrightarrow & \mathfrak{L}, \\
                           \psi_2 & \longrightarrow & (0,\psi_2), \end{array}.
$$
Then the transmission and reflection operators are defined as the operators acting on $\mathfrak{l}$ by
$$
  T_k^L(\lambda) = P_1 S_k(\lambda) P_1^*, \quad T_k^R(\lambda) = P_2 S_k(\lambda) P_2^*,
$$
$$
  R_k(\lambda) = P_1 S_k(\lambda) P_2^*, \quad L_k(\lambda) = P_2 S_k(\lambda) P_1^*.
$$
\end{definition}

On one hand, the transmission operators $T_k^{L,R}(\lambda)$ measure the part of a signal having energy $\lambda$ that is transmitted from an asymptotic region - say the cosmological horizon for $T_k^R$ - to the other - the event horizon for $T_k^R$ - in a scattering experiment. On the other hand, the reflection operators $R_k(\lambda)$ (resp. $L_k(\lambda)$) measure the part of a signal having energy $\lambda$ that is reflected from the cosmological horizon (resp. event horizon) to itself in a scattering process. We shall be able to prove that the knowledge of one of these operators at a fixed energy and for two different angular modes is enough to determine uniquely the metric of a KN-dS black hole. For this, we shall use several properties of the transmission and reflection operators that are easily consequences of the previous definitions. We summarize these properties in the next proposition.

\begin{prop} \label{TkRkLk}
Let $k \in 1/2 + \Z$ be fixed. For all $l \in \N^*$, we introduce the notation $Y_{kl}(\lambda) = (Y_{kl}^1(\lambda), Y_{kl}^2(\lambda))$ for the corresponding generalized spherical harmonics. Then, we can prove the following results. \\

\noindent 1) [\textbf{Hilbert bases of $\mathfrak{l} = L^2((0,\pi); \C)$}] The families $\{Y_{kl}^1(\lambda)\}_{l \in \N^*}$ and $\{Y_{kl}^2(\lambda)\}_{l \in \N^*}$ form Hilbert bases of $\mathfrak{l} = L^2((0,\pi); \C)$; precisely for all $\psi_k \in \mathfrak{l}$, we can decompose $\psi_k$ as
$$
  \psi_k = \sum_{l \in \N^*} \psi_{kl}^{j} Y_{kl}^{j}(\lambda), \quad j=1,2,
$$
with
$$
  \| \psi_k \|^2 = \frac{1}{2} \sum_{l \in \N^*} |\psi_{kl}^{j} |^2.
$$

\noindent 2) [\textbf{Transmission operators}] For all $\psi_k \in \mathfrak{l}$, we have
\begin{equation} \label{TkR}
  T_k^L(\lambda) \psi_k = T_k^L(\lambda) \left( \sum_{l \in \N^*} \psi_{kl}^1 Y_{kl}^1(\lambda) \right) = \sum_{l \in \N^*} \left( T_{kl}(\lambda) \psi_{kl}^1 \right) Y_{kl}^1(\lambda),
\end{equation}
where $T_{kl}(\lambda)$ is defined in Theorem \ref{SM-Global-Physical}. Also
\begin{equation} \label{TkL}
  T_k^R(\lambda) \psi_k = T_k^R(\lambda) \left( \sum_{l \in \N^*} \psi_{kl}^2 Y_{kl}^2(\lambda) \right) = \sum_{l \in \N^*} \left( T_{kl}(\lambda) \psi_{kl}^2 \right) Y_{kl}^2(\lambda).
\end{equation}
In other words, we have
\begin{equation} \label{Trk}
  T_k^L(\lambda) Y_{kl}^1 = T_{kl}(\lambda) Y_{kl}^1, \quad \quad T_k^R(\lambda) Y_{kl}^2 = T_{kl}(\lambda) Y_{kl}^2, \quad \forall l \in \N^*.
\end{equation}

\noindent 3) [\textbf{Reflection operators}] For all $\psi_k \in \mathfrak{l}$, we have
\begin{equation} \label{Rk}
  R_k(\lambda) \psi_{k} = R_k(\lambda) \left( \sum_{l \in \N^*} \psi_{kl}^2 Y_{kl}^2(\lambda) \right) = \sum_{l \in \N^*} \left( R_{kl}(\lambda) \psi_{kl}^2 \right) Y_{kl}^1(\lambda),
\end{equation}
and
\begin{equation} \label{Lk}
  L_k(\lambda) \psi_{k} = L_k(\lambda) \left( \sum_{l \in \N^*} \psi_{kl}^1 Y_{kl}^1(\lambda) \right) = \sum_{l \in \N^*} \left( L_{kl}(\lambda) \psi_{kl}^1 \right) Y_{kl}^2(\lambda),
\end{equation}
where $R_{kl}(\lambda), L_{kl}(\lambda)$ are defined in Theorem \ref{SM-Global-Physical}. In other words,
\begin{equation} \label{Rek}
  R_k(\lambda) Y_{kl}^2 = R_{kl}(\lambda) Y_{kl}^1, \quad \quad L_k(\lambda) Y_{kl}^1 = L_{kl}(\lambda) Y_{kl}^2, \quad \forall l \in \N^*.
\end{equation}

4) [\textbf{Adjoints}] Using that
$$
 T_k^L(\lambda)^* = P_1 S_k^*(\lambda) P_1^*, \quad T_k^R(\lambda)^* = P_2 S_k^*(\lambda) P_2^*,
$$
we get
\begin{equation} \label{Trk*}
  T_k^L(\lambda)^* Y_{kl}^1 = \overline{T_{kl}}(\lambda) Y_{kl}^1, \quad \quad T_k^R(\lambda)^* Y_{kl}^2 = \overline{T_{kl}}(\lambda) Y_{kl}^2, \quad \forall l \in \N^*.
\end{equation}

Similarly, using that
$$
   R_k(\lambda)^* = P_2 S_k^*(\lambda) P_1^*, \quad L_k(\lambda)^* = P_1 S_k^*(\lambda) P_2^*,
$$
we get
\begin{equation} \label{Rek*}
  R_k(\lambda)^* Y_{kl}^1 = \overline{R_{kl}}(\lambda) Y_{kl}^2, \quad \quad L_k(\lambda)^* Y_{kl}^1 = \overline{L_{kl}}(\lambda) Y_{kl}^1, \quad \forall l \in \N^*.
\end{equation}

\noindent 5) [\textbf{Unitarity}] At last, the unitarity relations for $S_k(\lambda)$ become in terms of the transmission and reflection operators
\begin{eqnarray*} 
  (i) & T_k^L(\lambda) T_k^L(\lambda)^* + R_k(\lambda) R_k(\lambda)^* = Id = T_k^R(\lambda)^* T_k^R(\lambda) + R_k(\lambda)^* R_k(\lambda), \\
  (i) & T_k^L(\lambda)^* T_k^L(\lambda) + L_k(\lambda)^* L_k(\lambda) = Id = T_k^R(\lambda) T_k^R(\lambda)^* + L_k(\lambda) L_k(\lambda)^*, \\
  (iii) & T_k^R(\lambda) R_k(\lambda)^* + L_k(\lambda) T_k^L(\lambda)^* = 0 = T_k^R(\lambda)^* L_k(\lambda) + R_k(\lambda)^* T_k^L(\lambda).
\end{eqnarray*}
\end{prop}

We shall need in Section \ref{InverseAngular} the following simple but crucial corollary of the previous Proposition.

\begin{coro} \label{SpectralProp-TrRe}
The operators $T_k^L(\lambda)$, $R_k(\lambda) R_k(\lambda)^*$ and $L_k(\lambda)^* L_k(\lambda)$ are diagonalizable on the Hilbert basis of eigenfunctions $(Y_{kl}^1)_{l \in \N^*}$ associated to the eigenvalues $T_{kl}(\lambda)$, $|R_{kl}(\lambda)|^2$ and $|L_{kl}(\lambda)|^2$ respectively. Similarly, the operators $T_k^R(\lambda)$, $R_k(\lambda)^* R_k(\lambda)$ and $L_k(\lambda) L_k(\lambda)^*$ are diagonalizable on the Hilbert basis of eigenfunctions $(Y_{kl}^2)_{l \in \N^*}$ associated to the eigenvalues $T_{kl}(\lambda)$, $|R_{kl}(\lambda)|^2$ and $|L_{kl}(\lambda)|^2$ respectively.
\end{coro}

%
%

\subsubsection{Dependence on the Regge-Wheeler variable} \label{DependenceRW}

We finish this Section explaining a subtle point in the definition of the reduced scattering matrices given in Proposition \ref{SM-Unphysical} and the global scattering matrix given in Theorem \ref{SM-Physical} that will play an important role in the statement of our main inverse result, Thm \ref{Main}. As mentioned in the Introduction, the stationary expression of $S_{kl}(\lambda)$ (and thus the one of $S(\lambda)$) depends on the choice of the radial Regge-Wheeler coordinate through the constant of integration $c$ in (\ref{RW}).

To see this, let us assume that the KN-dS black hole we are considering is described with a Regge-Wheeler coordinate $\tilde{x} = x + c$ instead of the previous $x$ Regge-Wheeler coordinate. Here $c$ is any constant of integration. We first analyse the influence of this change of variable on the unphysical scattering matrix $\hat{S}(\lambda,k,z)$ defined in (\ref{SM-Ak}) through the chain of equations (\ref{FL})-(\ref{SR-SM2}). We denote by $\hat{\tilde{S}}(\lambda,k,z)$ the scattering matrix obtained using the variable $\tilde{x}$. Note that in what follows, we shall add a subscript $\,\tilde{}$ to any quantities related to the choice of the $\tilde{x}$ radial variable. When described with the variable $\tilde{x}$, it is easy to see (\cite{DN3}) that the potentials $a(x)$ and $c(x,k)$ become new potentials defined by
\begin{equation} \label{Pot-a-c-tilde}
  \tilde{a}(x) = a(x-c), \qquad \tilde{c}(x,k) = c(x-c,k),
\end{equation}
 and the stationary scattering is governed by the new stationary equation
$$
[\Ga D_x  - z \tilde{a}(x) \Gb + \tilde{c}(x,k)] \psi =0.
$$
The potentials $\tilde{a}$ and $\tilde{c}$ still have the asymptotics (\ref{Asymp-a}) and (\ref{Asymp-c}) where, according to (\ref{Omega}) and with obvious notations, $\tilde{\Omega}_\pm(k) = \Omega_\pm(k)$. So, the primitive (\ref{Pot-C}) of the new potential
$\tilde{c}(x,k)$ becomes
\begin{equation}
  \tilde{C}(x,k) = \int_{-\infty}^x [\tilde{c}(s,k) - \Omega_-(k)] \ ds + \Omega_-(k) x + K.
\end{equation}
A straightforward calculation gives :
\begin{equation} \label{Pottilde-C}
\tilde{C}(x,k) = C(x-c,k) + c  \Omega_-(k).
\end{equation}
Then, $\tilde{q}(x,k):= \tilde{a}(x)\  e^{2i\tilde{C}(x,k)} = q(x-c,k) \ e^{2i c \Omega_-(k)}$. Thus, we obtain
\begin{equation} \label{Pottilde-V}
\tilde{V}_k(x) =  e^{i c \Omega_-(k)\Ga} \  V_k(x-c) \ e^{-i c \Omega_-(k)\Ga}.
\end{equation}
Using (\ref{IE-FL}) and (\ref{ALRepresentation}), we deduce
\begin{eqnarray}\label{FLtilde}
\tilde{F}_L(x,\lambda,k,z) &=&  e^{i c \Omega_-(k)\Ga} \ F(x-c,\lambda,k,z)  \ e^{-i c \Omega_-(k)\Ga} e^{i\lambda c \Ga}, \\
\tilde{A}_L(\lambda,k,z) &=&  e^{-i c(\lambda- \Omega_-(k))\Ga} \ A_L(\lambda,k,z) \  e^{i c(\lambda- \Omega_-(k))\Ga}.
\end{eqnarray}
We conclude with (\ref{SR-SM1}) and (\ref{SR-SM2}) that the stationary expression of
$\hat{\tilde{S}}(\lambda,k,z)$ is then given in term of $\hat{S}(\lambda,k,z)$ by
\begin{equation} \label{RW-SM0}
  \hat{\tilde{S}}(\lambda,k,z) = e^{-i c \lambda^-(k)\Ga}\  \hat{S}(\lambda,k,z) \ e^{i c \lambda^-(k) \Ga},
\end{equation}
where
\begin{equation}\label{lambdapmk}
\lambda^\pm(k) = \lambda - \Omega_\pm(k),
\end{equation}
or written in components - using the notations (\ref{SR-SM2}) - by
\begin{equation} \label{RW-ScatCoeff}
  \left[ \begin{array}{cc} \hat{\tilde{T}}(\lambda,k,z)&\hat{\tilde{R}}(\lambda,k,z)\\ \hat{\tilde{L}}(\lambda,k,z)&\hat{\tilde{T}}(\lambda,k,z) \end{array} \right] =
\left[ \begin{array}{cc} \hat{T}(\lambda,k,z)&e^{-2i\lambda^-(k) c} \hat{R}(\lambda,k,z)\\ e^{2i\lambda^-(k) c} \hat{L}(\lambda,k,z)&\hat{T}(\lambda,k,z) \end{array} \right].
\end{equation}

The scattering matrix (\ref{RW-ScatCoeff}) is not however the physical scattering matrix which is given in fact by
\begin{equation} \label{RW-SM1}
  S(\lambda,k,z) = \left[ \begin{array}{cc} T(\lambda,k,z)&R(\lambda,k,z)\\ L(\lambda,k,z)&T(\lambda,k,z) \end{array} \right] = \left[ \begin{array}{cc} e^{-i\beta(k)} \hat{T}(\lambda,k,z) & e^{-2i(\beta(k) + K)} \hat{R}(\lambda,k,z)\\ e^{2iK} \hat{L}(\lambda,k,z) & e^{-i\beta(k)} \hat{T}(\lambda,k,z) \end{array} \right].
\end{equation}
We thus analyse the influence of the change of variable $\tilde{x} = x + c$ on the constant $\beta(k)$ appearing above. Recall that
\begin{equation} \label{beta}
    \beta(k) = \int_{-\infty}^0 \left[ c(s,k) - \Omega_-(k) \right] ds + \int_0^{+\infty} \left[ c(s,k) - \Omega_+(k) \right] ds.
\end{equation}
Hence, plugging (\ref{Pot-a-c-tilde}) into (\ref{beta}), we get for the new constant $\tilde{\beta}(k)$
\begin{eqnarray}
  \tilde{\beta}(k) & = & \int_{-\infty}^0 \left[ \tilde{c}(s,k) - \tilde{\Omega}_-(k) \right] ds + \int_0^{+\infty} \left[ \tilde{c}(s,k) - \tilde{\Omega}_+(k) \right] ds, \nonumber \\
  & = & \int_{-\infty}^0 \left[ c(s - c,k) - \Omega_-(k) \right] ds + \int_0^{+\infty} \left[ c(s - c,k) - \Omega_+(k) \right] ds,
\end{eqnarray}
 A change of variable in the above integral shows that
\begin{eqnarray}
  \tilde{\beta}(k) & = & \beta(k) + \int_0^{-c} \left[ c(s,k) - \Omega_-(k) \right] ds + \int_{-c}^0 \left[ c(s,k) - \Omega_+(k) \right] ds, \nonumber \\
  & = & \beta(k) + \theta(k) c, \label{Beta-k-tilde}
\end{eqnarray}
with
\begin{equation} \label{Theta-k}
  \theta(k) = \Omega_-(k) - \Omega_+(k).
\end{equation}
Coming back to (\ref{RW-SM1}), we see first that the physical scattering matrix $\tilde{S}(\lambda,k,z)$ associated to $\tilde{x}$ is expressed as
\begin{equation} \label{RW-SM2}
  \tilde{S}(\lambda,k,z) = \left[ \begin{array}{cc} \tilde{T}(\lambda,k,z)&\tilde{R}(\lambda,k,z)\\ \tilde{L}(\lambda,k,z)&\tilde{T}(\lambda,k,z) \end{array} \right] = \left[ \begin{array}{cc} e^{-i\tilde{\beta}(k)} \hat{\tilde{T}}(\lambda,k,z) & e^{-2i(\tilde{\beta}(k) + K)} \hat{\tilde{R}}(\lambda,k,z)\\ e^{2iK} \hat{\tilde{L}}(\lambda,k,z) & e^{-i\tilde{\beta}(k)} \hat{\tilde{T}}(\lambda,k,z) \end{array} \right].
\end{equation}
But, in view of (\ref{RW-ScatCoeff}) and (\ref{Beta-k-tilde}) - (\ref{Theta-k}), the scattering matrix $\tilde{S}(\lambda,k,z)$ can be written in terms of the components of the unphysical scattering matrix $\hat{S}(\lambda,k,z)$ by
\begin{equation} \label{RW-SM-Ak}
  \tilde{S}(\lambda,k,z) = \left[ \begin{array}{cc} e^{-i(\beta(k) + \theta(k)c)} \hat{T}(\lambda,k,z) & e^{-2i(\beta(k) + \theta(k)c + K + \lambda^-(k) c)} \hat{R}(\lambda,k,z)\\ e^{2i(K + \lambda^-(k) c)} \hat{L}(\lambda,k,z) & e^{-i(\beta(k) + \theta(k) c)} \hat{T}(\lambda,k,z) \end{array} \right],
\end{equation}
or in terms of the physical scattering matrix $S(\lambda,k,z)$ by
\begin{equation} \label{RW-SM}
  \left[ \begin{array}{cc} \tilde{T}(\lambda,k,z)&\tilde{R}(\lambda,k,z)\\ \tilde{L}(\lambda,k,z)&\tilde{T}(\lambda,k,z) \end{array} \right] = \left[ \begin{array}{cc} e^{-i\theta(k)c} \,T(\lambda,k,z) & e^{-2i\lambda^+(k) c} \,R(\lambda,k,z)\\ e^{2i\lambda^-(k) c} L(\lambda,k,z) & e^{-i\theta(k) c} \, T(\lambda,k,z) \end{array} \right].
\end{equation}

Let us summarize these results for the true scattering data in a Proposition.
\begin{prop} \label{SM-Red-RW}
  Let the exterior region of a KN-dS black hole be described by two Regge-Wheeler coordinates $x$ and $\tilde{x}$ related by a constant $c$ such that
  $$
    \tilde{x} = x + c.
  $$
  For all $(k,l) \in I$, we denote by
  $$
    S_{kl}(\lambda) = \left[ \begin{array}{cc} T_{kl}(\lambda)&R_{kl}(\lambda)\\ L_{kl}(\lambda)&T_{kl}(\lambda) \end{array} \right], \qquad \tilde{S}_{kl}(\lambda) = \left[ \begin{array}{cc} \tilde{T}_{kl}(\lambda)&\tilde{R}_{kl}(\lambda)\\ \tilde{L}_{kl}(\lambda)&\tilde{T}_{kl}(\lambda) \end{array} \right],
  $$
  the corresponding physical reduced scattering matrices at fixed energy $\lambda$ associated to massless Dirac fields as given in Thm \ref{SM-Physical}. Then we have
  \begin{equation} \label{RW-SM-Red}
  \left[ \begin{array}{cc} \tilde{T}_{kl}(\lambda)&\tilde{R}_{kl}(\lambda)\\ \tilde{L}_{kl}(\lambda)&\tilde{T}_{kl}(\lambda) \end{array} \right] = \left[ \begin{array}{cc} e^{-i c(\Omega_+ (k)-\Omega_-(k))} \,T_{kl}(\lambda) & e^{-2ic\lambda^+(k)} \, R_{kl}(\lambda)\\ e^{2ic\lambda^-(k)} L_{kl}(\lambda) & \,
  e^{-i c(\Omega_+(k)-\Omega_-(k))} \,T_{kl}(\lambda) \end{array} \right],
\end{equation}
where
$$
  \lambda^{\pm}(k) = \lambda- \Omega_{\pm}(k).
$$
Moreover, the equality (\ref{RW-SM-Red}) holds with the underscript $kl$ replaced only by $k$, that is for the scattering matrices reduced onto an angular mode $\{ e^{-ik\varphi}\}$ only. At last, if
$$
    \hat{S}_{kl}(\lambda) = \left[ \begin{array}{cc} \hat{T}_{kl}(\lambda) & \hat{R}_{kl}(\lambda)\\ \hat{L}_{kl}(\lambda)& \hat{T}_{kl}(\lambda) \end{array} \right], \qquad \hat{\tilde{S}}_{kl}(\lambda) = \left[ \begin{array}{cc} \hat{\tilde{T}}_{kl}(\lambda)&\hat{\tilde{R}}_{kl}(\lambda)\\ \hat{\tilde{L}}_{kl}(\lambda)&\hat{\tilde{T}}_{kl}(\lambda) \end{array} \right],
  $$
denote the corresponding unphysical reduced scattering matrices at fixed energy $\lambda$ associated to massless Dirac fields as given in Thm \ref{SM-Unphysical}, then we have
\begin{equation} \label{RW-Unphys-SM-Red}
  \left[ \begin{array}{cc} \hat{\tilde{T}}_{kl}(\lambda)&\hat{\tilde{R}}_{kl}(\lambda)\\ \hat{\tilde{L}}_{kl}(\lambda)&\hat{\tilde{T}}_{kl}(\lambda) \end{array} \right] =
\left[ \begin{array}{cc} \hat{T}_{kl}(\lambda)&e^{-2i\lambda^+(k) c} \hat{R}_{kl}(\lambda)\\ e^{2i\lambda^-(k) c} \hat{L}_{kl}(\lambda)&\hat{T}_{kl}(\lambda) \end{array} \right].
\end{equation}
\end{prop}

The dependence of the expression of the scattering matrix under a change of radial variable will have important consequences for our main result. Since the exterior region of a KN-dS black hole can be described \emph{uniquely} by any choice of the Regge-Wheeler variable $x$, we shall \emph{identify} all the possible forms of the reduced scattering matrices (as given in Prop \ref{SM-Red-RW}) in the statement of our main inverse result (see Thm \ref{Main}). However and in a first reading, the reader is invited to take $c=0$ in the next formulae to simplify the statement of our results.

%
%
%
%

\subsection{Time-dependent expression of the scattering matrix} \label{Time-Dependent}

In this Section, we give an alternative definition of the scattering matrix $S(\lambda)$ by purely time-dependent methods. This definition will be given in terms of time-dependent wave operators associated to the Dirac Hamiltonian $H$ and the \emph{natural} asymptotic Hamiltonians corresponding to the evolution of the Dirac waves in the neighbourhood of the event and cosmological horizons. We shall then introduce the corresponding scattering operator $S$ and finally, define the time-dependent scattering matrix $S(\lambda)$ by a standard procedure. In the following Section, we shall prove the equivalence of this definition of the time-dependent scattering matrix with the stationary definition given in Section \ref{Direct-scat}.

First recall that the Dirac Hamiltonian $H$ is given by (see Section \ref{Hamilton-Form} and more precisely (\ref{Op-H}))
$$
  H = J^{-1} H_0, \quad H_0 = \Ga D_x + a(x) H_{\S^2} + c(x,D_\varphi),
$$
where
$$
  J^{-1} = (1 - \alpha^2) (I_2 - \alpha \Gc), \quad \alpha(x,\theta) = a(x) b(\theta) = \frac{\sqrt{\Delta_r}}{r^2 + a^2}\frac{a \sin^2\theta}{\sqrt{\Delta_\theta}}.
$$
Observe that, by (\ref{Asymp-J}), we have $J^{-1} = I_2 + O(e^{- \kappa |x|}), \ \kappa = \min(|\kappa_-|, |\kappa_+|)  > 0$ when $x \to \pm \infty$. Hence, the operator $H$ can be viewed as a \emph{short-range perturbation of order $1$} of the operator $H_0$, \textit{i.e}
$$
  H = H_0 + O(e^{- \kappa |x|}) H_0,
$$
and we expect that we can compare the dynamics generated by $H$ and $H_0$ at late times. In fact, we are able to prove a complete scattering theory for the pair of selfadjoint operators $(H,H_0)$. Before stating this first result, recall that the Hamiltonians $H_0$ and $H$ act on the different Hilbert spaces $\H$ and $\G$ respectively (defined before Theorem \ref{Spectra-H-H0}). In consequence, we use the two Hilbert spaces scattering formalism as exposed in \cite{Y} to state our first Theorem. Precisely, we have
\begin{theorem} \label{WO-H-H0}
  The Hamiltonians $H_0$ and $H$ have purely absolutely continuous spectra, precisely
  $$
    \sigma(H_0)= \sigma_{ac}(H_0) = \R, \quad \sigma(H)= \sigma_{ac}(H) = \R,
  $$
  and the following wave operators
  $$
    W^\pm(H,H_0,I_2) := s-\lim_{t \to \pm \infty} e^{itH} e^{-itH_0},
  $$
  exist as operators from $\H$ to $\G$ and are asymptotically complete, \textit{i.e.} they are isometries from $\H$ to $\G$ and their inverse wave operators given by
  $$
    (W^\pm(H,H_0,I_2))^* = W^\pm(H_0,H,J) := s-\lim_{t \to \pm \infty} e^{itH_0} J e^{-itH},
  $$
  exist as operators from $\G$ to $\H$. (Note that the identity operator $I_2: \ \H \longrightarrow \G$ has been used as identification operator between $\H$ and $\G$ in the definition of the direct wave operators, whereas the dual operator $(I_2)^* = J: \ \G \longrightarrow \H$ appears in the definition of the inverse wave operators).
\end{theorem}
\begin{proof}
    The proof of this Theorem given in Appendix \ref{LAP-Mourre} is a direct consequence of the Limiting Absorption Principles (LAP) for the Hamiltonians $H$ and $H_0$ and the theory of $H$-smooth operators as exposed in \cite{RS}. These LAPs in turn are obtained with the help of a non trivial Mourre theory similar to the ones obtained in \cite{HaN, Da2}.
\end{proof}

The interest in comparing first the dynamic $e^{-itH}$ with the dynamic $e^{-itH_0}$ at late times is that the operator $H_0$ can be considerably simplified if we decompose it onto its restrictions to the generalized spherical harmonics $Y_{kl}(0)$, that are common eigenfunctions of the operators $H_{\S^2} = A_{\S^2}(0)$ and $D_\varphi$. Indeed the Hamiltonian $H_0$ is nothing but the Hamiltonian $H(0)$ introduced in (\ref{Hlambda}) for which separation of variable is available. Referring to Theorems \ref{Spectrum-A} and \ref{Separation} for the details, we consider precisely the following decomposition of the Hilbert space $\H$
$$
  \H = \oplus_{(k,l) \in I} \H_{kl}(0),
$$
where
$$
  I = (\frac{1}{2} + \Z) \times \N^*,
$$
and
$$
  \H_{kl}(0) = L^2(\R, \C^2) \otimes Y_{kl}(0) \simeq L^2(\R,\C^2) := \h.
$$
These reduced Hilbert spaces remain invariant under the action of $H_0$ and we are led to study the family of one-dimensional Dirac operators
$$
  H_{0 | \H_{kl}(0)} = H_{kl}(0) := \Ga D_x + \mu_{kl}(0) a(x) \Gb + c(x,k).
$$

Hence, in order to obtain a complete scattering theory for $H_0$, it is enough to obtain a complete scattering theory for each of the one-dimensional Dirac operators $H_{kl}(0)$, a much simpler problem already studied in \cite{Da1} and which is a particular case of the one studied in Section \ref{Direct-Stat-Scat}. Recall indeed that the Hamiltonian $H_{kl}(0)$ defined above is simply the Hamiltonian $H_{kz}$ with $z = \mu_{kl}(0)$ defined in (\ref{Hkz}), \textit{i.e.}
$$
  H_{kl}(0) = H_k(\mu_{kl}(0)),
$$
and for which a complete scattering theory has been obtained in Theorem \ref{WO-Hkz-Hkpm}. We recall these results here in our particular case. Introducing the asymptotic Hamiltonians
$$
  H^\pm_k = (\Ga D_x + \Omega_+(k)) P_\pm + (\Ga D_x + \Omega_-(k)) P_\mp,
$$
where $P_\pm = \mathbf{1}_{\R^\pm}(\Ga)$, we can prove the following Proposition

\begin{prop} \label{WO-H0-Hpm-kl}
  For each $(k,l) \in I$, the Hamiltonians $H_{kl}(0)$ and $H^\pm_k$ have purely absolutely continuous spectra, precisely
  $$
    \sigma(H_{kl}(0))= \sigma_{ac}(H_{kl}(0)) = \R, \quad \sigma(H^\pm_k)= \sigma_{ac}(H^\pm_k) = \R,
  $$
  and the wave operators
  $$
    W^\pm(H_{kl}(0), H^\pm_k) := s-\lim_{t \to \pm \infty} e^{itH_{kl}(0)} e^{-itH^\pm_k},
  $$
  exist on $\h$ and are asymptotically complete, \textit{i.e.} they are isometries on $\h$ and their inverse wave operators given by
  $$
    (W^\pm(H_{kl}(0), H^\pm_k))^* = W^\pm(H^\pm_{k}, H_{kl}(0)) := s-\lim_{t \to \pm \infty} e^{itH^\pm_{k}} e^{-itH_{kl}(0)},
  $$
  also exist on $\h$.
\end{prop}

By summing over all the generalized spherical harmonics $Y_{kl}(0)$, it is immediate to re-write this result in a global form as follows. Define the global asymptotic Hamiltonians
\begin{equation} \label{Hpm}
  H^\pm = (\Ga D_x + \Omega_+(D_\varphi)) P_\pm + (\Ga D_x + \Omega_-(D_\varphi)) P_\mp.
\end{equation}
Then
\begin{theorem} \label{WO-H0-Hpm}
  The wave operators
  $$
    W^\pm(H_0,H^\pm) = s-\lim_{t \to \pm \infty} e^{itH_0} e^{-itH^\pm},
  $$
  exist on $\H$ and are asymptotically complete, \textit{i.e.} they are isometries on $\H$ and their inverse wave operators given by
  $$
    (W^\pm(H_0,H^\pm))^* = W^\pm(H^\pm,H_0) = s-\lim_{t \to \pm \infty} e^{itH^\pm} e^{-itH_0},
  $$
  also exist on $\H$.
\end{theorem}

Finally we can finish the construction and the proof of existence of global wave operators associated to our Dirac Hamiltonian $H$ and to the natural asymptotic Hamiltonians $H^\pm$. Using Theorems \ref{WO-H-H0} and \ref{WO-H0-Hpm} and the chain-rule for wave operators, we obtain

\begin{theorem} \label{WO}
  The global wave operators
  $$
    W^\pm(H,H^\pm,I_2) = s-\lim_{t \to \pm \infty} e^{itH} e^{-itH^\pm},
  $$
  exist as operators from $\H$ to $\G$ and are asymptotically complete, \textit{i.e.} they are isometries from $\H$ to $\G$ and their inverse wave operators given by
  $$
    (W^\pm(H,H^\pm,I_2)^* = W^\pm(H^\pm,H,J) = s-\lim_{t \to \pm \infty} e^{itH^\pm} J e^{-itH},
  $$
  also exist as operators from $\G$ to $\H$.
\end{theorem}

We are now in position to define the global time-dependent scattering operator by the standard rule
\begin{equation} \label{TD-ScatOp}
  S = (W^+(H,H^+,I_2))^* \ W^-(H,H^-,I_2),
\end{equation}
which is clearly a unitary operator from $\H$ to $\H$. In order to finish the construction of the corresponding time-dependent scattering matrix $S(\lambda)$ at energy $\lambda$, we introduce the following unitary operators on $\H$
\begin{equation} \label{F+}
  F^+ \psi(\lambda) = \frac{1}{\sqrt{2\pi}} \int_\R  \left( \begin{array}{cc} e^{-ix (\lambda - \Omega_+(D_\varphi))}&0\\0&e^{ix (\lambda - \Omega_-(D_\varphi))} \end{array} \right) \psi(x) dx,
\end{equation}
and
\begin{equation} \label{F-}
  F^- \psi(\lambda) = \frac{1}{\sqrt{2\pi}} \int_\R  \left( \begin{array}{cc} e^{-ix (\lambda - \Omega_-(D_\varphi))}&0\\0&e^{ix (\lambda - \Omega_+(D_\varphi))} \end{array} \right) \psi(x) dx.
\end{equation}
As already explained in Section \ref{Direct-Stat-Scat} and more precisely in (\ref{F+k})-(\ref{F-k}), these operators diagonalize the Hamiltonians $H^+$ and $H^-$ respectively. Hence, we define the global scattering matrix at energy $\lambda$ in a natural way by the rule
\begin{equation} \label{TD-Scat}
  S = (F_+)^* S(\lambda) F_-.
\end{equation}

At this stage, we have obtained a complete time-dependent scattering theory for the Dirac Hamiltonian $H$. The meaning of the above definitions for the wave operators is the following. At late times (from the point of view of a stationary observer for which the variable $t$ corresponds to proper time), the energy of massless Dirac fields escape towards the two asymptotic regions: the event and cosmological horizons. This is a direct consequence of the absence of pure point spectrum for $H$. Moreover, in these regions, the massless Dirac fields are shown to obey simpler evolutions governed by the asymptotic Hamiltonians $H^\pm$. From the expressions (\ref{Hpm}) of $H^\pm$ and (\ref{NullVectors}) for the principal null geodesics, we imediately see that, at late times and from the point of view of a stationary observer, the Dirac fields simply obey a system of transport equations along the incoming and ougoing principal null geodesics in the neighbourhood of the event and cosmological horizons.

%
%

\subsection{Link between the stationary and time-dependent expressions of the scattering matrices}

In this section, we make the link between the stationary scattering matrix $S(\lambda)$ defined by (\ref{SM-Global-Physical}) in Theorem \ref{SM-Physical} and the time-dependent scattering matrix $S(\lambda)$ defined by (\ref{TD-Scat}). We shall show that the two definitions coincide, justifying then the use of the same notation for the two objects.

To do this, we start from the time-dependent scattering matrix $S(\lambda)$ defined by (\ref{TD-Scat}) and shall obtain a stationary representation of $S(\lambda)$ following the approach of Kuroda \cite{Ku} or Isozaki-Kitada \cite{Is-Ki}. We also refer to Yafaev (\cite{Y}) where similar formulae are proved. The procedure is quite standard in the case of a general Hamiltonian $H$ which is a short-range perturbation by a potential of a given Hamiltonian $H_0$. This is not so easy in our case however for two reasons.
\begin{itemize}
\item The Hamiltonian $H$ is a short-range perturbation \emph{of order 1} of the asymptotic Hamiltonians $H^\pm$.
\item The dynamics $e^{-itH}$ is compared with \emph{different} dynamics $e^{-itH^\pm}$ when $t \to \pm \infty$.
\end{itemize}
Both problems entail technical difficulties in adapting the methods of \cite{Ku, Is-Ki,Y}. Finally, the stationary representation of $S(\lambda)$ obtained in this way from (\ref{TD-Scat}) will be shown to coincide with the stationary representation (\ref{SM-Global-Physical}) obtained through the separation of variables procedure.

%
%

\subsubsection{First simplifications}

In order to avoid the second above difficulty, we shall slightly simplify the problem as follows. Recall that the fact that we compare the dynamics $e^{-itH}$ with different dynamics $e^{-itH^\pm}$ at late times follows from the presence of the long-range potential $c(x,D_\varphi)$ in the expression of $H$. But, in Section \ref{Simplified-SM}, we removed this potential of the equation by introducing a convenient unitary transform (\ref{Unitary-Uk}) on each generalized spherical harmonics $Y_{kl}(\lambda)$. This leaded us to define a simplified stationary scattering matrix $\hat{S}(\lambda)$ expressed in terms of the Jost functions solutions of the simplified stationary equation. We shall follow the same approach here, that is first remove the potential $c(x,D_\varphi)$ from the expression of $H$ and obtain the corresponding expression for the time-dependent scattering matrix, and second obtain a stationary expression of this time-dependent scattering matrix by the route proposed  in \cite{Ku, Is-Ki,Y}.

Let us then introduce the unitary transform $U$ from $\H$ to $\H$ defined by
\begin{equation} \label{Unitary-U}
  U \psi = e^{-i C(x,D_\varphi) \Ga} \psi,
\end{equation}
where $C(x,D_\varphi)$ is defined by (\ref{Pot-C}). Note that this unitary transform is simply the unitary transform (\ref{Unitary-Uk}) when restricted onto the angular modes $\{e^{ik\varphi}\}, \ k \in 1/2 + \Z$. Let us introduce some new notations
\begin{eqnarray}
  \hat{J} & = & U^{-1} J U, \\
  \hat{H_0} & = & U^{-1} H_0 U = \Ga D_x + a(x) U^{-1} H_{\S^2} U, \\
  \hat{H} & = & U^{-1} H U = \hat{J} \hat{H_0}, \\
  H_\infty & = & \Ga D_x.
\end{eqnarray}
Note that the Hamiltonian $\hat{H}$ is still a short-range perturbation of order $1$ of the Hamiltonian $\hat{H_0}$, which in turn can now be viewed as a short-range perturbation of the Hamiltonian $H_\infty$ after decomposition onto convenient generalized spherical harmonics. In other words, after conjugation by the unitary transform $U$, we are able to compare the dynamics generated by $\hat{H}$ with a \emph{single} dynamics generated by $H_\infty$ at late times. Thus we have removed the second above difficulty.

Let us now write the details. We define the new Hilbert space
$$
  \hat{\G} = L^2(\R \times \S^2, dx d\theta d\varphi; \C^2),
$$
equipped with the scalar product $(.,\hat{J}.)_\H$. Then from Thm \ref{Spectra-H-H0} we immediately have

\begin{lemma}
  1) The operators $\hat{J}, \hat{H_0}, H_\infty$ are selfadjoint on $\H$. \\
  2) The operator $\hat{H}$ is selfadjoint on $\hat{\G}$. \\
  3) The transform $U$ is unitary on $\H$ and isometric from $\hat{\G}$ to $\G$. \\
  4) The dynamics $e^{-it\hat{H}} = U^{-1} e^{-itH} U$ is unitary on $\hat{\G}$.
\end{lemma}

Let us also introduce the operators
\begin{equation} \label{Gpm}
  G^\pm = s-\lim_{t \to \pm \infty} e^{itH_\infty} U^{-1} e^{-itH^\pm},
\end{equation}
and
\begin{equation} \label{BetaDphi}
  \beta(D_\varphi) = \int_{-\infty}^0 \left[ c(s,D_\varphi) - \Omega_-(D_\varphi) \right] ds + \int_0^{+\infty} \left[ c(s,D_\varphi) - \Omega_+(D_\varphi) \right] ds.
\end{equation}
These operators are simply the operators (\ref{Gkpm}) and (\ref{Betak}) when restricted onto the angular modes $\{e^{ik\varphi}\}, \ k \in 1/2 + \Z$. Using Lemma \ref{G-k}, we thus get

\begin{lemma} \label{G}
  The operators $G^\pm$ are unitary on $\H$ and we have
  \begin{equation} \label{G-1}
    G^\pm = e^{i\Ga \left[ \left( \Omega_+(D_\varphi) + \beta(D_\varphi) + K \right) P_\pm \ + \ \left( \Omega_-(D_\varphi) + K \right) P_\mp \right]},
  \end{equation}
  where $K$ is the constant of integration in (\ref{Pot-C}) and $P_\pm = \mathbf{1}_{\R^\pm}(\Ga)$. Moreover, we have
  \begin{equation} \label{G-2}
    \F G^+ = e^{i\Ga K} \left( \begin{array}{cc} e^{i\beta(D_\varphi)}&0\\0&1 \end{array} \right) F^+,
  \end{equation}
  and
  \begin{equation} \label{G-3}
    \F G^- = e^{i\Ga K} \left( \begin{array}{cc} 1&0\\0&e^{-i\beta(D_\varphi)} \end{array} \right) F^-.
  \end{equation}
\end{lemma}

We are now able to prove

\begin{theorem} \label{WOhat}
  The wave operators
  $$
    W^\pm(\hat{H},H_\infty,I_2) = s-\lim_{t \to \pm \infty} e^{it\hat{H}} e^{-itH_\infty},
  $$
  exist as operators from $\H$ to $\hat{\G}$ and are asymptotically complete, \textit{i.e.} they are isometries from $\H$ to $\hat{\G}$ and their inverse wave operators given by
  $$
    (W^\pm(\hat{H},H_\infty,I_2)^* = W^\pm(H_\infty,\hat{H},\hat{J}) = s-\lim_{t \to \pm \infty} e^{itH_\infty} \hat{J} e^{-it\hat{H}},
  $$
  also exist as operators from $\hat{\G}$ to $\H$. Moreover, we have
  $$
    W^\pm(\hat{H},H_\infty,I_2) = U^{-1} W^\pm(H,H^\pm,I_2) G^\pm,
  $$
  $$
    W^\pm(H_\infty,\hat{H},\hat{J}) = G^\pm W^\pm(H^\pm,H,J) U.
  $$
\end{theorem}
\begin{proof}
  The first part of the Thm is a direct consequence of Thm \ref{WO}, Lemma \ref{G} and the unitarity of $U$. The second part of the Thm is the same as the calculations obtained in (\ref{WO-1}) on a fixed generalized spherical harmonics. We omit the details.
\end{proof}

Using Thm \ref{WOhat}, we can thus define the scattering operator unitary on $\H$
\begin{equation} \label{SOhat}
  \hat{S} = (W^+(\hat{H},H_\infty,I_2))^* W^-(\hat{H},H_\infty,I_2).
\end{equation}
Recalling that the unitary transform $\F$ given by (\ref{F}) diagonalizes $H_\infty$, we define the corresponding scattering matrix $\hat{S}(\lambda)$ unitary on $L^2(\S^2, \C^2)$ by
\begin{equation} \label{SMhat}
  \hat{S} = \F^* \hat{S}(\lambda) \F.
\end{equation}

Finally, using the same calculation as in (\ref{SM-1}), we can make the link between the physical time-dependent scattering matrix $S(\lambda)$ and the simplified one $\hat{S}(\lambda)$. Precisely, we get
\begin{equation} \label{LinkSMhat}
  S(\lambda) = \left( \begin{array}{cc} e^{-i\beta(D_\varphi)}&0\\0&1 \end{array} \right) e^{-i\Ga K} \hat{S}(\lambda) e^{i\Ga K} \left( \begin{array}{cc} 1&0\\0&e^{-i\beta(D_\varphi)} \end{array} \right)
\end{equation}

%
%

\subsubsection{Stationary formulation for $\hat{S}(\lambda)$} \label{Link}

In this Section, we show that the simplified time-dependent scattering matrix $\hat{S}(\lambda)$ obtained in (\ref{SMhat}) is nothing but the simplified stationary scattering matrix $\hat{S}(\lambda)$ obtained in Thm \ref{SM-Unphysical} by means of separation of variables. Hence the equivalence of the two definitions and our use of the same notation.

Let us obtain thus a stationary formulation for the one shell energy scattering matrix $\hat{S}(\lambda)$ defined by (\ref{SMhat}). To do this,  we adapt the well-known Kuroda's approach in our context, (\cite{Ku}, \cite{Is-Ki}).

\vspace{0.5cm} \noindent For $l \in \mathbb{R}$, we define
\begin{equation}
\H_l = L^2 (\mathbb{R}\times \S^2; <x>^{2l} dx d\theta d\varphi; \mathbb{C}^2),
\end{equation}
and for $ l>\frac{1}{2}$, we set
\begin{eqnarray}\label{fo}
\Gamma_0(\lambda) : &\mathcal{H}_l& \ \rightarrow \  L^2(\S^2;\C^2) \nonumber \\
                      & \Psi & \ \rightarrow \ \frac{1}{\sqrt{2\pi}} \int_{\R} e^{-i\Gamma^1x\lambda}\ \Psi(x,\theta,\varphi)\ dx.
\end{eqnarray}
Clearly, $\Gamma_0(\lambda)$ is a bounded operator from $\H_l$ to $\H_{\S^2}=L^2(\S^2, \C^2)$, and we have for $ l>\frac{1}{2}$,
\begin{eqnarray}\label{foad}
\Gamma_0(\lambda)^* : &L^2(\S^2;\C^2)& \ \rightarrow \  \mathcal{H}_{-l} \nonumber \\
                      & \Psi & \ \rightarrow  \ e^{i\Gamma^1x\lambda}\ \Psi(\theta,\varphi).
\end{eqnarray}
The operator $\Gamma_0(\lambda)$ is the one-shell energy restriction of the unitary transform $\F$ defined in (\ref{F}) that diagonalizes the operator $H_\infty$. We shall use constantly the spectral decomposition for $H_\infty$, that is the following relation: for all $u,v \in \H_l, \ l> 1/2$, we have
\begin{equation} \label{Gamma0}
  (u,v) = (\F u, \F v) = \int_{\R_\lambda} < \Gamma_0(\lambda) u, \Gamma_0(\lambda) v > d\lambda,
\end{equation}
where $<.,.>$ denotes here the scalar product in $ \H_{\S^2}$.

Since $\hat{S}$ commutes with $H_{\infty}$ and
\begin{equation} \label{Gamma0Hinfty}
  \go H_{\infty}= \lambda \go,
\end{equation}
on a suitable domain, $\hat{S}$ is a decomposable operator (see \cite{RS}). It means in particular that for a.e $\lambda \in \R$, there exists unitary operators ${\hat{S}}(\lambda)$ on $\H_{\S^2}$ such that, for $f=u_1\otimes v_1, \ g=u_2\otimes v_2$, with $\hat{u}_j \in C_0^{\infty}(\R, \C^2)$ and $v_j \in L^2(\S^2,\C^2)$,

\begin{equation} \label{decomposition}
((\hat{S}-Id)f, g ) =  \int_{-\infty}^{+\infty} < (\hat{S}(\lambda)-Id) \go f, \go g> \ d\lambda.
\end{equation}


Recall that the space $\H_{\S^2}$ can be decomposed first onto the angular modes $\{e^{ik\varphi}\}, \ k \in 1/2 + \Z$, then onto the generalized spherical harmonics $Y_{kl}(\lambda_0)$ where $\lambda_0 \in \R$ is any fixed energy. Precisely, we first have the decomposition
$$
  \H_{\S^2} = \oplus_{k \in \frac{1}{2} + \Z} \H_{\S^2}^k, \quad \H_{\S^2}^k = L^2((0,\pi); \C^2) \otimes e^{ik\varphi} \simeq L^2((0,\pi); \C^2).
$$
Second, each $\H_{\S^2}^k$ can be decomposed onto
$$
\H_{\S^2}^k = \oplus_{l \in \N^*} \H_{\S^2}^{kl}(\lambda_0), \quad \H_{\S^2}^{kl}(\lambda_0) = \C^2 \otimes Y_{kl}(\lambda_0) \simeq \C^2.
$$
This above decomposition explains why we shall study the quantity $((\hat{S}-Id)f, g )$ with test functions $f$ and $g$ having the form
$ f= \phi\otimes Y_{kl}(\lambda_0)$ and $ g= \psi\otimes Y_{pq}(\lambda_0)$ with $\hat{\phi}, \hat{\psi}$ with compact support.

We start rewriting the wave operators $W^{\pm} := W^{\pm}(\hat{H}, H_{\infty}; Id)$ defined in Thm \ref{WOhat} with their so-called abelian limits in order to systematically deal with well-defined integral expressions. Precisely, we have
\begin{eqnarray*}
W^{\pm} f   &=& \lim_{t \rightarrow \pm \infty} e^{it\hat{H}} \ Id \ e^{-itH_{\infty}} f  \\
            &=& \lim_{\mu \rightarrow 0^+} \  \mu \ \int_0^{\pm \infty} e^{-\mu \mid t\mid} \ e^{it\hat{H}} \ Id \ e^{-itH_{\infty}} f \ dt,
\end{eqnarray*}
where we recall that the operator $Id$ is used as identification operator between the Hilbert spaces $\H$ and $\hat{\G}$. Integrating by parts, we obtain
\begin{equation}\label{abelien}
(W^{\pm} - Id)f \ =\  \lim_{\mu \rightarrow 0^+}  \ i\  \int_0^{\pm \infty} e^{-\mu \mid t\mid} \ e^{it\hat{H}} \ T \ e^{-itH_{\infty}} f \ dt,
\end{equation}
where we have set $T := \hat{H} \, Id - Id \, H_{\infty}$. Thus, using the classical intertwining property for the wave operators $W^\pm$, we can write
\begin{eqnarray}
((\hat{S}-Id)f, g )  &=& \left( W^{+*}(W^--W^+) f,g \right) \nonumber \\
                    &=& - \lim_{\mu \rightarrow 0^+}  \ \left( W^{+*} (\  i \int_{-\infty}^{+ \infty} e^{-\mu \mid t\mid} \ e^{it\hat{H}} \
                         T \ e^{-itH_{\infty}} f \ dt), g\right) \nonumber \\
                    &=& - \lim_{\mu \rightarrow 0^+}  \ \left(  ( \ i \int_{-\infty}^{+ \infty} e^{-\mu \mid t\mid} \ e^{itH_{\infty}} W^{+*}\
                         T \ e^{-itH_{\infty}} f \ dt), g\right) \nonumber \\
                    &=& - \lim_{\mu \rightarrow 0^+} \ (I_1(\mu)+I_2(\mu)), \label{I1I2}
\end{eqnarray}
where
\begin{eqnarray*}
I_1(\mu) &=& \left(  ( \ i \int_{-\infty}^{+ \infty} e^{-\mu \mid t\mid} \ e^{itH_{\infty}} \hat{J}\
                         T \ e^{-itH_{\infty}} f \ dt), g\right), \\
I_2(\mu) &=& \left(  ( \ i \int_{-\infty}^{+ \infty} e^{-\mu \mid t\mid} \ e^{itH_{\infty}} (W^{+*}- \hat{J})\
                         T \ e^{-itH_{\infty}} f \ dt), g\right).
\end{eqnarray*}



Let us first examine $I_1(\mu)$. Using the spectral decomposition (\ref{Gamma0})-(\ref{Gamma0Hinfty}) for $H_{\infty}$ and the resolvent formula, we have :
\begin{eqnarray*}
I_1(\mu) &=& \int_{-\infty}^{+\infty} < \go \hat{J} T \left( i\int_{-\infty}^{+\infty} e^{-\mu \mid t\mid} \ e^{-it(H_{\infty}-\lambda)} f\ dt \right) ,\go g> \
             d\lambda \\
         &=& \int_{-\infty}^{+\infty} < \go \hat{J} T \left( R_{\infty}(\lambda+i\mu) - R_{\infty}(\lambda-i\mu)\right)f ,\go g> \ d\lambda,
\end{eqnarray*}
where $R_{\infty}(z) = (H_\infty - z)^{-1}$ denotes the resolvent for $H_{\infty}$.  In order to use carefully the limiting absorption principle for $H_{\infty}$, we remark that the unbounded operator $\hat{J} T$ can be written as
\begin{eqnarray}\label{jyvais}
\hat{J}T &=& \hat{H}_0 - \hat{J} H_{\infty} \nonumber\\
        &=& a(x)\  e^{2iC(x,k) \Gamma^1} \ (H_{\S^2} -b(\theta) \Gamma^3\Gamma^1 D_x).
\end{eqnarray}
Hence, the term
\begin{eqnarray*}
\hat{J} T  (R_{\infty}(\lambda+i\mu) - R_{\infty}(\lambda-i\mu))f &= &
a(x)\  e^{2iC(x,k) \Gamma^1} \ [ (R_{\infty}(\lambda+i\mu) - R_{\infty}(\lambda-i\mu))\phi \otimes (H_{\S^2}Y_{kl}(\lambda_0))   \\
& & - \Gamma^3 (R_{\infty}(\lambda+i\mu) - R_{\infty}(\lambda-i\mu)) (\Gamma^1 D_x\phi) \otimes (b(\theta) Y_{kl}(\lambda_0))],
\end{eqnarray*}
is well-defined for all $\mu > 0 $ since $ f= \phi\otimes Y_{kl}(\lambda_0)$ with $\hat{\phi}$ with compact support. Thus, we can apply the limiting absorption principle for $H_\infty$ and Stone's theorem to deduce
\begin{eqnarray}\label{simpl1}
\lim_{\mu \rightarrow 0^+} I_1(\mu) &=& \int_{-\infty}^{+\infty} < \go \hat{J} T \left( R_{\infty}(\lambda+i0) - R_{\infty}(\lambda-i0)\right)f ,\go g>
                                        \ d\lambda \nonumber \\
                                     &=&  2i\pi \ \int_{-\infty}^{+\infty} < \go \hat{J} T \go^* \go f ,\go g> \ d\lambda.
\end{eqnarray}
Now, using (\ref{jyvais}) again and the fact that the operators $\Gamma_0(\lambda)$ and $\Gamma_0^*(\lambda)$ let invariant the special form of $f$, for instance
$$
  \Gamma_0(\lambda) f = (\Gamma_0(\lambda) \phi) \otimes Y_{kl}(\lambda_0),
$$
we see that
\begin{eqnarray}\label{simpl2}
\hat{J} T \go^* \go f &=& a(x)\  e^{2iC(x,k) \Gamma^1} \ \left( H_{\S^2} -\lambda b(\theta) \Gamma^3 \right) \go^* \go f \nonumber \\
                      &=& a(x)\  e^{2iC(x,k) \Gamma^1} \ \left( (H_{\S^2} -\lambda_0 b(\theta) \Gamma^3) + (\lambda_0-\lambda) b(\theta) \Gamma^3\right) \go^* \go f
                      \nonumber \\
                      &=& a(x)\  e^{2iC(x,k) \Gamma^1} \ \left( A_{\S^2}(\lambda_0) + (\lambda_0-\lambda) b(\theta) \Gamma^3\right) \go^* \go f \nonumber \\
                      &=& a(x)\  e^{2iC(x,k) \Gamma^1} \ \left( \mu_{kl}(\lambda_0)\Gamma^2 + (\lambda_0-\lambda) b(\theta) \Gamma^3)\right) \go^* \go f \nonumber \\
                      &=& a(x)\ \left( \mu_{kl}(\lambda_0)\hat{\Gamma}^2 + (\lambda_0-\lambda) b(\theta) \hat{\Gamma}^3\right) \go^* \go f \nonumber \\
                      &=& \left( \mu_{kl}(\lambda_0) a(x) \hat{\Gamma}^2 + (\lambda_0-\lambda) (\hat{J}-1)\right) \go^* \go f,
\end{eqnarray}
where
$$
  \hat{\Gamma}^2 = e^{2iC(x,k) \Gamma^1} \Gb, \quad \hat{\Gamma}^3 = e^{2iC(x,k) \Gamma^1} \Gc.
$$
Thus, we deduce from (\ref{simpl1}) and (\ref{simpl2}) that
\begin{eqnarray}\label{limI1}
\lim_{\mu \rightarrow 0^+} I_1(\mu)  &=&  2i\pi \ \int_{-\infty}^{+\infty} < \go \big[ \mu_{kl}(\lambda_0) a(x) \hat{\Gamma}^2 \nonumber \\
                                      & & \hspace{2cm} + (\lambda_0-\lambda) (\hat{J}-1) \big]\go^* \go f ,\go g> \ d\lambda.
\end{eqnarray}

\vspace{0.5cm}
Now, let us study $I_2(\mu)$. As in (\ref{abelien}), the existence of $W^{+*}$ ensures that
\begin{equation}
W^{+*}- \hat{J}  =  s-\lim_{\nu \rightarrow 0^+} \ -i  \int_0^{+\infty}  e^{-s\nu} \  e^{isH_{\infty}} T^* e^{-is\hat{H}}   ds.
\end{equation}
Thus,
\begin{eqnarray}
I_2(\mu)& =& - \lim_{\nu \rightarrow 0^+} \Big(  \Big( \ i \int_{-\infty}^{+\infty} e^{-\mu \mid t\mid} \ e^{itH_{\infty}}
              \Big[ i  \int_0^{+\infty}  e^{-s\nu} \  e^{isH_{\infty}} T^* e^{-is\hat{H}}   ds \Big]  \nonumber \\
        & &   \hspace{2cm}  T \ e^{-itH_{\infty}} f \ dt \Big), g \Big) \nonumber \\
        & =& - \lim_{\nu \rightarrow 0^+} \int_{-\infty}^{+\infty} \Big< \go  T^* \Big[i  \int_0^{+\infty}  e^{-s\nu} \ e^{-is(\hat{H}-\lambda)}  ds \Big] T \nonumber \\
        & & \hspace{2cm}  \ \Big(i \int_{-\infty}^{+ \infty} e^{-\mu \mid t\mid} \ e^{-it(H_{\infty}-\lambda)}f \ dt \Big), \go g \Big>\ d\lambda,
\end{eqnarray}
where we have used again the spectral decomposition (\ref{Gamma0})-(\ref{Gamma0Hinfty}) for $H_{\infty}$. Thus, setting $\hat{R}(z)= (\hat{H}-z)^{-1}$, we obtain using the resolvent formula
\begin{eqnarray}
I_2(\mu)&=&  - \lim_{\nu \rightarrow 0^+} \int_{-\infty}^{+\infty} <\go  T^* \hat{R}(\lambda+i\nu) T \nonumber \\
        & & \hspace{2.3cm}  \ [R_{\infty}(\lambda+i\mu) - R_{\infty}(\lambda-i\mu)]f , \go g >\ d\lambda,\nonumber\\
        &=& - \lim_{\nu \rightarrow 0^+} \ I(\mu, \nu),
\end{eqnarray}
where
\begin{eqnarray}
I(\mu, \nu) &=&  \int_{-\infty}^{+\infty} \Big( \hat{R}(\lambda+i\nu) T  \ [R_{\infty}(\lambda+i\mu) - R_{\infty}(\lambda-i\mu) ]f, T \go^* \go g \Big)_{\hat{\G}} \ d\lambda
\nonumber \\
            &=&  \int_{-\infty}^{+\infty} \Big( \hat{R}(\lambda+i\nu) T  \ [R_{\infty}(\lambda+i\mu) - R_{\infty}(\lambda-i\mu) ]f, \nonumber \\
            & & \hspace{1cm} \hat{J} T \go^* \go g \Big)_{\H} \ d\lambda
\end{eqnarray}

\vspace{0.3cm}
Now, let us recall that $J= I_2 +\alpha(r,\theta) \Gamma^3$ with ${\displaystyle{\sup_{r,\theta} \ \alpha(r,\theta) <1}}$. Then there exists $c>0$ such that $\hat{J}^{-1} \geq c$ in the sense of operators on $\hat{\G}$. For $\nu>0$, we set
\begin{equation}
H(\nu) = \hat{H} -i\nu \hat{J}^{-1}.
\end{equation}
It follows from (\cite{Roy}, Lemma 2.1) that $H(\nu)$ is a maximal dissipative operator with domain $D(H(\nu)) = D(\hat{H})$,
and for all $\lambda \in \R$, $H(\nu) -\lambda$ is invertible with bounded inverse. Then, using the resolvent identity, we can split
$I(\mu,\nu)$ into two terms by $I(\mu,\nu) = I_1(\mu,\nu)+ I_2(\mu,\nu)$ where
\begin{eqnarray}
I_1(\mu, \nu) &=&  \int_{-\infty}^{+\infty} \Big( (H(\nu)-\lambda)^{-1} T  \ [R_{\infty}(\lambda+i\mu) - R_{\infty}(\lambda-i\mu) ]f, \nonumber \\
              & &  \hspace{1.4cm} \hat{J} T \go^* \go g \Big)_{\H} \ d\lambda, \\
I_2(\mu,\nu)  &=&  i\nu \ \int_{-\infty}^{+\infty} \Big(\hat{R}(\lambda+i\nu) (\hat{J}^{-1}-1) (H(\nu)-\lambda)^{-1} T  \nonumber \\
              & & \hspace{2cm}     [R_{\infty}(\lambda+i\mu) - R_{\infty}(\lambda-i\mu) ]f, \hat{J} T \go^* \go g \Big)_{\H}\ d\lambda.
\end{eqnarray}

\vspace{0.3cm} \noindent
First, let us study $I_1(\mu, \nu)$. We easily see that
\begin{eqnarray}
I_1(\mu, \nu) &=&  \int_{-\infty}^{+\infty} \Big((\hat{H}_0 -\lambda \hat{J} -i\nu)^{-1} \hat{J} T  \ [R_{\infty}(\lambda+i\mu) - R_{\infty}(\lambda-i\mu) ]f, \nonumber \\
              & &  \hspace{1.4cm} \hat{J} T \go^* \go g \Big)_{\H} \ d\lambda.
\end{eqnarray}
Thus, using the limiting absorption principle for $\hat{H}_0 -\lambda \hat{J}$ given in Proposition \ref{LAP-L0} and the same simplifications as in (\ref{simpl2}), we obtain
\begin{eqnarray}
\lim_{\mu,\nu \rightarrow 0^+} \ I_1(\mu, \nu) &=&  2i\pi \int_{-\infty}^{+\infty} \Big((\hat{H}_0 -\lambda \hat{J} -i0)^{-1} \hat{J} T  \
                                                    \go^* \go f, \nonumber \\
                                               & &  \hspace{2cm} \hat{J} T \go^* \go g \Big)_{\H} \ d\lambda. \nonumber \\
                                               &=&  2i\pi \int_{-\infty}^{+\infty}  \Big((\hat{H}_0 -\lambda \hat{J} -i0)^{-1}
                                                    [\mu_{kl}(\lambda_0) a(x) \hat{\Gamma}^2 + (\lambda_0-\lambda) (\hat{J}-1)] \go^* \go f, \nonumber \\
                                               & & \hspace{2cm} [ \mu_{pq}(\lambda_0) a(x) \hat{\Gamma}^2 + (\lambda_0-\lambda)
                                               (\hat{J}-1)] \go^* \go g \Big)_{\H} \ d\lambda.
\end{eqnarray}
We now observe that the operator $\hat{H}_0-\lambda \hat{J}   = \Gamma^1 D_x + a(x) e^{2iC(x,k) \Gamma^1} A_{S^2}(\lambda) -\lambda$ acts on $\H_{\S^2}^{kl}(\lambda)$ by
$$
  (\hat{H}_0-\lambda \hat{J}) (\phi \otimes Y_{kl}(\lambda)) = (\hat{H}_{kl}(\lambda) \phi) \otimes Y_{kl}(\lambda),
$$
where we have used the notation from (\ref{HkzHat})
\begin{equation} \label{HatV}
  \hat{H}_{kl}(\lambda)= \Gamma^1 D_x +\mu_{kl}(\lambda) a(x) \hat{\Gamma}^2 = \Ga D_x + \mu_{kl}(\lambda) V_k(x).
\end{equation}
Then we have,
\begin{eqnarray}
\lim_{\mu,\nu \rightarrow 0^+} \ I_1(\mu, \nu) &=&  2i\pi \int_{-\infty}^{+\infty} \Big((\hat{H}_{kl}(\lambda) -\lambda -i0)^{-1}
                                               [ \mu_{kl}(\lambda_0) a(x) \hat{\Gamma}^2 + (\lambda_0-\lambda) (\hat{J}-1)]\  \go^* \go f, \nonumber \\
                                               & & \hspace{2cm} [ \mu_{pq}(\lambda_0) a(x) \hat{\Gamma}^2 + (\lambda_0-\lambda)
                                               (\hat{J}-1)] \go^* \go g \Big)_{\H} \ d\lambda.
\end{eqnarray}
Thus, we have obtained
\begin{eqnarray}
\lim_{\mu,\nu \rightarrow 0^+} \ I_1(\mu, \nu) &=&  2i\pi \int_{-\infty}^{+\infty} \Big< \go \ [ \mu_{pq}(\lambda_0) a(x) \hat{\Gamma}^2 + (\lambda_0-\lambda)
                                               (\hat{J}-1)] \ (\hat{H}_{kl}(\lambda) -\lambda -i0)^{-1} \nonumber \\
                                               & & \hspace{2cm}[ \mu_{kl}(\lambda_0) a(x) \hat{\Gamma}^2 + (\lambda_0-\lambda) (\hat{J}-1)] \
                                                   \go^* \go f,    \go g \Big> \ d\lambda. \nonumber
\end{eqnarray}
Similarly, we show easily that
\begin{equation}
\lim_{\mu,\nu \rightarrow 0^+} \ I_2(\mu, \nu) =0.
\end{equation}
In consequence, we get
\begin{eqnarray}
\lim_{\mu\rightarrow 0^+} \ I_2(\mu) &=&  -2i\pi \int_{-\infty}^{+\infty} \Big< \go \ [ \mu_{pq}(\lambda_0) a(x) \hat{\Gamma}^2 + (\lambda_0-\lambda)
                                               (\hat{J}-1)] \ (\hat{H}_{kl}(\lambda) -\lambda -i0)^{-1} \nonumber \\
                                               & & \hspace{1cm}[ \mu_{kl}(\lambda_0) a(x) \hat{\Gamma}^2 + (\lambda_0-\lambda) (\hat{J}-1)] \
                                                   \go^* \go f,    \go g \Big> \ d\lambda. \label{limI2}
\end{eqnarray}

Coming back to (\ref{decomposition})-(\ref{I1I2}) and using (\ref{Gamma0}), (\ref{limI1}), (\ref{limI2}), we have shown that for all $f =\phi \otimes Y_{kl}(\lambda_0)$ and $g= \psi \otimes Y_{pq}(\lambda_0)$, and for a.e $\lambda \in \mathbb{R}$,
\begin{eqnarray*}
<(\hat{S}(\lambda)- Id) \go f, \go g > &=& -2i\pi < \go \ [\mu_{kl}(\lambda_0) a(x) \hat{\Gamma}^2
                              + (\lambda_0-\lambda) (\hat{J}-1) ] \ \go^* \go f ,\go g> \nonumber \\
                             & & + 2i\pi  < \go \ [ \mu_{pq}(\lambda_0) a(x) \hat{\Gamma}^2 + (\lambda_0-\lambda)
                             (\hat{J}-1)] \ (\hat{H}_{kl}(\lambda) -\lambda -i0)^{-1} \nonumber \\
                             & & \hspace{2,5cm}[ \mu_{kl}(\lambda_0) a(x) \hat{\Gamma}^2 + (\lambda_0-\lambda) (\hat{J}-1)] \nonumber \\
                             & & \hspace{2,5cm} \go^* \go f,    \go g >.
\end{eqnarray*}
In particular, choosing $\lambda=\lambda_0$, we obtain
\begin{eqnarray}\label{block}
<(\hat{S}(\lambda)- Id) \go f, \go g > &=& -2i\pi < \go \ \mu_{kl}(\lambda) a(x) \hat{\Gamma}^2  \ \go^* \go f ,\go g> \nonumber \\
                             & & + 2i\pi  < \go \  \mu_{pq}(\lambda) a(x) \hat{\Gamma}^2  \ (\hat{H}_{kl}(\lambda) -\lambda -i0)^{-1} \nonumber \\
                             & & \hspace{2,5cm} \mu_{kl}(\lambda) a(x) \hat{\Gamma}^2 \go^* \go f,    \go g >.
\end{eqnarray}
Thus, since the $\go f, \go g$ are dense in $\H_{\S^2}^{kl}(\lambda)=\C^2 \otimes Y_{kl}(\lambda)$ and the spaces $\H_{\S^2}^{kl}(\lambda)$ form an orthogonal Hilbert decomposition of $\H_{\S^2}$, we can write (\ref{block}) concisely
\begin{equation} \label{StationaryShat}
\hat{S}(\lambda)= \oplus_{k,l} \left( \hat{S}_{kl}(\lambda) \otimes Y_{kl}(\lambda) \right),
\end{equation}
where $\hat{S}_{kl}(\lambda)$ are operators acting on $\C^2$ given by
\begin{equation} \label{Link1}
\hat{S}_{kl}(\lambda) = Id - 2i\pi  \go \ \mu_{kl}(\lambda) a(x) \hat{\Gamma}^2  \big[ Id - (\hat{H}_{kl}(\lambda) -\lambda)^{-1} \mu_{kl}(\lambda) a(x) \hat{\Gamma}^2 \big] \go^*,
\end{equation}
or using the notation (\ref{HatV})
\begin{equation} \label{Link2}
\hat{S}_{kl}(\lambda) = Id - 2i\pi  \go \ \mu_{kl}(\lambda) V_k(x)  \big[ Id - (\hat{H}_{kl}(\lambda) -\lambda)^{-1} \mu_{kl}(\lambda) V_k(x) \big] \go^*.
\end{equation}

We emphasize that (\ref{StationaryShat}) is the stationary expression of the simplified time-dependent scattering matrix $\hat{S}(\lambda)$ defined in (\ref{SMhat}). On each generalized spherical harmonics $Y_{kl}(\lambda)$, we have shown that its expression is simplified into $\hat{S}_{kl}(\lambda)$ given in (\ref{Link2}) which is nothing but the usual stationary expression for the scattering matrix associated to the pair of $1$-dimensional Hamiltonians $(\hat{H}_{kl}(\lambda), H_{\infty})$. Representation formulae for $\hat{S}_{kl}(\lambda)$ in terms of stationary solutions of the equation (the Jost functions) are well known and are studied for instance in \cite{DM}. In particular, the expression (\ref{Link2}) for $\hat{S}_{kl}(\lambda)$ coincide with the stationary definition in Thm \ref{SM-Unphysical} given by the separation of variables procedure. Whence the link between the two a priori distinct definitions for the scattering matrix $S(\lambda)$.


\Section{Uniqueness results in the inverse scattering problem at fixed energy} \label{MainResult}

In this Section, we state our main Theorem, that is a uniqueness result that, roughly speaking, asserts that the scattering matrix $S(\lambda)$ at a fixed energy $\lambda \in \R$ associated to massless Dirac fields evolving in a KN-dS black hole determines uniquely the parameters of the black hole and thus its metric. In fact, we shall obtain more and better results in the course of the proof of our main Thm.

First, it is enough to know \emph{both} reduced transmission operators $T_k^{R/L}(\lambda)$ or \emph{one} of the reduced reflection operators $R_k(\lambda)$ or $L_k(\lambda)$ and this, only for \emph{two} different angular modes $\{e^{ik\varphi}\}, \ k \in \frac{1}{2} + \Z$ in order to determine uniquely the black hole. We refer to Thm \ref{SM-Physical} for the notations. Remark that we cannot assume the knowledge of the fully reduced scattering coefficients $T_{kl}(\lambda)$, $R_{kl}(\lambda)$ or $L_{kl}(\lambda)$ onto the generalized spherical harmonics $Y_{kl}(\lambda)$ since the latters depend on the black hole, precisely on the two parameters $a$ and $\Lambda$, that we are trying to determine uniquely.

Second, we are able to recover more than only the four parameters $M,Q^2, a, \Lambda$ that characterize the black hole. The Complex Angular Momentum (CAM) method of Section \ref{Complexification} allows us indeed to determine \emph{functions} depending on the radial variable (up to diffeomorphisms), that is infinite dimensional objects. For instance, we are able to determine the function $\frac{\lambda - c(x,k)}{a(x)}$ (up to diffeomorphisms). Note that, in the particular case of Kerr-de-Sitter black hole ($Q=0$), we are able to determine the functions $a(x)$ and $c(x,k)$ (up to translations) separatly. That is we recover the potentials appearing in the separated radial equation (\ref{Radial-Eq}).

Let us state now our main Thm.

\begin{theorem} \label{Main}
Let $(M,Q^2,a,\Lambda)$ and $(\tilde{M},\tilde{Q}^2,\tilde{a},\tilde{\Lambda})$ be the parameters of two a priori different KN-dS black holes. Let $\lambda \in \R$ and denote by $S(\lambda)$ and $\tilde{S}(\lambda)$ the corresponding scattering matrices at fixed energy $\lambda$ (given by Thm \ref{SM-Physical}). More generally, we shall add a symbol $\ \tilde{}$ to all the relevant scattering quantities corresponding to the second black hole. Assume that both reduced transmission operators $T_k^{R}(\lambda)$, $T_k^{L}(\lambda)$ or one reduced reflection operators $R_k(\lambda)$ or $L_k(\lambda)$ are known in the sense that there exist constants $c_T(\lambda,k), c_R(\lambda,k), c_L(\lambda,k) \in \R$ such that one of the following equalities is fulfilled
\begin{eqnarray}
  T_{k}^{R/L}(\lambda) & = & e^{i c_T(\lambda,k)} \, \tilde{T}_{k}^{R/L}(\lambda), \nonumber \\
  R_{k}(\lambda) & = & e^{ic_R(\lambda,k)} \, \tilde{R}_{k}(\lambda), \label{MainAssumption} \\
  L_{k}(\lambda) & = & e^{ic_L(\lambda,k)} \, \tilde{L}_{k}(\lambda), \nonumber
\end{eqnarray}
as operators on $\ls = L^2((0,\theta); \C)$ and for \emph{two} different values of $k \in \frac{1}{2} + \Z$. Then the parameters of the two black holes coincide, \textit{i.e.}
$$
  M = \tilde{M}, \ a = \tilde{a}, \ Q^2 = \tilde{Q}^2, \ \Lambda = \tilde{\Lambda}.
$$
\end{theorem}

\begin{remark}
We emphasize that we add the constants $c_T(\lambda,k)$, $c_R(\lambda,k)$ and $c_L(\lambda,k)$ in the assumption (\ref{MainAssumption}) to include the possibility that a KN-dS black hole be described by two different Regge-Wheeler variables $x$ and $\tilde{x} = x +c$ for a constant $c$. Precisely, these constants have the following expression
$$
  c_T(\lambda,k) = -(\Omega_+(k) - \Omega_-(k))c, \quad c_R(\lambda,k) = 2\lambda_+(k) c, \quad c_L(\lambda,k) = -2\lambda_-(k) c.
$$
We refer to Section \ref{DependenceRW} and precisely to Proposition \ref{SM-Red-RW} for this subtlety and the notations. For simplicity, the reader is invited to take $c = 0$ and thus $c_T = c_R = c_L = 0$ in a first reading.
\end{remark}

Let us now explain the strategy for the proof of Thm \ref{Main}. We recall from Corollary \ref{SpectralProp-TrRe} that the quantities $T_{kl}(\lambda)$ for $l \in \N^*$ are the eigenvalues of the operators $T_{k}^{R}(\lambda)$ and $T_{k}^{L}(\lambda)$, whereas $|R_{kl}(\lambda)|^2$ and $|L_{kl}(\lambda)|^2$ are the eigenvalues of the operators $R_{k}^*(\lambda) R_{k}(\lambda)$ or $R_{k}(\lambda) R_{k}^*(\lambda)$ for the formers and $L_{k}^*(\lambda) L_{k}(\lambda)$ and $L_{k}(\lambda) L_{k}^*(\lambda)$ for the latters. Since these operators are uniquely determined by our main assumption (\ref{MainAssumption}), their eigenvalues (and associated eigenspaces) would be also uniquely determined if they were simple! But it turns out that we can prove this only for large enough $l \in \N^*$. This will be showed using a precise study of the asymptotic behaviour of the scattering coefficients $T(\lambda, k,z)$, $R(\lambda,k,z)$ and $L(\lambda,k,z)$ when $z \to +\infty$ given in Section \ref{AL-Strictly-Increasing}. At last, since $\{Y_{kl}^1\}_{l \in \N^*}$ and/or $\{Y_{kl}^2\}_{l \in \N^*}$ are the eigenfunctions associated to the "simple" eigenvalues $T_{kl}(\lambda)$, $|R_{kl}(\lambda)|^2$ and $|L_{kl}(\lambda)|^2$ for $l$ large enough, we shall obtain the following result

\begin{prop} \label{Consequence1}
  Under the assumption (\ref{MainAssumption}), there exists $L > 0$ such that for all $l \in \N^*, \ l \geq L$, one of the following conditions holds
  \begin{eqnarray}
  T_{kl}(\lambda) & = & e^{i c_T(\lambda,k)} \, \tilde{T}_{kl}(\lambda), \label{T=T} \\
  |R_{kl}(\lambda)| & = & |\tilde{R}_{kl}(\lambda)|, \\
  |L_{kl}(\lambda)| & = & |\tilde{L}_{kl}(\lambda)|.
  \end{eqnarray}
  Moreover for all $l \geq L$, there exists $\alpha_{kl}^{j}$ with $|\alpha_{kl}^{j}| = 1$ for $j=1,2$ such that
  \begin{equation} \label{UniqueY}
    Y_{kl}^1 = \alpha_{kl}^1 \tilde{Y}_{kl}^1, \quad Y_{kl}^2 = \alpha_{kl}^2 \tilde{Y}_{kl}^2.
  \end{equation}
\end{prop}

To go further, let us now analyse for instance some consequences of (\ref{T=T}) with $c_T(\lambda,k)$ supposed to be $0$ for simplicity. Recall that $T_{kl}(\lambda) = T(\lambda, k , \mu_{kl}(\lambda))$ where $\mu_{kl}(\lambda)$ are the eigenvalues of the angular operator $A_k(\lambda)$ (given in (\ref{Ak})) labeled according to their increasing order, \textit{i.e.} $\muk < \mu_{k,l+1}(\lambda), \ \forall l \in \N^*$. Then (\ref{T=T}) reads
\begin{equation} \label{Proof-Test1}
  \forall l \geq L, \quad T(\lambda, k , \mu_{kl}(\lambda)) =  \tilde{T}(\lambda, k , \tilde{\mu}_{kl}(\lambda)).
\end{equation}
The main idea of this work is to use next the CAM method, that is to complexify the angular momentum $\muk$ and study the analytic properties of the corresponding scattering coefficients. Precisely, we shall show for instance in Section \ref{Complexification} that the function
$$
  z \longrightarrow \frac{1}{T(\lambda, k , z)},
$$
is entire with respect to $z \in \C$ (the other parameters $\lambda$ and $k$ being fixed) and belongs to a certain class of analytic functions - the Nevanlinna class - which possess good uniqueness properties: a Nevanlinna function is indeed uniquely determined by its values on any sequences $(\alpha_l)$ of reals numbers satisfying a M\"untz condition $\sum_{l \in \N^*} \frac{1}{\alpha_l} = \infty$. But we see from (\ref{Proof-Test1}) that the Nevanlinna functions
$$
  \frac{1}{T(\lambda, k , z)}, \quad \textrm{and} \quad \frac{1}{\tilde{T}(\lambda, k , z)},
$$
take the same values on the - a priori different - sequences of eigenvalues $\mu_{kl}(\lambda)$ and $\tilde{\mu}_{kl}(\lambda)$. Note that these eigenvalues satisfy the M\"untz condition above since they grow linearly with respect to $l \in \N^*$ (see Appendix \ref{Estimate-mukl}). We thus infer that, if we could prove the equality between the eigenvalues $\mu_{kl}(\lambda)$ and $\tilde{\mu}_{kl}(\lambda)$ for an infinite subset of $l \in \N^*$ satisfying the M\"untz condition, then we would have
\begin{equation} \label{CAM1}
  \forall z \in \C, \quad \frac{1}{T(\lambda, k , z)} = \frac{1}{\tilde{T}(\lambda, k , z)}.
\end{equation}
We emphasize here that (\ref{CAM1}) is one of the essential results given by the CAM method, result which turns out to be useful in the study of inverse scattering problems (see \cite{FY, DN3} and below) since we can play with the - now - \emph{complex} angular momentum to obtain more informations from the scattering coefficients.

The second step is thus naturally to show the equalities
\begin{equation} \label{DiscreteSpectra}
  \mu_{kl}(\lambda) = \tilde{\mu}_{kl}(\lambda),
\end{equation}
between the angular momenta for all large enough $l \in \N^*$ from the data given in Proposition \ref{Consequence1}. The first idea is to use the asymptotics of the scattering coefficients $T_{kl}(\lambda)$, $R_{kl}(\lambda)$ and $L_{kl}(\lambda)$ but unfortunately, these asymptotics (obtained in Proposition \ref{asymptoticsal}) do not give us enough information to prove this result. Instead, we use the second information given in Proposition \ref{Consequence1}, namely the uniqueness (\ref{UniqueY}) of the eigenfunctions $Y_{kl}^j, \ j=1,2$, to prove that the parameters $a$ and $\Lambda$ are uniquely defined. Clearly such a result would entail (\ref{DiscreteSpectra}) since the angular operator (\ref{Ak}) only depends on these two parameters.

Therefore we need to perform a detailed analysis of the angular eigenvalues equation (\ref{Angular-Eq}) and its associated eigenfunctions. This is what we do in Section \ref{InverseAngular}. Recasting this equation into a convenient form, we get a system of ODEs of Fuschian type with weakly singularities at $0$ and $\pi$. Hence we can use the Frobenius method to construct a system of fundamental solutions (SFS) for this equation, that is a system of linearly independent solutions. It turns out that only one of the solutions of the SFS belongs to $L^2$ in a neighbourhood of $0$ and thus is a constant multiple of the generalized spherical harmonics $Y_{kl}(\lambda)$. Moreover the Frobenius method allows to construct the $Y_{kl}(\lambda)$'s as singular power series in the variable $\theta$. The asymptotic expansion of $Y_{kl}(\lambda)(\theta,\varphi)$ when $\theta \to 0$ together with (\ref{UniqueY}) lead to the proof that the two parameters $a$ and $\Lambda$ are uniquely determined. As already said, this entails (\ref{DiscreteSpectra}) and also the equality (up to multiplicative constants of modulus $1$) between the generalized spherical harmonics $Y_{kl}(\lambda)$ and $\tilde{Y}_{kl}(\lambda)$ for all values $k \in \frac{1}{2} + \Z$ and \emph{all} angular momenta $l \in \N^*$. Summarising, we shall prove in Section \ref{InverseAngular}

\begin{theorem} \label{Main-Frobenius}
  Let the assumptions of Thm \ref{Main} (and thus the results of Proposition \ref{Consequence1}) hold. Then
  $$
    a = \tilde{a}, \qquad \Lambda = \tilde{\Lambda}.
  $$
  Hence, for all $(k,l) \in (\frac{1}{2} + \Z) \times \N^*$, we have
  $$
    \mu_{kl}(\lambda) = \tilde{\mu}_{kl}(\lambda),
  $$
  and there exist $\alpha_{kl} \in \C, \ |\alpha_{kl}| = 1$ such that
  $$
    Y_{kl}(\lambda) = \alpha_{kl} \tilde{Y}_{kl}(\lambda).
  $$
  As a consequence, we also get
  \begin{eqnarray*}
    T_{kl}(\lambda) & = & e^{ic_T(\lambda,k)} \, \tilde{T}_{kl}(\lambda), \\
    R_{kl}(\lambda) & = & e^{ic_R(\lambda,k)} \, \tilde{R}_{kl}(\lambda), \\
    L_{kl}(\lambda) & = & e^{ic_L(\lambda,k)} \, \tilde{L}_{kl}(\lambda),
  \end{eqnarray*}
where $c_T(\lambda,k)$, $c_R(\lambda,k)$ and $c_L(\lambda,k)$ are the constants in Theorem \ref{Main}.
\end{theorem}

Hence we face now a situation quite similar to the one studied in \cite{DN3}. We are led indeed to study a uniqueness inverse scattering problem for the countable family of one-dimensional radial Dirac equations (\ref{Radial-Eq}) - (\ref{Radial}) parametrized by the \emph{uniquely determined} angular momenta $\mu_{kl}(\lambda)$! As explained previously and thanks to Theorem \ref{Main-Frobenius}, the CAM method allows us to prove rigorously (\ref{CAM1}) and similar results for the reflection coefficients $R_{kl}(\lambda)$ and $L_{kl}(\lambda)$ (see Propositions \ref{Uniqueness} and \ref{Uniqueness2}). From the explicit expressions of the reflection coefficients, we can easily show that the Fourier transforms of the potentials $a(x)$ and $c(x,k)$ at the fixed energy $2\lambda$ are uniquely determined (in fact \emph{almost} uniquely determined). Using the exponential decay of these potentials at the cosmological and event horizons and some straighforward additional work, we are already able to prove at this stage a nice uniqueness result \emph{localized in energy}. Precisely,

\begin{theorem} \label{Uniqueness-Localized}
  Let $(M,Q^2,a,\Lambda)$ and $(\tilde{M},\tilde{Q}^2,\tilde{a},\tilde{\Lambda})$ be the parameters of two a priori different KN-dS black holes. We denote by $I$ a (possibly small) open interval of $\R$. Assume that, there exists a constant $c \in \R$ such that for all $\lambda \in I$ and for two different $k \in \frac{1}{2} + \Z$, one of the following conditions holds
\begin{eqnarray*}
  L_{kl} (\lambda) &=& e^{-2i \lambda_-(k) c}\  \tilde{L}_{kl} (\lambda),  \\
  R_{kl} (\lambda) &=& e^{2i\lambda_+(k) c}\   \tilde{R}_{kl} (\lambda),
\end{eqnarray*}
for all $l \in \mathcal{L}_k$ where the sets $\mathcal{L}_k \subset \N^* $ satisfy a M\"untz condition
$$
  \sum_{l \in \mathcal{L}_k} \frac{1}{l} = \infty.
$$
Then, we have
\begin{equation} \label{Uniqueness-ac}
  a(x) = \tilde{a}(x-c), \quad c(x,k) = \tilde{c}(x-c,k).
\end{equation}
In particular, using the particular form of the potential $a(x)$ and $c(x,k)$, we can show that the parameters of the two black holes coincide.
\end{theorem}

\begin{remark}
  1) Note that we used here the precise form of the constants
  $$
    c_R(\lambda,k) = 2\lambda_+(k) c, \quad c_L(\lambda,k) = -2\lambda_+(k) c,
  $$
  of the main assumption (\ref{MainAssumption}). Then the uniqueness result (\ref{Uniqueness-ac}) shows that the potentials $a(x)$ and $c(x,k)$ are uniquely determined up to a translation and that this translation is precisely given by the constant $c$ that encodes the possibility to describe the same KN-dS black hole by two Regge-Wheeler variables. \\
  2) We mention that we are not able to prove a localized in energy uniqueness result from the transmission coefficient $T_{kl}(\lambda)$ alone. This is a pure technical problem (see Remark \ref{NotLocalizedUniquenessT}). However, with more work, we shall be able to prove a uniqueness result from the knowledge of $T_{kl}(\lambda)$ for a fixed energy $\lambda$!
\end{remark}

To obtain a uniqueness result from the scattering coefficients at a \emph{fixed} energy, we need more informations on the properties of the scattering data with respect to the complexified angular momentum $z$. In particular, using a convenient change of variable (a Liouville transformation) and the corresponding form of the radial Dirac equation (\ref{Rad-0}), we shall obtain precise asymptotics of the scattering data $T(\lambda, k, z), \ R (\lambda, k ,z), \ L (\lambda, k ,z)$ when $z \to +\infty$. On one hand, these asymptotics are an important ingredient in the proof of Proposition \ref{Consequence1}. On the other hand, these asymptotics, CAM's results such that (\ref{CAM1}) together with a standard technique (as exposed first in \cite{FY} and used in this setting in \cite{DN3}) will lead to the unique determination of certain scalar functions depending on the radial variable (up to diffeomorphisms). From the explicit form of these functions, we prove the uniqueness of the parameters of the two black holes. Precisely, we prove

\begin{theorem} \label{Uniqueness-Fixed}
  Consider two a priori different KN-dS black holes and add a $\ \tilde{}$ to the quantities related to the second black hole. Assume that for a fixed energy $\lambda \in \R$ and for two different $k \in \frac{1}{2} + \Z$, there exist sets $\mathcal{L}_k \subset \N^* $ satisfying a M\"untz condition $\displaystyle{\sum_{l \in \mathcal{L}_k} \frac{1}{l} = \infty}$ such that
  $$
    \forall l \in \L_k, \quad \muk = \tilde{\mu}_{kl}(\lambda),
  $$
  and one of the following condition holds
  \begin{eqnarray}
    T_{kl}(\lambda) & = & e^{ic_T(\lambda,k)} \, \tilde{T}_{kl}(\lambda), \nonumber \\
    R_{kl}(\lambda) & = & e^{ic_R(\lambda,k)} \, \tilde{R}_{kl}(\lambda), \label{MainAssumption-1} \\
    L_{kl}(\lambda) & = & e^{ic_L(\lambda,k)} \, \tilde{L}_{kl}(\lambda). \nonumber
\end{eqnarray}
Here the constants $c_T(\lambda,k)$, $c_R(\lambda,k)$ and $c_L(\lambda,k)$ are the same as in Thm \ref{Main}. Then
\begin{equation} \label{Uniqueness-Function1}
  \frac{\lambda - c(\tilde{x},k)}{a(\tilde{x})} = \frac{\lambda - c(x,k)}{a(x)},
\end{equation}
where $\tilde{x}= \tilde{x}(x)$ is a diffeomorphism on $\R$. Moreover, from the explicit forms of the potentials, the parameters of the two black holes are shown to coincide,
$$
  M = \tilde{M}, \ a = \tilde{a}, \ Q^2 = \tilde{Q}^2, \ \Lambda = \tilde{\Lambda}.
$$
Finally, in the particular case of Kerr-de-Sitter black holes ($Q = 0$), we get more precise results. Precisely, under the same assumption, there exists a constant $\sigma \in \R$ such that
\begin{eqnarray}
  \tilde{x} & = & x + \sigma, \nonumber \\
  \tilde{a}(x) & = & a(x - \sigma), \label{Uniqueness-Function2} \\
  \tilde{c}(x,k) & = & c(x - \sigma,k). \nonumber
\end{eqnarray}
Note that this last result holds true in the case of KN-dS black hole if we assume that (\ref{MainAssumption-1}) is known for two different energies $\lambda \in \R$.
\end{theorem}

The proofs of Thm \ref{Uniqueness-Localized} and \ref{Uniqueness-Fixed} will be the object of Sections \ref{Complexification}, \ref{AsymptoticsSD} and \ref{Inverse}.


\Section{The angular equation and partial inverse result} \label{InverseAngular}

\subsection{The Frobenius method: construction and asymptotics of the generalized spherical harmonics $Y_{kl}(\lambda)$} \label{Frobenius}

In this Section, we study rigorously the angular equation (\ref{Angular-Eq}) that arises in the separation of variables procedure exposed in Section \ref{Sep-Var}. Let us recall first some notations and definitions. For each fixed energy $\lambda \in \R$, we define for all $(k,l) \in I = (\frac{1}{2} + \Z) \times \N^*$ the generalized spherical harmonics $Y_{kl}(\lambda)$ as normalized eigenfunctions of the angular (selfadjoint) operator $A_{\S^2}(\lambda)$ given by (\ref{AS2}) and (\ref{HS2}) associated to its positive eigenvalues $\mu_{kl}(\lambda)$. Using the cylindrical symmetry of the equation and the corresponding decomposition onto the angular modes $\{e^{ik\varphi}\}$ with $k \in \frac{1}{2} + \Z$, it was shown in (\ref{Ykl}) that the $Y_{kl}(\lambda)$'s can be written as
$$
  Y_{kl}(\lambda)(\theta, \varphi) = u_{kl}^\lambda(\theta) e^{ik\varphi},
$$
where the $u_{kl}^\lambda(\theta)$ are normalized eigenfunctions of the reduced angular equation
\begin{equation} \label{AE}
  A_k(\lambda) u_{kl}^\lambda(\theta) = \mu_{kl}(\lambda) \, u_{kl}^\lambda(\theta).
\end{equation}
Here the operator $A_k(\lambda)$ is given by
\begin{equation} \label{Ukl}
  A_k(\lambda) = \sqrt{\Delta_\theta} \left[ \Gb D_\theta + \Gb \frac{i\Lambda a^2 \sin(2\theta)}{12 \Delta_\theta}  + \frac{k \Gc}{\sin{\theta}} + \Gc \frac{ \Lambda a^2 k \sin(\theta)}{3 \Delta_\theta} - \lambda \frac{a \sin \theta}{\Delta_\theta} \Gc \right].
\end{equation}
and is selfadjoint on $\L := L^2((0,\pi), d\theta; \C^2)$. Recall at last that the spectrum of $A_k(\lambda)$ is composed of discrete simple eigenvalues and thus, $u_{kl}^\lambda(\theta)$ is the \emph{unique} (up to a multiplicative constant of modulus $1$) normalized solution of (\ref{AE}) that belongs to $\L$.

Our aim is to construct and give the asymptotics of the eigenfunctions $u_{kl}^\lambda(\theta)$ when $\theta \to 0$. For this, we use the explicit form of the above angular equation and recast it into a usual first-order system of ODEs. The resulting system possesses weakly singular points at $\theta = 0$ and $\theta = \pi$ and turns out to be a system of ODEs of Fuschian type. We refer for instance \cite{W}, chapter V, for a classical presentation of these systems of singular ODEs with some explicit examples very well adapted to our case. Thanks to the classical Frobenius method, we are able to construct a fundamental system of solutions and provide a series expansion for these solutions. In particular, choosing the unique solution that belongs to $\L$ to be $u_{kl}^\lambda(\theta)$ , we obtain easily the asymptotics of $u_{kl}^\lambda(\theta)$ when $\theta \to 0$ from its series expansion.

We introduce the notations $\zeta = \frac{\Lambda a^2}{3}$ and we fix $(k,l) \in I$. Then the eigenvalue equation (\ref{AE}) - (\ref{Ukl}) can be written as the equivalent system of ODEs
\begin{equation} \label{AES}
  \partial_\theta u = A(\theta) u,
\end{equation}
with
\begin{equation} \label{AT}
  A(\theta) = \frac{1}{\theta} \left[ \frac{i \mu_{kl}(\lambda) \theta}{\sqrt{\Delta_\theta}} \Gb - \frac{ k \theta}{\sin \theta} \Ga + \frac{\zeta \theta \sin(2\theta)}{4 \Delta_\theta} I_2 - \frac{(\zeta k - a \lambda) \theta \sin \theta}{\Delta_\theta} \Ga   \right].
\end{equation}
Note that the singularity at $\theta = 0$ is given by the term $\frac{1}{\theta}$ whereas the term $\theta A(\theta)$ is analytic in the (complex) variable $\theta$ in a neighbourhood of $0$. Hence the latter can be written as a power serie in $\theta$, \textit{i.e} $\theta A(\theta) = \ds \sum_{n = 0}^{\infty} A_n \theta^n$ where the $A_n$'s are $2 \times 2$ matrix-valued coefficients. In what follows, we shall need the first terms in this power serie up to order $3$. Recalling that $\Delta_\theta = 1 + \zeta \cos^2\theta$, we get from (\ref{AT})
\begin{equation} \label{AT-Asymp}
  \theta A(\theta) = A_0 + A_1 \theta + A_2 \theta^2 + O(\theta^3), \quad \theta \to 0,
\end{equation}
where the coefficients $A_0, A_1, A_2$ are given by
\begin{equation} \label{A0A1}
 A_0 = \left[ \begin{array}{cc} -k & 0 \\ 0 & k \end{array} \right], \quad  A_1 = \left[ \begin{array}{cc} 0 & \frac{i \mu_{kl}(\lambda)}{\sqrt{1 + \zeta}} \\ \frac{i \mu_{kl}(\lambda)}{\sqrt{1 + \zeta}} & 0 \end{array} \right],
\end{equation}
\begin{equation} \label{A2}
 A_2 = \left[ \begin{array}{cc} -\frac{k}{6} + \frac{\zeta}{2(1 + \zeta)} - \frac{\zeta k - a \lambda}{1 + \zeta} & 0 \\ 0 & \frac{k}{6} + \frac{\zeta}{2(1 + \zeta)} + \frac{\zeta k - a \lambda}{1 + \zeta} \end{array} \right]
\end{equation}

In the following we assume that $k \in \frac{1}{2} + \N$ (the other case is handled similarly) and we use the Frobenius method. According to \cite{Wa}, chap. V, there exist two linearly independent solutions of (\ref{AES}) that have the following form
\begin{eqnarray}
  v_{kl}(\theta) & = & \theta^k h_0(\theta), \label{FrobeniusSol} \\
  w_{kl}(\theta) & = & \theta^{-k} ( h_1(\theta) + \log \theta h_2(\theta) ), \nonumber
\end{eqnarray}
where the vector-valued functions $h_j$ are analytic in $\theta$ in a neighbourhood of $0$. Since $k \in \frac{1}{2} + \N$, we see that $v_{kl}$ is the unique solution (up to a multiplicative constant) that belongs to $L^2$ in a neighbourhood of $0$. Hence, the eigenfunction $u_{kl}^\lambda(\theta)$ must be a constant multiple of $v_{kl}(\theta)$.

Even better, according to \cite{Wa}, chap. V, the solution $v_{kl}$ takes the form
$$
  v_{kl}(\theta) = \theta^k \sum_{n = 0}^{+\infty} v_{kl}^n \, \theta^n,
$$
where the vectors $v_{kl}^n \in \C^2$ can be explicitly constructed by the following induction.
\begin{itemize}
\item The vector $v_{kl}^0$ is an eigenvector of $A_0$ corresponding to the eigenvalue $k$.
\item For $n \geq 1$, the vectors $v_{kl}^n$ are uniquely determined by
$$
  \left( (k+n) I_2 - A_0 \right) v_{kl}^n = \sum_{j = 0}^{n-1} A_{n-j} v_{kl}^j.
$$
\end{itemize}
Let us apply this procedure to our case. From (\ref{A0A1}), we choose
\begin{equation} \label{Vkl0}
  v_{kl}^0 = \left[ \begin{array}{c} 0 \\ 1 \end{array} \right],
\end{equation}
and applying the above rules, we successively obtain
\begin{equation} \label{Vkl12}
  v_{kl}^1 = \left[ \begin{array}{c} \frac{i \mu_{kl}(\lambda)}{(2k+1) \sqrt{1+ \zeta}} \\ 0 \end{array} \right], \quad v_{kl}^2 = \left[ \begin{array}{c} 0 \\ \frac{k}{12} + \frac{\zeta}{4(1 + \zeta)} + \frac{\zeta k - a\lambda}{2(1 + \zeta)} - \frac{\mu_{kl}(\lambda)^2}{2(2k+1)(1 + \zeta)} \end{array} \right].
\end{equation}
We conclude that the solution $v_{kl}$ has the asymptotic expansion when $\theta \to 0$
\begin{equation} \label{Vkl}
  v_{kl}= \theta^k \left( v_{kl}^0 + v_{kl}^1 \theta + v_{kl}^2 \theta^2 + O(\theta^3) \right),
\end{equation}
where the $v_{kl}^j$'s are given by (\ref{Vkl0}) - (\ref{Vkl12}). Since the eigenfunction $u_{kl}^\lambda(\theta)$ is a constant multiple of $v_{kl}(\theta)$, we have proved the following Proposition

\begin{prop} \label{Asymptotics-Ykl}
  For all $\lambda \in \R$ and $(k,l) \in (\frac{1}{2} + \N) \times \N^*$, there exist constants $c_{kl}^{\lambda} \in \C$ such that
  \begin{eqnarray}
    u_{kl}^\lambda(\theta) & = & c_{kl}^{\lambda} \Bigg\{ \left( \begin{array}{c} 0 \\ 1 \end{array} \right) \theta^k +  \frac{i \mu_{kl}(\lambda)}{(2k+1) \sqrt{1+ \zeta}} \left( \begin{array}{c} 1 \\ 0 \end{array} \right) \theta^{k+1} \label{Ukll} \\
                           &   & \quad + \frac{1}{2} \left[ \frac{k}{6} + \frac{\zeta}{2(1 + \zeta)} + \frac{\zeta k - a\lambda}{1 + \zeta} - \frac{\mu_{kl}(\lambda)^2}{(2k+1)(1 + \zeta)} \right] \left( \begin{array}{c} 0 \\ 1 \end{array} \right) \theta^{k+2} \ + \ O(\theta^{k+3}) \Bigg\}, \nonumber
\end{eqnarray}
as $\theta \to 0$. Moreover
$$
  Y_{kl}(\lambda)(\theta, \varphi) = u_{kl}^\lambda(\theta) e^{ik\varphi}.
$$
\end{prop}

We finish this short Section mentioning that similar asymptotics are obtained when $k \in \frac{1}{2} - \N^*$ simply by inverting the roles of the solutions $v_{kl}$ and $w_{kl}$ in (\ref{FrobeniusSol}). We omit these calculations.


\subsection{Application to the uniqueness inverse problem: proofs of Theorems \ref{Consequence1} and \ref{Main-Frobenius}} \label{ApplicationFrobenius}

\begin{proof}[Proof of Theorem \ref{Consequence1}]
Assume for instance that for $\lambda \in \R$ fixed and two different $k \in \frac{1}{2} + \Z$
\begin{equation} \label{P1}
  T_k^L(\lambda) = e^{i c_T(\lambda,k)} \tilde{T}_k^L(\lambda),
\end{equation}
where $c_T(\lambda,k)$ is a constant. Our assumption (\ref{P1}) implies that the operators $T_k^L(\lambda)$ and $e^{i c_T(\lambda,k)} \tilde{T}_k^L(\lambda)$ have the same eigenvalues with the same mutiplicities. But from Corollary \ref{SpectralProp-TrRe}, we know that the $T_{kl}(\lambda)$'s are the eigenvalues of the operator $T_k^L(\lambda)$ and similarly that the $e^{i c_T(\lambda,k)} \tilde{T}_{kl}(\lambda)$'s are the eigenvalues of $e^{i c_T(\lambda,k)} \tilde{T}_k^L(\lambda)$. Hence, we immediately get

\begin{lemma} \label{Consequence2}
  There exists a bijective map $\varphi: \ \N^* \longrightarrow \N^*$ such that for all $l \in \N^*$,
  $$
    T_{kl}(\lambda) = e^{i c_T(\lambda,k)} \tilde{T}_{k \varphi(l)}(\lambda).
  $$
\end{lemma}

The next step is to prove that the $T_{kl}(\lambda)$'s are in fact simple eigenvalues of $T_k^L(\lambda)$ when $l$ is large enough. To prove this, we need some additional a priori results on the transmission coefficients
\begin{equation} \label{bo0}
  T_{kl}(\lambda) = T(\lambda,k,\muk),
\end{equation}
when $l$ is large. The next Lemma is a direct consequence of the study of the asymptotics of the scattering data given in Section \ref{AsymptoticsSD}.

\begin{lemma} \label{PropTkl}
  For $\lambda \in \R$ and $k \in \frac{1}{2} + \Z$ fixed, \\
  (i) $T_{kl}(\lambda) \to 0$ when $l \to \infty$. \\
  (ii) For all $l \in \N^*$, $T_{kl}(\lambda) \ne 0$. \\
  (iii) There exists $L>0$ such that the map $l \in \N^* \longrightarrow |T_{kl}(\lambda)|$ is strictly decreasing for $l \geq L$. \\
  As a consequence, the eigenvalues $T_{kl}(\lambda)$ of $T_k^L(\lambda)$ are simple for $l$ large enough.
\end{lemma}
\begin{proof}
We shall systematically use the fact that for fixed $\lambda$ and $k$
\begin{equation} \label{bo1}
  \muk \to +\infty, \quad l \to +\infty.
\end{equation}
Hence, the first point (i) is a consequence of (\ref{bo0}) and of the asymptotics of $T(\lambda,k,z)$ when $z \to +\infty$ given in Theorem \ref{asymptoticssc}. The second point (ii) is clear since
\begin{equation} \label{bo2}
  T(\lambda,k,z) = \frac{1}{a_{L1}(\lambda,k,z)},
\end{equation}
$T_{kl}(\lambda)= T(\lambda, k, \muk)$ and $a_{L1}(\lambda,k,z)$ is entire in $z \in \C$. Eventually, the last point (iii) follows imediately from (\ref{bo1}), (\ref{bo2}) and Proposition \ref{AL1-Increasing}.
\end{proof}

Using the unitarity of the reduced scattering matrix $S_{kl}(\lambda)$, and in particular
$$
  |T_{kl}(\lambda)|^2 + |R_{kl}(\lambda)|^2 = 1 = |T_{kl}(\lambda)|^2 + |L_{kl}(\lambda)|^2,
$$
we also have an analogous statement for the reflection coefficients $R_{kl}(\lambda)$ and $L_{kl}(\lambda)$ that we shall need later.

\begin{lemma} \label{PropRkl}
  For $\lambda \in \R$ and $k \in \frac{1}{2} + \Z$ fixed, \\
  (i) $|R_{kl}(\lambda)| \to 1$ when $l \to \infty$. \\
  (ii) For all $l \in \N^*$, $|R_{kl}(\lambda)| \ne 1$. \\
  (iii) There exists $L>0$ such that the map $l \in \N^* \longrightarrow |R_{kl}(\lambda)|$ is strictly increasing for $l \geq L$. \\
  Moreover the same is true of we replace the reflection coefficient from the right $R_{kl}(\lambda)$ by the reflection coefficient from the left $L_{kl}(\lambda)$. In consequence, the eigenvalues $|R_{kl}(\lambda)|^2$ of $R_k(\lambda) R_k^*(\lambda)$ and $|L_{kl}(\lambda)|^2$ of $L_k^*(\lambda) L_k(\lambda)$ are simple for $l$ large enough.
\end{lemma}

Let us come back to the proof of Theorem \ref{Consequence1}. From Lemma \ref{PropTkl}, we can prove easily

\begin{coro} \label{Consequence3}
  (i) The bijective map $\varphi: \ \N^* \longrightarrow \N^*$ from Lemma \ref{Consequence2} is also strictly increasing for $l$ large enough. As a consequence, there exists $L>0$ such that for all $l \geq L$, $\varphi(l) = l$. \\

\noindent (ii) For all $l \geq L$,
$$
  T_{kl}(\lambda) = e^{i c_T(\lambda,k)} \tilde{T}_{kl}(\lambda).
$$
(iii) There exists $\alpha_{kl}^1 \in \C, \ |\alpha_{kl}^1| = 1$ such that
$$
  Y_{kl}^1 = \alpha_{kl}^1 \tilde{Y}_{kl}^1.
$$
\end{coro}
\begin{proof}
  The first assertion (i) is a direct consequence of Lemma \ref{PropTkl}. Since $\varphi$ is bijective and strictly increasing for large enough $l$, we conclude that $\varphi(l) = l$ for $l$ large enough. \\
  The second assertion is then a consequence of Lemma \ref{Consequence2}. \\
  The third assertion comes from the fact that the family $\{Y_{kl}^1\}_{l \geq L}$ are the eigenfunctions of the operator $T_k^L(\lambda)$ associated to the simple eigenvalues $T_{kl}(\lambda)$ according to Lemma \ref{PropTkl}.
\end{proof}

Now assume that, besides (\ref{P1}), we also have
\begin{equation} \label{P2}
  T_k^R(\lambda) = e^{i c_T(\lambda,k)} \tilde{T}_k^R(\lambda).
\end{equation}
Recall that $T_k^R(\lambda)$ is diagonalizable on the Hilbert basis $\{Y_{kl}^2\}_{l \in \N^*}$ associated to the eigenvalues $T_{kl}(\lambda)$. Then the same argument as above leads to the same conclusions than the ones in Corollary \ref{Consequence3} with the obvious replacement that there exist $\alpha_{kl}^2 \in \C$, $|\alpha_{kl}^2| =1$ such that
$$
  Y_{kl}^2 = \alpha_{kl}^2 \tilde{Y}_{kl}^2.
$$
We conclude that under the assumptions (\ref{P1}) and (\ref{P2}), Theorem \ref{Consequence1} is proved. \\

Let us assume now that
\begin{equation} \label{P3}
  R_k(\lambda) = e^{i c_R(\lambda,k)} \tilde{R}_k(\lambda),
\end{equation}
for a certain constant $c_R(\lambda,k)$. Hence we have
\begin{equation} \label{P4}
  R_k(\lambda) R_k^*(\lambda) = \tilde{R}_k(\lambda) \tilde{R}_k^*(\lambda), \quad R_k^*(\lambda) R_k(\lambda) = \tilde{R}_k^*(\lambda) \tilde{R}_k(\lambda).
\end{equation}
We recall that $|R_{kl}(\lambda)|^2$ are the eigenvalues of $R_k(\lambda) R_k^*(\lambda)$ and $R_k^*(\lambda) R_k(\lambda)$ associated to the eigenfunctions $Y_{kl}^1$ and $Y_{kl}^2$ respectively. Hence using a similar argument than the one given previously and Lemma \ref{PropRkl}, we can show that

\begin{coro} \label{Consequence4}
  There exists $L>0$ such that for all $l \geq L$,
  $$
    |R_{kl}(\lambda)| = | \tilde{R}_{kl}(\lambda)|,
  $$
  and there exists $\alpha_{kl}^j \in \C, \ |\alpha_{kl}^j| = 1$ such that
  $$
    Y_{kl}^j = \alpha_{kl}^j \tilde{Y}_{kl}^j, \quad j=1,2.
  $$
\end{coro}

This proves Theorem \ref{Consequence1} from the assumption (\ref{P3}). Of course, the proof of Theorem \ref{Consequence1} from the knowledge of the reflection operator from the left $L_k(\lambda)$ is similar and we omit it.

\end{proof}

\begin{remark}
  Note a certain assymmetry in the results of the above proof of Theorem \ref{Consequence1}. If we assume that one of the reflection operator $R_k(\lambda)$ or $L_k(\lambda)$ is known, then we get informations on both eigenfunctions $Y_{kl}^1$ and $Y_{kl}^2$. This is due to the the fact that the reflection operators - the anti-diagonal elements of the scattering matrix $S_k(\lambda)$ - transform \emph{by definition} $Y_{kl}^1$ into $Y_{kl}^2$ or the converse. On the other hand, if we only assume the knowledge of one of the transmission operators $T_k^R$ or $T_k^L$, then we obtain informations on either $Y_{kl}^1$ or $Y_{kl}^2$. However, for pure technical reasons, we need to know both $Y_{kl}^1$ and $Y_{kl}^2$ to prove the following uniqueness inverse result. This is why we assume that both $T_k^R$ and $T_k^L$ are known in our main assumption (\ref{MainAssumption}).
\end{remark}

We are now in position to prove Theorem \ref{Main-Frobenius}.

\begin{proof}[Proof of Theorem \ref{Main-Frobenius}]
According to Theorem \ref{Consequence1}, for $\lambda \in \R$ fixed and two different $k \in \frac{1}{2} + \Z$, there exists $L > 0$ such that for all $l \geq L$, there exist constants $\alpha_{kl}^1, \alpha_{kl}^2 \in \C$ with $|\alpha_{kl}^j| =1$ such that
\begin{equation} \label{P8}
  Y_{kl}^1 = \alpha_{kl}^1 \tilde{Y}_{kl}^1, \quad Y_{kl}^2 = \alpha_{kl}^2 \tilde{Y}_{kl}^2.
\end{equation}
We also assume that $k \in \frac{1}{2} + \N$ since the other case is treated similarly. According to Proposition \ref{Asymptotics-Ykl}, we know the asymptotics of the eigenfunctions $Y_{kl}$ when $\theta \to 0$. Precisely, there exists $c_{kl}^\lambda \in \C$ such that
\begin{eqnarray}
    Y_{kl}(\theta,\varphi) & = & c_{kl}^{\lambda} e^{ik\varphi} \Bigg\{ \left( \begin{array}{c} 0 \\ 1 \end{array} \right) \theta^k +  \frac{i \mu_{kl}(\lambda)}{(2k+1) \sqrt{1+ \zeta}} \left( \begin{array}{c} 1 \\ 0 \end{array} \right) \theta^{k+1} \label{P5} \\
                           &   & \quad + \frac{1}{2} \left[ \frac{k}{6} + \frac{\zeta}{2(1 + \zeta)} + \frac{\zeta k - a\lambda}{1 + \zeta} - \frac{\mu_{kl}(\lambda)^2}{(2k+1)(1 + \zeta)} \right] \left( \begin{array}{c} 0 \\ 1 \end{array} \right) \theta^{k+2} \ + \ O(\theta^{k+3}) \Bigg\}, \nonumber
\end{eqnarray}
where we use the notation
\begin{equation} \label{zeta}
  \zeta = \frac{a^2 \Lambda}{3}.
\end{equation}
In fact, the constant $c_{kl}^\lambda$ plays no role in what follows. Hence, we take it equal to $1$. It follows from (\ref{P5}) that when $\theta \to 0$,
\begin{equation} \label{P6}
  Y_{kl}^1(\theta,\varphi) = e^{ik\varphi} \frac{i \mu_{kl}(\lambda)}{(2k+1) \sqrt{1+ \zeta}} \theta^{k+1} + O(\theta^{k+3}),
\end{equation}
and
\begin{equation} \label{P7}
    Y_{kl}^2(\theta,\varphi) = e^{ik\varphi} \Bigg\{ \theta^k  + \left[ \frac{k}{12} + \frac{\zeta}{4(1 + \zeta)} + \frac{\zeta k - a\lambda}{2(1 + \zeta)} - \frac{\mu_{kl}(\lambda)^2}{2(2k+1)(1 + \zeta)} \right] \theta^{k+2} \ + \ O(\theta^{k+4}) \Bigg\}.
\end{equation}

Now using (\ref{P8}), (\ref{P6}) and (\ref{P7}) and equating the terms with same orders in $\theta \to 0$, we get
\begin{equation} \label{As1}
  \frac{\mu_{kl}(\lambda)}{\sqrt{1+ \zeta}} = \alpha_{kl}^1 \frac{\tilde{\mu}_{kl}(\lambda)}{\sqrt{1 + \tilde{\zeta}}},
\end{equation}
\begin{equation} \label{As2}
  1 = \alpha_{kl}^2,
\end{equation}
\begin{equation} \label{As3}
  \left[ \frac{\zeta}{2(1+\zeta)} + \frac{\zeta k - a \lambda}{1 + \zeta} - \frac{\mu_{kl}(\lambda)^2}{(2k+1)(1+\zeta)}  \right] = \left[ \frac{\tilde{\zeta}}{2(1+\tilde{\zeta})} + \frac{\tilde{\zeta} k - \tilde{a} \lambda}{1 + \tilde{\zeta}} - \frac{\tilde{\mu}_{kl}(\lambda)^2}{(2k+1)(1+\tilde{\zeta})}  \right].
\end{equation}
Taking the modulus of (\ref{As1}), we obtain
$$
  \frac{\mu_{kl}(\lambda)}{\sqrt{1+ \zeta}} = \frac{\tilde{\mu}_{kl}(\lambda)}{\sqrt{1 + \tilde{\zeta}}},
$$
and putting this in (\ref{As3}), we get
\begin{equation} \label{P9}
  \frac{\zeta}{2(1+\zeta)} + \frac{\zeta k - a \lambda}{1 + \zeta} = \frac{\tilde{\zeta}}{2(1+\tilde{\zeta})} + \frac{\tilde{\zeta} k - \tilde{a} \lambda}{1 + \tilde{\zeta}}.
\end{equation}
Since (\ref{P9}) is true for two different values of $k$, we thus obtain two decoupled equalities.
\begin{eqnarray}
  \frac{\zeta}{1 + \zeta} & = & \frac{\tilde{\zeta}}{1 + \tilde{\zeta}}, \label{P10} \\
  \frac{\zeta}{2(1+\zeta)} - \frac{a \lambda}{1 + \zeta} & = & \frac{\tilde{\zeta}}{2(1+\tilde{\zeta})} - \frac{\tilde{a} \lambda}{1 + \tilde{\zeta}}. \label{P11}
\end{eqnarray}
From (\ref{P10}), we get
\begin{equation} \label{P12}
  \zeta = \tilde{\zeta}.
\end{equation}
Using then (\ref{P11}), we get
\begin{equation} \label{P13}
  a = \tilde{a}.
\end{equation}
Finally, (\ref{P12}), (\ref{P13}) and (\ref{zeta}) lead to
$$
  \Lambda = \tilde{\Lambda}.
$$
Hence, the parameters $a$ and $\Lambda$ are uniquely determined from our main assumption (\ref{MainAssumption}). At last, since the angular operator $A_k(\lambda)$ appearing in (\ref{Angular-Eq}) only depends on the parameters $a$ and $\Lambda$, the other assertions of Theorem \ref{Main-Frobenius} follow immediately from the uniqueness of $a$ and $\Lambda$.

\end{proof}

\begin{remark}
  \begin{enumerate}
  \item We emphasize that we used the fact that the scattering coefficients are known for two distinct $k \in \frac{1}{2} + \Z$ in the above proof. Otherwise we wouldn't be able to determine uniquely $a$ and $\Lambda$.
  \item Our method also depends on the small number of parameters to recover. If the angular operator had depended not on a few parameters, but on - say - a scalar function with respect to the variable $\theta$, our method would have of course failed down. Instead, we should try to show that the eigenvalues $\muk$ are uniquely determined from our main assumption (\ref{MainAssumption}), at least for large enough $l \in \N^*$. From this and the fact that the eigenfunctions $Y_{kl}$ are uniquely determined for large enough $l$ (see Proposition \ref{Consequence1}), we could easily determine the unknown scalar function. Unfortunately, we have not been able to prove a similar claim in that particular model. We conjecture however that this is true. 
  \end{enumerate}
\end{remark}


\Section{The radial equation: complexification of the angular momentum.} \label{Complexification}

In this section, we follow the strategy exposed in \cite{DN3} and we allow the eigenvalues $\muk$ of the angular operator $A_{\S^2}(\lambda)$ to be complex. We shall denote by $z$ the complexified angular momenta and study the analytic properties (with respect to the variable $z$) of all the relevant scattering quantities such as the Jost functions $F_L(x,\lambda,k, z)$, $F_R(x,\lambda,k, z)$ and the matrix $A_L(\lambda,k, z)$ introduced in Section \ref{Direct-Stat-Scat}. The main result of this Section is that the coefficients of the matrix $A_L(\lambda,k,z)$ belong to the Nevanlinna class (see Section \ref{Nevanlinna}). This will be a crucial ingredient in the proof of our main Theorem. As a by-product, we prove a first inverse uniqueness result \emph{localized in energy}.


\subsection{Analytic properties of the Jost functions and matrix $A_L(\lambda,k, z)$.}

For all complexified angular momentum $z \in \C$, the Jost functions $F_L(x,\lambda,k,z)$ and $F_R(x,\lambda,k, z)$ are solutions of the stationary equation
\begin{equation} \label{PartialSE1}
  [ \Ga D_x - z V_k(x)  ] \psi = \lambda \psi, \quad \forall z \in \C.
\end{equation}
with prescribed asymptotics at $x \to \pm \infty$ given by (\ref{FL}) and (\ref{FR}). Recall from (\ref{Pot-V}) that
$$
  V_k(x) = \left( \begin{array}{cc} 0&q_k(x) \\ \bar{q}_k(x) &0 \end{array} \right), \qquad q_k(x) = a(x) e^{2iC(x,k)}.
$$
As in \cite{DN3}, we introduce the Faddeev matrices
$M_L(x,\lambda,k, z)$ and $M_R(x,\lambda,k, z)$ defined by
\begin{equation} \label{JostFaddeev}
  M_L(x,\lambda,k, z) = F_L(x,\lambda,k, z) e^{-i\Ga \lambda x}, \quad M_R(x,\lambda,k, z) = F_R(x,\lambda,k, z)
  e^{-i\Ga \lambda x},
\end{equation}
which satisfy the boundary conditions
\begin{eqnarray}
  M_L(x,\lambda,k, z) & = & I_2 + o(1), \ x \to +\infty, \\
  M_R(x,\lambda,k, z) & = & I_2 + o(1), \ x \to -\infty.
\end{eqnarray}
We shall systematically use the notations
\begin{equation} \label{Faddeev}
M_L(x,\lambda,k, z) = \left[\begin{array}{cc} m_{L1}(x,\lambda,k, z)&m_{L2}(x,\lambda,k, z)\\
m_{L3}(x,\lambda,k, z)&m_{L4}(x,\lambda,k, z) \end{array} \right],
\end{equation}
\begin{equation} \label{Faddeev1}
M_R(x,\lambda,k, z) = \left[\begin{array}{cc} m_{R1}(x,\lambda,k, z)&m_{R2}(x,\lambda,k, z)\\
m_{R3}(x,\lambda,k, z)&m_{R4}(x,\lambda,k, z) \end{array} \right].
\end{equation}
From (\ref{IE-FL}) and (\ref{IE-FR}), the Faddeev matrices satisfy the integral equations
\begin{equation} \label{IE-ML}
  M_L(x,\lambda,k, z) = I_2 - i z \Ga \int_x^{+\infty} e^{-i\Ga \lambda (y-x)} V_k(y)  M_L(y,\lambda,k, z)
  e^{i\Ga \lambda (y-x)} dy,
\end{equation}
\begin{equation} \label{IE-MR}
  M_R(x,\lambda,k, z) = I_2 + i z \Ga \int_{-\infty}^x e^{-i\Ga \lambda (y-x)}V_k(y)  M_R(y,\lambda,k, z)
  e^{i\Ga \lambda (y-x)}dy.
\end{equation}
Iterating (\ref{IE-ML}) and (\ref{IE-MR}) once, we get the uncoupled systems
\begin{eqnarray}
  m_{L1}(x,\lambda,k, z)  &= & 1 + z^2 \int_x^{+\infty} \int_y^{+\infty} e^{2 i \lambda (t - y)} q_k(y)
  {\overline{q_k(t)}} m_{L1}(t,\lambda,k, z) dt dy, \label{IE-ML1} \\
  m_{L2}(x,\lambda,k, z) & = & -iz \int_x^{+\infty} e^{-2 i \lambda (y-x)} q_k(y) dy \nonumber \\
                      &  & + z^2 \int_x^{+\infty} \int_y^{+\infty} e^{-2 i \lambda (y-x)} q_k(y) {\overline{q_k(t)}}
                      m_{L2}(t,\lambda,k, z) dt dy, \label{IE-ML2}\\
  m_{L3}(x,\lambda,k, z) & = & iz \int_x^{+\infty} e^{2 i \lambda (y-x)} {\overline{q_k(y)}} dy \nonumber \\
                      & & + z^2 \int_x^{+\infty} \int_y^{+\infty} e^{2 i \lambda (y-x)} {\overline{q_k(y)}}
                      q_k(t) m_{L3}(t,\lambda,k, z) dt dy, \label{IE-ML3}\\
  m_{L4}(x,\lambda,k, z) & = & 1 + z^2 \int_x^{+\infty} \int_y^{+\infty} e^{-2 i \lambda (t - y)}
                       {\overline{q_k(y)}} q_k(t) m_{L4}(t,\lambda,k, z) dt dy, \label{IE-ML4}
\end{eqnarray}
and
\begin{eqnarray}
  m_{R1}(x,\lambda,k, z) & = & 1 + z^2 \int_{-\infty}^x \int_{-\infty}^y e^{-2 i \lambda (y-t)}
                          q_k(y) {\overline{q_k(t)}} m_{R1}(t,\lambda,k, z) dt dy, \label{IE-MR1} \\
  m_{R2}(x,\lambda,k, z) & = & iz \int_{-\infty}^x e^{2 i \lambda (x-y)} q_k(y) dy \nonumber \\
                         & & + z^2 \int_{-\infty}^x \int_{-\infty}^y e^{2 i \lambda (x - y)} q_k(y)
                      {\overline{q_k(t)}} m_{R2}(t,\lambda,k, z) dt dy, \label{IE-MR2}\\
  m_{R3}(x,\lambda,k, z) & = & -iz \int_{-\infty}^x e^{-2 i \lambda (x-y)} {\overline{q_k(y)}} dy \nonumber \\
                         & &+ z^2 \int_{-\infty}^x \int_{-\infty}^y e^{-2 i \lambda (x-y)}
                        {\overline{q_k(y)}} q_k(t) m_{R3}(t,\lambda,k, z) dt dy, \label{IE-MR3}\\
  m_{R4}(x,\lambda,k,z) & = & 1 + z^2 \int_{-\infty}^x \int_{-\infty}^y e^{2 i \lambda (y - t)}
                        {\overline{q_k(y)}} q_k(t) m_{R4}(t,\lambda,k, z) dt dy. \label{IE-MR4}
\end{eqnarray}
The equations (\ref{IE-ML1})-(\ref{IE-ML4}) are of Volterra type and we can solve them by iteration. We obtain easily the following lemma
\begin{lemma} \label{ML-Analytic}
  (i) Set $m_{L1}^0(x,\lambda,k) = 1$ and for all $n \geq 1$
  $$
    m_{L1}^n(x,\lambda,k) = \int_x^{+\infty} \int_y^{+\infty} e^{2 i \lambda (t - y)} q_k(y) {\overline{q_k(t)}} m_{L1}^{n-1}(t,\lambda,k) dt dy.
  $$
  Then we get by induction
  $$
    |m_{L1}^n(x,\lambda,k)| \leq \frac{1}{(2n)!} \Big( \int_x^{+\infty} a(y) dy \Big)^{2n}.
  $$
  For $x, \lambda \in \R$ fixed, the serie $m_{L1}(x,\lambda,k,z) = \displaystyle\sum_{n=0}^\infty m_{L1}^n(x,\lambda,k) z^{2n}$
  converges absolutely and uniformly on each compact subset of $\C$ and satisfies the estimate
  $$
    |m_{L1}(x,\lambda,k,z)| \leq \cosh\Big(|z| \int_x^{+\infty} a(s)ds \Big), \quad \forall x \in \R, \ z \in \C .
  $$
  Moreover, the application $z \longrightarrow m_{L1}(x,\lambda,k,z)$ is entire and even. \\

  \noindent (ii) Set $m_{L2}^0(x,\lambda,k) = -i \int_x^{+\infty} e^{-2 i \lambda (y-x)} q_k(y) dy$ and for all $n \geq 1$
  $$
    m_{L2}^n(x,\lambda,k) = \int_x^{+\infty} \int_y^{+\infty} e^{2 i \lambda (x - y)} q_k(y) {\overline{q_k(t)}} m_{L2}^{n-1}(t,\lambda,k) dt dy.
  $$
  Then we get by induction
  $$
    |m_{L2}^n(x,\lambda,k)| \leq \frac{1}{(2n+1)!} \Big( \int_x^{+\infty} a(y) dy \Big)^{2n+1}.
  $$
  For $x,\lambda \in \R$ fixed, the serie $m_{L2}(x,\lambda,k,z) = \displaystyle\sum_{n=0}^\infty m_{L2}^n(x,\lambda,k)
  z^{2n+1}$ converges absolutely and uniformly on each compact subset of $\C$ and satisfies the estimate
  $$
    |m_{L2}(x,\lambda,k, z)| \leq \sinh\Big(|z| \int_x^{+\infty} a(s)ds \Big), \quad \forall x \in \R, \ z \in \C .
  $$
  Moreover, the application $z \longrightarrow m_{L2}(x,\lambda,k,z)$ is entire and odd.  \\

  \noindent (iii) Set $m_{L3}^0(x,\lambda,k) = i \int_x^{+\infty} e^{2 i \lambda (y-x)} {\overline{q_k(y)}} dy$ and for all $n \geq 1$                    $$
    m_{L3}^n(x,\lambda,k) = \int_x^{+\infty} \int_y^{+\infty} e^{2 i \lambda (y-x)} {\overline{q_k(y)}} q_k(t) m_{L3}^{n-1}(t,\lambda,k) dt dy.
  $$
  Then we get by induction
  $$
    |m_{L3}^n(x,\lambda,k)| \leq \frac{1}{(2n+1)!} \Big( \int_x^{+\infty} a(y) dy \Big)^{2n+1}.
  $$
  For $x, \lambda \in \R$ fixed, the serie $m_{L3}(x,\lambda,k,z) = \displaystyle\sum_{n=0}^\infty m_{L3}^n(x,\lambda,k)
  z^{2n+1}$ converges absolutely and uniformly on each compact subset of $\C$ and satisfies the estimate
  $$
    |m_{L3}(x,\lambda,k,z)| \leq \sinh\Big(|z| \int_x^{+\infty} a(s)ds \Big), \quad \forall x \in \R, \ z \in \C.
  $$
  Moreover, the application $z \longrightarrow m_{L3}(x,\lambda,k,z)$ is entire and odd. \\

  \noindent (iv) Set $m_{L4}^0(x,\lambda,k) = 1$ and for all $n \geq 1$
  $$
    m_{L4}^n(x,\lambda,k) = \int_x^{+\infty} \int_y^{+\infty} e^{-2 i \lambda (t - y)} {\overline{q_k(y)}} q_k(t) m_{L4}^{n-1}(t,\lambda,k) dt dy.
  $$
  Then we get by induction
  $$
    |m_{L4}^n(x,\lambda,k)| \leq \frac{1}{(2n)!} \Big( \int_x^{+\infty} a(y) dy \Big)^{2n}.
  $$
  For $x,\lambda \in \R$ fixed, the serie $m_{L4}(x,\lambda,k,z) = \displaystyle\sum_{n=0}^\infty m_{L4}^n(x,\lambda,k)
  z^{2n}$ converges absolutely and uniformly on each compact subset of $\C$ and satisfies the estimate
  $$
    |m_{L4}(x,\lambda,k, z)| \leq \cosh\Big(|z| \int_x^{+\infty} a(s)ds \Big), \quad \forall x \in \R, \ z \in \C .
  $$
  Moreover, the application $z \longrightarrow m_{L4}(x,\lambda,k, z)$ is entire and even. \\

  \noindent (v) Note at last the obvious symmetries
  \begin{eqnarray}
  m_{L1}(x,\lambda,k,z) = \overline{m_{L4}(x,\lambda,k,\bar{z})}, \quad \forall z \in \C, \\
  m_{L2}(x,\lambda,k,z) = \overline{m_{L3}(x,\lambda,k,\bar{z})}, \quad \forall z \in \C.
  \end{eqnarray}
\end{lemma}
Of course we have similar results for the Faddeev functions $m_{Rj}(x,\lambda,k,z), \ j=1,..,4$.
\begin{remark}
It is clear that the Jost functions $F_L(x,\lambda,k,z)$ and $F_R(x,\lambda,k,z)$  are also entire in $z \in \C$. Moreover,
using (\cite{AKM}, Prop. 2.2) and the analytic continuation, we have
\begin{equation} \label{determinant}
  det(F_L(x,\lambda,k,z)) = det(F_R(x,\lambda,k,z)) = 1, \quad \forall x \in \R, \ z \in \C.
\end{equation}
Using the notations
  $$
F_L(x,\lambda,k,z) = \left[\begin{array}{cc} f_{L1}(x,\lambda,k,z)&f_{L2}(x,\lambda,k,z)\\f_{L3}(x,\lambda,k,z)&f_{L4}(x,\lambda,k,z) \end{array} \right],
$$
$$
  F_R(x,\lambda,k,z) = \left[\begin{array}{cc} f_{R1}(x,\lambda,k,z)&f_{R2}(x,\lambda,k,z)\\f_{R3}(x,\lambda,k,z)&f_{R4}(x,\lambda,k,z) \end{array} \right],
  $$
we see that
 \begin{eqnarray}
  f_{Lj}(x,\lambda,k,z) = e^{i\lambda x} m_{Lj}(x,\lambda,k,z), & f_{Rj}(x,\lambda,k,z) = e^{i\lambda x} m_{Rj}(x,\lambda,k,z),
  & j=1,3, \label{fj-mj1} \\
  f_{Lj}(x,\lambda,k,z) = e^{-i\lambda x} m_{Lj}(x,\lambda,k,z), & f_{Rj}(x,\lambda,k,z) = e^{-i\lambda x} m_{Rj} (x,\lambda,k,z), & j=2,4. \label{fj-mj2}
\end{eqnarray}
  Secondly, using the integral equations (\ref{IE-ML1})-(\ref{IE-MR4}), we prove easily that the $f_{Lj}(x,\lambda,k,z)$ and $f_{Rj}(x,\lambda,k,z)$ satisfy second order differential equations with complex potentials.  For instance, the components $f_{Lj}(x,\lambda,k,z)$ and $f_{Rj}(x,\lambda,k,z)$, $j=1,2$ satisfy
\begin{equation} \label{2ndOrder1}
  \Big[ -\frac{d^2}{dx^2} + \frac{q'_k(x)}{q_k(x)} \frac{d}{dx} + z^2 a^2(x)
  - i\lambda \frac{q'_k(x)}{q_k(x)} \Big] f = \lambda^2 f,
\end{equation}
where
\begin{equation}
\frac{q'_k(x)}{q_k(x)}= \frac{a'(x)}{a(x)} +2i c(x,k).
\end{equation}
Similarly, the components $f_{Lj}(x,\lambda,k,z)$ and $f_{Rj}(x,\lambda,k,z)$, $j=3,4$ satisfy
\begin{equation} \label{2ndOrder2}
  \Big[ -\frac{d^2}{dx^2} + \overline{\left(\frac{q'_k(x)}{q_k(x)}\right)}  \frac{d}{dx} + z^2 a^2(x)
  + i\lambda \overline{\left(\frac{q'_k(x)}{q_k(x)}\right)}  \Big] f = \lambda^2 f.
\end{equation}
\end{remark}

Now, let us study the matrix
$$
  A_L(\lambda,k,z) = \left[\begin{array}{cc} a_{L1}(\lambda,k, z)&a_{L2}(\lambda,k, z)\\
a_{L3}(\lambda,k, z)&a_{L4}(\lambda,k, z) \end{array} \right],
$$
for $z \in \C$. Using (\ref{ALRepresentation}) and (\ref{Faddeev}), we express
the components of $A_L(\lambda,k,z)$ by means of the Faddeev functions $m_{Lj}(x,\lambda,k,z)$ as
\begin{eqnarray}
  a_{L1}(\lambda,k,z) & = & 1 - i z \int_\R q_k(x) m_{L3}(x,\lambda,k,z) dx, \label{al1} \\
  a_{L2}(\lambda,k,z) & = & - i z \int_\R e^{-2i \lambda x} q_k(x) m_{L4}(x,\lambda,k,z) dx, \label{al2} \\
  a_{L3}(\lambda,k,z) & = & i z \int_\R e^{2i \lambda x} \overline{q_k(x)} m_{L1}(x,\lambda,k,z) dx, \label{al3} \\
  a_{L4}(\lambda,k,z) & = & 1 + i z \int_\R \overline{q_k(x)} m_{L2}(x,\lambda,k,z) dx. \label{al4}
\end{eqnarray}
Hence we get using Lemma \ref{ML-Analytic}
\begin{lemma} \label{AL-Analytic}
  (i) For $\lambda \in \R$ fixed and for all $z \in \C$,
  \begin{eqnarray*}
    a_{L1}(\lambda,k,z) & = & 1 - i \sum_{n=0}^\infty \Big( \int_\R q_k(x) m_{L3}^n(x,\lambda,k) dx \Big) z^{2n+2}, \\
    a_{L2}(\lambda,k,z) & = & - i \sum_{n=0}^\infty \Big( \int_\R e^{-2i \lambda x} q_k(x) m_{L4}^n(x,\lambda,k) dx \Big) z^{2n+1}, \\
    a_{L3}(\lambda,k,z) & = & i \sum_{n=0}^\infty \Big( \int_\R e^{2i \lambda x} \overline{q_k(x)} m_{L1}^n(x,\lambda,k) dx \Big) z^{2n+1}, \\
    a_{L4}(\lambda,k,z) & = & 1 + i \sum_{n=0}^\infty \Big( \int_\R \overline{q_k(x)} m_{L2}^n(x,\lambda,k) dx \Big) z^{2n+2}.
  \end{eqnarray*}
  (ii) Set $A = \displaystyle\int_\R a(x) dx$. Then
  \begin{eqnarray}
    |a_{L1}(\lambda,k,z)|, \ |a_{L4}(\lambda,k,z)| \leq \cosh(A|z|), \quad \forall z \in \C, \label{AL-ExpType1}\\
    |a_{L2}(\lambda,k,z)|, \ |a_{L3}(\lambda,k,z)| \leq \sinh(A|z|), \quad \forall z \in \C. \label{AL-ExpType2}
  \end{eqnarray}
  (iii) The functions $a_{L1}(\lambda,k,z)$ and $a_{L4}(\lambda,k,z)$ are entire and even in $z$ whereas the functions
  $a_{L2}(\lambda,k,z)$ and $a_{L3}(\lambda,k,z)$ are entire and odd in $z$. Moreover they satisfy the symmetries
  \begin{eqnarray}
    a_{L1}(\lambda,k,z) & = & \overline{a_{L4}(\lambda,k,\bar{z})}, \quad \forall z \in \C, \label{ALSym1}\\
    a_{L2}(\lambda,k,z) & = & \overline{a_{L3}(\lambda,k,\bar{z})}, \quad \forall z \in \C. \label{ALSym2}
  \end{eqnarray}
  (iv) The following relations hold for all $z \in \C$
  \begin{eqnarray}
    a_{L1}(\lambda,k,z) \overline{a_{L1}(\lambda,k,\bar{z})} - a_{L3}(\lambda,k,z) \overline{a_{L3}(\lambda,k,\bar{z})}
    & = & 1, \label{SymAL1-AL3}\\
    a_{L4}(\lambda,k,z) \overline{a_{L4}(\lambda,k,\bar{z})} - a_{L2}(\lambda,k,z) \overline{a_{L2}(\lambda,k,\bar{z})}
    & = & 1. \label{SymAL2-AL4}
  \end{eqnarray}
\end{lemma}
\begin{proof}
The proof is identical to \cite{DN3}, Lemma 3.3. The relation (\ref{SymAL1-AL3})-(\ref{SymAL2-AL4}) are simply the expression of the unitarity of the scattering matrix associated to $A_L$.
\end{proof}

Note in particular that the components of the matrix $A_L(\lambda,k,z)$ are entire functions of exponential type in the variable $z$.
Precisely, from (\ref{AL-ExpType1}) and (\ref{AL-ExpType2}), we have
\begin{equation} \label{AL-ExpType}
  |a_{Lj}(\lambda,k,z)| \leq e^{A |z|}, \quad \forall z \in \C, \ j=1,..,4,
\end{equation}
where  ${\displaystyle{A = \int_\R a(x) dx}}$. Using the relations (\ref{SymAL1-AL3}), (\ref{SymAL2-AL4}), the parity properties of the $a_{Lj}(\lambda,k,z)$ and the Phragm\'en-Lindel\"of theorem, we can slightly improve this estimate as in \cite{DN3}. Precisely, we have

\begin{lemma} \label{MainEsti}
  Let $\lambda \in \R$ be fixed. Then for all $z \in \C$
  \begin{equation} \label{MainEst}
    |a_{Lj}(\lambda,k,z)| \leq e^{A |Re(z)|}, \quad j=1,..,4.
  \end{equation}
\end{lemma}

\begin{proof}
We refer the reader to \cite{DN3}, Lemma 3.4.
\end{proof}

\begin{remark}
For later use, we mention that we have also the corresponding estimates for the Jost functions $f_{Lj}(x,\lambda,k,z)$ and $f_{Rj}(x,\lambda,k,z)$. Precisely. for all $j=1,..,4$ and for all $x \in \R$,
\begin{eqnarray}\label{MainEstiF}
  |f_{Lj}(x,\lambda,k,z)| \leq C \, e^{|Re(z)| \int_x^{\infty} a(s) ds}, \\
  |f_{Rj}(x,\lambda,k,z)| \leq C \, e^{|Re(z)| \int_{-\infty}^x a(s) ds}.
\end{eqnarray}
\end{remark}

\begin{remark}
In \cite{DN3}, Lemma 3.6, it was shown that for Reissner-Nordström-de-Sitter black holes ($a = 0$ in our model), the scattering matrix at energy $0$ could be explicitely calculated and given precisely by
$$
  S(0,k,z) = \left[\begin{array}{cc} \cosh(zA)& i\sinh(zA)\\
-i \sinh(zA)&\cosh(zA) \end{array} \right], \quad A = \int_\R a(x) dx.
$$
As a by product, the full scattering matrix $S(0)$ did not determine uniquely the parameters of a RN-dS black hole. For Kerr-Newmann-de-Sitter
black holes ($a \ne 0$), the situation is completely different. First, due to the presence of the non vanishing phase $C(x,k)$ appearing in
the potential $q(x,k)$, we can not obtain explicit formulae. Second, we shall see in Section \ref{Inverse} that, even in the case of zero energy $\lambda=0$, the parameters of a KN-dS black hole are uniquely determined by the (partial) knowledge of the scattering matrix.
\end{remark}


\subsection{Nevanlinna class.} \label{Nevanlinna}

Let us denote the right half complex plane by $\Pi^+ = \{z \in\C: \ Re(z) >0\}$. We recall that the Nevanlinna class $N(\Pi^+)$ is defined as the set of all analytic functions $f(z)$ on $\Pi^+$ that satisfy the estimate
$$
  \sup_{0<r<1} \int_{-\pi}^{\pi} \ln^+ \Big| f\Big(\frac{1 - re^{i\theta}}{1+re^{i\theta}} \Big) \Big| d\theta < \infty,
$$
where $\ln^+(x) = \left\{ \begin{array}{cc} \ln x, & \ln x \geq 0,\\ 0, & \ln x <0. \end{array} \right.$ In \cite{Ra}, (see also \cite{DN3}, Lemma 3.8), the following lemma is proved

\begin{lemma} \label{NevanlinnaCriterion}
  Let $h \in H(\Pi^+)$ be an holomorphic function in $\Pi^+$ satisfying
  \begin{equation} \label{Esti}
   |h(z)| \leq C e^{A \,Re(z)}, \quad \forall z \in \Pi^+,
  \end{equation}
  where $A$ and $C$ are two constants. Then $h \in N(\Pi^+)$.
\end{lemma}
Thanks to Lemmata \ref{MainEsti} and \ref{NevanlinnaCriterion}, we thus get in our model
\begin{coro} \label{AL-Nevanlinna}
  For each $\lambda \in \R$ fixed, the applications $z \longrightarrow a_{Lj}(\lambda,k,z)_{|\Pi^+}$ belong to $N(\Pi^+)$.
\end{coro}

Let us prove now a usefull uniqueness theorem involving functions in the Nevanlinna class $N(\Pi^+)$. This result is very close to
\cite{Ra}, Theorem 1.3, and uses  the following result :

\begin{theorem}[\cite{Ru}, Thm 15.23] \label{UniquenessNevanlinnadisc}
Let $f$ be a holomorphic function in the unit disc $D$. Assume that $f$ belongs to the Nevanlinna class in $D$, that is
$$
\sup_{0<r<1} \ \int_{-\pi}^{\pi} \ln^+( \mid f(re^{i\theta} \mid) \ d\theta \ < \ \infty.
$$
Assume also that $f(\alpha_l)=0$ for a sequence $\alpha_l \in D$ satisfying
$$
\sum_{l=0}^{\infty} (1- |\alpha_l|) = \infty.
$$
Then, $f \equiv 0$ in $D$.
\end{theorem}

Now, we can establish the following theorem.

\begin{theorem} \label{UniquenessNevanlinna}
  For $\lambda \in \R$ and $k \in \frac{1}{2} + \Z$ fixed, let $h \in N(\Pi^+)$ satisfy $h(\muk) = 0$ for all $l \in \mathcal{L}$ where $\mathcal{L} \subset \N^*$ with $\displaystyle\sum_{l \in \mathcal{L}} \frac{1}{l} = \infty$. Then $h \equiv 0$ in $\Pi^+$.
\end{theorem}

\begin{proof}
The function $g : D \rightarrow \Pi^+$ defined by
$$
g(z) = \frac{1-z}{1+z}
$$
maps conformally $D$ onto $\Pi^+$ and for $u \in \Pi^+$,
$$
g^{-1}(u) = \frac{1-u}{1+u}.
$$
Let us define $f : D \rightarrow \C$ by ${\displaystyle{f(z)= h(g(z))= h(\frac{1-z}{1+z}) }}$. Clearly, $f$ belongs to the Nevanlinna class in $D$ if $h$ belong s to the Nevanlinna class in $\Pi^+$. Moreover, from our assumptions, we have $f(\alpha_l)=0$ for ${\displaystyle{\alpha_l = \frac{1-\muk}{1+\muk}}}$. But according to Appendix \ref{Estimate-mukl}, there exist suitable constants $ 0 < C_1 < C_2$ such that for all $l \in \N^*$
$$
  C_1 l \leq \muk \leq C_2 l
$$
Hence we get for $l$ large enough,
$$
1-| \alpha_l | = 1- \left| \frac{1-\muk}{1+\muk} \right| = \frac{2}{1+\muk} \geq \frac{2}{1 + C_2 l},
$$
Thus,
$$
\sum_{l \in \mathcal{L}} \ (1- |\alpha_l|) = \infty.
$$
From Theorem \ref{UniquenessNevanlinnadisc}, we obtain $ f \equiv 0$ and then $ h \equiv 0$.
\end{proof}

Now we deduce from Corollary \ref{AL-Nevanlinna} and Theorem \ref{UniquenessNevanlinna}
\begin{coro} \label{AL-Uniqueness}
  Consider two KN-dS black holes and denote by $a_{Lj}(\lambda,k,z)$ and $\tilde{a}_{Lj}(\lambda,k,z)$ the corresponding scattering data for fixed $\lambda$ and $k$. Let $\mathcal{L} \subset \N^*$ satisfying $\displaystyle\sum_{l \in \mathcal{L}} \frac{1}{l} = \infty$. Assume that one of the following equality holds
$$
  a_{Lj}(\lambda,k,\muk) = \tilde{a}_{Lj}(\lambda,k,\muk), \quad \forall l \in \mathcal{L}, \quad j=1,..,4.
$$
Then
$$
  a_{Lj}(\lambda,k,z) = \tilde{a}_{Lj}(\lambda,k,z), \quad \forall z \in \C, \quad j=1,..,4.
$$
\end{coro}

\begin{proof}
Using Lemma \ref{NevanlinnaCriterion}, we see that $a_{Lj}(\lambda,k,z) - \tilde{a}_{Lj}(\lambda,k,z)$ belongs to $N(\Pi^+)$. Then, we apply Theorem \ref{UniquenessNevanlinna}.
\end{proof}

\begin{remark}
Till here, we have seen that the scattering data $a_{Lj}(\lambda,k,z)$ are uniquely determined as functions of $z \in \C$ if we know their values on the physical angular momenta $\muk, \ \forall l \in \cal{L}$. We now improve this uniqueness result assuming that only the physical scattering data - precisely the reflection coefficients $\hat{L}_{kl}(\lambda)$ or $\hat{R}_{kl}(\lambda)$ -  are known (up to phase factors in order to obtain a result independent of our choice of radial coordinates, (see (\ref{RW-Unphys-SM-Red})).
\end{remark}


\begin{prop} \label{Uniqueness}
  Consider two KN-dS black holes and denote by $Z$ and $\tilde{Z}$ all the corresponding scattering data. Assume that, for $\lambda$ and $k$ fixed, there exists a constant $d(\lambda,k) \in \R$ such that one of the following equality holds for all $l \in \mathcal{L}$
  with $\mathcal{L} \subset \N^*$ satisfying $\displaystyle\sum_{l \in \mathcal{L}} \frac{1}{l} = \infty$:
  \begin{eqnarray}
    \hat{L}_{kl}(\lambda) & = & e^{-id(\lambda,k)} \  \hat{\tilde{L}}_{kl}(\lambda), \label{L} \\
    \hat{R}_{kl}(\lambda) & = & e^{id(\lambda,k)} \ \hat{\tilde{R}}_{kl}(\lambda). \label{R}
  \end{eqnarray}
  Assume moreover that $\muk = \tilde{\mu}_{kl}(\lambda), \ \forall l \in \mathcal{L}$. Then for all $z \in \C$,
  \begin{equation} \label{ALUniqueness}
    \left[ \begin{array}{cc} a_{L1}(\lambda,k,z)& a_{L2}(\lambda,k,z)\\ a_{L3}(\lambda,k,z)&a_{L4}(\lambda,k,z) \end{array}
    \right] = \left[ \begin{array}{cc} \tilde{a}_{L1}(\lambda,k,z)&  e^{id(\lambda,k)}\  \tilde{a}_{L2}(\lambda,k,z)\\
    e^{-id(\lambda,k)}\  \tilde{a}_{L3}(\lambda,k,z)& \tilde{a}_{L4}(\lambda,k,z) \end{array} \right].
  \end{equation}
  \end{prop}

\begin{proof}
The proof is identical to \cite{DN3}, Proposition 3.12. We use Corollary \ref{AL-Nevanlinna} and the equality
$$
\hat{L}_{kl}(\lambda) = \frac{a_{L3}(\lambda,k, \muk)}{a_{L1}(\lambda,k, \muk)}.
$$
\end{proof}

\par
Analogously, if we assume that only the transmission coefficient $\hat{T}_{kl}(\lambda)$ is known, we have the following result

\begin{prop} \label{Uniqueness2}
  Consider two KN-dS black holes and denote by $Z$ and $\tilde{Z}$ all the corresponding scattering data. Assume that for $\lambda$ and $k$ fixed, we have
$$
  \hat{T}_{kl}(\lambda) = \hat{\tilde{T}}_{kl}(\lambda), \quad \muk = \tilde{\mu}_{kl}(\lambda),
$$
for all $l \in \mathcal{L} \subset \N^*$ satisfying $\displaystyle\sum_{l \in \mathcal{L}} \frac{1}{l} = \infty$. Then, there exists $d(\lambda,k)\in \R$ such  that for all $z \in \C$
\begin{equation} \label{ALUniqueness2}
    \left[ \begin{array}{cc} a_{L1}(\lambda,k,z)& a_{L2}(\lambda,k,z)\\ a_{L3}(\lambda,k,z)&a_{L4}(\lambda,k,z) \end{array}
    \right] = \left[ \begin{array}{cc} \tilde{a}_{L1}(\lambda,k,z)& e^{i d(\lambda,k)} \tilde{a}_{L2}(\lambda,k,z)\\
    e^{-i d(\lambda,k)} \tilde{a}_{L3}(\lambda,k,z)& \tilde{a}_{L4}(\lambda,k,z) \end{array} \right].
\end{equation}
\end{prop}

\begin{proof}
 The proof of (\ref{ALUniqueness2}) is identical to \cite{DN3}, Proposition 3.13.
\end{proof}


\subsection{A first uniqueness result.}

We finish this section with a first application of the previous uniqueness results to the study of an inverse scattering
problem in which the {\it{reflection}} coefficients $L_{kl}(\lambda)$ or $R_{kl}(\lambda)$ are supposed to be known on an \emph{interval}
of energy (and not simply at a fixed energy $\lambda$). Precisely we prove Theorem \ref{Uniqueness-Localized} that, for the convenience of the reader,  we recall here

\begin{theorem}
  Consider two KN-dS black holes and denote by $Z$ and $\tilde{Z}$ all the corresponding scattering data. Let $I$ be a (possibly small) open interval of energy. Assume that, there exists $c \in \R$ such that, for all $\lambda \in I$ and for two different $k$, one of the following conditions holds for all $l \in \mathcal{L}_k$ where the subsets $\mathcal{L}_k \subset \N^* $ satisfy the M\"untz condition $\displaystyle{\sum_{l \in \mathcal{L}_k} \frac{1}{l} = \infty}$ :
  \begin{eqnarray}
   L_{kl} (\lambda) &=& e^{-2i\lambda^-(k) c} \tilde{L}_{kl} (\lambda),  \\ \label{uniciteL}
   R_{kl} (\lambda) &=& e^{2i\lambda^+(k) c} \tilde{R}_{kl} (\lambda), \label{uniciteR}
  \end{eqnarray}
  where $\lambda^\pm(k) = \lambda - \Omega_\pm(k)$. Then for all $x \in \R$,
  $$
    a(x) = \tilde{a}(x-c), \quad c(x,k) = \tilde{c}(x-c,k),
  $$
  and the parameters of the two KN-dS black holes coincide.
\end{theorem}

\begin{proof}
Let us prove for instance the corollary under condition (\ref{uniciteR}). As it was explained in Proposition \ref{SM-Red-RW}, the condition (\ref{uniciteR}) on the physical reflection operators is equivalent to the following equality on the unphysical reflection coefficients
\begin{equation}
\hat{R}_{kl} (\lambda) = e^{2i \lambda^+(k) c} \hat{\tilde{R}}_{kl} (\lambda).
\end{equation}
Applying Proposition \ref{Uniqueness}, we see in particular that
$$
  a_{L2}(\lambda,k,z) = e^{2i \lambda c}\  e^{-2i \Omega_+(k) c} \ \tilde{a}_{L2}(\lambda,k,z), \quad \forall z \in \C, \ \forall \lambda \in I.
$$
From the first term of the series defining $a_{L2}(\lambda,k,z)$ and $\tilde{a}_{L2}(\lambda,k, z)$ (see Lemma \ref{AL-Analytic}, (ii)), we thus obtain
\begin{equation} \label{FourierPot-a}
  \hat{q}_k(2\lambda) = e^{2i \lambda c} \  e^{-2i \Omega_+(k) c} \  \hat{\tilde{q}}_k(2\lambda), \quad \forall \lambda \in I,
\end{equation}
where $ \hat{q}_k$ and $ \hat{\tilde{q}}_k$ denote the Fourier transforms of the potentials $q_k$ and $\tilde{q}_k$. Since these potentials are exponentially decreasing at both horizons, their Fourier transforms are analytic in a small strip $K$ around the real axis. Thus the equality (\ref{FourierPot-a}) extends analytically to the whole strip $K$. In particular, we have
\begin{equation}
  \hat{q}_k(2\lambda) = e^{2i \lambda c} \  e^{-2i \Omega_-(k) c } \ \hat{\tilde{q}}_k(2\lambda), \quad \forall \lambda \in \R,
\end{equation}
and therefore
\begin{equation} \label{EqualityPotentials}
  q_k(x) = \  e^{-2i \Omega_-(k) c}  \ \tilde{q}_k(x-c), \quad \forall x \in \R.
\end{equation}
Note that from the modulus of (\ref{EqualityPotentials}), we obtain
\begin{equation}\label{unicitea}
a(x)= \tilde{a}(x-c).
\end{equation}
Also, taking the logarithmic derivative with respect to $x$ in (\ref{EqualityPotentials}), we get
\begin{equation}
\frac{q'_k(x)}{q_k(x)} = \frac{\tilde{q}_k'(x-c)}{\tilde{q}_k(x-c)},
\end{equation}
or equivalently,
\begin{equation} \label{log}
\frac{a'(x)}{a(x)} +2ic(x,k) = \frac{\tilde{a}'(x-c)}{\tilde{a}(x-c)} + 2i\tilde{c}(x-c,k).
\end{equation}
Thus, from the imaginary parts of (\ref{log}), we obtain
\begin{equation} \label{uniciteck}
c(x,k) = \tilde{c}(x-c,k).
\end{equation}

Now, let us use the particular forms of the potentials $a(x)$ and $c(x,k)$ to recover the parameters of the black hole. Recall that
\begin{equation} \label{defc}
c(x,k) = \frac{aEk +qQr}{r^2+a^2} \ ,\ E= 1+ \frac{a^2\Lambda}{3}.
\end{equation}
Using that (\ref{uniciteck}) is true for two different $k$, and using the fact that $a=\tilde{a}$, $\Lambda = \tilde{\Lambda}$ (proved in Thm \ref{Main-Frobenius}), we obtain the following decoupled equations
\begin{eqnarray}\label{systemedecouple}
\frac{qQr}{r^2+a^2} & = & \frac{q\tilde{Q}\tilde{r}}{\tilde{r}^2+a^2}, \label{systemedecouple1}\\
\frac{aE}{r^2+a^2} & = & \frac{aE}{\tilde{r}^2+a^2}, \label{systemedecouple2}
\end{eqnarray}
where we use the notation $\tilde{r} = \tilde{r}(x-c)$. Since $a\not=0$ and $E\not=0$, we deduce from (\ref{systemedecouple2}) that $r = \tilde{r}$.
Now, from (\ref{systemedecouple1}), we obtain $Q=\tilde{Q}$ when $q\ne 0$.\\

Similarly, recall that
\begin{equation}\label{pot}
a(x) = \frac{ \sqrt{\Delta_r}}{r^2+a^2},
\end{equation}
with
\begin{equation} \label{Delta}
\Delta_r = (r^2+a^2)(1-\frac{\Lambda}{3}) -2Mr +Q^2.
\end{equation}
Since $r=\tilde{r}$, we obtain using (\ref{unicitea}),
\begin{equation} 
\Delta_r = \tilde{\Delta}_{\tilde{r}} =  \tilde{\Delta}_r.
\end{equation}
Then, from (\ref{Delta}), we obtain $M=\tilde{M}$ and $Q^2 = \tilde{Q}^2$ in the case $q = 0$.
\end{proof}

\begin{remark} \label{NotLocalizedUniquenessT}
It would be also interesting to determine uniquely the parameters of a KN-dS black hole from the knowledge of the transmission coefficients
$T_{kl}(\lambda)$ (up to a phase factor) on an energy interval $I$. More precisely, assume that for two different $k$ (or more), there exist a constant $c$ and subsets $\mathcal{L}_k \subset \N^* $ that satisfy the M\"untz condition $\displaystyle{\sum_{l \in \mathcal{L}_k} \frac{1}{l} = \infty}$ such that
\begin{equation}
T_{kl}(\lambda) = e^{i(\Omega_+(k) - \Omega_-(k))c} \ \tilde{T}_{kl}(\lambda) \ ,\ \forall \lambda \in I. \label{eqT}
\end{equation}
We recall from Proposition \ref{SM-Red-RW} that the equality (\ref{eqT}) between the physical transmission coefficients is equivalent to the following equality between the unphysical transmission coefficients
\begin{equation}
\hat{T}_{kl}(\lambda) = \hat{\tilde{T}}_{kl}(\lambda).
\end{equation}
Applying Proposition \ref{Uniqueness2} and following the same strategy as for (\ref{FourierPot-a}), we see there exists a constant $d(\lambda,k)$ such that
\begin{equation} \label{FourierPot-q}
  \hat{q}_k(2\lambda) = e^{2i d(\lambda,k)} \   \  \hat{\tilde{q}}_k(2\lambda), \quad \forall \lambda \in I.
\end{equation}
But, since we do not know how the function $d(\lambda,k)$ depends on the variable $\lambda$, we cannot hope to obtain a similar result as (\ref{EqualityPotentials}). We would like to mention that the corresponding result obtained in \cite{DN3} (see Corollary 3.14) for the transmission coefficient seems to be wrong as stated. Nevertheless, we shall see later that this uniqueness result is in fact even true at a fixed energy $\lambda$!  To prove this, we need a precise asymptotic expansion of the scattering data. This is the object of the next section.
\end{remark}



\Section{Large $z$ asymptotics of the scattering data} \label{AsymptoticsSD}

In this Section, we follow the results of the Section 4 in \cite{DN3} to obtain the asymptotic expansion of the scattering data when the angular momentum $z \rightarrow +\infty$, $z$ real. To do this, we introduce a convenient change of variable $X=g(x)$ which turns out to be a simple Liouville transformation between some underlying Sturm-Liouville equations that appear in our model (see below). This Liouville transformation as well as the asymptotics of the scattering data will be useful later to study the inverse scattering problem.

It is convenient in this Section not to work with the simplified stationary equation (\ref{Stat-Eq-Ak}) but rather with the initial stationary equation (\ref{Rad-0}). Following the reverse procedure given in Section \ref{Simplified-SM}, we introduce the following functions
\begin{eqnarray}
H_L(x,\lambda,k,z) &=& e^{-iC(x,k) \Gamma^1} F_L(x,\lambda,k,z), \label{modJostL}\\
H_R(x,\lambda,k,z) &=& e^{-iC(x,k) \Gamma^1} F_R(x,\lambda,k,z), \label{modJostR}
\end{eqnarray}
that are matrix-valued solutions of the stationary equation
$$
  [\Ga D_z - z a(x) \Gb + c(x,k)] H_{L/R} = \lambda H_{L/R}.
$$
Clearly, from (\ref{determinant}), we see that
\begin{equation}\label{determinantH}
det \ H_L(x,\lambda,k, z) = det \ H_R(x,\lambda,k, z) =1.
\end{equation}
We shall also use the notations
\begin{equation} \label{MatriceHL}
H_L(x,\lambda,k, z) = \left[\begin{array}{cc} h_{L1}(x,\lambda,k, z)&h_{L2}(x,\lambda,k, z)\\
h_{L3}(x,\lambda,k, z)&h_{L4}(x,\lambda,k, z) \end{array} \right],
\end{equation}
\begin{equation} \label{MatriceHR}
H_R(x,\lambda,k, z) = \left[\begin{array}{cc} h_{R1}(x,\lambda,k, z)&h_{R2}(x,\lambda,k, z)\\
h_{R3}(x,\lambda,k, z)&h_{R4}(x,\lambda,k, z) \end{array} \right].
\end{equation}
Thus, for $j = 1,2$, we have
\begin{equation}
  h_{Lj}(x,\lambda,k,z) = e^{-iC(x,k)} f_{Lj}(x,\lambda,k,z), \ h_{Rj}(x,\lambda,k,z) = e^{-iC(x,k)} f_{Rj}(x,\lambda,k,z), \label{hj-fj1}
\end{equation}
whereas for $j= 3,4$, we have
\begin{equation}
  h_{Lj}(x,\lambda,k,z) = e^{+ iC(x,k)} f_{Lj}(x,\lambda,k,z), \ h_{Rj}(x,\lambda,k,z) = e^{ +iC(x,k)} f_{Rj}(x,\lambda,k,z). \label{hj-fj2}
\end{equation}

We recall that the $f_{Lj}(x,\lambda,k,z)$ and $f_{Rj}(x,\lambda,k,z)$ satisfy the second order differential equations (\ref{2ndOrder1}) and (\ref{2ndOrder2}).
Thus, using (\ref{hj-fj1}) and (\ref{hj-fj2}), it is easy to see that the components $h_{Lj}(x,\lambda,k,z)$ and $h_{Rj}(x,\lambda,k,z)$, $j=1,2$ satisfy
\begin{equation} \label{2ndOrder1h}
  \Big[ -\frac{d^2}{dx^2} + \frac{a'(x)}{a(x)} \frac{d}{dx} + z^2 a(x) -i\lambda \frac{a'(x)}{a(x)} +s(x,\lambda,k) \Big] f = \lambda^2 f,
\end{equation}
where
\begin{equation}\label{potentiels}
s(x,\lambda, k) = -ic'(x,k)-c^2(x,k)+ i \frac{a'(x)}{a(x)}c(x,k)+2\lambda c(x,k).
\end{equation}
In the same way, the components $h_{Lj}(x,\lambda,k,z)$ and $h_{Rj}(x,\lambda,k,z)$, $j=3,4$ satisfy
\begin{equation} \label{2ndOrder2h}
   \Big[ -\frac{d^2}{dx^2} + \frac{a'(x)}{a(x)} \frac{d}{dx} + z^2 a(x) +i\lambda \frac{a'(x)}{a(x)} +\overline{s(x, \lambda,k)} \Big] f = \lambda^2 f.
\end{equation}

\subsection{The Liouville transformation.}

We follow the same strategy  as in \cite{DN3}, (see also \cite{CMu}, \cite{CKM}). Considering the differential equations (\ref{2ndOrder1h}) and (\ref{2ndOrder2h}) satisfied by the modified Jost functions $h_{Lj}(x,\lambda,z)$ and $h_{Rj}(x,\lambda,z)$, we use a Liouville transformation $X = g(x)$, that transforms the equations (\ref{2ndOrder1h}) and (\ref{2ndOrder2h}) into singular Sturm-Liouville differential equations in which the angular momentum $z$ becomes the \emph{spectral parameter} (see Lemma \ref{Sturm} below).

We define the variable $X$ by
\begin{equation} \label{Liouville}
  X= g(x)=\int_{-\infty}^x a(t) \ dt.
\end{equation}
Clearly, $g$ is a $C^1$-diffeomorphism from $\R$ to $(0,A)$ where
\begin{equation} \label{DefA}
  A=\int_{-\infty}^{+\infty} a(t) \ dt.
\end{equation}
We denote by $h=g^{-1}$ the inverse diffeomorphism of $g$ and we shall use the notation $y'(X)$ for the derivative of $y(X)$ with respect to $X$. We also define for $j=1,...,4$, and for $X \in ]0,A[$,
\begin{equation} \label{Fj}
  \fj = h_{Lj} (h(X), \lambda,k, z).
\end{equation}
\begin{equation} \label{Gj}
  \gj = h_{Rj} (h(X), \lambda,k, z).
\end{equation}

Let us begin by an elementary result.

\begin{lemma} \label{Sturm} \hfill
\begin{enumerate}
\item For $j=1,2$, $\fj$ and  $\gj$ satisfy on $]0,A[$ the Sturm-Liouville equation
\begin{equation} \label{SL1}
  y'' +Q(X,k)y = z^2 y.
\end{equation}
\item For $j=3,4$, $\fj$ and $\gj$ satisfy on $]0,A[$ the Sturm-Liouville equation
\begin{equation} \label{SL2}
  y'' +\overline{Q(X,k)}y = z^2 y,
\end{equation}
\end{enumerate}
where the potential $Q(X,k)$ is given by
$$
  Q(X,k) = (\lambda-c(h(X),k))^2 h'(X)^2 -i(\lambda-c(h(X),k))h''(X) +i c'(h(X),k)h'(X)^2,
$$
or in the initial Regge-Wheeler variable $x$
$$
  Q(X,k) = \frac{(\lambda-c(x,k))^2}{a^2(x)} +i \left(\frac{a(x)c'(x,k) +a'(x)(\lambda-c(x,k))} {a^3(x)} \right).
$$
\end{lemma}

\begin{proof}
Using (\ref{2ndOrder1h}), (\ref{2ndOrder2h}) and ${\displaystyle{\frac{dX}{dx} = a(x)}}$, a straightforward calculation gives (\ref{SL1}) and (\ref{SL2}).
 \end{proof}

In order to obtain the asymptotics of the modified Jost functions $f_j$ and $g_j$ when $z$ is large and real, we shall solve the Sturm-Liouville equations (\ref{SL1})- (\ref{SL2}) by a perturbative argument. Following \cite{DN3}, we consider $Q(X,k)$ as a perturbation of a function $Q_+(X,k)$ where $Q_+(X,k)$ is given by the same formula as $Q(X,k)$, but with $c(h(X),k)$  replaced by $\Omega_+(k)$ - that is its equivalent at $+\infty$ - and $h(X) = g^{-1}(X)$ is replaced by another diffeomorphism denoted $h^+(X) = g_+^{-1}(X)$ where the diffeomorphism $g_+$ is defined in the same manner as $g$, with $a(x)$ replaced by its equivalent at $+\infty$.

More precisely, if we write $A-g(x) = \int_x^{+\infty}  a(t) \ dt$ and recalling the asymptotics of $a(x)$ given in (\ref{Asymp-a}), it is natural to set for $X \in ]0,A[$
$$
A-g_+(x) = \int_x^{+\infty}  a_+ \ e^{\kappa_+ t} \ dt = -\frac{a_+}{\kappa_+} e^{\kappa_+ x}.
$$
So, we define
\begin{equation}\label{newdiffeo}
h_+ (X)= g_+^{-1} (X) =  \frac{1}{\kappa_+} \ \log (A-X) + \,C_+,
\end{equation}
with
\begin{equation}\label{cstplus}
C_+ =  \frac{1}{\kappa_+} \log (-\frac{\kappa_+}{a_+}).
\end{equation}
Then, we can set
\begin{equation}
Q_+(X,k) =  (\lambda-\Omega_+(k))^2 h_+'(X)^2 -i(\lambda-\Omega_+(k))h_+''(X).
         \end{equation}
Note that an elementary calculation gives
\begin{equation}
Q_+(X,k) = \frac{\omega_+(k)}{(A-X)^2},
\end{equation}
where
\begin{eqnarray}
    \omega_{+}(k)& = & \left(\frac{\lambda^+(k)}{\kappa_\pm}\right)^2 + i\  \frac{\lambda^+(k)}{\kappa_{\pm}}, \\
    \lambda^+(k) & = &\lambda-\Omega_+(k).
\end{eqnarray}

Now, let us prove that $Q_+(X,k)$  is a suitable perturbation of $Q(X,k)$, in the sense that in the variable $X$, the  modified Jost functions $f_j (X,\lambda,k,z)$ and $g_j (X,\lambda,k,z)$ satisfy Sturm Liouville equations with potentials having quadratic singularities at the boundary $X=A$. To do this, we have to control $h(X)-h_+(X)$. \\

We recall the following result (\cite{DN3}, Lemma 4.7):

\begin{lemma}\label{estdiffeo}
Let $X_0 \in]0,A[$ fixed. Then, there exists $C>0$ such that for $k =0,1,2$
\begin{equation}\label{estimdiffeo}
\mid h^{(k)}(X)-h_+^{(k)}(X) \mid \ \leq\ C\ (A-X)^{2-k}\ \ ,\ \ \forall X \in ]X_0,A[.
\end{equation}
\end{lemma}
Then, we can prove
\begin{lemma}
Let $X_0 \in]0,A[$ fixed. Then
\begin{equation} \label{omega}
       Q(X,k) = Q_+(X,k) + \ q_+(X,k) \ , \ \ {\rm{with}} \ \ q_+(X,k) \in L^{\infty}(X_0,A),
\end{equation}
\end{lemma}

\begin{proof}
We recall that the potential $a(x)$ has the asymptotics (\ref{Asymp-a}) when $x \to +\infty$
\begin{equation}
a(x) =  a_+ \ e^{\kappa_+ x} + \ O(e^{3\kappa_+ x}).
\end{equation}
When $x \rightarrow + \infty$, or equivalently when $X \rightarrow A$, it thus follows that
$$ 
A-X=\int_x^{+\infty} a(t) \ dt = \frac {a_+}{\kappa_+} \ e^{\kappa_+ x} + \ O(e^{3\kappa_+ x}) = \frac {1}{\kappa_+} \ a(x) + \ O(e^{3\kappa_+ x}),
$$
from which we also have
$$
  e^{\kappa_+ x} = O(A-X), \quad X \to A.
$$
Hence, when $X \rightarrow A$,
\begin{equation} \label{aX}
  a(x) = \kappa_+ (A-X) + O((A-X)^3).
\end{equation}
Similarly, we recall the asymptotics (\ref{Asymp-c}) of the function $c(x,k)$ when $x \rightarrow + \infty$
\begin{eqnarray}
  c(x,k) &=& \Omega_+ (k)+ c_+ e^{2\kappa_+ x} +O(e^{4\kappa_+ x}), \label{cx1} \\
  c'(x,k) &=& 2c_+ \kappa_+  e^{2\kappa_+ x} +O(e^{4\kappa_+ x}). \label{cx2}
\end{eqnarray}
Hence, we obtain when $X \rightarrow A$:
\begin{eqnarray}
c(x,k) & = & \Omega_+ (k) +O((A-X)^2), \label{cX1} \\
c'(x,k)& = & O((A-X)^2). \label{cX2}
\end{eqnarray}
Moreover, it follows from Lemma \ref{estdiffeo} that for $X\in (X_0,A)$,
\begin{equation} \label{h}
h'(X) = O((A-X)^{-1}) \ , \  h''(X) = O((A-X)^{-2}).
\end{equation}
Now, using (\ref{aX}), (\ref{cX1}), (\ref{cX2}), (\ref{h}) and Lemma \ref{estdiffeo}, we write for $X\in (X_0,A)$ :
\begin{eqnarray*}
Q(X,k) & = & (\lambda-c(h(X),k))^2 h'(X)^2 -i(\lambda-c(h(X),k))h''(X) +i c'(h(X),k)h'(X)^2,\\
       & = & (\lambda-c(h(X),k))^2 h'(X)^2 -i(\lambda-c(h(X),k))h''(X) +O(1),\\
       & = & [\lambda^+(k)-(c(h(X),k)-\Omega^+(k))]^2 h'(X)^2 \nonumber \\
       &   &-i[\lambda^+(k)-(c(h(X),k)-\Omega^+(k))]h''(X) +O(1),\\
       & = &  (\lambda^+(k))^2 h'(X)^2 -i\lambda^+(k)h''(X) +O(1),\\
       & = &  (\lambda^+(k))^2 h_+'(X)^2 -i\lambda^+(k)h_+''(X) +O(1),\\
       & = & Q_+(X,k) +O(1).
\end{eqnarray*}

\end{proof}

Let us calculate here some Wronskians between modified Jost functions that will be useful later. We recall that the Wronskian of two functions $f, g$ is given by $W(f, g) = fg'-f'g$.

\begin{lemma} \label{wronskien} For $z \in \C$, we have :
$$
  W(f_1, f_2)=W(g_1, g_2)= W(f_3, f_4) =W(g_3, g_4) = iz.
$$
\end{lemma}

\begin{proof}
For example, let us calculate $W(f_1, f_2)$. Using ${\displaystyle{\frac{dX}{dx} = a(x)}}$ again, it is clear that
\begin{eqnarray}
W(f_1, f_2) & = &  \frac{1}{a(x)} W(e^{-iC(x,k)} f_{L1}(x,\lambda,k,z), \ e^{-iC(x,k)} f_{L2}(x,\lambda,k,z)) \\
            & = &  \frac{e^{-2iC(x,k)}}{a(x)} W(f_{L1}(x,\lambda,k,z), \ f_{L2}(x,\lambda,k,z)).
\end{eqnarray}
Using (\ref{PartialSE1}) and (\ref{determinant}), we obtain easily :
\begin{eqnarray*}
W(f_{L1}(x,\lambda,z), f_{L2}(x,\lambda,z)) & = & f_{L1} [i\lambda f_{L2} +iz q(x,k) f_{L4}] - f_{L2} [i\lambda f_{L1} +iz q(x,k) f_{L3}] \\
                  & = & iz q(x,k)\ det\ F_L(x,\lambda,k,z) = iz q(x,k),
\end{eqnarray*}
which implies the result.
\end{proof}

\subsection{Asymptotics of the Jost functions.}

In this section, we determine the asymptotics of the modified Jost functions $f_j(X,\lambda,k,z)$ and $g_j(X,\lambda,k,z)$ when $z \rightarrow +\infty$. We emphasize here that the same type of asymptotics for solutions of singular Sturm-Liouville equations more general than (\ref{SL1}) - (\ref{SL2}) have been yet studied by Freiling and Yurko in \cite{FY}. However their asymptotics are given up to multiplicative constants that, in our case, should be determined precisely in order to solve our inverse scattering problem. So, we prefer to follow a self-contained approach which only uses the series expansion for the Jost functions obtained in Section \ref{Complexification} and a simple perturbative argument.

Let us study first the modified Jost function $h_{L1}(x, \lambda,k,z)$. From Lemma \ref{ML-Analytic} and (\ref{fj-mj1}), we recall that for $z \in \C$,

\begin{equation} \label{jostgauche}
h_{L1}(x, \lambda,k,z) = e^{-iC(x,k)} e^{i\lambda x}\ \sum_{n=0}^{+\infty} \ m_{L1}^n (x, \lambda,k) \ z^{2n}
\end{equation}
where
\begin{eqnarray*}
m_{L1}^0 (x,\lambda,k) &=& 1, \\
m_{L1}^n (x,\lambda,k) &=&  \int_x^{+\infty} \int_y^{+\infty}\ e^{2i\lambda (t-y)} \ q_k(y)\ \overline{q_k(t))} \ m_{L1}^{n-1} (t, \lambda,k) \ dt \ dy \ ,
\ {\rm{for}} \ n \geq 1.
\end{eqnarray*}

So, using the Liouville transformation, we immediately obtain

\begin{lemma} \label{JostX} \hfill
\begin{equation}\label{funX}
\fun = \sum_{n=0}^{+\infty} a_n (X,\lambda,k) \ z^{2n},
\end{equation}
where
\begin{eqnarray} \label{anX}
a_0 (X,\lambda,k) &=& e^{-iC(h(X),k)}\  e^{i\lambda h(X)}, \nonumber \\
a_n (X,\lambda,k) &=& e^{-iC(h(X),k)} \ e^{i\lambda h(X)} \ \int_X^A \int_Y^A e^{-2i\lambda h(Y)} \ e^{i\lambda h(T)} e^{2iC(h(Y),k)} \ e^{-iC(h(T),k)} \nonumber \\
                & & \hspace{4,4cm}  \ a_{n-1} (T, \lambda,k) \ dT \ dY \ , \ {\rm{for}} \ n \geq 1.
\end{eqnarray}
\end{lemma}

As previously, we shall consider $\fun$ as a perturbation of a function $\funp$ where $\funp$ is defined by the same series (\ref{funX})-(\ref{anX})
where $h(X)$ is replaced by $h^+(X)$ and $C(h(X),k)$ is replaced by a new function $C^+(h^+(X),k)$ which we define now. Recall from (\ref{CxkBetak}) that
\begin{equation}
C(x,k) = \beta(k)  - \int_x^{+\infty} [c(t,k)-\Omega_+ (k)] \ dt + \Omega_+ (k)x. \label{Cxk}
\end{equation}
So, replacing $c(t,k)$ by its equivalent at $+\infty$, it is natural set
\begin{equation}\label{Cplus}
C^+(h^+ (X),k) = \beta(k) + \Omega_+ (k) h^+(X).
\end{equation}
The function $\funp$ is thus defined by
\begin{equation} \label{funplus}
\funp = \sum_{n=0}^{+\infty} a_n^+ (X,\lambda,k) \ z^{2n},
\end{equation}
where
$$
a_0^+ (X,\lambda,k) = e^{-i\beta(k)}\ e^{i(\lambda-\Omega_+ (k)) h^+(X)},
$$
and for $n \geq 1$
$$
a_n^+ (X,\lambda,k) = e^{i(\lambda-\Omega_+ (k)) h^+(X)} \ \int_X^A \int_Y^A e^{-2i(\lambda -\Omega_+ (k)) h^+(Y)}
\ e^{i(\lambda-\Omega_+ (k)) h^+(T)} \ a_{n-1}^+ (T, \lambda,k) \ dT \ dY.
$$

\begin{remark}
It is clear by definition that the function $\funp$ satisfies the following singular Sturm-Liouville equation
\begin{equation} \label{SL1p}
  y'' + Q_+(X,k)y = z^2 y \ \Longleftrightarrow \ y'' + \frac{\omega_+}{(A-X)^2} y = z^2 y.
\end{equation}
Exact solutions of (\ref{SL1p}) are well known and called modified Bessel functions. We refer for instance to \cite{Le} for a description of these functions. In the following, we prove that $\fun$ is a \emph{small} perturbation of $\funp$. In consequence, the asymptotics of $\fun$ when $z \to +\infty$ will be read from the corresponding well known asymptotics for $\funp$.
\end{remark}

Thanks to our choice of diffeomorphism $h^+$, the coefficients of the serie (\ref{funplus}) can be explicitely calculated. Precisely, denoting by $\Gamma$ the well-known Euler Gamma function, we have

\begin{lemma} \label{Jostplus}
For $X\in ]0,A[$, $z \in \C$ and for all $n \geq 0$
$$
a_n^+ (X,\lambda,k) =  e^{-i\beta(k)} \ (-\frac{\kappa_+}{a_+})^{i\frac{\lambda^+(k)}{\kappa_+}} \ \Gamma(1-\nu_+(k))\
\frac{1}{2^{2n}\ n! \ \Gamma(n+1-\nu_+(k) )} \ (A-X)^{2n+i\frac{\lambda^+(k)}{\kappa_+}}
$$
with
\begin{equation} \label{nuplus}
\nu_+ (k) = \frac{1}{2} - i \  \frac{\lambda^+(k)}{\kappa_+}.
\end{equation}
\end{lemma}

\begin{proof}
We prove the formula by induction. For $n=0$, the result is clear by (\ref{newdiffeo}), (\ref{cstplus}) and (\ref{Cplus}).
For $n\geq 1$, an elementary calculation gives
\begin{equation} \label{formedev}
a_n^+ (X,\lambda,k) = (-\frac{\kappa_+}{a_+})^{i\frac{\lambda^+(k)}{\kappa_+}} \ \frac{1}{(1+\frac{2i\lambda^+(k)}{\kappa_+}) \cdots
(2n-1+\frac{2i\lambda^+(k)}{\kappa_+}) \ 2 \cdots(2n)} \ (A-X)^{2n+i\frac{\lambda^+(k)}{\kappa_+}}
\end{equation}
Using the functional equality $\Gamma(z+1) = z \Gamma(z)$, Lemma \ref{Jostplus} is proved.
\end{proof}

Now, it turns out that the expressions for the coefficients $a_n^+(X,\lambda,k)$ can be written in terms of the modified Bessel function $I_{-\nu}(x)$. Let us recall its definition (\cite{Le}, Eq. $(5.7.1)$, p. $108$),
\begin{equation} \label{Besselmod}
I_{-\nu}(x) =\sum_{n=0}^{+\infty}\  \frac{1}{\Gamma(n-\nu +1)\  n!}\ \left(\frac{x}{2}\right)^{-\nu+2n}, \ \ x \in \C, \ \mid Arg\ x \mid < \pi.
\end{equation}
As in \cite{DN3}, Corollary 4.5, we have

\begin{coro} \label{asjostp} \hfill
\begin{enumerate}
\item For $X \in ]0,A[$ and $z \in \C$,
$$ 
\funp = e^{-i\beta(k)} \left( -\frac{\kappa_+}{a_+} \right)^{i\frac{\lambda^+(k)}{\kappa_+}} \Gamma(1-\nu_+ (k)) \sqrt{A-X} \ \left( \frac{z}{2} \right)^{\nu_+ (k)} I_{-\nu_+ (k)} (z(A-X)).
$$
\item Let $X_1 \in ]0,A[$ fixed. Then, for $j=0,1$ and uniformly in $X \in ]X_1, A[$, the following asymptotics hold when $z \rightarrow + \infty$, $z$ real
$$  
 f_1^{+(j)}(X,\lambda,k, z)= (-1)^j\ \frac{2^{-\nu_+(k)}}{\sqrt{2\pi}}\ (-\frac{\kappa_+}{a_+})^{i\frac{\lambda^+(k)}{\kappa_+}} \ \Gamma(1-\nu_+(k))
\ z^{j-i\frac{\lambda^+(k)}{\kappa_+}} \ e^{z(A-X)} \ \Big(1+O(\frac{1}{z})\Big).
$$
\end{enumerate}
\end{coro}

Now, we prove some useful properties for $a_n(X, \lambda,k)$ and  $a_n^+(X, \lambda,k)$ with $k$ fixed.

\begin{lemma} \label{Estimations}\hfill
\begin{enumerate}
\item For $n \geq 0$ and $X \in]0,A[$, we have
\begin{equation}\label{estimAn}
\mid a_n (X,\lambda,k) \mid \leq \frac {(A-X)^{2n}}{(2n)!},
\end{equation}
\item Let $X_0 \in (0,A)$ be fixed. For $n \geq 0$ and for all $X \in]X_0,A[$,
\begin{equation}\label{estimAn'}
a_n' (X,\lambda,k) = O\left((A-X)^{2n-1}\right).
\end{equation}
\end{enumerate}
The same estimates hold for $a_n^+(X,\lambda,k)$.
\end{lemma}

\begin{proof}
The first point is clear by a simple induction. Let us prove the second one. For $n=0$, we observe that
\begin{equation}
a_0 '(X,\lambda,k) = i (\lambda-c(h(X),k))h'(X)\  a_0 (X,\lambda, k).
\end{equation}
Using $c(h(X),k)=O(1)$ and the estimate
\begin{equation} \label{h'(X)}
 {\displaystyle{h'(X) = O\left( \frac{1}{A-X}\right)}}, \quad \forall X \in ]X_0,A[,
\end{equation}
we obtain $ a_0 '(X,\lambda,k)=O((A-X)^{-1})$. For $n\geq1$, we have
\begin{eqnarray}
a_n' (X,\lambda,k) &=& i(\lambda - c(h(X),k)) h'(X)\ a_n (X\lambda,k) \nonumber \\
                   & &  - e^{iC(h(X),k)}\ e^{-i\lambda h(X)}\ \int_X^A e^{i\lambda h(T)}\ e^{-iC(h(T),k)}\ a_{n-1} (T,\lambda,k) \ dT.
\end{eqnarray}
Thus, the result follows from (\ref{estimAn}) and (\ref{h'(X)}) by induction. For $a_n^+(X,\lambda,k)$,
the proof is identical.
\end{proof}

Now, we want to control the difference $\fun- \funp$. To do this, we set
\begin{equation}\label{error}
e_n(X,\lambda,k) =a_n(X, \lambda,k) - a_n^+(X, \lambda,k),
\end{equation}
and thus, we have
\begin{equation}\label{difference}
\fun- \funp = \sum_{n=0}^{+\infty} \ e_n(X,\lambda,k) \  z^{2n}.
\end{equation}
In the next lemma, we show that $a_n(X,\lambda,k), \ a_n^+(X,\lambda,k)$ and $e_n(X,\lambda,k)$ satisfy second order differential equations.
\begin{lemma}\label{eqdiffs} \hfill
\begin{enumerate}
\item For $n\geq 1$, $a_n(X,\lambda,k)$ and $a_n^+(X,\lambda,k)$ satisfy on $]0,A[$,
\begin{eqnarray}
a_n''(X, \lambda,k) + Q(X,k) \ a_n (X, \lambda,k) & = & a_{n-1}(X, \lambda,k), \label{eqdiffAn}\\
{a_n^+}''(X, \lambda,k) + Q_+(X,k) \ a_n^+ (X, \lambda,k) &=& a_{n-1}^+(X, \lambda,k), \label{eqdiffAn+}
\end{eqnarray}
\item For $n\geq 1$, $e_n(X,\lambda,k)$ satisfies on $]0,A[$,
\begin{equation}\label{eqdiffen}
{e_n}''(X, \lambda,k) + Q_+ (X,k)   \ e_n (X, \lambda,k) =
e_{n-1}(X, \lambda,k) - q_+(X,k) \ a_n (X, \lambda,k)
\end{equation}
where
\begin{equation}\label{qzero}
q_+(X,k) = Q(X,k)- Q_+(X,k).
\end{equation}
\end{enumerate}
\end{lemma}

\begin{proof}
Since ${\displaystyle{\fun= \sum_{n=0}^{+\infty} \ a_n(X,\lambda,k) \  z^{2n}}}$, (\ref{eqdiffAn}) follows directly
from (\ref{SL1}). The proof of (\ref{eqdiffAn+}) is identical. At last, (\ref{eqdiffen}) is a direct consequence
of (\ref{eqdiffAn}) and (\ref{eqdiffAn+}).
\end{proof}

Now we show that the equation (\ref{eqdiffen}) can be rewritten as an integral equation which will be useful to estimate the error term $e_n$.

\begin{lemma}\label{equationerreur}
For $n\geq 1$ and for all $X \in ]X_0,A[$, $e_n(X,\lambda,k)$ satisfies  the integral equation
\begin{eqnarray}\label{eqintegrale}
e_n (X,\lambda,k) &=& e^{i\lambda^+(k) h_+ (X)} \ \int_X^A \int_Y^A e^{-2i\lambda^+(k) h_+ (Y)}\  e^{i\lambda^+(k) h_+ (T)} \nonumber \\
                  & & \hspace{2 cm} [ e_{n-1}(T,\lambda,k) -q_+ (T,k)\ a_n (T,\lambda,k) ] \  dT\ dY.
\end{eqnarray}

\end{lemma}

\begin{proof}
We denote by $f_n (X,\lambda,k)$ the (R.H.S) of (\ref{eqintegrale}). Using (\ref{eqdiffen}), we have :
$$
f_n (X,\lambda,k) = e^{i\lambda^+(k) h_+ (X)} \ \int_X^A \int_Y^A e^{-2i\lambda^+(k) h_+ (Y)}\  e^{i\lambda^+(k) h_+ (T)} \ {e_n}''(T, \lambda,k)\  dT\ dY
\hspace{3cm}
$$
$$
  + \ e^{i\lambda^+(k) h_+ (X)} \ \int_X^A \int_Y^A e^{-2i\lambda^+(k) h_+ (Y)}\  e^{i\lambda h_+ (T)} \
    Q_+ (T,k)   \ e_n (T, \lambda,k)     \  dT\ dY.
$$
Integrating by part twice the first integral and using Lemma \ref{Estimations} yield (\ref{eqintegrale}).
\end{proof}

In the next lemma, for $k$ fixed, we estimate $e_n(X,\lambda,k)$ and its derivative.

\begin{lemma}\label{estimationerreur}
Let $X_0 \in (0,A)$ be fixed. There exists a constant $C>0$ such that for all $n \geq 0$ and for all $X \in ]X_0, A[$,
\begin{equation}\label{estimen}
\mid e_n(X,\lambda,k) \mid \ \leq \ C \ (n+1) \ \frac {(A-X)^{2n+2}}{(2n+2)!}
\end{equation}
\begin{equation}\label{estimend}
\mid {e_n}'(X,\lambda,k) \mid \ \leq \ C \ (n+1)   \  \frac {(A-X)^{2n+1}}{(2n+1)!}
\end{equation}
\end{lemma}

\begin{proof}
We prove the lemma by induction. For $n=0$, we have
\begin{eqnarray}
e_0(X,\lambda,k) & = &  a_0(X,\lambda,k)- a_0^+(X,\lambda,k) \\
                 & = &  e^{-iC(h(X),k)} \ e^{i\lambda h(X)}- e^{-i \beta(k)}\ e^{i\lambda^+(k) h_+(X)}.
\end{eqnarray}
Using (\ref{Cxk}),
\begin{eqnarray}
C(h(X),k) & = &  \beta(k)  + \Omega_+(k) h(X)- \int_X^A [c(h(T),k)-\Omega_+(k)] \ \frac{1}{a(h(T))} \ dT, \\
          & = &  \beta(k) + \Omega_+(k) h(X) +O ((A-X)^2).
          \end{eqnarray}
It follows from Lemma \ref{estdiffeo} that, for a suitable constant $C>0$,
$$
\mid e_0(X,\lambda,k) \mid \leq \frac{C}{2} \ (A-X)^2, \quad \forall X \in (X_0,A).
$$
Assuming that the property is true for $n-1$, we have by  (\ref{estimAn}) and (\ref{eqintegrale}),
$$
\mid e_n(X,\lambda,k)\mid \leq \int_X^A \int_Y^A \left( C\ n \ \frac{(A-T)^{2n}}{(2n)!} + C \ \frac{(A-T)^{2n}}{(2n)!} \right) \ dT\ dY,
\quad \forall X \in (X_0,A),
$$
where we have supposed that ${\displaystyle{C \ \geq \ \mid\mid q_+(.,k) \mid\mid_{L^\infty (X_0,A)}}}$. So,
$$
\mid e_n(X,\lambda,k) \mid \leq C \ (n+1) \frac {(A-X)^{2n+2}}{(2n+2)!}, \quad \forall X \in (X_0,A).
$$
We prove (\ref{estimend}) similarly.
\end{proof}

Now, we can establish the main result of this section.

\begin{prop}\label{asymptoticfun}\hfill
\begin{enumerate}
\item There exists $C>0$ such that for $j=0,1$, for all $X \in ]X_0,A[$ and all $z >0$,
\begin{equation}\label{differenceAsymp}
\mid f_1^{(j)}(X,\lambda,k, z) - f_1^{+ (j)} (X,\lambda,k, z) \mid \ \leq \ C \ (A-X) \ z^{j-1} \ e^{z(A-X)}.
\end{equation}
\item For fixed $X_0<X_1$ with $X_j \in ]0, A[$, $\forall j=0,1$, $\forall X\in]X_0,X_1[$, we have the following asymptotics,
when $z \rightarrow +\infty$,
\begin{equation} \label{asymflun}
 f_1^{(j)}(X,\lambda,k, z)= (-1)^j \frac{2^{-\nu_+(k)}}{\sqrt{2\pi}} \left( -\frac{\kappa_+}{a_+} \right)^{i\frac{\lambda^+(k)}{\kappa_+}}  \Gamma(1-\nu_+(k)) z^{j-i\frac{\lambda^+(k)}{\kappa_+}}  e^{z(A-X)}  \Big(1+O(\frac{1}{z})\Big).
\end{equation}
\end{enumerate}
\end{prop}

\begin{proof}
The proof is identical to \cite{DN3}, Proposition 4.12, but for the convenience of the reader, we repeat the argument. We only prove (\ref{differenceAsymp}) in the case $j=0$ since the case $j=1$ is similar. By  Lemma \ref{estimationerreur} and for $z>0$, we have
\begin{eqnarray*}
\mid \fun -\funp \mid \ & \leq & \ \sum_{n=0}^{+\infty} \mid e_n (x,\lambda,k)\mid \ z^{2n}\\
                        & \leq & \ \sum_{n=0}^{+\infty} C\ (n+1) \ \frac {(A-X)^{2n+2}}{(2n+2)!}\ z^{2n}\\
                        & \leq & \ \frac{C}{z} \ \ (A-X) \ \sum_{n=0}^{+\infty}  \frac{(n+1)}{2n+2} \ \frac {(z(A-X))^{2n+1}}{(2n+1)!} \\
                        & \leq & \ \frac{C}{2z} \ (A-X) \ \sinh(z(A-X)) \\
                        & \leq & \ \frac{C}{z} \ (A-X) \  e^{z(A-X)}.
\end{eqnarray*}
Now, since $z$ is real, (\ref{asymflun}) follows from Corollary \ref{asjostp} and (\ref{differenceAsymp}).
\end{proof}

In order to calculate the asymptotics of $f_2(X,\lambda,k,z)$, we follow the same method as for $f_1(X,\lambda,k,z)$. We thus only give the final results omitting the details. First, we construct $f_2^+(X,\lambda,k,z)$ which approximates $f_2(X,\lambda,k, z)$ as in (\ref{differenceAsymp}). We obtain

\begin{equation}\label{jostplusdeux}
 f_2^+(X,\lambda,k,z) = -i\ (-\frac{\kappa_+}{a_+})^{-i\frac{\lambda^+(k)}{\kappa_+}} \ \Gamma(1-\mu_+(k)) \ \sqrt{A-X} \ (\frac{z}{2})^{\mu_+(k)} \
 I_{1-\mu_+(k)} (z(A-X)),
\end{equation}
where
\begin{equation}\label{muplus}
\mu_+(k) = \frac{1}{2} +i \ \frac{\lambda^+(k)}{\kappa_+}\ .
\end{equation}
Then, using the well-known asymptotics for the modified Bessel functions \cite{Le}, we deduce

\begin{prop}\label{asymptoticfdeux} For fixed $X_0<X_1$ with $X_j \in ]0, A[$, $\forall j=0,1$, $\forall X\in]X_0,X_1[$,
we have the following asymptotics,
when $z \rightarrow +\infty$,
\begin{equation} \label{asymfldeux}
 f_2^{(j)}(X,\lambda,k, z)= (-1)^{j+1}\ i\ \frac{2^{-\mu_+(k)}}{\sqrt{2\pi}} \left( -\frac{\kappa_+}{a_+} \right)^{-i\frac{\lambda^+(k)}{\kappa_+}} \ \Gamma(1-\mu_+(k))
\ z^{j+i\frac{\lambda^+(k)}{\kappa_+}} \ e^{z(A-X)} \ \Big(1+O(\frac{1}{z})\Big).
\end{equation}
\end{prop}

In order to obtain the asymptotics of the scattering data, we need to calculate the asymptotics of the Jost functions $g_j(X,\lambda,k,z)$. Since the procedure is the same as the one for the $f_j(X,\lambda,k,z)$, we give without proof the main steps to obtain the asymptotics of $g_j(X,\lambda,k,z), \ j=1,2$, when $z \rightarrow +\infty$. Since $g_j (X,\lambda,k,z)$ satisfies (\ref{SL1}) with a boundary condition at $X=0$, we work with an other diffeomorphism, denoted by $h_-(X)$, in order to construct the functions $g_j^- (X,\lambda,k,z)$ that approximate $g_j (X,\lambda,k,z)$. This new diffeomorphism is defined as follows
\begin{equation}\label{diffeomoins}
h_-(X) = \frac{1}{\kappa_-} \ \log X + C_-,
\end{equation}
where
\begin{equation}\label{constantemoins}
C_- =\frac{1}{\kappa_-}\ \log \left(\frac{\kappa_-}{a_-}\right).
\end{equation}
Also, recalling that
\begin{equation}
C(x,k)= \int_{-\infty}^x [c(t,k)-\Omega_- (k)] \ dt + \Omega_- (k)x,
\end{equation}
we replace $C(h(X),k)$ appearing in the definition of the modified Jost functions $g_j (X,\lambda,k,z)$ by $C^- (h^- (X),k)= \Omega_- (k)h^- (X) $.

As previously, we can calculate $g_j^- (X,\lambda,k,z)$ explicitely and we obtain the following equalities
$$
g_1^- (X,\lambda,k,z) =  \left( \frac{\kappa_-}{a_-} \right)^{i\frac{\lambda^-(k)}{\kappa_-}} \ \sqrt{X} \ \Gamma(1-\nu_-(k)) \ \left( \frac{z}{2} \right)^{\nu_-(k)} \ I_{-\nu_-(k) } (zX),
$$
where
\begin{eqnarray}\label{numoins}
\lambda^-(k) & = & \lambda-\Omega_- (k), \\
\nu_-(k)     & = & \frac{1}{2} -i \frac{\lambda^-(k)}{\kappa_-}.
\end{eqnarray}
Similarly, we have
$$
g_2^- (X,\lambda,k,z) = i\ \left( \frac{\kappa_-}{a_-} \right)^{-i\frac{\lambda^-(k)}{\kappa_-}} \ \sqrt{X} \
\Gamma(1-\mu_-(k)) \ \left( \frac{z}{2} \right)^{\mu_-(k)} \ I_{1-\mu_-(k) } (zX),
$$
where
\begin{equation}\label{mumoins}
\mu_-(k) = \frac{1}{2} +i \frac{\lambda^-(k)}{\kappa_-}.
\end{equation}

The $g_j^- (X,\lambda,k,z)$ are perturbations of the $g_j(X,\lambda,k,z)$. Precisely, we have for fixed $k$,

\begin{lemma}\label{gdeuxmoins}
For $X_1 \in ]0,A[$ fixed, there exists $C>0$ such that for $p=0,1$, for all $X \in ]0, X_1[$ and for all $z >0$,
\begin{equation}\label{diff}
\mid g_j^{(p)}(X,\lambda,k, z) - g_j^{- (p)} (X,\lambda,k, z) \mid \ \leq \ C \ X \ z^{p-1} \ e^{zX}, \quad j =1,2.
\end{equation}
\end{lemma}
\noindent Then, using the asymptotics of the modified Bessel functions \cite{Le}, we obtain

\begin{prop}\label{asymptoticgj}
For fixed $X_0<X_1$ with $X_j \in ]0, A[$, $\forall j=0,1$, $\forall X\in]X_0,X_1[$, we have the following asymptotics,
when $z \rightarrow +\infty$,
\begin{equation} \label{asymgun}
 g_1^{(j)}(X,\lambda,k, z)= \frac{2^{-\nu_-(k)}}{\sqrt{2\pi}}\ \left( \frac{\kappa_-}{a_-} \right)^{i\frac{\lambda^-(k)}{\kappa_-}} \ \Gamma(1-\nu_-(k))
\ z^{j-i\frac{\lambda^-(k)}{\kappa_-}} \ e^{zX} \ \Big(1+O(\frac{1}{z})\Big),
\end{equation}
\begin{equation} \label{asymgdeux}
 g_2^{(j)}(X,\lambda,k, z)= i \ \frac{2^{-\mu_-(k)}}{\sqrt{2\pi}}\ \left( \frac{\kappa_-}{a_-} \right)^{-i\frac{\lambda^-(k)}{\kappa_-}} \ \Gamma(1-\mu_-(k)) z^{j+i\frac{\lambda^-(k)}{\kappa_-}} \ e^{zX} \ \Big(1+O(\frac{1}{z})\Big).
\end{equation}
\end{prop}

\subsection{Asymptotics of the scattering data.}

In this section, we put together all the previous results and calculate the asymptotics of the scttering data $a_{Lj}(\lambda,k,z), \ j=1,...,4$ when $z \rightarrow +\infty$.

First, we recall that for $j = 1,2,3,4$
\begin{eqnarray}\label{relation}
f_j(X,\lambda,k,z ) & = &   e^{-iC(h(X),k)} \  f_{Lj}(X,\lambda,k,z ),  \label{relation1}\\
g_j(X,\lambda,k,z ) & = &   e^{-iC(h(X),k)} \  f_{Rj}(X,\lambda,k,z ).  \label{relation2}
\end{eqnarray}
Second, we recall that for all $x \in \R$,
\begin{equation}\label{rappel}
F_L (x,\lambda, z) = F_R (x,\lambda, z)\ A_L(\lambda, z).
\end{equation}
Calculating (\ref{rappel}) components by components, it follows from (\ref{relation1}) and (\ref{relation2}) that (in the variable $X$)
$$
f_1(X,\lambda,k,z) = \alun \ g_1(X,\lambda,k,z) + \alt \ g_2(X,\lambda,k,z)
$$
$$
f_2(X,\lambda,k,z) = \ald \ g_1(X,\lambda,k,z) + \alq \ g_2(X,\lambda,k,z)
$$
So, by Lemma \ref{wronskien}, we obtain for $z\not=0$ :
\begin{eqnarray*}
\alun &=& \frac{1}{iz} \ W(f_1,g_2)\ ,\quad \ald = \frac{1}{iz} \ W(f_2,g_2),\\
\alt &=& - \frac{1}{iz} \ W(f_1,g_1)\ ,\quad \alq = - \frac{1}{iz} \ W(f_2,g_1).
\end{eqnarray*}
The following theorem is then an easy consequence of Propositions \ref{asymptoticfun}, \ref{asymptoticfdeux}, \ref{asymptoticgj} and of the preceding formulae

\begin{theorem}\label{asymptoticsal}
When $z \rightarrow + \infty$, we have
\begin{eqnarray*}
\alun & \sim & \ \frac{e^{-i\beta(k)}}{2\pi}\ \left(-\frac{\kappa_+}{a_+}\right)^{i\frac{\lambda^+(k)}{\kappa_+}}
\left(\frac{\kappa_-}{a_-}\right)^{-i\frac{\lambda^-(k)}{\kappa_-}} \nonumber \\
      &     & \Gamma\left(\frac{1}{2}-i\frac{\lambda^-(k)}{\kappa_-}\right) \Gamma\left(\frac{1}{2}+i\frac{\lambda^+(k)}{\kappa_+}\right)
      \left(\frac{z}{2}\right)^{i\left(\frac{\lambda^-(k)}{\kappa_-} - \frac{\lambda^+(k)}{\kappa_+}\right)}  e^{zA}, \\
\ald & \sim & \ -i\ \frac{e^{i\beta(k)}}{2\pi}\ \left(-\frac{\kappa_+}{a_+}\right)^{-i\frac{\lambda^+(k)}{\kappa_+}}
\left(\frac{\kappa_-}{a_-}\right)^{-i\frac{\lambda^-(k)}{\kappa_-}} \nonumber \\
      &     & \Gamma\left(\frac{1}{2}-i\frac{\lambda^-(k)}{\kappa_-}\right) \Gamma\left(\frac{1}{2}-i\frac{\lambda^+(k)}{\kappa_+}\right)
      \left(\frac{z}{2}\right)^{i\left(\frac{\lambda^-(k)}{\kappa_-} + \frac{\lambda^+(k)}{\kappa_+}\right)}  e^{zA} \\
\alt & \sim & \ i\ \frac{e^{-i\beta(k)}}{2\pi}\ \left(-\frac{\kappa_+}{a_+}\right)^{i\frac{\lambda^+(k)}{\kappa_+}}
\left(\frac{\kappa_-}{a_-}\right)^{i\frac{\lambda^-(k)}{\kappa_-}} \nonumber \\
      &     & \Gamma\left(\frac{1}{2}+i\frac{\lambda^-(k)}{\kappa_-}\right) \Gamma\left(\frac{1}{2}+i\frac{\lambda^+(k)}{\kappa_+}\right)
      \left(\frac{z}{2}\right)^{-i\left(\frac{\lambda^-(k)}{\kappa_-} + \frac{\lambda^+(k)}{\kappa_+}\right)}  e^{zA} \\
\alq & \sim & \ \frac{e^{i\beta(k)}}{2\pi}\ \left(-\frac{\kappa_+}{a_+}\right)^{-i\frac{\lambda^+(k)}{\kappa_+}}
\left(\frac{\kappa_-}{a_-}\right)^{i\frac{\lambda^-(k)}{\kappa_-}} \nonumber \\
      &     & \Gamma\left(\frac{1}{2}+i\frac{\lambda^-(k)}{\kappa_-}\right) \Gamma\left(\frac{1}{2}-i\frac{\lambda^+(k)}{\kappa_+}\right)
      \left(\frac{z}{2}\right)^{-i\left(\frac{\lambda^-(k)}{\kappa_-} - \frac{\lambda^+(k)}{\kappa_+}\right)}  e^{zA}
\end{eqnarray*}
\end{theorem}

From the asymptotics of the $a_{Lj}$s' when $z \to +\infty$, we can also obtain the asymptotics of the simplified scattering coefficients $\hat{T}(\lambda,k,z), \ \hat{R}(\lambda,k,z), \hat{L}(\lambda,k,z)$ using (\ref{SR-SM2}). Further, we finally obtain the asymptotics of the true scattering coefficients using (\ref{Link-SM}) and given by

\begin{theorem} \label{asymptoticssc}
When $z \rightarrow + \infty$, we have
\begin{eqnarray*}
T(\lambda,k,z) & \sim & \ \frac{2 \pi  \left(-\frac{\kappa_+}{a_+}\right)^{-i\frac{\lambda^+(k)}{\kappa_+}}
\left(\frac{\kappa_-}{a_-}\right)^{i\frac{\lambda^-(k)}{\kappa_-}}}{ \Gamma\left(\frac{1}{2}-i\frac{\lambda^-(k)}{\kappa_-}\right) \Gamma\left(\frac{1}{2}+i\frac{\lambda^+(k)}{\kappa_+}\right)}
      \left(\frac{z}{2}\right)^{-i\left(\frac{\lambda^-(k)}{\kappa_-} - \frac{\lambda^+(k)}{\kappa_+}\right)}  e^{-zA}, \\
R(\lambda,k,z) & \sim & \ i \  \left(-\frac{\kappa_+}{a_+}\right)^{-2i\frac{\lambda^+(k)}{\kappa_+}} \frac{ \Gamma\left(\frac{1}{2}-i\frac{\lambda^+(k)}{\kappa_+}\right)}{\Gamma\left(\frac{1}{2}+i\frac{\lambda^+(k)}{\kappa_+}\right)}
      \left(\frac{z}{2}\right)^{2i \frac{\lambda^+(k)}{\kappa_+}}, \\
L(\lambda,k,z) & \sim & \ i \  \left(-\frac{\kappa_-}{a_-}\right)^{2i\frac{\lambda^-(k)}{\kappa_-}} \frac{ \Gamma\left(\frac{1}{2}+i\frac{\lambda^-(k)}{\kappa_-}\right)}{\Gamma\left(\frac{1}{2}-i\frac{\lambda^-(k)}{\kappa_-}\right)}
      \left(\frac{z}{2}\right)^{-2i \frac{\lambda^-(k)}{\kappa_-}}.
\end{eqnarray*}
\end{theorem}

\begin{remark}
  Thanks to Theorems \ref{Main-Frobenius} and \ref{asymptoticssc}, the following parameters
  $$
    A = \int_\R a(x)  dx, \quad \frac{\lambda^-(k)}{\kappa_-} = \frac{\lambda - \Omega_-(k)}{\kappa_-} = \frac{\lambda}{\kappa_-} - \frac{aEk}{(r_-^2 + a^2)\kappa_-} - \frac{qQr_-}{(r_-^2 + a^2) \kappa_-},
  $$
  $$
     \frac{\lambda^+(k)}{\kappa_+} = \frac{\lambda - \Omega_+(k)}{\kappa_+} = \frac{\lambda}{\kappa_+} - \frac{aEk}{(r_+^2 + a^2)\kappa_+} - \frac{qQr_+}{(r_+^2 + a^2) \kappa_+},
  $$
  are uniquely determined from the knowledge of the scattering operators $T_k(\lambda), \ R_k(\lambda), \ L_k(\lambda)$. If we suppose that these scattering operators are known for two different $k$ and $\lambda$, we can thus determine uniquely the following parameters
  $$
    A, \quad \kappa_\pm, \quad r_\pm,
  $$
  which clearly have an important physical meaning.
\end{remark}

\subsection{Derivatives of the scattering data with respect to $z$.} \label{AL-Strictly-Increasing}

The previous theorem shows that for all $j=1,...,4$, $\exists C_j >0$ such that $|a_{Lj}(\lambda, k, z)|^2 \sim C_j \ e^{2zA}$ when $z \rightarrow +\infty$. It is reasonable to think that $\frac{d}{dz} (|a_{Lj}(\lambda, k, z)|^2) \sim 2A C_j \ e^{2zA}$. This would be true if we could take the derivatives with respect to z of the asymptotics of $a_{Lj}(\lambda,k,z)$ in Theorem \ref{asymptoticsal}. In this Section, we prove that this is indeed the case. In consequence, it follows that $z \rightarrow |a_{Lj}(\lambda, k, z)|$ is a strictly increasing function for large real $z$. We emphasize that this result is one of the crucial ingredient in the proof of Proposition \ref{Consequence1} and more generally in the proof of our inverse problem.

In what follows, we shall use the notation $\dot{f} = \frac{d}{dz} f$ and prove the Lemma

\begin{lemma} \label{AL1-Increasing}
There exists a real $z_0$ large enough such that, for $z \geq z_0, \  z \rightarrow \mid a_{Lj}(\lambda, k, z)\mid$ is a strictly increasing function.
\end{lemma}

\begin{proof}
For instance, let us prove the lemma in the case $j=1$. Clearly, it suffices to show  that
\begin{equation}\label{derivpositive}
\frac{d}{dz} \left( |a_{L1}(\lambda, k, z)|^2 \right)= 2 Re \  (\dot{a}_{L1}(\lambda,k,z) \overline{a_{L1}(\lambda,k,z)})>0 \ \ ,\ \ z \gg 1.
\end{equation}
By (\ref{rappel}), one has :
\begin{equation}
a_{L1}(\lambda, k, z) = f_1(X,\lambda,k,z)g_4(X,\lambda,k,z)-f_3(X,\lambda,k,z)g_2(X,\lambda,k,z).
\end{equation}
So,
\begin{eqnarray}
\dot{a}_{L1}(\lambda, k, z) &=& \dot{f}_1(X,\lambda,k,z)g_4(X,\lambda,k,z)+ f_1(X,\lambda,k,z)\dot{g} _4(X,\lambda,k,z) \nonumber \\
&- & \dot{f}_3(X,\lambda,k,z)g_2(X,\lambda,k,z)- f_3(X,\lambda,k,z)\dot{g} _2(X,\lambda,k,z).
\end{eqnarray}
Recalling that
\begin{eqnarray}
g_4(X,\lambda,k,z) &=& \overline{g_1(X,\lambda,k,z)}, \\
f_3(X,\lambda,k,z) &=& \overline{f_2(X,\lambda,k,z)},
\end{eqnarray}
we only have to study the asymptotics of $\dot{f}_j(X,\lambda,k,z)$ and $\dot{g}_j(X,\lambda,k,z)$ for $j=1,2$ when  $z \to +\infty$. For instance, let us study $\dot{f}_1(X,\lambda,k,z)$ since the other cases are similar. First, (\ref{difference}) implies
\begin{equation}\label{derivdifference}
\dot{f}_1(X,\lambda,k,z) - \dot{f}_1^+(X,\lambda,k,z) = \sum_{n=1}^{+\infty} \ 2n \ e_n(X,\lambda,k) \  z^{2n-1}.
\end{equation}
Using (\ref{estimen}), we easily obtain in the same way  as in Proposition \ref{asymptoticfun},  the following estimate for $X \in ]X_0,A[$,
\begin{equation}\label{estimederiv}
\mid  \dot{f}_1(X,\lambda,k,z) - \dot{f}_1^+(X,\lambda,k,z) \mid \leq  \frac {C (A-X)^2}{z} \ e^{z(A-X)}.
\end{equation}
Now, let us  recall that
\begin{equation}
\funp = e^{-i\beta(k)} \ (-\frac{\kappa_+}{a_+})^{i\frac{\lambda^+(k)}{\kappa_+}} \ \Gamma(1-\nu_+ (k)) \ \sqrt{A-X} \ (\frac{z}{2})^{\nu_+ (k)}
\ I_{-\nu_+ (k)} (z(A-X)).
\end{equation}
So, using (see \cite{Le}, Eq. (5.7.9)), the well-known asymptotics for the modified Bessel functions as well as the relation
\begin{equation}
\frac{d}{dw} \left( w^{\nu} I_{-\nu}(w) \right) = w^{\nu} I_{-\nu+1}(w),
\end{equation}
and finally (\ref{estimederiv}), we obtain in a similar way as in (\ref{asymflun}),
\begin{equation} \label{derivasymflun}
 \dot{f}_1(X,\lambda,k, z)= \frac{2^{-\nu_+(k)}}{\sqrt{2\pi}}\ (-\frac{\kappa_+}{a_+})^{i\frac{\lambda^+(k)}{\kappa_+}} \ \Gamma(1-\nu_+(k))
\ z^{-i\frac{\lambda^+(k)}{\kappa_+}} \ (A-X) e^{z(A-X)} \ \Big(1+O(\frac{1}{z})\Big).
\end{equation}
This means that $\dot{f}_1(X,\lambda,k, z)$ has precisely the expected asymptotics, that is the asymptotics we would obtain taking the derivative (with respect to $z$) of the asymptotics (\ref{asymflun}). Similarly, the asymptotics of $\dot{f}_2(X,\lambda,k, z)$, $\dot{g}_1(X,\lambda,k, z)$ and $\dot{g}_2(X,\lambda,k, z)$ are simply the derivatives of the asymptotics (\ref{asymfldeux}), (\ref{asymgun}) and (\ref{asymgdeux}). Thus (\ref{derivpositive}) follows easily.

\end{proof}


\Section{The inverse scattering problem.} \label{Inverse}

In this section, we prove Theorem \ref{Uniqueness-Fixed} and thus finish the proof of our main result Theorem \ref{Main}. Precisely, we prove the uniqueness of the mass $M$, the charge $Q$, the cosmological constant $\Lambda$ and the angular momentum per unit mass $a$ of a KN-dS black hole from the knowledge of either the transmission coefficient $T_{kl}(\lambda)$, or one of the reflection coefficients $L_{kl}(\lambda)$ and $R_{kl}(\lambda)$, at a {\it{fixed energy}} $\lambda \in \R$, for {\it{two distinct}} $k \in \frac{1}{2} + \Z$ and for all $l \in \mathcal{L}_k \subset \N^*$ satisfying a M\"untz condition
$$
  \sum_{l \in \mathcal{L}_k} \frac{1}{l} = \infty.
$$
In fact, we recover more than only $4$ parameters as explained in Theorem \ref{Uniqueness-Fixed}. We recover some scalar functions (up to some diffeomorphisms) that appear in the potentials of the separated radial equation (see below).

\begin{proof}[Proof of Theorem \ref{Uniqueness-Fixed}]
We recall that we have yet recovered the two parameters $a$ and $\Lambda$ from the study of the angular separated equation and the Frobenius's method (see Section \ref{Frobenius} and Theorem \ref{Main-Frobenius}). Thus we only have to recover the remaining parameters $Q$ and $M$ from the scattering data. As usual let us consider two KN-dS black holes with parameters $(M,Q,\Lambda, a)$ and $(\tilde{M},\tilde{Q},\Lambda, a)$ respectively. We shall use the notation $\tilde{Z}$ for all the data associated with the parameters $(\tilde{M},\tilde{Q},\Lambda, a)$.

We also recall that we want to include the possiblity of describing the same KN-dS black hole by two different RW variables and make our result coordinates independent (see Section \ref{DependenceRW}). Hence, we assume that there exists $c\in \R$ such that for two distinct $k \in \frac{1}{2} + \Z$, one of the following equalities holds for all $l \in \mathcal{L}_k \subset \N^*$

\begin{equation}\label{egalitecoefS}
\left\{ \begin{array}{ccc}
T_{kl}(\lambda) & = & e^{i c (\Omega_+ (k)-\Omega_-(k))} {\tilde{T}}_{kl}(\lambda), \\
R_{kl}(\lambda) &=& e^{2ic\lambda^+(k)}\ {\tilde{R}}_{kl}(\lambda), \\
L_{kl}(\lambda) &=& e^{-2ic\lambda^-(k)}\ {\tilde{L}}_{kl}(\lambda).
\end{array} \right.
\end{equation}
As shown in Proposition \ref{RW-SM-Red}, these equalities are equivalent to (in terms of the simplified scattering coefficients)
\begin{equation}\label{egalitecoefSAK}
\left\{ \begin{array}{ccc}
\hat{T}_{kl}(\lambda) & = &  {\hat{\tilde{T}}}_{kl}(\lambda), \\
\hat{R}_{kl}(\lambda) & = & e^{2i\lambda^+(k) c} \ {\hat{\tilde{R}}}_{kl}(\lambda), \\
\hat{L}_{kl}(\lambda) & = & e^{-2i \lambda^-(k) c} \ {\hat{\tilde{L}}}_{kl}(\lambda).
\end{array} \right.
\end{equation}
Recall moreover from Theorem \ref{Main-Frobenius} that $\muk = \tilde{\mu}_{kl}(\lambda), \ \forall l \in \mathcal{L}$. Hence, in any case, using one of the Propositions \ref{Uniqueness} or \ref{Uniqueness2}, we deduce from (\ref{egalitecoefSAK}) that there exists a suitable constant $d(\lambda,k) \in \R$ such that for  all $z \in \C$
\begin{eqnarray}
  a_{L1}(\lambda,k, z) = \tilde{a}_{L1} (\lambda,k, z) \ ,\quad a_{L2}(\lambda,k, z) = e^{2i d(\lambda,k)} \tilde{a}_{L2}(\lambda,k, z), \label{egalite0} \\
a_{L3}(\lambda,k, z) = e^{-2id(\lambda,k)} \ \tilde{a}_{L3} (\lambda,k, z) \ , \quad
a_{L4}(\lambda,k, z) = \tilde{a}_{L4}(\lambda,k, z). \label{egalite}
\end{eqnarray}
Hence, we first deduce from the asymptotics of Theorem \ref{asymptoticsal} that :
\begin{equation}\label{uniciteA}
A := \int_{-\infty}^{+\infty} \ a (x) \ dx \ =\ \int_{-\infty}^{+\infty} \ \tilde{a} (x) \ dx \  = \tilde{A}.
\end{equation}

As in \cite{DN3}, Section 5, (see also \cite{FY}), we can thus define the diffeomorphisms $h, \ \tilde{h} : \ ]0,A[ \rightarrow \R$ as the inverse diffeomorphisms of the Liouville transforms $g$ and $\tilde{g}$ given by (\ref{Liouville}) in which we use the potentials $a(x)$ and $\tilde{a}(x)$ respectively. We emphasize that these diffeomorphisms act on the same interval $]0,A[$ by (\ref{uniciteA}). Let us now introduce for $X \in ]0,A[$ the matrix
$$
  P(X,\lambda,k, z) = \left( \begin{array}{cc} P_1(X,\lambda,k, z) & P_2(X,\lambda,k, z) \\
                                            P_3(X,\lambda,k,z) & P_4(X,\lambda,k, z)
                                            \end{array} \right),
$$
defined by
\begin{equation}\label{matricepassage}
P(X,\lambda,k,z) \ \tilde{H}_R (\tilde{h} (X), \lambda,k, z) \ = \ H_R (h(X), \lambda,k, z)\ e^{i d(\lambda,k) \Gamma^1},
\end{equation}
where $H_R$ and $\tilde{H}_R$ are the modified Jost solutions from the right associated with $a(x)$ and $\tilde{a}(x)$ defined by
(\ref{modJostR}). As in the previous Section, for $k=1, ...,4$, we set
\begin{eqnarray}
g_k(X,\lambda,k,z) &=& h_{Rk}(h(X),\lambda,k,z), \\
\tilde{g}_k(X,\lambda,k,z) &=& \tilde{h}_{Rk}(\tilde{h}(X),\lambda,k,z).
\end{eqnarray}
Using (\ref{modJostL}) and (\ref{modJostR}) again, we remark that $H_L(x,\lambda,z)=H_R (x, \lambda, z)\ A_L (\lambda, z)$. So, we face exactly the same situation than the one exposed in \cite{DN3}, Section 5. Following line by line the proof of the uniqueness of the parameters exposed therein, we can show that the components $P_j (X,\lambda,k,z)$ of $P$ are independent of $z \in \C$. In particular, we obtain
\begin{equation}\label{thliouville}
P_j (X,\lambda,k,z)=P_j (X,\lambda,k,0), \quad \forall z \in \C, \quad \forall j =1,...,4.
\end{equation}
But it follows from the definition of the Jost functions that
$$
  H_R (x,\lambda,k, 0) = e^{-iC(x,k) \Gamma^1} \ e^{i \lambda \Gamma^1 x}, \quad \tilde{H}_R (x,\lambda,k, 0) = e^{-i\tilde{C}(x,k) \Gamma^1}\ e^{i  \lambda \Gamma^1 x}.
$$
Then, taking $z=0$ in (\ref{matricepassage}), we obtain
\begin{equation}\label{egalitemat}
P(X,\lambda,k, 0) = e^{i\  [\lambda \ ( h(X)-\tilde{h}(X))+ d(\lambda,k) -C(h(X),k) +\tilde{C}(\tilde{h}(X),k) \ ]\Gamma^1 }.
\end{equation}
Then, using  (\ref{thliouville}), (\ref{egalitemat}) and (\ref{matricepassage}) we get
\begin{equation}\label{egalitejost}
\left\{ \begin{array}{ccc}
\tilde{g}_{1}(X,\lambda,k, z) &=& e^{i \alpha(X,k)}  \ g_{1}(X,\lambda,k, z),   \\
\tilde{g}_{2}(X,\lambda,k, z) &=& e^{-2 i d(\lambda,k)}\ e^{i \alpha(X,k)} \ g_{2}(X,\lambda,k, z),
\end{array}
\right.
\end{equation}
where  $\alpha(X,k) = \lambda \ (\tilde{h}(X) - h (X))+C(h(X),k) -\tilde{C}(\tilde{h}(X),k)$.

By Lemma \ref{wronskien}, the wronskians $W(g_{1} , g_{2}) = W(\tilde{g}_{1}  , \tilde{g}_{2}) = iz$. Then, a straightforward calculation gives
\begin{equation}\label{egalitephase}
e^{2i (\alpha(X,k)- d(\lambda,k))} \ =\ 1.
\end{equation}
Thus, by a standard continuity argument, there exists $p_k \in \Z$ such that
\begin{equation}\label{diffeos}
\lambda (\tilde{h}(X)-h(X)) + C(h(X),k) -\tilde{C}(\tilde{h}(X),k) - d(\lambda,k) = p_k \pi \ \ ,\ \ \forall X \in ]0,A[,
\end{equation}
We differentiate (\ref{diffeos}) with respect to $X$ and we get
\begin{equation}\label{systemek}
\frac{\lambda}{\tilde{a}(\tilde{h}(X))} -\frac{\lambda}{a(h(X))} +\frac{c(h(X),k)}{a(h(X))}
- \frac{\tilde{c}(\tilde{h}(X),k)}{\tilde{a}(\tilde{h}(X))} =0.
\end{equation}
In other words, using the notations $x = h(X)$ and $\tilde{x} = \tilde{h}(X) = \tilde{h} \circ h^{-1}(x)$, we have
\begin{equation} \label{bip}
  \frac{\lambda - \tilde{c}(\tilde{x},k)}{\tilde{a}(\tilde{x})} = \frac{\lambda  - c(x,k)}{a(x)},
\end{equation}
which proves (\ref{Uniqueness-Function1}). Note in passing that if we know (\ref{bip}) for two different energies $\lambda \ne \lambda'$, then we obtain
\begin{equation}
  \frac{1}{\tilde{a}(\tilde{h}(X))} = \frac{1}{a(h(X))}.
\end{equation}
Since $\frac{d h(X)}{dX} = \frac{1}{a(h(X))}$, we thus see that there exists $\sigma \in \R$ such that $\tilde{h}(X) = h(X) + \sigma$ and also
$$
  \tilde{a}(x) = a(x - \sigma), \quad \tilde{c}(x,k) = c(x-\sigma,k).
$$

Now, let us from (\ref{systemek}) and the exact expressions of the potentials $a(x)$ and $c(x,k)$ recover the missing parameters $M$ and $Q$. We first recall that
$$
c(x,k) = \frac{aEk +qQr}{r^2+a^2},
$$
where $r$ stands for $r(x)$ the inverse of the Regge-Wheeler diffeomorphism and $E= 1+\frac{a^2\Lambda}{3}$. So, using (\ref{systemek}) for two different $k$, recalling that the angular momentum  $a$ and the cosmological constant $\Lambda$ are  unique, we obtain the following decoupled
equations :
\begin{eqnarray}\label{systemedec}
\frac{\lambda}{\tilde{a}(\tilde{h}(X))} -\frac{\lambda}{a(h(X))} + \frac{qQr}{(r^2+a^2)a(h(X))}
- \frac{q\tilde{Q}\tilde{r}}{(\tilde{r}^2+a^2)\tilde a(\tilde{h}(X))} & = &0, \label{systemedec1}\\
aE\left( \frac{1}{(r^2+a^2)a(h(X))}
- \frac{1}{(\tilde{r}^2+a^2)\tilde a(\tilde{h}(X))} \right) & =& 0, \label{systemedec2}
\end{eqnarray}
where we have denoted $\tilde{r} = \tilde{r}(\tilde{x})$. Recalling that
\begin{equation}\label{pot1}
a(x) = \frac{ \sqrt{\Delta_r}}{r^2+a^2},
\end{equation}
with
\begin{equation} \label{Deltar}
\Delta_r = (r^2+a^2)(1-\frac{\Lambda}{3}) -2Mr +Q^2,
\end{equation}
and using the fact that $a,E \not=0$, we obtain from (\ref{systemedec2}),
\begin{equation} \label{racine}
\sqrt{\Delta_r} = \sqrt{\tilde{\Delta}_{\tilde{r}}}\ .
\end{equation}
Using  ${\displaystyle{\frac{dX}{dx} = a(x)}}$, ${\displaystyle{\frac{dx}{dr}=\frac{r^2+a^2}{\Delta_r}}}$ and (\ref{pot1}), a straightforward calculation gives
\begin{equation}\label{deriveer}
\frac{d}{dX} \  [ r(h(X)) ] = \sqrt{\Delta_r}.
\end{equation}
So, (\ref{racine}) implies
\begin{equation}\label{uniciter}
  r(h(X)) =  r(h(X)) -b,
\end{equation}
for a suitable constant $b$. Then, using again (\ref{racine}), we obtain
\begin{equation}\label{unicitedelta}
  \Delta_r = \tilde{\Delta}_{\tilde{r}}= \tilde{\Delta}_{r+b}.
\end{equation}
Eventually, we use (\ref{Deltar}) and we identify  the terms in $r^3$ in the equality $ \Delta_r = \tilde{\Delta}_{r+b}$. We obtain easily
\begin{equation}
  - \frac{4\Lambda}{3} b =0.
\end{equation}
Hence, $b=0$ and $ \Delta_r = \tilde{\Delta}_{r}$ and this leads to the uniqueness of the parameters by (\ref{Deltar}).

To finish the proof of Theorem \ref{Uniqueness-Fixed}, let us consider separatly the case of Kerr-de-Sitter black holes. Assume thus that $Q = 0$. Recall then that
$$
  c(x,k) = \frac{aEk}{r^2+a^2} = c(x) k.
$$
From the equalities (\ref{systemek}) or equivalently (\ref{bip}) and playing with two different $k$, we determine uniquely and separatly
$$
  \frac{\lambda}{\tilde{a}(\tilde{h}(X))} = \frac{\lambda}{a(h(X))}, \quad \quad \frac{\tilde{c}(\tilde{h}(X))}{\tilde{a}(\tilde{h}(X))} = \frac{c(h(X))}{a(h(X))}.
$$
By the same argument as above, we conclude that there exists $\sigma \in \R$ such that
$$
  \tilde{h}(X) = h(X) +\sigma \ \Longleftrightarrow \ \tilde{x} = x + \sigma,
$$
and
$$
  \tilde{a}(x) = a(x - \sigma), \quad \quad \tilde{c}(x) = c(x-\sigma),
$$
which finishes the proof of the Theorem.

\end{proof}


\appendix 

\Section{Growth estimate of the eigenvalues $\muk$} \label{Estimate-mukl}

In this Appendix, we prove the following Proposition.

\begin{prop}[Estimate on $\muk$]
For all $\lambda \in \R$, for all $k \in \frac{1}{2} + \Z$ and for all $l \in \N^*$, there exist constants $C_1$ and $C_2$ independent of $k,l$ such that
$$
  \left( 2 - e^{\frac{1}{26}} \right) \left(|k| - \frac{1}{2} + l\right) - C_1 |k| - C_2 - |a \lambda| \leq \mu_{kl}(\lambda) \leq e^{\frac{1}{26}} \left(|k| - \frac{1}{2} + l\right) + C_1 |k| + C_2 + |a \lambda|.
$$
We thus conclude that for fixed $\lambda \in \R$ and $k \in \frac{1}{2} + \Z$,
$$
  \sum_{l \in \N^*} \frac{1}{\mu_{kl}(\lambda)} = + \infty.
$$
\end{prop}
\begin{proof}
  To do this, we shall apply the theory of analytic perturbation due to Kato \cite{Ka}. Recall first that for each $\lambda \in \R$ and $k \in 1/2 + \Z$, the $\mu_{kl}(\lambda)$ are eigenvalues of the selfadjoint operator $A_k(\lambda)$ given by (\ref{Ak}) on $\L := L^2((0,\pi), d\theta, \C^2)$. Introducing the notations
$$
  \zeta = \frac{a^2 \Lambda}{3}, \qquad \xi = a \lambda,
$$
the operator $A_k(\lambda)$ can be written as
\begin{equation} \label{Uk-1}
  A_k(\zeta,\xi) = \sqrt{1 + \zeta \cos^2\theta} \,\mathbb{D}_{\S^2}^k + i \Gb \frac{\zeta \sin(2\theta)}{4 \sqrt{1+\zeta \cos^2\theta}} + \Gc \frac{(\zeta k - \xi) \sin\theta}{\sqrt{1+\zeta\cos^2\theta}},
\end{equation}
where
\begin{equation} \label{Dk}
  \mathbb{D}_{\S^2}^k = \Gb D_\theta + \Gc\frac{k}{\sin\theta}.
\end{equation}
Note that according to (\ref{Cond-Lambda}), the physical parameters $\zeta$ and $\xi$ belong to $[0, 7 - 4 \sqrt{3}] \subset [0, \frac{1}{13,8}]$ and $\R$ respectively. In what follows, we shall allow these parameters to be complex, for instance in $B(0,\frac{1}{13}) \times S$ where $S$ is a narrow strip containing the real axis. Note at last that the operator
$$
  A_k(0,0) = \mathbb{D}_{\S^2}^k,
$$
is simply the restriction of the intrinsic Dirac operator on the $2$-sphere $\S^2$ to the angular mode $\{e^{ik\varphi}\}, \ k \in 1/2 + \Z$. Hence it is well known that $A_k(0,0)$ is selfadjoint on $\L$ with its natural domain given by
$$
  \D = \{ u \in \L, \ u \ \textrm{is absolutely continuous and} \ \mathbb{D}_{\S^2}^k u \in \L \ \textrm{and anti-periodic boundary conditions}\}.
$$

We introduce the notations
$$
  A_k(\zeta, \xi) = A(\zeta) \mathbb{D}_{\S^2}^k + B(\zeta,\xi),
$$
with
$$
  A(\zeta) = \sqrt{1 + \zeta \cos^2\theta}, \qquad B(\zeta,\xi) = i \Gb \frac{\zeta \sin(2\theta)}{4 \sqrt{1+\zeta \cos^2\theta}} + \Gc \frac{(\zeta k - \xi) \sin\theta}{\sqrt{1+\zeta\cos^2\theta}}.
$$
Clearly the operators $A(\zeta)$ and $B(\zeta, \xi)$ are bounded (matrix-valued) multiplication operators that are analytic in the variables $(\zeta, \xi) \in B(0, \frac{1}{13}) \times S$. Since the operator $A(\zeta)$ is also invertible, it follows that the operator domain of $A_k(\zeta,\xi)$ is always $\D$ and is thus independent of $(\zeta, \xi) \in B(0,\frac{1}{13}) \times S$. Since for all $u \in \D$, $A_k(\zeta, \xi) u$ is a vector-valued analytic function with respect to $(\zeta, \xi)$, we conclude that $A_k(\zeta, \xi)$ is an analytic family of type (A) in Kato's classification (\cite{Ka}, chap. VII, sect. 1). Moreover, if $(\zeta, \xi) \in [0, \frac{1}{13}] \times \R$, then $A_k(\zeta, \xi)$ is selfadjoint on $\L$ by \cite{BC1, BC2}. Hence, according to \cite{Ka}, Chap. VII, sect. 3, $A_k(\zeta, \xi)$ forms a \emph{self-adjoint holomorphic family of type (A)} in the variable $(\zeta, \xi) \in B(0,\frac{1}{13}) \times S$.

We also know that $A_k(0,0) = \mathbb{D}_{\S^2}^k$ has simple discrete spectrum given by (see \cite{BSW}, Appendix A)
\begin{equation} \label{Mu(0)}
  \mu_{k,l}(0,0) = sgn(l) \Big( |k| - \frac{1}{2} + |l| \Big), \quad l \in \Z^*.
\end{equation}
In particular, $A_k(0,0)$ has compact resolvent. From \cite{Ka}, Chap. V, sect. 2, Thm 2.4, it follows that $A_k(\zeta,\xi)$ has compact resolvent for all $(\zeta, \xi) \in B(0,\frac{1}{13}) \times S$. As a consequence, the spectrum of $A_k(\zeta,\xi)$ is discrete and since $A_k(\zeta, \xi)$ is in the limit point case at $\theta = 0$ and $\theta = \pi$ (see \cite{BC1, BC2}), it consists of simple eigenvalues for $(\zeta, \xi) \in [0,\frac{1}{13}] \times \R$. According to \cite{Ka}, chap. V, sect. 3, Thm 3.9, we conclude that, for a fixed $k \in 1/2 + \Z$, the eigenvalues
$$
  \mu_{kl}(\zeta, \xi), \quad k \in \frac{1}{2} + \Z, \quad l \in \Z^*,
$$
of $A_k(\zeta, \xi)$ are simple and depend holomorphically on $(\zeta, \xi)$ in a complex neighbourhood of $[0, \frac{1}{13}] \times \R$.

To estimate the growth of the eigenvalues $\mu_{kl}(\zeta,\xi)$ for fixed $k$ and large $l$, we use \cite{Ka}, Chap. VII, sect. 3, Thm 3.6. Let us first estimate the growth of the eigenvalues $\mu_{kl}(\zeta,0)$ when $l \to \infty$. For each $u \in \D$ and all $\zeta \in [0, \frac{1}{13}]$, we have from (\ref{Uk-1})
$$
  \partial_\zeta A_k(\zeta, 0) u = \frac{\cos^2\theta}{2 \sqrt{1 + \zeta \cos^2\theta}} \,\DS^k + \frac{1 + \frac{\zeta}{2} \cos^2\theta}{(1 + \zeta \cos^2\theta)^{3/2}} \left( i \frac{\sin(2\theta)}{4} \Gb + k \sin \theta \Gc \right),
$$
that can be re-written as
\begin{eqnarray*}
  \partial_\zeta A_k(\zeta, 0) u & = & \frac{\cos^2\theta}{2 (1 + \zeta \cos^2\theta)} \,A_k(\zeta, 0) - \frac{\cos^2\theta}{2 (1 + \zeta \cos^2\theta)} \frac{\zeta}{\sqrt{1+\zeta \cos^2\theta}} \left( i \frac{\sin(2\theta)}{4} \Gb + k \sin\theta \Gc \right) \nonumber \\  & & \qquad \qquad + \frac{1 + \frac{\zeta}{2} \cos^2\theta}{(1 + \zeta \cos^2\theta)^{3/2}} \left( i \frac{\sin(2\theta)}{4} \Gb + k \sin \theta \Gc \right),
\end{eqnarray*}
Hence for each $u \in \D$ and for all $\zeta \in [0,\frac{1}{13}]$, we get
\begin{equation} \label{Uk-2}
  \left\| \partial_\zeta A_k(\zeta, 0) u \right\| \leq a' \|u\| + b' \| A_k(\zeta,0) u\|,
\end{equation}
with
$$
  a' = (1 + \frac{1}{26}) ( |k| + \frac{1}{4}), \quad b' = \frac{1}{2}.
$$
We conclude from (\ref{Uk-2}) and \cite{Ka}, Chap. VII, sect. 3, Thm 3.6, that
\begin{equation} \label{Uk-3}
  |\mu_{kl}(\zeta,0) - \mu_{kl}(0,0) | \leq  \frac{1}{b'} ( a' + b' |\mu_{kl}(0,0)|) ( e^{b' |\zeta|} - 1).
\end{equation}
More precisely, recalling that $\zeta \in [0, \frac{1}{13}]$ and using (\ref{Mu(0)}), we obtain for all $k \in \frac{1}{2} + \Z$ and for all $l \in \N^*$
$$
  |\mu_{kl}(\zeta,0) - (|k| - \frac{1}{2} + l) | \leq \left( e^{\frac{1}{26}} - 1 \right) (|k| - \frac{1}{2} + l) + 2 \left( e^{\frac{1}{26}} - 1 \right) (1 + \frac{1}{26}) ( |k| + \frac{1}{4}).
$$
Thus there exist constants $C_1, \ C_2$ independent of $k,l$ such that for all $k \in \frac{1}{2} + \Z, \ l \in \N^*$ and for all $\zeta \in [0, \frac{1}{13}]$
\begin{equation} \label{Uk-4}
  \left( 2 - e^{\frac{1}{26}} \right) \left(|k| - \frac{1}{2} + l\right) - C_1 |k| - C_2 \leq \mu_{kl}(\zeta,0) \leq e^{\frac{1}{26}} \left(|k| - \frac{1}{2} + l\right) + C_1 |k| + C_2.
\end{equation}

To get an estimate of the growth of $\mu_{kl}(\zeta,\xi)$ for all $(\zeta,\xi) \in [0, \frac{1}{13}] \times \R$, for fixed $k \in \frac{1}{2} + \Z$ and large $l \in \N^*$, we use the same strategy. We first note that for all $u \in \D$, we have according to (\ref{Uk-1})
$$
  \partial_\xi A_k(\zeta,\xi) u = - \Gc \frac{\sin\theta}{\sqrt{1+\zeta\cos^2\theta}} u.
$$
Hence, for all $(\zeta,\xi) \in [0, \frac{1}{13}] \times \R$, we have
$$
  \| \partial_\xi A_k(\zeta,\xi) u \| \leq \| u \|,
$$
and we conclude from \cite{Ka}, Chap. VII, sect. 3, Thm 3.6, that
\begin{equation} \label{Uk-5}
  | \mu_{kl}(\zeta,\xi) - \mu_{kl}(\zeta,0) | \leq |\xi |.
\end{equation}
Putting together (\ref{Uk-4}) and (\ref{Uk-5}), we finally see that for all $(\zeta,\xi) \in [0, \frac{1}{13}] \times \R$, for all $k \in \frac{1}{2} + \Z$ and for all $l \in \N^*$, we have
\begin{equation} \label{Growth-mu}
  \left( 2 - e^{\frac{1}{26}} \right) (|k| - \frac{1}{2} + l) - C_1 |k| - C_2 - |\xi| \leq \mu_{kl}(\zeta,\xi) \leq e^{\frac{1}{26}} (|k| - \frac{1}{2} + l) + C_1 |k| + C_2 + |\xi|.
\end{equation}
Recalling that $\xi = a \lambda$, this finishes the proof of the Proposition.
\end{proof}


\Section{Limiting Absorption Principles and scattering theory for $H_0$ and $H$} \label{LAP-Mourre}

In this Appendix, we prove the main result that permitted us to establish a complete time-dependent scattering theory for the Dirac Hamiltonian $H$, namely Theorem \ref{WO-H-H0}. For definiteness, we recall it here.

\begin{theorem}
  The Hamiltonians $H_0$ and $H$ have purely absolutely continuous spectra, precisely
  $$
    \sigma(H_0)= \sigma_{ac}(H_0) = \R, \quad \sigma(H)= \sigma_{ac}(H) = \R,
  $$
  and the following wave operators
  $$
    W^\pm(H,H_0,I_2) := s-\lim_{t \to \pm \infty} e^{itH} e^{-itH_0},
  $$
  exist as operators from $\H$ to $\G$ and are asymptotically complete, \textit{i.e.} they are isometries from $\H$ to $\G$ and their inverse wave operators given by
  $$
    (W^\pm(H,H_0,I_2))^* = W^\pm(H_0,H,J) := s-\lim_{t \to \pm \infty} e^{itH_0} J e^{-itH},
  $$
  exist as operators from $\G$ to $\H$. (Note that the identity operator $I_2: \ \H \longrightarrow \G$ has been used as identification operator between $\H$ and $\G$ in the definition of the direct wave operators, whereas the dual operator $(I_2)^* = J: \ \G \longrightarrow \H$ appears in the definition of the inverse wave operators).
\end{theorem}

The results of this Theorem are quite similar to the results obtained in \cite{HaN, Da2} in the setting of Kerr and Kerr-Newman black holes ($\Lambda = 0$ in our model). In particular, the absence of pure point spectrum for $H$ is a consequence of the separation of variables and the integrability  of the potentials $a(x)$ and $c(x,k)$ at both horizons. This has been proved rigorously in this setting in \cite{BC2}. Also the absence of singular continuous spectrum for $H$ and the construction of wave operators for massless Dirac fields near the event horizon of a Kerr or Kerr-Newman black holes given in \cite{HaN, Da2} can be used almost without changes in the case of Kerr-Newman-de-Sitter black holes at both the event and the cosmological horizons. Therefore, in what follows, we shall only sketch the essential steps of the proof and refer to \cite{HaN, Da2} for some of the technical details.

Our strategy is first to prove limiting absorption principles (LAP) for the Hamiltonians $H_0$ and $H$ by means of a Mourre theory. Second we prove the existence and asymptotic completeness of the wave operators corresponding to the pair of Hamiltonians $(H, H_0)$ using the theory of $H$-smooth operators due to Kato and exposed for instance in \cite{RS}.

\subsection{Abstract Mourre theory}

The principle of Mourre theory for a selfadjoint operator $H$ on a Hilbert space $\H$ is to find a selfadjoint operator $A$ on $\H$ so that the pair $(H,A)$ satisfies the following assumptions (see \cite{M}). Let $I \subset \R$ an open interval.
\begin{description}
  \item [(M1)] $e^{-itA} D(H) \subset D(H)$.
  \item [(M2)] $i[H,A]$ defined as a quadratic form on $D(H) \cap
  D(A)$ extends to an element of $\B(D(H),\H)$.
  \item [(M3)] $[A,[A,H]]$ well defined as a quadratic form on
  $D(H) \cap D(A)$ by \textbf{(ii)}, extends to an element of
  $\B(D(H),D(H)^{*})$.
  \item [(M4)] There exists a strictly positive constant $\e$  and a compact operator $K$ such that
  \begin{equation}\label{MourreEstimate}
    \mathbf{1}_I(H) i[H,A] \mathbf{1}_I(H) \geq \e \mathbf{1}_I(H) + \mathbf{1}_I(H) K \mathbf{1}_I(H).
  \end{equation}
\end{description}

The fundamental assumption here is the Mourre estimate (M4). Its meaning is that we must find an observable $A$ which essentially increases along the evolution $e^{-itH}$. The other conditions are more technical and turn out to be difficult to check directly in the case where $A$ and $H$ are unbounded selfadjoint operators having no explicitely known domains. We give below some other criteria to verify them. If the pair $(H,A)$ satisfy these assumptions then we say that $A$ is a \emph{conjugate operator} for $H$ on $I$. The existence of a conjugate operator provides important informations on the spectrum of $H$. Precisely, we have (see for example \cite{ABG1, M})

\begin{theorem} \label{MO}
  Let $H,A$ two selfadjoint operators on $\H$. Assume that $A$ is a conjugate operator for $H$ on the interval $I$. Then $H$ has no singular continuous spectrum in $I$. Furthermore, the number of eigenvalues of $H$ in $I$ is finite (counting multiplicity).
\end{theorem}

Furthermore, a Limiting Absorption Principle holds \cite{ABG1, RS, Y}. Precisely, assume that $A$ is a conjugate operator for $H$ on an interval $I$. By shrinking the length of the interval $I$ and using the compactness of $K$ in the Mourre estimate (\ref{MourreEstimate}), there exists a positive constant $\delta$ such that for small enough $I$
\begin{equation}\label{MO-Est}
  \mathbf{1}_I(H) i[H,A] \mathbf{1}_I(H) \geq \delta \mathbf{1}_I(H).
\end{equation}

\begin{theorem}[Limiting Absorption Principle] \label{LAP}
  Let $H,A$ two selfadjoint operators on $\H$. Assume that $A$ is a conjugate operator for $H$ on the interval $I$ and that the Mourre estimate (\ref{MO-Est}) holds. Then for all closed interval $J \subset I$ and for $s > \frac{1}{2}$, there exists a constant $C$ such that
  \begin{equation} \label{LAP-MO}
    \sup_{\lambda \in J} \| \l A \r^{-s} (H - \lambda)^{-1} \l A \r^{-s} \| \leq C < \infty.
  \end{equation}
\end{theorem}

Before we find a conjugate operators for the Hamiltonians $H_0$ and $H$, let us give some precisions concerning the conditions (M1), (M2) and (M3) of Mourre's theory. One of the difficulties in Mourre theory consists in working with commutators $i[H,A]$ (see (M2)) between unbounded selfadjoint operators. We have to be careful to define correctly such quantities since $D(H)$ and $D(A)$ can be unknown or have an intersection which is not even dense in $\H$. Similarly, the assumption (M1) is not easy to prove since the action of $e^{isA}$ may also be unknown. Therefore, it is useful to have different criterion.

Let us first give a definition. For a selfadjoint operator $A$, we say that another selfadjoint operator $H$ belongs to $C^{k}(A), \ k \in \N,$ if and only if $\exists z \in \C \setminus \sigma(H), s \longrightarrow e^{isA}(H-z)^{-1}e^{-isA}$ belongs to $C^{k}(\R_{s};\mathcal{B}(\H))$ for the strong topology of $\B(\H)$. It has been shown in \cite{ABG1} that we can replace the assumptions (M1) and (M2) by the unique assumption $H \in C^1(A)$ without changing the conclusions of Theorem \ref{MO}. Roughly speaking, this condition allows that the following equality
\begin{displaymath}
  [A,(z-H)^{-1}] = (z-H)^{-1} [A,H] (z-H)^{-1},
\end{displaymath}
makes sense on $\H$. Also from \cite{ABG1}, $H \in C^1(A)$ is equivalent to
\begin{eqnarray*}
  \mathbf{(ABG1)} &  \exists z \in \C \setminus \sigma(H), \ (H-z)^{-1}D(A) \subset D(A), \ (H-\overline{z})^{-1}D(A) \subset D(A),
  \\
  \mathbf{(ABG2)} & |(Hu,Au)-(Au,Hu)| \leq C(\|Hu\|^{2}+\|u\|^{2}),
  \quad \forall u \in D(H) \cap D(A),
\end{eqnarray*}
which is close to both conditions (M1) and (M2) and are easier to check.

At last, we mention that the condition $H \in C^1(A)$ together with the condition (M2) imply the condition (M1) thanks to a result due to Gérard and Georgescu \cite{GG}
\begin{lemma}[Georgescu, Gérard] \label{GeorgescuGerard}
  Let $H$ and $A$ two self-adjoint operators such that $H \in  C^{1}(A)$ and $i[H,A] \in \B(D(H),\H)$ then $e^{isA} D(H) \subset D(H)$ for all $s \in \R$.
\end{lemma}
%
%


\subsection{Conjugate operators for $H_0$ and $H$} \label{ConjugateOperator}

How can we choose a conjugate operator? For Schr\"odinger or Dirac operators in flat spacetime, the usual generator of dilations is a good choice although not at all the only one (see \cite{ABG1}). In the particular case of Dirac operators for instance, there exist other choices (see \cite{BMP}, \cite{Da1}, \cite{Da2}, \cite{GM}) which could be used. For the Dirac Hamiltonian $H_0$ appearing in our KN-dS model, precisely
$$
  H_0 = \Ga D_x + a(x) H_{\S^2} + c(x,D_\varphi),
$$
the situation is in some sense similar to that of Dirac operators in flat spacetime thanks to the possibility of decomposing the problem into the generalized spherical harmonics $Y_{kl} := Y_{kl}(0)$ (see Theorem \ref{Spectrum-A}). We are led to study a countable family of one-dimensional Dirac operators
$$
  H_0^{kl} = \Ga D_x + \mu_{kl} a(x) \Gb + c(x,k),
$$
parametrized by $\mu_{kl} = \mu_{kl}(0) \in \R^+$ with $k \in \frac{1}{2} + \Z$ and $l \in \N^*$. Here the operators $a(x) H_{\S^2}$ and $c(x,D_\varphi)$ can be thus treated as mere potentials that are exponentially decreasing on $\R$. The existence of locally conjugate operators for each one-dimensional Dirac operators $H_0^{kl}$ is easy to establish. For instance, the conjugate operator
$$
  A = \Ga x,
$$
already used in a similar situation in \cite{Da1} is particularly adapted to this type of Dirac operators. On the other hand, although we can consider $H$ as a short-range perturbation of $H_0$ of order $1$, recall that
$$
  H = J^{-1} H_0, \quad J^{-1} = \frac{1}{1 - (a(x) b(\theta))^2} \left( I_2 - a(x) b(\theta) \Gc \right),
$$
and by (\ref{Pot-J})
$$
  \sup_\theta \| J^{-1} - I_2 \| = O(e^{-c|x|}), \quad x \to \pm \infty,
$$
this perturbation breaks the "symmetries" of $H_0$ and thus, we cannot use a decomposition into the generalized spherical harmonics $Y_{kl}$. Instead the full operator $H_{\S^2}$ must be conserved in the course of the calculations, or equivalently, we must be able to control the one-dimensional Mourre estimates relatively to the indices $k,l$.

To be more specific, recall that the evolution described by $H$ can be understood as an evolution on a Riemanniann manifold (given here by $\Sigma = \R_x \times \S^2$) having two ends: the cosmological and event horizons $\{x = \pm \infty\}$. At both horizons, the metric is exponentially large or, in other words, the geometry of the two ends is asymptotically hyperbolic. The choice of a conjugate operator turns out to be more complicated in this kind of situation but is by now well studied. For instance, analogous situations have been treated before, first by Froese and Hislop \cite{FH} in the case of a second-order-elliptic Hamiltonian, then by De Bièvre, Hislop and Sigal \cite{DHS} for the wave equation, more recently, by H\"afner and
Nicolas \cite{HaN} for a massless Dirac equation in a Kerr background and the first author \cite{Da2} for a massive Dirac equation in Kerr-Newman black hole. For the wave equation, we also mention more recent results due to Bouclet \cite{Bou} and Isozaki, Kurylev \cite{IK} who also study the corresponding inverse scattering problem. We follow here the presentation given in \cite{HaN, Da2} in the case of Dirac operators.

Let us study a toy model of the situation at the black hole event horizon and consider the operator
$$
  H_0 = \Ga D_x + e^{\kappa_- x} H_{\S^2} + \Omega_-(D_\varphi),
$$
acting on $L^2(\R_- \times \S^2; \C^2)$. In this toy model, we are only interested indeed in what happens in a neighbourhood of $x = -\infty$. Hence we simply replaced the potentials $a(x)$ and $c(x,D_\varphi)$ by their asymptotics at $x = - \infty$ in the expression og $H_0$.
Let us try to use the operator $A = \Ga x$ as a conjugate operator for $H_0$. This very simple operator is indeed well adapted to the case of massless Dirac operator. Since $H_{\S^2}$ anti-commutes with $\Ga$, we get
\begin{displaymath}
  i[H_0,A] = 1 + 2 x e^{\kappa_- x} H_{\S^2} \Ga.
\end{displaymath}
The first term is a positive constant (what we want) but the second term has no sign and is not controlled by $H_0$. The origin of the problem comes from the fact that the term $x e^{\kappa_- x}$ does not decay faster than $e^{\kappa_- x}$ when $x \to -\infty$. Note that the same problem happens if we use the generator of dilations instead of $A$. Therefore, the Mourre estimate has no chance to hold if we proceed in this manner. Instead, we introduce the unitary transformation
\begin{displaymath}
  U = e^{- \frac{i}{\kappa_-} \ln |H_{\S^2}| D_x},
\end{displaymath}
and we try to find a conjugate operator for
\begin{displaymath}
  \hat{H_0} = U^* H_0 U = \Ga D_x + e^{\kappa_- x} \frac{H_{\S^2}}{|H_{\S^2}|}.
\end{displaymath}
This is equivalent to the initial problem since $U$ is unitary. Now, if we compute the commutator between $\hat{H}_0$ and $A$, we obtain
\begin{displaymath}
  i[\hat{H}_0, A] = 1 + 2 x e^{\kappa_- x} \frac{H_{\S^2}}{|H_{\S^2}|} \Ga.
\end{displaymath}
Since $\frac{H_{\S^2}}{|H_{\S^2}|}$ is bounded, the term $\chi(\hat{H}_0) x e^{\kappa_- x} \frac{H_{\S^2}}{|H_{\S^2}|} \Ga \chi(\hat{H}_0)$ is now compact for any $\chi \in C_0^\infty(\R)$ by standard compactness criterion. Therefore, we can use
$$
  U A U^* = \Ga (x + \kappa_-^{-1} \ln |H_{\S^2}|),
$$
as conjugate operator for the toy Hamiltonian $H_0$ on  $L^2(\R^- \times \S^2; \C^2)$. The conjugate operator for the true Hamiltonians $H_0$ and $H$ we propose below, closely follows this procedure. Roughly speaking, we glue together two locally conjugate operators of the above type, each one corresponding to a different end: the event and cosmological horizons. The technical difficulties we shall avoid here lie of course in this "gluing together" operation. \\

In order to separate the problems at the event horizon and at the cosmological horizon, we define two cut-off functions $j_\pm \in C^\infty(\R)$ satisfying
\begin{eqnarray*}
  j_-(x) = 1, \ \textrm{for} \ x \leq \half, \quad j_-(x) = 0, \ \textrm{for} \ x \geq 1, \\
  j_+(x) = 1, \ \textrm{for} \ x \geq -\frac{1}{2}, \quad j_+(x) = 0, \ \textrm{for} \ x \leq -1.
\end{eqnarray*}
Now, for $S \geq 1$, we set
$$
  R_-(x,H_{\S^2}) = (x + \kappa_-^{-1} \ln |H_{\S^2}|) \, j_-^2 \Big( \frac{x + \kappa_-^{-1} \ln |H_{\S^2}|}{S} \Big),
$$
and
$$
  R_+(x,H_{\S^2}) = (x + \kappa_+^{-1} \ln |H_{\S^2}|) \, j_+^2 \Big( \frac{x + \kappa_+^{-1} \ln |H_{\S^2}|}{S} \Big).
$$
The addition of a parameter $S$ is part of the technical difficulties mentioned above. Following our previous discussion, we define the conjugate operators $A$ by

$$
  A = R_-(x,H_{\S^2}) \Ga + R_+(x,H_{\S^2}) \Ga.
$$

Before showing the Mourre estimate for the couple of operators $(H_0, A)$, let us give some of the important properties of the conjugate operator $A$ proved in \cite{Da2}. Note first that on each generalized spherical harmonics $Y_{kl}$, $A$ reduces to the operator of multiplication by
$$
  A = R_-(x,\mu_{kl}) \Ga + R_+(x,\mu_{kl}) \Ga,
$$
on $L^2(\R; \C^2)$ where $\mu_{kl}$ denotes the positive eigenvalues of the angular operator $A_{\S^2}(0)$. Thanks to this, we can prove easily

\begin{lemma} \label{DA}
  For all $S \geq 1$ and $i =1,2$,
  \begin{eqnarray*}
    |R_\pm(x,\mu_{kl})| \leq C \l x \r, \ \textrm{uniformly in k,l}, \\
    |R_\pm^{(i)}(x,\mu_{kl})| \leq C, \ \textrm{uniformly in k,l}.
  \end{eqnarray*}
  As a consequence, the domains of the operators $R_\pm(x,H_{\S^2})$ and thus $A$ contain $D(\l x \r)$.

  Moreover, if $j_1, j_2 \in C^\infty(\R)$ satisfy $j_1(x) = 1$ for $x \leq 1$ and $j_1(x) = 0$ for $x \geq \frac{3}{2}$, and $j_2(x) = 1$ for $x \geq -1$ and $j_2(x) = 0$ for $x \leq -\frac{3}{2}$ then
  \begin{eqnarray*}
    R^{(i)}_-(x, H_{\S^2}) & = & R^{(i)}_-(x,H_{\S^2}) j_1^2(\frac{x}{S}), \ i= 0,1,2, \\
    R^{(i)}_+(x, H_{\S^2}) & = & R^{(i)}_+(x,H_{\S^2}) j_2^2(\frac{x}{S}), \ i= 0,1,2, \\
  \end{eqnarray*}
\end{lemma}

\begin{proof} The proof is almost the same as the one given in \cite{Da2}, Lemma 4.4. It suffices to replace the Dirac operator $\mathbb{D}_{\S^2}$ on $\S^2$ in \cite{Da2} by our angular Dirac operator $H_{\S^2}$ and also the angular momenta $k = l + \frac{1}{2}$ of $\mathbb{D}_{\S^2}$ in \cite{Da2} by $\mu_{kl}$.

%
%

\end{proof}

\begin{lemma} \label{C1AH0}
  $H_0 \in C^1(A)$. Moreover, the commutator $i[H_0,A]$ belongs to $\B(D(H_0), \H)$. Also, the double commutator $[i[H_0,A],A]$ extends to a bounded operator in $\B(D(H_0), \H)$. Consequently, assumptions (M1),(M2) and (M3) of Mourre theory are satisfied.
\end{lemma}

\begin{proof}
  We refer to \cite{Da2}, lemmata 4.5, 4.6, 4.7 and 4.8 for these technical results.
\end{proof}

\begin{lemma} \label{MO-Local}
  Set $A_\pm = R_\pm(x,H_{\S^2}) \Ga$. Let $\lambda_0 \in \R$. Then there exists a function $\chi \in C_0^\infty(\R)$ with $supp\,\chi$ containing $\lambda_0$, a strictly positive constant $\e$ and compact operators generically denoted by $K$ such that for $S \geq 1$ large enough
  $$
    \chi(H_0) i[H_0,A_-] \chi(H_0) \geq \e \,\chi(H_0) j_1^2(\frac{x}{S}) \chi(H_0) + \,\chi(H_0) K \chi(H_0),
  $$
  $$
    \chi(H_0) i[H_0,A_+] \chi(H_0) \geq \e \,\chi(H_0) j_2^2(\frac{x}{S}) \chi(H_0) + \,\chi(H_0) K \chi(H_0),
  $$
where the functions $j_1, \ j_2$ are given the ones given in Lemma \ref{DA}.
\end{lemma}

\begin{proof}
  We refer to \cite{Da2}, Lemma 4.11.
\end{proof}

We can put together the previous results in order to obtain a Mourre estimate for $(H_0,A)$. Precisely, we can prove

\begin{prop} \label{MourreEstimateH0}
  Let $\lambda_0 \in \R$. Then there exists a function $\chi \in C_0^\infty(\R)$ with $supp\,\chi$ containing $\lambda_0$, a strictly
  positive constant $\e$ and a compact operator $K$ such that
  \begin{equation} \label{MO-H0}
    \chi(H_0) i[H_0,A] \chi(H_0) \geq \e \,\chi^2(H_0) + \,\chi(H_0) K \chi(H_0),
  \end{equation}
for $S \geq 1$ large enough.
\end{prop}

\begin{proof}
  Using the same notations as in Lemma \ref{MO-Local} and the fact that $A = A_- + A_+$, we have
  $$
    \chi(H_0) i[H_0,A] \chi(H_0) \geq \e \,\chi(H_0) \left( j_1^2(\frac{x}{S}) + j_2^2 (\frac{x}{S}) \right) \chi(H_0) + \,\chi(H_0) K \chi(H_0).
  $$
  We set $j_0^2 = 1 - j_1^2 - j_2^2$. Then, from the defining properties of $j_1, \ j_2$, $j_0 \in C_0^\infty(\R)$. Hence,
  $$
    \chi(H_0) i[H_0,A] \chi(H_0) \geq \e \, \chi^2(H_0) + \ \chi(H_0) j_0^2(\frac{x}{S}) \chi(H_0) + \, \chi(H_0) K \chi(H_0).
  $$
  We conclude using the standard compactness criterion
  \begin{equation} \label{CC}
    f(x) g(H_0) \ \textrm{compact if} \ f \in C_\infty(\R), \ g \in C_\infty(\R).
  \end{equation}
\end{proof}

From the Mourre estimate for $(H_0,A)$, we obtain easily a Mourre estimate for $(H,A)$. Recall first that $H$ is selfadjoint on $\G = L^2(\R \times \S^2; \C^2)$ equipped with the scalar product $(.,J.)_\H$. According to Theorem \ref{Spectra-H-H0}, the norms $\|. \|_\H$ and $\|.\|_\G$ are equivalent and the domains $D(H_0)$ and $D(H)$ coincide. Hence, it is immediate that the pair $(H,A)$ satisfies the assumptions (M1)-(M3) of Mourre theory. Let us check that the Mourre estimate of Proposition \ref{MourreEstimateH0} still holds. Precisely, we prove

\begin{prop} \label{MourreEstimateH}
  Let $\lambda_0 \in \R$. Then there exists a function $\chi \in C_0^\infty(\R)$ with $supp\,\chi$ containing $\lambda_0$, a strictly
  positive constant $\e$ and a compact operator $K$ such that
  \begin{equation} \label{MO-H}
    \chi(H) i[H,A] \chi(H) \geq \e \,\chi^2(H) + \,\chi(H) K \chi(H),
  \end{equation}
for $S \geq 1$ large enough.
\end{prop}

\begin{proof}
  We use first the following result proved in \cite{Da2}, Corollary 3.2. For all $\chi \in C_0^\infty(\R)$, the operator $\chi(H) - \chi(H_0)$ is compact. This is indeed a consequence of the Helffer-Sjostrand formula and the fact that
\begin{equation} \label{bu1}
  H = H_0 + (J^{-1}- I_2) H_0 = H_0 + O(e^{-c|x|}) H_0,
\end{equation}
according to (\ref{Pot-J}). Now we get
$$
  \chi(H) i[H,A] \chi(H)  = \chi(H_0) i[H,A] \chi(H_0)  + \chi(H) K \chi(H),
$$
for a compact operator $K$. Using (\ref{bu1}), we then obtain
\begin{eqnarray} \label{bu2}
  \chi(H) i[H,A] \chi(H)  & = & \chi(H_0) i[H_0,A] \chi(H_0)  + \chi(H_0) i[(J^{-1} - I_2),A] H_0 \chi(H_0) \\
                          &   & \quad + \chi(H_0) (J^{-1} - I_2) i[H_0,A] \chi(H_0) + \chi(H) K \chi(H). \nonumber
\end{eqnarray}
The last three terms in (\ref{bu2}) are compact by the compacness criterion (\ref{CC}). Hence, we conclude from Proposition \ref{MourreEstimateH0} that
$$
    \chi(H) i[H,A] \chi(H) \geq \e \,\chi^2(H) + \,\chi(H) K \chi(H),
  $$
for $S \geq 1$ large enough.

\end{proof}

\begin{remark}
  We emphasize here that the Mourre estimates (\ref{MO-H0}) and (\ref{MO-H}) are obtained in small neighbourhoods of any given energy $\lambda \in \R$ with no exception. In particular, the choice of the conjugate operator $A$ avoids to make appear a threshold at the energy $0$. This wouldn't be the case if we had used a conjugate operator constructed around the generator of dilations as in \cite{HaN}.
\end{remark}

Combining Theorem \ref{MO} with the absence of pure point spectrum for $H_0$ and $H$, we thus have proved the first assertion in Theorem \ref{WO-H-H0}. Using Theorem \ref{LAP}, we also have proved the following LAPs

\begin{prop} \label{LAP-H-H0}
  For all $\lambda \in \R$ and for $s > \frac{1}{2}$, there exists a small enough interval $I$ containing $\lambda$ such that
  \begin{equation} \label{LAP-H}
    \sup_{\lambda \in I} \| \l A \r^{-s} (H_0 - \lambda)^{-1} \l A \r^{-s} \| < \infty, \quad \sup_{\lambda \in I} \| \l A \r^{-s} (H - \lambda)^{-1} \l A \r^{-s} \| < \infty.
  \end{equation}
\end{prop}

\begin{remark}
We finish this Section with a technical result used in Section \ref{Link}. Precisely we prove a LAP at energy $0$ for the operator $L_0 = H_0 - \lambda J$ when $\lambda \in \R$ is fixed.

\begin{prop} \label{LAP-L0}
  Let $\lambda \in \R$ fixed. Then, for $s > \frac{1}{2}$,
  $$
    \| \l A \r^{-s} (L_0 - 0)^{-1} \l A \r^{-s} \| < \infty.
  $$
\end{prop}

\begin{proof}
  Let us write $L_0 = H_0 - \lambda + \lambda (I_2 - J)$. Recalling that $I_2 - J = O(e^{-c|x|})$ by (\ref{Pot-J}), we see easily that the operator $A$ is still a conjugate operator for $L_0$. Using the standard compactness criterion (\ref{CC}), we deduce the result from Theorem \ref{LAP}.
\end{proof}

\end{remark}

\subsection{$H$-smooth theory and wave operators}

Let us recall the definition of locally $H$-smooth operator. We say that a bounded operator $K$ is $H$-smooth on an interval $I$ if there exists a constant $C$ such that
\begin{equation} \label{Hsmooth}
  \int_\R \left\| K e^{-itH} \mathbf{1}_I(H) u \right\|^2 dt \leq C \| u \|^2.
\end{equation}
It is well known that LAPs such as the ones given in Proposition \ref{LAP-H-H0} entail $H_0$ and $H$-smoothness for the operator $\l A \r^{-s}, \ s > \frac{1}{2}$. Precisely, we have (see \cite{RS, Y})

\begin{prop}[$H$-smoothness] \label{A-Hsmooth}
  Let $(H,A)$ be selfadjoint operators on $\H$ such that the estimate (\ref{LAP-H}) holds on an interval $I$. Then the operator $\l A \r^{-s}, \ s > \frac{1}{2}$ is $H_0$ and $H$-smooth on $I$.
\end{prop}

There is a deep connection between $H_0$ and $H$-smoothness and the existence and asymptotic completeness of wave operators for the pair $(H_0,H)$. We shall use the following result from \cite{Y}

\begin{theorem}[Wave Operators by the $H$-smooth method] \label{WO-Hsmooth}
  Let $H_0, H$ be selfadjoint operators on Hilbert spaces $\H$ and $\G$ respectively. Let $P: \ \H \longrightarrow \G$ be a bounded identification operator. Assume that $H - H_0 = K^* K_0$ where $K_0$ and $K$ are $H_0$ and $H$-smooth operators on an interval $J$. Then the wave operators
  $$
    W^\pm(H,H_0,J \mathbf{1}_J(H_0)) = s-\lim_{t \to \pm \infty} e^{itH} J e^{-itH_0} \mathbf{1}_J(H_0),
  $$
  and
  $$
    W^\pm(H_0,H,J^* \mathbf{1}_J(H)) = s-\lim_{t \to \pm \infty} e^{itH_0} J^* e^{-itH} \mathbf{1}_J(H),
  $$
  exist.
\end{theorem}

Let us use the Theorem \ref{WO-Hsmooth} and the Proposition \ref{A-Hsmooth} to prove the existence and asymptotic completeness of the wave operators in Theorem \ref{WO-H-H0}. We do this in three steps. \\

\noindent \textbf{Step 1}. We show that for $s > \frac{1}{2}$, the operator $\l x \r^{-s}$ is $H_0$ and $H$-smooth on any compact interval $I \subset \R$. We only give the proof for $H_0$ since the proof is identical for $H$. We have to show that
\begin{equation} \label{ba1}
  \int_\R \left\| \l x \r^{-s} e^{-itH_0} \mathbf{1}_I(H_0) u \right\|^2 dt \leq C \| u \|^2.
\end{equation}
Using the compactness of $I$, we can suppose that the length of $I$ is as small as we want in (\ref{ba1}). Let $\chi \in C_0^\infty(\R)$ such that $\chi = 1$ on $I$ and $supp\,\chi$ is small. It is thus enough to show
\begin{equation} \label{ba2}
  \int_\R \left\| \l x \r^{-s} e^{-itH_0} \chi(H_0) u \right\|^2 dt \leq C \| u \|^2.
\end{equation}
From Lemma \ref{DA}, we see easily that for all $s > 0$ the operator $ \l x \r^{-s} \l A \r^{s}$ is a bounded operator on $\H$. Let us fix an $s$ such that $s > \frac{1}{2}$ and use Proposition \ref{A-Hsmooth}. We get (the constants $C$ differ from line to line in what follows)
\begin{eqnarray*}
  \int_\R \left\| \l x \r^{-s} e^{-itH_0} \chi(H_0) u \right\|^2 dt & \leq & \int_\R \left\| \l x \r^{-s} \l A \r^s \l A \r^{-s} e^{-itH_0} \chi(H_0) u \right\|^2 dt, \\
  & \leq & C \int_\R \left\| \l A \r^{-s} e^{-itH_0} \chi(H_0) u \right\|^2 dt, \\
  & \leq & C \|u \|^2.
\end{eqnarray*}
This proves (\ref{ba2}) and thus our claim. \\

\noindent \textbf{Step 2}. By definition, we have $ H - H_0 = (J^{-1} - I_2) H_0$. By (\ref{Pot-J}), we see that $J^{-1} - I_2  = O(e^{-c|x|})$ where the positive constant $c$ is given by $c = \min (\kappa_-, - \kappa_+) > 0$. We thus define
$$
  K_0 = e^{\frac{c}{2}|x|} (J^{-1} - I_2) H_0, \quad K = K^* = e^{-\frac{c}{2}|x|},
$$
in such a way that
$$
  H - H_0 = K^* K_0.
$$
We now prove that the operators $K_0$ and $K$ are $H_0$ and $H$-smooth respectively on any compact intervals $I$. We only give the proof for $K_0$ since the proof for $K$ is identical. As above, it suffices to consider a cut-off function $\chi \in C_0^\infty(\R)$ such that $supp\,\chi$ is small enough and $\chi = 1$ on $I$. Let $s > \frac{1}{2}$. We have
\begin{equation} \label{ba3}
  \int_\R \left\| K_0 e^{-itH_0} \chi(H_0) u \right\|^2 dt = \int_\R \left\| e^{\frac{c}{2}|x|} (J^{-1} - I_2) \l x \r^{s} \l x \r^{-s} e^{-itH_0} H_0 \chi(H_0) u \right\|^2 dt.
\end{equation}
We denote $\tilde{chi}(H_0) = H_0 \chi(H_0)$. Then the function $\tilde{\chi} \in C_0^\infty(\R)$ and $supp\, \tilde{\chi}$ is still small. Using (\ref{Pot-J}), we also have
$$
  \| e^{\frac{c}{2}|x|} (J^{-1} - I_2) \l x \r^{s} \| \leq C.
$$
Hence, from (\ref{ba2}), we obtain
\begin{eqnarray*}
  \int_\R \left\| K_0 e^{-itH_0} \chi(H_0) u \right\|^2 dt & \leq & C \int_\R \left\| \l x \r^{-s} e^{-itH_0} \tilde{\chi}(H_0) u \right\|^2 dt, \\
  & \leq & C \| u \|^2,
\end{eqnarray*}
which proves the assertion. \\

\noindent \textbf{Step 3. (Conclusion)}: The existence of the wave operators in Theorem \ref{WO-H-H0}
$$
  W^\pm(H,H_0,I_2) := s-\lim_{t \to \pm \infty} e^{itH} e^{-itH_0},
$$
$$
  (W^\pm(H,H_0,I_2))^* = W^\pm(H_0,H,J) := s-\lim_{t \to \pm \infty} e^{itH_0} J e^{-itH},
$$
follows then from Steps 1 and 2, Theorem \ref{WO-Hsmooth} (which gives the local existence on an interval of energy $I$) and finally, an easy density argument (for the global existence).

\end{document}